\documentclass[preprint,pteplogo]{ptephy_v2}

\preprintnumber{XXXX-XXXX} 

\usepackage{amsmath,amssymb,braket,fancybox}
\usepackage{amsfonts}
\usepackage{mathrsfs}
\usepackage{graphicx}
\usepackage{comment}
\usepackage{ascmac}
\usepackage{amsthm}
\usepackage{bm}
\usepackage[dvipsnames]{xcolor}
\usepackage{footmisc}
\usepackage{hyperref}
\usepackage{enumerate}

\numberwithin{equation}{section}

\theoremstyle{definition}
\newtheorem{theorem}{Theorem}

\newtheorem*{theorem*}{Theorem}

\newtheorem*{definition*}{Definition}

\newtheorem{lemma}[theorem]{Lemma}
\newtheorem*{lemma*}{Lemma}
\newtheorem{proposition}[theorem]{Proposition}
\newtheorem*{proposition*}{Proposition}

\newcommand{\av}[1]{\overline{#1}}
\newcommand{\mc}[1]{\mathcal{#1}}
\newcommand{\ml}[1]{\mathcal{#1}}

\newcommand{\mr}[1]{\mathrm{#1}}
\newcommand{\mbb}[1]{\mathbb{#1}}

\newcommand{\ketbra}[1]{\ket{#1}\bra{#1}}

\newcommand{\lrs}[1]{\left( #1 \right)}
\newcommand{\lrm}[1]{{\left\{ #1 \right\}}}
\newcommand{\lrl}[1]{\left[ #1 \right]}

\newcommand{\braketL}[1]{\left\langle #1 \right\rangle}

\newcommand{\aln}[1]{
\begin{align}
#1
\end{align}
}

\newcommand{\pmat}[1]{
\begin{pmatrix}
#1
\end{pmatrix}
}

\newcommand{\ra}{\rightarrow}

\newcommand{\Tr}{\mr{Tr}}

\begin{document}

\title{An introduction to monitored quantum systems and quantum trajectories: spectrum, typicality, and phases}


\author[1,2,*]{Ryusuke Hamazaki}
\author[3,1]{Ken Mochizuki}
\author[3,1]{Hisanori Oshima}
\author[3]{Yohei~Fuji}

\affil[1]{Nonequilibrium Quantum Statistical Mechanics RIKEN Hakubi Research Team, RIKEN Pioneering Research Institute (PRI), RIKEN, Wako 351-0198, Saitama, Japan}
\affil[2]{RIKEN Center for Interdisciplinary Theoretical and Mathematical Sciences(iTHEMS), RIKEN, Wako 351-0198, Saitama, Japan}
\affil[3]{Department of Applied Physics, University of Tokyo, Tokyo 113-8656, Japan}
\affil[*]{email: ryusuke.hamazaki@riken.jp}


\begin{abstract}%
Thanks to recent experimental advances in simulating and detecting quantum dynamics with high precision and controllability, our understanding of the physics of monitored quantum systems has considerably deepened over the past decades.
In this article, we provide an introductory theoretical review on the basic formalisms governing open quantum dynamics under measurement, along with recent developments in their spectral and typical aspects.
After reviewing quantum measurement theory, we
introduce the concept of quantum trajectories, which are the conditional dynamics of monitored states shaped by a set of measurement outcomes.
We then discuss the spectral properties of the dynamical map describing the evolution averaged over measurement outcomes.
As has recently been recognized, these spectral features are intimately connected to whether quantum trajectories exhibit typical behaviors, such as ergodicity and purification.
Moreover, we introduce Lyapunov exponents of typical quantum trajectories and discuss how these quantities serve as indicators of measurement-induced phase transitions in monitored quantum many-body systems.

\end{abstract}

\subjectindex{A13,A50,A58,A63,A64}

\maketitle

\tableofcontents

\section{Introduction}
Quantum measurement theory~\cite{von2018mathematical, peres2002quantum, braginsky1995quantum, nielsen2010quantum} has been at the heart of the foundation of quantum mechanics since its inception in the early 20th century.
In the last few decades, the study of monitored quantum systems has seen remarkable progress and has become an essential topic even for understanding the dynamics of open quantum systems~\cite{breuer2002theory, rivas2012open, wiseman2009quantum}.
These developments have been driven by groundbreaking experimental advances in precisely manipulating and detecting quantum systems, from the seminal experiments~\cite{nagourney1986shelved, sauter1986observation, bergquist1986observation} to the advent of quantum simulators and computers~\cite{barreiro2011open, georgescu2014quantum, altman2021quantum, harrington2022engineered}.
A striking example is the observation that additional measurements on many-body unitary dynamics within a quantum processor can give rise to novel non-equilibrium phases, characterized by their entanglement structure~\cite{koh2023measurement, Hoke23}. 
As seen from this example, monitored quantum dynamics is now recognized as linking diverse research areas, ranging from quantum information physics and condensed-matter physics to thermodynamics.

Conceptually, monitored quantum systems have an interesting structure that is not present in open quantum systems coupled to uncontrollable environment.
While both systems evolve under non-unitary dynamics, monitored systems retain access to the measurement outcomes, which can be recorded in a classical register.
The time evolution is then conditioned by the measurement outcomes unless we discard them, forming what is known as a quantum trajectory~\cite{plenio1998quantum, daley2014quantum}.
Importantly, quantum trajectories can remain in pure states throughout the evolution, in contrast to open quantum systems driven by the environment where the state eventually becomes mixed.
This feature leads to  phenomena unique to monitored systems, including the measurement-induced entanglement phase transitions~\cite{Potter2022, Lunt2022Quantum, fisher2023random}.

While (continuous-time) quantum trajectories were first introduced in the 1990s~\cite{ueda1990nonequilibrium, dalibard1992wave, dum1992monte, carmichael2009open, plenio1998quantum}, uncovering their mathematical properties remains an intriguing topic even today~\cite{attal2015central, benoist2019invariant, carollo2019unraveling, bernard2021can, benoist2021invariant, benoist2023limit, tindall2023generality, girotti2023concentration, benoist2024quantum}.
One notable feature is that most quantum trajectories display universal behaviors for certain quantities, which we call ``typical" behaviors of quantum trajectories in this article\footnote{
We note that, while the term ``typical" is often used in physics (e.g., statistical mechanics), it may be a rather informal word from a mathematical point of view. 
For example, if we consider the convergence of sequences of random variables in mathematics, there are several distinct notions, such as the convergence in probability or almost-sure convergence (see footnote~\ref{foot:TypicalVSAlmostAll} for their precise statements).
In contrast, such exact meanings are not usually considered when one says ``typical" in physics.
Still, we here adopt this terminology to introduce the intuitive notion in a physicist-friendly manner.
In this review, ``convergence for typical (almost all) quantum trajectories" basically means the almost-sure convergence; when we want to stress this fact, especially in Chapters~\ref{sec:linear-quantity_purification} and~\ref{sec:nonlear-quantity_Lyapunov-spectrum}, we will explicitly use the term ``almost surely." 
}.
Note that justifying such typical properties is by no means a simple task, partly because measurement probabilities determined by the Born rule depend nontrivially on the quantum state.
Nevertheless, many conditions for the emergence of typical behaviors have been identified, such as the uniqueness of the stationary state of the ensemble-averaged dynamics.
Examples that manifest typical features include outcome statistics, ergodic properties, state purification, and relaxation timescales governed by the Lyapunov spectrum.
Interestingly, some of these quantities also serve as crucial indicators for characterizing novel measurement-induced phase transitions in many-body systems.

In this review article, we give an introduction to monitored quantum systems and quantum trajectories from a theoretical viewpoint, emphasizing their spectra, typical properties, and phases.
Our aim is twofold:
First, we explain the basic formalisms of measurement theory and quantum trajectories in a pedagogical manner\footnote{
Indeed, this article is originally based on a lecture given by the first author at ``Summer Lecture Camp of the Hatano Laboratory" in Aug. 2024, although a significant amount of material has been added to this article.
}.
Accordingly, the content in the first half of the article partially overlaps with the existing literature~\cite{daley2014quantum, wiseman2009quantum, landi2024current, wolf2012quantum, watrous2018theory}.
We also emphasize that this article is not intended to be a comprehensive survey of the field of open quantum systems, for which numerous excellent books and reviews are already available~\cite{breuer2002theory, gardiner2004quantum, rivas2012open, daley2014quantum, wiseman2009quantum, breuer2016colloquium, sieberer2016keldysh, de2017dynamics, ashida2020non, weimer2021simulation, milz2021quantum, chruscinski2022dynamical, landi2022nonequilibrium, harrington2022engineered, landi2024current, sieberer2025universality, mori2023floquet, fisher2023random, albarelli2024pedagogical, fazio2025many}.
Second, which we believe will make this review unique, is the highlighting of recent advances concerning the typical properties of quantum trajectories and their connections to the spectral properties and non-equilibrium phases manifesting in monitored systems. 
Despite recent progress on this topic in the community of mathematical physics, these developments are not widely shared among many physicists studying monitored quantum systems.
We aim to bridge this gap by presenting core ideas in an intuitive and physicist-friendly manner, occasionally at the expense of full mathematical rigor.

The rest of this review is organized as follows.
In Chapter~\ref{sec:quantum-measurement}, we introduce basic concepts of quantum measurement theory. 
After reviewing simple projective measurements, we discuss some fundamental concepts such as indirect measurements, positive operator valued measures (POVMs), CP-instruments, and completely-positive trace-preserving (CPTP) maps.
We also explain how these maps and instruments can be represented in various ways.
In Chapter~\ref{sec:trajectory_master-equation}, we overview the formalism of quantum trajectories.
By taking the continuous-time limit of a repeated measurement protocol, we derive a stochastic equation for a quantum trajectory, whose ensemble average leads to the quantum master equation.
We then discuss physical quantities characterizing quantum trajectories, especially nonlinear observables and the statistics of quantum jumps.
We also briefly explain some related concepts, such as numerical methods for quantum trajectories and the quantum diffusion equation.

In Chapter~\ref{sec:CPTPspectra}, we discuss the spectral properties of CPTP maps and quantum master equations.
We review some important conditions for the steady states of CPTP maps, such as irreducibility and primitivity, and provide a detailed discussion of the rigorous criteria for these properties.
We also explain their connection to the steady-state properties of quantum master equations.
While most of the content in this chapter is devoted to rigorous discussions of steady-state properties, e.g., uniqueness, we also include one section that overviews miscellaneous recent topics beyond the steady-state properties, such as the spectral gap and spectral statistics.

In Chapter~\ref{sec:linear-quantity_purification}, we discuss the typical properties of quantum trajectories, focusing on the ergodicity of linear observables and purification, on the basis of the seminal results by K\"ummerer and Maassen. 
We explain the notions of these properties and their relation to the steady-state properties of the averaged dynamics discussed in Chapter~\ref{sec:CPTPspectra}.
Instead of rigorous formulations, we try to provide physical intuition for these concepts by presenting simple examples.
In Chapter~\ref{sec:nonlear-quantity_Lyapunov-spectrum}, we review recent developments on the typical properties of quantum trajectories, i.e., the ergodicity of nonlinear quantities and Lyapunov exponents.
Key results are summarized in Table~\ref{tab:ergodicity}.
These concepts play a foundational role in understanding measurement-induced phase transitions in Chapter~\ref{sec:mipt}.

\begin{table}
\begin{center}
\begin{tabular}{l|ccc}
\hline \hline
 & Unique SS & Full-rank SS & Purification \\
\hline
Ergodicity of linear observables (Sec.~\ref{sec:ergodicity_linear-observable}) & $\checkmark$ && \\
Ergodicity of nonlinear observables (Sec.~\ref{sec:ergodicity_nonlinear-quantity}) & $\checkmark$ & & $\checkmark$ \\
Convergence of LEs (Sec.~\ref{sec:typical-convergence_Lyapunov-spectrum}) & $\checkmark$ & $\checkmark$ & \\
Nonzero Lyapunov gap (Sec.~\ref{sec:spectral_gap_purification}) & $\checkmark$ & $\checkmark$ & $\checkmark$ \\
\hline \hline
\end{tabular}
\caption{
\label{tab:ergodicity}
Sufficient conditions for a typical trajectory to satisfy the ergodicity of observables linear or nonlinear in density operators and for the convergence of Lyapunov exponents (LEs) and a nonzero Lyapunov gap.
Here, the ergodicity means the equivalence between the long-time average in a single trajectory and the long-time ensemble average over all trajectories.
These conditions involve the uniqueness and full rankness (positive definiteness) of a steady state (SS) of the corresponding CPTP map, which will be discussed in Chapter~\ref{sec:CPTPspectra}, and the purification property of quantum trajectories, which will be discussed in Sec.~\ref{sec:pur}.
}
\end{center}
\end{table}

In Chapter~\ref{sec:mipt}, we overview the measurement-induced phase transitions of quantum trajectories, which have attracted recent attention as a novel type of non-equilibrium phase transition in quantum many-body systems.
While entanglement and purification transitions are well-known examples of measurement-induced phase transitions, we also discuss the relevance of the Lyapunov spectrum of quantum trajectories, which has been uncovered only recently.
In Chapter~\ref{sec:conclusion}, we conclude this review article and state some future prospects.

\section{Basics of quantum measurements}\label{sec:quantum-measurement}
In this section, we discuss some of the basics of quantum measurement theory, starting from the review of projective measurements.
We especially introduce some primary concepts required to understand the following sections, such as indirect measurements, CP-instruments, CPTP maps, and various representations of channels and instruments; an interested reader may refer to, e.g., Refs.~\cite{wiseman2009quantum, wolf2012quantum} for further details.

\subsection{Projective measurement}
Let us first consider a projective measurement of an observable $\hat{A}$, which is a Hermitian operator in a set of linear operators acting on the Hilbert space $\ml{H}$, denoted by $\mbb{B}[\ml{H}]$.
We assume that $\ml{H}$ is finite-dimensional throughout this paper.
Prepare a density matrix $\hat{\rho}$, which is a normalized (i.e., $\mr{Tr}[\hat{\rho}]=1$) positive semidefinite operator in $\mbb{B}[\ml{H}]$.
Then, the projective measurement of $\hat{A}$ probabilistically transforms $\hat{\rho}$ into one of the states corresponding to the projectors $\hat{P}_\eta$ onto the eigenspaces of $\hat{A}$.
That is, assume that $\hat{A}$ is decomposed as
\aln{
\hat{A}=\sum_\eta a_\eta\hat{P}_\eta,
}
where $a_\eta$ are the eigenvalues of $\hat{A}$ and $\hat{P}_\eta\:(=\hat{P}_\eta^\dag)$ are the Hermitian orthogonal projectors satisfying 
\aln{
\hat{P}_\eta\hat{P}_{\eta'}&=\delta_{\eta\eta'}\hat{P}_\eta,\\
\sum_\eta\hat{P}_\eta&=\hat{\mbb{I}}.
}
Then, when we measure $\hat{\rho}$ in the basis of $\{\hat{P}_\eta\}_\eta$, the normalized post-measurement state becomes
\aln{
\hat{\rho}'_\eta=\frac{\hat{P}_\eta\hat{\rho}\hat{P}_\eta}{p_\eta}
}
with the Born probability
\aln{
p_\eta=\Tr[\hat{P}_\eta \hat{\rho} \hat{P}_\eta]=\Tr[\hat{\rho}\hat{P}_\eta].
}
Note that when $a_\eta$ is degenerate, the rank of $\hat{P}_\eta$ becomes equal to the degree of the degeneracy. 
In particular, if $\hat{A}$ is nondegenerate, i.e., $\hat{P}_\eta=\ket{\eta}\bra{\eta}$, $\hat{\rho}_\eta'$ and $p_\eta$ are simply given by $\hat{\rho}_\eta'=\ket{\eta}\bra{\eta}$ and $p_\eta=\braket{\eta|\hat{\rho}|\eta}$, respectively.

\subsection{Indirect measurement}\label{sec:indirect}
While the projective measurement scheme discussed above completely destroys the pre-measurement state, we can instead consider measurements for which the quantum back-action is moderate.
Here, we specifically consider indirect measurements, where we first attach a meter M to the system S, let them interact, and perform the projective measurement on the meter. 
See Fig.~\ref{fig1}(a) for the case with a pure state.

\begin{figure}[!h]
\centering\includegraphics[width=\linewidth]{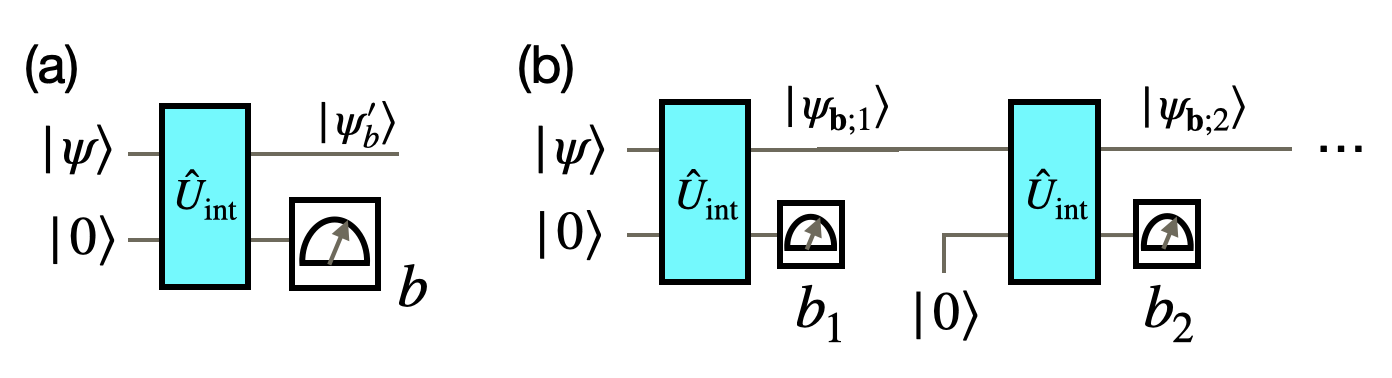}
\caption{
(a) Indirect measurement for the case with a pure state. 
We let the state $\hat{\rho}=\ket{\psi}\bra{\psi}$ of the system and the state $\hat{\sigma}_\mr{M}=\ket{0}\bra{0}$ of the meter interact via $\hat{U}_\mr{int}$ and measure the meter's state. 
If the outcome is $b$, the post-measurement state for the system becomes $\hat{\rho}_b'=\ket{\psi_b'}\bra{\psi_b'}$.
(b) Repeated measurement: we repeat the above process (a) and obtain the quantum trajectory $\hat{\rho}_{\bm{b};n}=\ket{\psi_{\bm{b};n}}\bra{\psi_{\bm{b};n}}$ depending on the measurement outcomes.
}
\label{fig1}
\end{figure}

To be more precise, we consider a composite Hilbert space $\ml{H}=\ml{H}_\mr{S}\otimes \ml{H}_\mr{M}$, where $\ml{H}_\mr{S} \:(\ml{H}_\mr{M})$ denotes the Hilbert space for the system (meter).
The initial state is assumed to be
\aln{\label{indirecteq}
\hat{\rho}\otimes \hat{\sigma}_\mr{M}\in \mbb{B}[\ml{H}],
}
where the system and meter are decoupled. 
The initial state $\hat{\rho}\otimes\hat{\sigma}_\mathrm{M}$ is evolved by a joint unitary operator $\hat{U}_\mr{int}$ as $\hat{U}_\mr{int}(\hat{\rho}\otimes \hat{\sigma}_\mr{M})\hat{U}_\mr{int}^\dag$, which entangles the system and meter.
For simplicity, we assume that the projection operator for the meter is a rank-one operator, 
\aln{
\hat{P}_{\mr{M}b}=\ket{q_b}\bra{q_b},
}
where $\ket{q_b}\in\ml{H}_\mr{M}$ and $b$ denotes the label of the measurement outcome.
Then, the post-measurement state reads
\aln{
\hat{\rho}'_b\otimes\ket{q_b}\bra{q_b}= \frac{1}{p_b}\braket{q_b|\hat{U}_\mr{int}(\hat{\rho}\otimes \hat{\sigma}_\mr{M})\hat{U}_\mr{int}^\dag|q_b}\otimes  \ket{q_b}\bra{q_b}
}
with probability
\aln{
p_b=\Tr_\mr{SM}[(\hat{\mbb{I}}_\mr{S}\otimes\hat{P}_{\mr{M}b})\hat{U}_\mr{int}(\hat{\rho}\otimes \hat{\sigma}_\mr{M})\hat{U}_\mr{int}^\dag]
=\Tr_\mr{S}[\braket{q_b|\hat{U}_\mr{int}(\hat{\rho}\otimes \hat{\sigma}_\mr{M})\hat{U}_\mr{int}^\dag|q_b}],
}
where $\mr{Tr}_\mr{SM}[\cdots]$ and $\mr{Tr}_\mr{S}[\cdots]$ denote the trace of the composite system (S and M) and the system S alone, respectively. 
By writing
\aln{
\hat{\sigma}_\mr{M}=\sum_a\sigma_a\ket{\sigma_a}\bra{\sigma_a},
}
and introducing the measurement operator
\aln{
\hat{M}_{ab}=\sqrt{\sigma_a}\braket{q_b|\hat{U}_\mr{int}|\sigma_a},
}
we have\footnote{
In the following, we remove the subscript S for the trace of the system for simplicity, when no ambiguity would arise.
} 
\aln{
\begin{split}p_b&=\Tr\lrl{\hat{\rho}\sum_a\hat{M}_{ab}^\dag\hat{M}_{ab}},\\
\hat{\rho}'_b&=\frac{1}{p_b}\sum_a\hat{M}_{ab}\hat{\rho}\hat{M}_{ab}^\dag.
\end{split}}
If we take the average over the outcomes $\{b\}$, we have
\aln{\label{eq:kraus}
\hat{\rho}'=\sum_bp_b\hat{\rho}'_b=\sum_{ab}\hat{M}_{ab}\hat{\rho}\hat{M}_{ab}^\dag.
}
Note that $\sum_{ab}\hat{M}_{ab}^\dag\hat{M}_{ab}=\hat{\mbb{I}}_\mr{S}$ holds, which follows from $\hat{U}_\mr{int}^\dag \hat{U}_\mr{int}=\hat{\mbb{I}}_\mr{SM}$ and $\sum_b\ket{q_b}\bra{q_b}=\hat{\mbb{I}}_\mr{M}$.
Then, the mapping of $\hat{\rho}$ in the form of Eq. \eqref{eq:kraus} is called the Kraus representation.
In that context, $\hat{M}_{ab}$ is also called the Kraus operator.

In a special case, we can consider a situation where $\hat{\sigma}_\mr{M}=\ket{\sigma}\bra{\sigma}$ is a pure state. Then, we can omit the label $a$ and find
\aln{\label{krauspure}
\begin{split}
p_b&=\Tr\lrl{\hat{\rho}\hat{M}_{b}^\dag\hat{M}_{b}},\\
\hat{\rho}'_b&=\frac{1}{p_b}\hat{M}_{b}\hat{\rho}\hat{M}_{b}^\dag
\end{split}}
with
\aln{\label{qbusi}
\hat{M}_{b}=\braket{q_b|\hat{U}_\mr{int}|\sigma}.
}

\subsubsection{Positive operator valued measure (POVM)}
Let us define an operator
\aln{
\hat{E}_b=\sum_a\hat{M}_{ab}^\dag\hat{M}_{ab},
}
which satisfies $p_b=\Tr[\hat{E}_b\hat{\rho}]$.
Then, the operators $\{\hat{E}_b\}$ satisfy the following two properties:
\begin{itemize}
\item 
The resolution of identity:
\aln{
\sum_b\hat{E}_b=\hat{\mbb{I}}.
}
\item 
Positive semidefiniteness:
\aln{
\hat{E}_b\succeq 0.
}
\end{itemize}
Here, $\hat{X}\succeq 0$ means $\braket{\phi|\hat{X}|\phi}\geq 0$ for all $\ket{\phi}\in\mathcal{H}$. 
The above positive semidefiniteness obviously follows from $\braket{\phi|\hat{E}_b|\phi}=\sum_a\|\hat{M}_{ab}\ket{\phi}\|^2\geq 0$.

Conversely, if the above two conditions, i.e., the resolution of identity and positive semidefiniteness, are met, the set $\{\hat{E}_b\}$ is called a positive operator valued measure (POVM).
Note that the POVM is more general than the set of projection operators $\{\hat{P}_\eta\}$, which is a special type of POVM.
Indeed, elements of the POVM do not need to be orthogonal to each other $\hat{E}_b\hat{E}_{b'}\neq \hat{E}_b\delta_{bb'}$ in general, in contrast to the set of projection operators.

\subsubsection{Examples}
Let us provide two examples of measurement operators and POVMs.
As a first example, we consider a measurement with an error $\epsilon\:(0\leq \epsilon <1)$.
We set the meter's initial state to $\hat{\sigma}_\mr{M}=\ket{0}\bra{0}$ and the measurement basis to $\ket{q_b}=\ket{0}$ or $\ket{1}$.
The joint unitary operator $\hat{U}_\mr{int}$ tries to copy the basis states $\ket{0},\ket{1}$ of the system onto the meter; the measurement becomes projective if the copy is perfect {($\epsilon=0$)}, while it becomes weaker if there is an error in the copying process.
To account for this process, we can set the joint unitary $\hat{U}_\mr{int}$ such that 
\aln{
\begin{split}\hat{U}_\mr{int}\ket{0}\otimes \ket{0}&=\ket{0}\otimes (\sqrt{1-\epsilon}\ket{0}+\sqrt{\epsilon}\ket{1}),\\
\hat{U}_\mr{int}\ket{1}\otimes \ket{0}&=\ket{1}\otimes (\sqrt{1-\epsilon}\ket{1}+\sqrt{\epsilon}\ket{0}).
\end{split}}
Using Eq.~\eqref{qbusi}, we find
\aln{
\begin{split}\hat{M}_0&=\sqrt{1-\epsilon}\ket{0}\bra{0}+\sqrt{\epsilon}\ketbra{1},\\
\hat{M}_1&=\sqrt{1-\epsilon}\ket{1}\bra{1}+\sqrt{\epsilon}\ketbra{0}.
\end{split}}
The corresponding POVM reads
\aln{
\begin{split}\hat{E}_0&=({1-\epsilon})\ket{0}\bra{0}+{\epsilon}\ketbra{1},\\
\hat{E}_1&=({1-\epsilon})\ket{1}\bra{1}+{\epsilon}\ketbra{0}
\end{split}}
with $\hat{E}_0+\hat{E}_1=\hat{\mbb{I}}_\mr{S}$.
Note that $\hat{E}_0\hat{E}_1\neq 0$ unless $\epsilon=0$ (projective measurement).

Another example is the swapping between the system and meter.
We again set the meter's initial state to $\hat{\sigma}_\mr{M}=\ket{0}\bra{0}$ and the measurement basis to $\ket{q_b}=\ket{0}$ or $\ket{1}$.
The joint unitary operator describing the interaction is taken as
\aln{
\begin{split}\hat{U}_\mr{int}\ket{0}\otimes \ket{0}&=\ket{0}\otimes \ket{0},\\
\hat{U}_\mr{int}\ket{1}\otimes \ket{0}&=\cos\theta \ket{1}\otimes\ket{0}+\sin\theta  \ket{0}\otimes\ket{1}.
\end{split}}
In this case, the measurement operators become
\aln{
\begin{split}\hat{M}_0&=\ket{0}\bra{0}+\cos\theta\ketbra{1},\\
\hat{M}_1&=\sin\theta\ket{0}\bra{1}.
\end{split}}
The POVM reads
\aln{
\begin{split}\hat{E}_0&=\ket{0}\bra{0}+\cos^2\theta\ketbra{1},\\
\hat{E}_1&=\sin^2\theta\ket{1}\bra{1}
\end{split}}
with $\hat{E}_0+\hat{E}_1=\hat{\mbb{I}}_\mr{S}$.
We find that $\hat{E}_0\hat{E}_1\neq 0$ unless $\theta=m\pi/2$ with $m$ being integer.
Physically, when $\sin\theta=\pm1$, this can be regarded as a toy model for the spontaneous emission of an atom.
Namely, an excited atom $\ket{1}$ (system) becomes a ground state $\ket{0}$ with emitting a photon to the vacuum $\ket{0}$ (meter).
Note that a more sophisticated treatment to describe the spontaneous emission is discussed in Sec. \ref{qtex} using continuous-time quantum trajectories.

\subsection{CP-instrument and CPTP map}
Let us next discuss how the measurement process and the change of a quantum state are characterized in a more abstract manner.
For this purpose, we first note that the post-measurement state $\hat{\rho}_b'$ is rewritten as
\aln{\label{change}
\hat{\rho}_b'=\frac{1}{p_b}\ml{E}_b[\hat{\rho}],
}
where $\ml{E}_b:\mbb{B}[\ml{H}_\mr{S}]\ra\mbb{B}[\ml{H}_\mr{S}]$ is given by
\aln{\label{cpik}
\ml{E}_b[\hat{\rho}]=\sum_a\hat{M}_{ab}\hat{\rho}\hat{M}_{ab}^\dag.
}
Here, $\{\ml{E}_b\}$ satisfies the following properties:
\begin{enumerate}[(i)]
\def\labelenumi{(\theenumi)}
\item 
The sum of $\mc{E}_b$,
\aln{
\ml{E}=\sum_b\ml{E}_b,
}
is trace preserving (TP): a linear map $\mc{F}: \mbb{B}[\mc{H}_\mr{S}] \to \mbb{B}[\mc{H}_\mr{S}]$ is called TP if
\aln{
\Tr[\ml{F}[\hat{X}]]=\Tr[\hat{X}]
}
for any $\hat{X} \in \mbb{B}[\mc{H}_\mr{S}]$.

\item 
Every $\mc{E}_b$ is completely positive (CP): a linear map $\mc{F}: \mbb{B}[\mc{H}_\mr{S}] \to \mbb{B}[\mc{H}_\mr{S}]$ is called CP if
\aln{
\label{eq:DefCP}
(\ml{F}\otimes \ml{I}_\mr{A})[\hat{R}]\succeq 0
}
for any auxiliary system A with Hilbert space $\mc{H}_\mr{A}$ and for any positive semidefinite operator $\hat{R} \in \mbb{B}[\mc{H}_\mr{S} \otimes \mc{H}_\mr{A}]$ (i.e., $\hat{R} \succeq 0$).
Here, $\mc{I}_\mr{A}: \mbb{B}[\mc{H}_\mr{A}] \to \mbb{B}[\mc{H}_\mr{A}]$ is the identity map on A.
\end{enumerate}

In general, the set of linear maps $\{\ml{E}_b\}$ satisfying (i) and (ii) above is called a CP-instrument and characterizes measurement processes.
For the Kraus representation in Eq.~\eqref{cpik}, (i) readily follows from the condition $\sum_{ab}\hat{M}_{ab}^\dag \hat{M}_{ab}=\hat{\mbb{I}}_\mr{S}$, and (ii) follows because we can write
\aln{\label{kracp}
(\ml{E}_b\otimes \ml{I}_\mr{A})[\hat{R}]=\sum_a(\hat{M}_{ab}\otimes\hat{\mbb{I}}_\mr{A})\hat{R}(\hat{M}_{ab}\otimes\hat{\mbb{I}}_\mr{A})^\dag=\sum_a \hat{X}_{ab}^\dag \hat{X}_{ab}\succeq 0,
}
where $\hat{X}_{ab}=\sqrt{\hat{R}}(\hat{M}_{ab}\otimes\mbb{I}_\mr{A})^\dag$ is well-defined since  $\hat{R}$ is positive semidefinite.

Let us explain the physical meanings of (i) and (ii). 
Since the change of a state averaged over the measurement outcomes is given by $\ml{E}[\hat{\rho}]$, the TP condition for $\ml{E}$ means that the net probability is kept normalized.
For CP, we first note that it is stronger than the condition for the positivity of a map, which requires that $\ml{E}_b[\hat{\rho}]\succeq 0$ for all $\hat{\rho}\succeq 0$.
The positivity ensures that a state after applying the map is positive semidefinite, given that the initial state is positive semidefinite.
However, CP requires more than that, and we allow attaching any auxiliary systems.

One concrete example of a map that is positive but not CP is the matrix transpose in a certain basis, $\ml{T}[\hat{\rho}]=\hat{\rho}^\mathsf{T}$\footnote{
For qubits, we expand $\hat{\rho}$ in the computational basis states as $\hat{\rho}=\sum_{ij} \rho_{ij}\ket{i}\bra{j}$ and define $\hat{\rho}^\mathsf{T}=\sum_{ij} \rho_{ji}\ket{i}\bra{j}$.
}.
Specifically, let us consider one qubit for S and another qubit for A.
If we apply the map $(\ml{T}\otimes \ml{I}_\mr{A})$ to the Bell state 
\aln{
\hat{R}=\frac{\ket{00}\bra{00}+\ket{00}\bra{11}+\ket{11}\bra{00}+\ket{11}\bra{11}}{2}\succeq 0,
}
we have
\aln{
(\ml{T}\otimes \ml{I}_\mr{A})[\hat{R}]=\frac{\ket{00}\bra{00}+\ket{10}\bra{01}+\ket{01}\bra{10}+\ket{11}\bra{11}}{2},
}
which has a negative eigenvalue and thus is not positive semidefinite.
Therefore, $\ml{T}$ is not CP, while it is clearly positive.
Note that if $\hat{R}$ is taken to be a separable state $\hat{R}=\sum_k q_k\hat{\rho}_{\mr{S},k}\otimes \hat{\rho}_{\mr{A},k}$ with some probability $\{q_k\}$, $(\ml{T}\otimes \ml{I}_\mr{A})[\hat{R}]=\sum_k q_k\hat{\rho}_{\mr{S},k}^\mathsf{T}\otimes \hat{\rho}_{\mr{A},k}$ becomes positive semidefinite.
Therefore, the negative eigenvalues of $(\ml{T}\otimes \ml{I}_\mr{A})[\hat{R}]$ mean that $\hat{R}$ is entangled. 
This is known as the positive partial transpose (PPT) criterion~\cite{peres1996separability, horodecki2001separability} to detect the entanglement of mixed states\footnote{
Note that this criterion is independent of the basis used to define the transposition, since a basis change does not alter the eigenvalues of the partially transposed matrix~\cite{wolf2012quantum}.
}.

Similar to the CP-instrument, we can consider a linear map $\ml{E}$ that is both TP and CP.
Such a map is called a CPTP map and characterizes general quantum channels.
By definition, a quantum channel $\ml{E}=\sum_b\ml{E}_b$ is a CPTP map if the set $\{\ml{E}_b\}$ is a CP-instrument. 
However, some CPTP maps may not be associated with a physicalmeasurement process and can describe more general dynamics in open quantum systems.

\subsection{Representations}
In the previous section, we have seen that the map in Eq.~\eqref{cpik}, which is given in the Kraus representation, satisfies the conditions for a CP-instrument. 
Here, we discuss that the converse is also true: any CP-instrument can be represented using the Kraus operators.
For this purpose, we introduce another representation, called the Choi-Jamiołkowski representation.
We also discuss other types of representations, such as the Stinespring representation and the ``natural" representation. 
See Refs.~\cite{wolf2012quantum, watrous2018theory} for further details.

\subsubsection{Choi-Jamiołkowski representation}
We first consider how CP of a linear map $\ml{F}$ can be characterized.
To this end, we introduce the Choi-Jamiołkowski representation $\hat{F}\in\mbb{B}[\ml{H}_\mr{S}\otimes \ml{H}_\mr{S}]$ of a map $\ml{F}: \mbb{B}[\ml{H}_\mr{S}]\ra \mbb{B}[\ml{H}_\mr{S}]$\footnote{
Two remarks are in order. 
First, in general, we can consider maps $\ml{F}: \mbb{B}[\ml{H}_1]\ra \mbb{B}[\ml{H}_2]$ where $\ml{H}_1$ and $\ml{H}_2$ are different. 
However, for simplicity, we only consider the case where they are the same.
Second, while Eq.~\eqref{choi} is often called the Choi-Jamiołkowski isomorphism, the original formulations by Choi~\cite{choi1975completely} and Jamiołkowski~\cite{jamiolkowski1972linear} are slightly different~\cite{jiang2013channel}.
} defined as
\aln{\label{choi}
\hat{F}=(\ml{F}\otimes \ml{I}_\mr{S})[\ket{\Phi}\bra{\Phi}]=\frac{1}{d}\sum_{ij}\ml{F}[\ket{i}\bra{j}]\otimes\ket{i}\bra{j},
}
where $\ket{\Phi}$ is the maximally entangled state between the original system and the copied system,
\aln{
\label{eq:MaximallyEntangledState}
\ket{\Phi}=\frac{1}{\sqrt{d}}\sum_i\ket{ii}\in \ml{H}_\mr{S}\otimes \ml{H}_\mr{S},
}
with $d=\mr{dim}[\ml{H}_\mr{S}]$.
Note that we can easily show that
\aln{
\Tr_\mr{S}[\hat{A}\ml{F}[\hat{\rho}]]=d\Tr_{\mr{SS}^c}[\hat{F}(\hat{A}\otimes\hat{\rho}^\mathsf{T})]
}
for every $\hat{A}\in \mbb{B}[\ml{H}_\mr{S}]$ and $\hat{\rho}\in \mbb{B}[\ml{H}_\mr{S}]$, where the subscript $\mr{S}$ ($\mr{S}^c$) means that the trace is taken over the original (copied) system.
This means that 
\aln{
\ml{F}[\hat{\rho}]=d\Tr_{\mr{S}^c}[\hat{F}(\hat{\mbb{I}}_\mr{S}\otimes\hat{\rho}^\mathsf{T})],
}
which indicates that $\hat{F}$ and $\ml{F}$ are in one-to-one correspondence.

Notably, CP becomes evident in the Choi-Jamiołkowski representation.
That is,
\aln{\label{positive_CP}
\hat{F}\succeq 0  \Longleftrightarrow \text{$\ml{F}$ is CP.}
}
The fact for $(\Longleftarrow)$ follows from the definition of CP. 
To see the other direction $(\Longrightarrow)$~\cite{wolf2012quantum}, we first notice that $(\ml{F}\otimes\ml{I}_\mr{A})[\hat{R}]\succeq 0$ is ensured if $(\ml{F}\otimes\ml{I}_\mr{A})[\ket{\phi_k}\bra{\phi_k}]\succeq 0$ for every $k$, where $\ket{\phi_k}\in\ml{H}_\mr{S}\otimes\ml{H}_\mr{A}$ are the eigenstates of $\hat{R}$, i.e.,  $\hat{R}=\sum_kr_k\ket{\phi_k}\bra{\phi_k}$ with $r_k\geq 0$.
Then, for each $\ket{\phi_k}$, we can find a map $\hat{V}_k:\ml{H}_\mr{S}\ra\ml{H}_\mr{A}$ such that\footnote{
To see this, let us consider the Schmidt decomposition of $\ket{\phi}$ between S and A, which is given by $\ket{\phi}=\sum_{l=1}^d\beta_l\ket{e_l}\otimes \ket{E_l}$ where $\beta_l\geq 0$ are the Schmidt coefficients. 
Here, we have omitted the subscript $k$ for simplicity. 
Introducing $\hat{v}_1=\sum_{l=1}^d\ket{e_l}\bra{l}$ and $\hat{v}_2=\sum_{l=1}^d\beta_l\ket{E_l}\bra{l}$, we find $\ket{\phi}=\sqrt{d}\hat{v}_1\otimes\hat{v}_2\ket{\Phi}=\sqrt{d}(\hat{\mbb{I}}_\mr{S}\otimes\hat{v}_2)(\hat{v}_1\otimes\hat{\mbb{I}}_\mr{S})\ket{\Phi}$. 
Since $\ket{\Phi}$ is the maximally entangled state, we have $(\hat{v}_1\otimes\hat{\mbb{I}}_\mr{S})\ket{\Phi}=(\hat{\mbb{I}}_\mr{S}\otimes \hat{v}_1^\mathsf{T})\ket{\Phi}$. 
Then, we can take $\hat{V}=\sqrt{d}\hat{v}_2\hat{v}_1^\mathsf{T}$.
} 
\aln{\label{phikisvkphi}
\ket{\phi_k}=(\hat{\mbb{I}}_\mr{S}\otimes\hat{V}_k)\ket{\Phi}.
}
Therefore, 
\aln{
(\ml{F}\otimes\ml{I}_\mr{A})[\ket{\phi_k}\bra{\phi_k}]
=(\hat{\mbb{I}}_\mr{S}\otimes\hat{V}_k)(\ml{F}\otimes \ml{I}_\mr{S})[\ket{\Phi}\bra{\Phi}]
(\hat{\mbb{I}}_\mr{S}\otimes\hat{V}_k)^\dag\succeq 0
}
if $\hat{F}=(\ml{F}\otimes \ml{I}_\mr{S})[\ket{\Phi}\bra{\Phi}]\succeq 0$.

\subsubsection{Kraus representation}
Using the above result, we can further show the equivalence between the Kraus representation and CP map~\cite{wolf2012quantum}:
\aln{
\ml{F} \text{ is represented as } \ml{F}[\hat{\rho}]=\sum_a\hat{M}_a\hat{\rho}\hat{M}_a^\dag \Longleftrightarrow \text{$\ml{F}$ is CP.}
}
The fact for $(\Longrightarrow)$ was already discussed in Eq.~\eqref{kracp}.
To see $(\Longleftarrow)$, we first note that the corresponding Choi-Jamiołkowski representation of the CP map, $\hat{F}$, satisfies $\hat{F}\succeq 0$, as described in Eq.~\eqref{positive_CP}.
By decomposing $\hat{F}$ as $\hat{F}=\sum_{a=1}^rf_a\ket{f_a}\bra{f_a}\:(f_a\geq 0)$ and using $\ket{f_a}=(\hat{m}_a\otimes\hat{\mbb{I}}_\mr{S})\ket{\Phi}$ for some operator $\hat{m}_a$ as in Eq.~\eqref{phikisvkphi}, we have
\aln{\label{F_diagonal}
\hat{F}=\sum_{a=1}^r (\hat{M}_a\otimes\hat{\mbb{I}}_\mr{S})\ket{\Phi}\bra{\Phi}(\hat{M}_a\otimes\hat{\mbb{I}}_\mr{S})^\dag,
}
where $\hat{M}_a=\sqrt{f_a}\hat{m}_a$.
Comparing Eq.~\eqref{F_diagonal} with the definition of $\hat{F}$ in Eq.~\eqref{choi} and recalling that $\hat{F}$ and $\ml{F}$ are in one-to-one correspondence, we can conclude that $\ml{F}$ admits a Kraus representation.
From this construction, we can see that $r\geq \mr{rank}[\hat{F}]$, where $\mr{rank}[\hat{F}]\leq d^2$.

Therefore, each element in a CP-instrument $\{\ml{E}_b\}_b$ can be given in the Kraus representation by Eq.~\eqref{cpik} because of CP of $\ml{E}_b$.
Furthermore, the TP condition for $\sum_b\ml{E}_b$ leads to $\sum_{ab}\hat{M}_{ab}^\dag\hat{M}_{ab}=\hat{\mbb{I}}_\mr{S}$.
Likewise, a CPTP map $\ml{E}$ is represented as 
\aln{\label{eq:KrausRep}
\ml{E}[\hat{\rho}]=\sum_k\hat{M}_k\hat{\rho}\hat{M}_k^\dag,
}
where
\aln{\label{eq:TPKraus}
\sum_{k}\hat{M}_{k}^\dag\hat{M}_{k}=\hat{\mbb{I}}_\mr{S}
}
is satisfied.

\subsubsection{Stinespring representation}
We have another representation of a map, which is relevant for the setup of open quantum systems where the bath is traced out after unitary time evolution.
Let $\ml{F}:\mbb{B}[\ml{H}_\mr{S}]\ra\mbb{B}[\ml{H}_\mr{S}]$ be a CP map.
Then, there exist a Hilbert space $\ml{H}_\mr{E}$ and an operator $\hat{V}:\ml{H}_\mr{S}\ra \ml{H}_\mr{S}\otimes\ml{H}_\mr{E}$ such that
\aln{\label{stine}
\ml{F}[\hat{\rho}]=\Tr_\mr{E}[\hat{V}\hat{\rho}\hat{V}^\dag],
}
which is called the Stinespring representation~\cite{stinespring1955positive}.
This representation can be constructed as follows. 
Since $\ml{F}$ is CP, it admits a Kraus representation $\ml{F}[\hat{\rho}]=\sum_{a=1}^r\hat{M}_a\hat{\rho}\hat{M}_a^\dag$.
If we define
\aln{\label{eq:sitespring_V}
\hat{V}=\sum_{a=1}^r\hat{M}_a\otimes \ket{a},
}
where $\{\ket{a}\}$ is an orthonormal basis of $\ml{H}_\mr{E}$, we find the representation in Eq.~\eqref{stine}.
This construction implies that the Stinespring representation is possible whenever $\mr{dim}[\ml{H}_\mr{E}]\geq \mr{rank}[\hat{F}]$ (since the minimum value of $r$ is $\mr{rank}[\hat{F}]$).
Moreover, if $\ml{F}$ is a CPTP map, we find from Eq.~\eqref{eq:sitespring_V} that $\hat{V}$ becomes an isometry, i.e., $\hat{V}^\dag\hat{V}=\hat{\mbb{I}}_\mr{S}$.

Now, since the isometry $\hat{V}$ can be written as $\hat{V}=\hat{U}(\hat{\mbb{I}}_\mr{S}\otimes \ket{\psi})$, where $\hat{U}\in\mbb{B}[\ml{H}_\mr{S}\otimes\ml{H}_\mr{E}]$ is unitary and $\ket{\psi}\in \ml{H}_\mr{E}$, we find~\cite{wolf2012quantum, barberena2024overview}
\aln{\label{Erhocptp}
\ml{E}[\hat{\rho}]=\mr{Tr}_\mr{E}[\hat{U}(\hat{\rho}\otimes\ket{\psi}\bra{\psi})\hat{U}^\dag]
}
for a CPTP map $\ml{E}$.
This representation, which is slightly different from the original Stinespring representation, offers an intuitive picture for describing open quantum systems: attach an environment $\ket{\psi}\in\ml{H}_\mr{E}$ to the system's state $\hat{\rho}$, let them interact via the unitary $\hat{U}$, and trace out the environmental degree's of freedom.

Similarly, it is known that a CP-instrument is represented as~\cite{wolf2012quantum, sagawa2013second}
\aln{
\ml{E}_b[\hat{\rho}]=\mr{Tr}_\mr{E}[\hat{U}(\hat{\rho}\otimes\ket{\psi}\bra{\psi})\hat{U}^\dag(\hat{\mbb{I}}_\mr{S}\otimes \hat{P}_{E,b})],
}
which corresponds to the decomposition of Eq.~\eqref{Erhocptp} into $\ml{E}_b$.
Here, $\hat{P}_{E,b}$ is the projection operator onto the basis $b$ in the environment. 
This representation reminds us of the setup of the indirect measurement, which we saw in Eq.~\eqref{indirecteq}.

\subsubsection{Natural representation}
Finally, we show a simple representation of a map, which is called, e.g., the natural representation~\cite{watrous2018theory}.
We first consider the vectorization of an operator $\hat{A}\in\mbb{B}[\ml{H}_\mr{S}]\ra \ml{H}_\mr{S}\otimes\ml{H}_\mr{S}$ as 
\aln{
\hat{A}=\sum_{ij}A_{ij}\ket{i}\bra{j}\longmapsto \ket{A}=\sum_{ij}A_{ij}\ket{i}\otimes\ket{j}.
}
In this formulation, we can easily confirm that
\aln{\label{henkan}
\hat{B}\hat{A}\hat{C}\longmapsto (\hat{B}\otimes\hat{C}^\mathsf{T})\ket{A}.
}

Now, a linear map $\ml{F}:\mathbb{B}[\ml{H}_\mr{S}]\ra \mathbb{B}[\ml{H}_\mr{S}]$ from $\hat{\rho}$ to $\hat{\rho}'=\ml{F}[\hat{\rho}]$
can be represented by $\hat{\mc{F}}:\ml{H}_\mr{S}\otimes \ml{H}_\mr{S}\ra\ml{H}_\mr{S}\otimes \ml{H}_\mr{S}$ from $\ket{\rho}$ to $\ket{\rho'}$.
When $\ml{F}$ is a CP map, using the Kraus representation of $\mc{F}$ and Eq.~\eqref{henkan},
we find
\aln{\label{naturalrep}
\hat{\mc{F}}=\sum_a\hat{M}_a\otimes \hat{M}_a^*.
}

We stress that the natural representation is different from the Choi-Jamiołkowski representation. 
While the natural representation is naive, it is not straightforward to judge CP of $\ml{F}$ from the natural representation $\mc{\hat{F}}$, unlike the Choi-Jamiołkowski representation $\hat{F}$ in Eq.~\eqref{choi}.
We also note that Eq.~\eqref{naturalrep} is intuitively understood as an interaction between the ``ket space" and ``bra space.'' 
If we consider a one-dimensional chain of qudits, the total system is given by a ladder composed of the two Hilbert spaces; therefore, this representation is also called the ladder representation~\cite{haga2023quasiparticles}.

\section{Quantum trajectories and quantum master equation} \label{sec:trajectory_master-equation}
In this section, we discuss the quantum dynamics driven by repeated measurements. 
Because of quantum back-action, the time-evolving state depends on the set of random measurement outcomes, which determines the so-called quantum trajectory\footnote{
The term ``quantum trajectory" is often used to describe the trajectory of a quantum stochastic equation under continuous-time measurements~\cite{wiseman2009quantum}. 
However, in this review, we also use the term to describe stochastic time evolution under discrete-time measurements.
}.
By taking the continuous-time limit, one finds a quantum stochastic equation that describes continuous-time quantum trajectories~\cite{ueda1990nonequilibrium, dalibard1992wave, dum1992monte, carmichael2009open} (see Refs.~\cite{pellegrini2008existence, pellegrini2010existence, pellegrini2010markov} for mathematical foundations).
If we average the dynamics over measurement outcomes, we obtain a quantum master equation, especially the Gorini–Kossakowski–Sudarshan–Lindblad (GKSL) equation.
The purpose of this section is to introduce these basic concepts, as well as to review some related topics.

\subsection{Repeated measurements}
As discussed in the previous chapter, a measurement on an initial state $\hat{\rho}_0$ with the outcome $b$ changes the state as in Eq.~\eqref{change}, where $\ml{E}_b$ is given in the Kraus representation by Eq.~\eqref{cpik}.
Averaging $\ml{E}_b[\hat{\rho}_0]$ over all possible measurement outcomes leads to the CPTP map 
\aln{
\hat{\rho}_1
=\sum_b\ml{E}_b[\hat{\rho}_0]=\ml{E}[\hat{\rho}_0].
}

We now consider repeating the above process. 
Assuming that the measurement on the meter is the same at each step, i.e., if all the CP-instruments $\{\ml{E}_{b_s}\}_{b_s=1}^B$ are the same at $s$th measurement with the same set of outcomes given by $b_s=0,\dots, B\in\mbb{N}$ with $s=1,2,\ldots$\footnote{Alternatively, we could consider a  correlated type of measurements. For example, we could consider a feedback operation, where the set $\{\ml{E}_{b_s}\}_{b_s}$ depends on the previous measurement outcomes $b_{s'}$ with $s'<s$. While such feedback operations can cause interesting dynamics, we do not discuss them in this review paper (also see the last paragraph in Chapter~\ref{sec:conclusion}).}, a quantum trajectory is described by an infinite sequence of measurement outcomes 
\begin{align}
 \bm{b}=(b_1,b_2,\cdots).
\label{eq:outcome-sequence_infinite}
\end{align} 
See Fig.~\ref{fig1}(b) for the pure-state case. 
If we assume the Kraus representation given in Eq.~\eqref{krauspure}\footnote{
In the following, we consider the case where the initial state of the meter is pure and thus the summation over $a$ is not necessary.
}, we find that a state after $n$ measurements reads
\aln{\label{qtmixed}
\hat{\rho}_{\bm{b};n}=\frac{\ml{E}_{b_n}\circ\cdots \circ\ml{E}_{b_1}[\hat{\rho}_0]}{p_{\bm{b};n}}=\frac{\hat{\mathsf{M}}_{\bm{b};n}\hat{\rho}_0\hat{\mathsf{M}}_{\bm{b};n}^\dagger}{p_{\bm{b};n}}.
}
Here, $\hat{\mathsf{M}}_{\bm{b};n}$ is the product of Kraus operators 
\begin{align}
  \hat{\mathsf{M}}_{\bm{b};n}=\hat{M}_{b_n}\hat{M}_{b_{n-1}}\cdots\hat{M}_{b_1},
  \label{eq:product_Kraus-operators}
\end{align}
corresponding to the sequence of measurement outcomes up to step $n$,
\begin{align}
\bm{b}_n=(b_1,\cdots,b_n),
\label{eq:outcome-sequence_finite}
\end{align}
and $p_{\bm{b};n}$ is the probability that $\bm{b}_n$ is realized:
\aln{
p_{\bm{b};n}=\Tr[\ml{E}_{b_n}\circ\cdots \circ\ml{E}_{b_1}[\hat{\rho}_0]].
}
In particular, if the initial state is a pure state, $\hat{\rho}_0=\ket{\psi_0}\bra{\psi_0}$, the post-measurement state remains pure and is given by
\aln{\label{qtpure}
\ket{\psi_{\bm{b};n}}=\frac{\hat{\mathsf{M}}_{\bm{b};n}\ket{\psi_0}}{\sqrt{p_{\bm{b};n}}}
}
with 
\aln{
p_{\bm{b};n}=\braket{\psi_0|\hat{\mathsf{M}}_{\bm{b};n}^\dag\hat{\mathsf{M}}_{\bm{b};n}|\psi_0}
= \| \hat{\sf{M}}_{\bm{b};n} | \psi_0 \rangle \|^2.
}
Equations~\eqref{qtmixed} and \eqref{qtpure} show that the time-evolving state is given by a conditional state that depends on the sequence of outcomes $\bm{b}$ under repeated measurements.
We refer to this sequence of conditional states as a quantum trajectory.
Note that the average over all possible measurement outcomes leads to the CPTP map,
\aln{
\hat{\rho}_{n}=\mbb{E}[\hat{\rho}_{\bm{b};n}]=\sum_{\bm{b}_n}p_{\bm{b};n}\hat{\rho}_{\bm{b};n}=\ml{E}^n[\hat{\rho}_0].
}
Here and hereafter, we denote by $\mbb{E}$ the average over all possible sequences of measurement outcomes $\bm{b}$.
We again stress that the quantum trajectory in Eq.~\eqref{qtpure} remains pure under time evolution, whereas the averaged dynamics generally leads to mixed states.

Before ending this section, we briefly comment on our notation, e.g., $
\hat{\rho}_{\bm{b};n}$.
While the state $ \hat{\rho}_{\bm{b};n}$ is essentially determined only by a finite sequence $\bm{b}_n=(b_1,\cdots,b_n)$, we keep the entire $\bm{b}$ in the subscript. 
This is in order to view this state as the $n$th-step state of a quantum trajectory determined from all infinite sequences of measurement outcomes $\bm{b}$. 
This view is important to define the probability measure of quantum trajectories and the average $\mbb{E}$ defined from it. 
Since this measure and average do not depend on the specific time step $n$, it is convenient to describe, e.g., multi-time correlation functions.

\subsection{Continuous-time limit of repeated measurements}
\subsubsection{Continuous-time limit}
We next consider a system that is continuously measured, which can be regarded as the continuous-time limit of the repeated indirect measurements discussed above~\cite{pellegrini2008existence, pellegrini2010existence, pellegrini2010markov, wiseman2009quantum, landi2024current}.
For this purpose, let us assume a simple situation where the state of the meter is given by
\aln{
\hat{\sigma}_\mr{M}=\ket{0}\bra{0},
}
and the interaction $\hat{V}_\mathrm{int}$ between the system and the meter before the measurement acts over an infinitesimal time interval $dt$.
Specifically, the joint unitary evolution is given by
\aln{
\hat{U}_\mr{int}=e^{-i\hat{V}_\mathrm{int}dt}\simeq 
\mbb{\hat{I}}_\mr{SM}+\ml{O}(dt).
}

Let us assume that the projective measurement on the meter is performed in an orthonormal basis $\{\ket{b}\}\:(b=0,\cdots, B)$, which includes the meter's initial state $\ket{0}$.
Then, since the interaction time is small, we expect that the measurement outcome is mostly $ b=0$, whereas the other outcomes $b \geq 1$ occur with a small probability proportional to $dt$.
From this observation, we can scale the measurement operators in Eq.~\eqref{krauspure} for $b\geq 1$ as
\aln{\label{jumpmeaskraus}
\hat{M}_{b\geq 1}=\sqrt{dt}\hat{L}_b.
}
Indeed, we find that the probability becomes proportional to $dt$ with this choice,
\aln{\label{probdt}
p_{b\geq 1}(t)=\Tr[\hat{\rho}(t)\hat{L}_b^\dag\hat{L}_b]dt,
}
where $\hat{\rho}(t)$ is the pre-measurement state at time $t$.
The post-measurement state is given by 
\aln{\label{click}
\hat{\rho}_{b\geq 1}(t+dt)=\frac{\hat{L}_b\hat{\rho}(t)\hat{L}_b^\dag}{\Tr[\hat{\rho}(t)\hat{L}_b^\dag\hat{L}_b]}.
}

Next, we notice that $\sum_{b=0}^B\hat{M}_b^\dag\hat{M}_b=\hat{\mbb{I}}_\mr{S}$.
By requiring this relation to hold up to the order of $dt$, we find $\hat{M}_0=\hat{\mbb{I}}_\mr{S}-i\hat{H}_\mr{S}dt$ in the absence of the measurement, where $\hat{H}_\mr{S}$ is an intrinsic Hamiltonian of the system.
Then, we find that $\hat{M}_0$ should have the form 
\aln{
\hat{M}_0=\hat{\mbb{I}}_\mr{S}-\frac{1}{2}\sum_{b\geq 1}\hat{L}_b^\dag\hat{L}_bdt-i\hat{H}_\mr{S}dt.
}
Up to the order of $dt$, the corresponding probability is given by
\aln{
p_0(t)=1-\sum_{b\geq 1}\Tr[\hat{\rho}(t)\hat{L}_b^\dag\hat{L}_b]dt,
}
and the post-measurement state is 
\aln{\label{noclick}
\hat{\rho}_{b=0}(t+dt)=\hat{\rho}(t)-i[\hat{H}_\mr{S},\hat{\rho}(t)]dt+\hat{\rho}(t)\sum_{b\geq 1}\Tr[\hat{\rho}(t)\hat{L}_b^\dag\hat{L}_b]dt-\frac{1}{2}\sum_{b\geq 1}\{\hat{\rho}(t),\hat{L}_b^\dag\hat{L}_b\}dt.
}

To summarize, after the small duration time $dt$, there are two possibilities. 
One is that the meter does not change its value from $\ket{0}$. 
In this case, the system undergoes a continuous-time evolution as given in Eq.~\eqref{noclick}.
The other is that the meter changes its value to $b\geq 1$. 
In this case, the system suddenly changes its state discontinuously as in Eq.~\eqref{click}, which is called a quantum jump.

\begin{figure}[!h]
\centering\includegraphics[width=\linewidth]{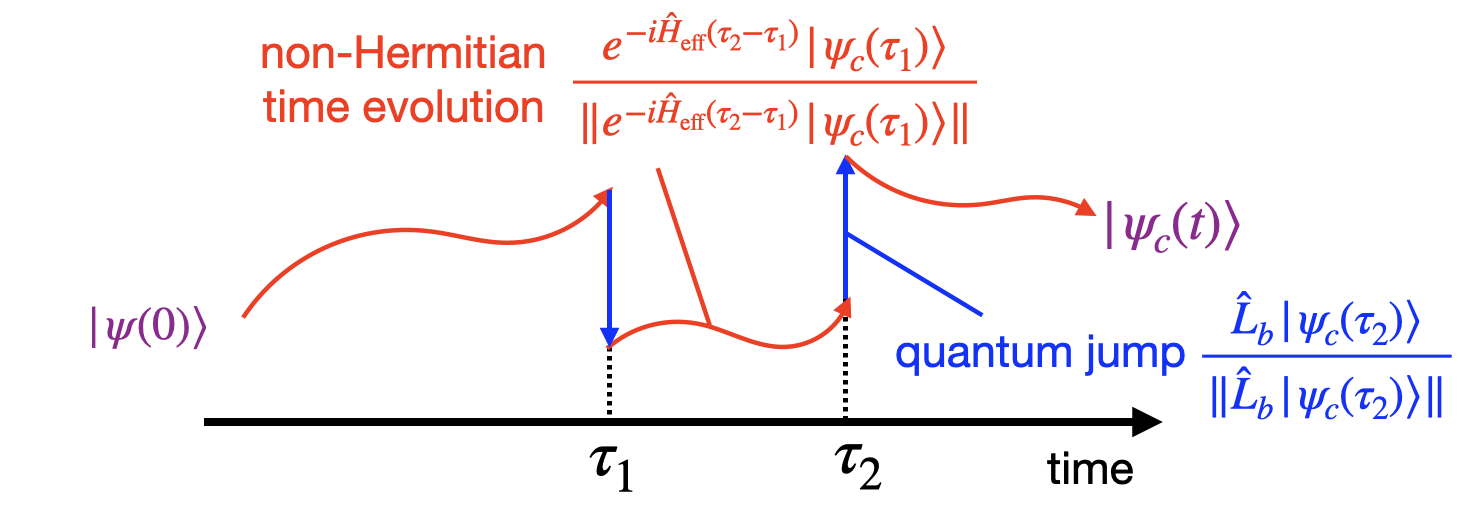}
\caption{
Schematic illustration of a quantum trajectory. 
Starting from $\ket{\psi(0)}$, the state undergoes a non-Hermitian continuous time evolution governed by $\hat{H}_\mathrm{eff}$, unless a quantum jump described by the jump operator $\hat{L}_b$ occurs, which suddenly changes the state. 
Here, the jump times are denoted by $\tau_1$ and $\tau_2$.
}
\label{fig2}
\end{figure}

If we assume that the meter's state is immediately reset to $\ket{0}$ after a quantum jump $b\geq1$, the time evolution of the system consists of discrete quantum jumps interspersed with continuous time evolutions, as illustrated in Fig.~\ref{fig2}.
Importantly, if we postselect a trajectory $\hat{\rho}_\mr{no}(t)$ free from quantum jumps at all times, i.e., $b=0$ for all $t$, the time evolution of the system is given by 
\aln{\label{nonhermdynamics}
\hat{\rho}_\mr{no}(t+dt)=\hat{\rho}_\mr{no}(t)-i\lrs{\hat{H}_\mr{eff}\hat{\rho}_\mr{no}(t)-\hat{\rho}_\mr{no}(t)\hat{H}_\mr{eff}^\dag}dt+\hat{\rho}_\mr{no}(t)\sum_{b\geq 1}\Tr[\hat{\rho}_\mr{no}(t)\hat{L}_b^\dag\hat{L}_b]dt,
}
where 
\aln{
\hat{H}_\mr{eff}=\hat{H}_\mr{S}-\frac{i}{2}\sum_{b\geq 1}\hat{L}_b^\dag\hat{L}_b
}
is the effective non-Hermitian Hamiltonian.
Equation~\eqref{nonhermdynamics} represents a simple non-Hermitian time evolution with the third term on the right-hand side ensuring the normalization of $\hat{\rho}_\mathrm{no}(t)$.
Indeed, the formal solution of Eq.~\eqref{nonhermdynamics} is 
\aln{
\hat{\rho}_\mr{no}(t)=\frac{e^{-i\hat{H}_\mr{eff} t}\hat{\rho}(0)e^{i\hat{H}_\mr{eff}^\dag t}}{\Tr[e^{-i\hat{H}_\mr{eff} t}\hat{\rho}(0)e^{i\hat{H}_\mr{eff}^\dag t}]}.
}
This type of non-Hermitian dynamics has been extensively investigated in recent years~\cite{ashida2020non}, while it requires postselecting a rare trajectory without any jumps, which occurs with an exponentially small probability in time.

\subsubsection{Stochastic equation}
Let us denote the number of jumps for $b\:(\geq 1)$ up to time $t$ by $N_b(t)$.
Then, we can consider the increment of the jump number, $dN_b(t)$, which is 1 only if there is a quantum jump at time $t$\footnote{
More precisely, we consider a quantum jump that occurs during $[t,t+dt]$.
} with type $b$ and is 0 otherwise.
From this definition, we have
\aln{
dN_b(t) dN_{b'}(t)=\delta_{bb'}dN_b(t).
}
Moreover, following the discussion in the previous subsection [especially Eq.~\eqref{probdt}], we obtain
\aln{\label{conditional}
p_b^c(t):=\mr{Prob}_t[dN_b(t)=1|\hat{\rho}_c(t)]=dt\mr{Tr}[\hat{\rho}_c(t)\hat{L}_b^\dag\hat{L}_b]=dt\braket{\hat{L}_b^\dag\hat{L}_b}_{c;t},
}
where $c$ indicates that we consider a quantum trajectory without average\footnote{
Note that the subscript/superscript $c$ will be omitted if it is apparent that we consider quantum trajectories from the context. 
We also sometimes omit the explicit time-dependence ($t$) of the state when it will raise no confusion.
}, $\braket{\cdots}_{c;t}:=\mr{Tr}[{\hat{\rho}_c(t)\cdots}]$, and $\mr{Prob}_t[\cdots|\hat{\rho}_c(t)]$ denotes a conditional probability at time $t$ for a given state $\hat{\rho}_c(t)$.

Using this notation, the possibilities given in Eqs.~\eqref{click} and~\eqref{noclick} are unified as
\aln{\label{stochasticeq}
d\hat{\rho}_c=-i[\hat{H}_\mr{S},\hat{\rho}_c]dt+\hat{\rho}_c\sum_{b\geq 1}\braket{\hat{L}_b^\dag\hat{L}_b}_{c;t}dt-\frac{1}{2}\sum_{b\geq 1}\{\hat{\rho}_c,\hat{L}_b^\dag\hat{L}_b\}dt+\sum_{b\geq 1}\lrs{\frac{\hat{L}_b\hat{\rho}_c\hat{L}_b^\dag}{\braket{\hat{L}_b^\dag\hat{L}_b}_{c;t}}-\hat{\rho}_c}dN_b.
}
We note that for pure states, the corresponding equation reads
\aln{
d\ket{\psi_c}=\lrs{-i\hat{H}_\mr{eff}+\frac{1}{2}\sum_{b\geq 1}\braket{\hat{L}_b^\dag\hat{L}_b}_{c;t}}\ket{\psi_c}dt+\sum_{b\geq 1}\lrs{\frac{\hat{L}_b}{\sqrt{\braket{\hat{L}_b^\dag\hat{L}_b}}_{c;t}}-1}\ket{\psi_c}dN_b.
}
We can see this by noticing, e.g., that $d[\ket{\psi_c}\bra{\psi_c}]=\ket{d\psi_c}\bra{\psi_c}+\ket{\psi_c}\bra{d\psi_c}+\ket{d\psi_c}\bra{d\psi_c}$, $dN_bdt=0$, and $dN_bdN_{b'}=\delta_{bb'}dN_b$.

A single quantum trajectory is characterized by the set of jump times $\{\tau_k\}_{k\in\mbb{N}}$ and the corresponding jump types $\{b_k\}_{k\in\mbb{N}}$, i.e., $\{(\tau_k,b_k)\}_{k\in\mbb{N}}$ such that $dN_{b_k}(\tau_k)=1$.
For pure states, the quantum trajectory after $K$ jumps is represented as
\aln{\label{conttraj}
\ket{\psi_c(t)}\propto e^{-i\hat{H}_\mr{eff}(t-\tau_K)}\hat{L}_{b_K}
e^{-i\hat{H}_\mr{eff}(\tau_K-\tau_{K-1})}\hat{L}_{b_{K-1}}\cdots
e^{-i\hat{H}_\mr{eff}(\tau_2-\tau_{1})}\hat{L}_{b_{1}}e^{-i\hat{H}_\mr{eff}\tau_{1}}\ket{\psi(0)}
}
with $\braket{\psi_c(t)|\psi_c(t)}=1$.

\subsubsection{Averaged dynamics} \label{sec:AveragedDynamics}
If we average Eq.~\eqref{stochasticeq} over the measurement outcomes, we obtain the renowned GKSL equation~\cite{lindblad1976generators, gorini1976completely}.
To be specific, we treat the dynamics of
\aln{
\hat{\rho}(t)=\mbb{E}[\hat{\rho}_c(t)],
}
where $\mbb{E}$ denotes the average over all quantum trajectories, i.e., over $dN_b(t)$ for all $t$.
To proceed, we note that for a function $g$ of $\hat{\rho}_c(t)$~\cite{landi2024current},
\aln{\label{eq:muzukasi}
\mbb{E}[g(\hat{\rho}_c(t))dN_b(t)]&=\mbb{E}_{[0,t)}[\mbb{E}_{t}[g(\hat{\rho}_c(t))dN_b(t)|\hat{\rho}_c(t)]]\nonumber\\
&=\mbb{E}_{[0,t)}[g(\hat{\rho}_c(t))\braket{\hat{L}_b^\dag\hat{L}_b}_{c;t}dt]\nonumber\\
&=\mbb{E}[g(\hat{\rho}_c(t))\braket{\hat{L}_b^\dag\hat{L}_b}_{c;t}]dt,
}
where $\mbb{E}_{[0,t)}$ and $\mbb{E}_t$ denote the average over $dN_b(t')$ with $0\leq t'<t$ and $t'=t$, respectively, and $\mathbb{E}_{t}[\cdots|\hat{\rho}_c(t)]$ denotes the conditional average for a given state $\hat{\rho}_c(t)$. 
In Eq.~\eqref{eq:muzukasi}, we first take the average over $dN_b(t)$ using Eq.~\eqref{conditional} for the given $\hat{\rho}_c(t)$ and then take the average over $\hat{\rho}_c(t)$ (i.e., taking the average over $dN_b(t')$ for $0\leq t'<t$).
We have also used the fact that the average over $dN_b(t')$ for $t< t'$, which is included in the average $\mbb{E}$, does not change the result in obtaining the third line from the second line.
Applying this equality to Eq.~\eqref{stochasticeq}, we find the GKSL equation
\aln{
\frac{d\hat{\rho}(t)}{dt}=\ml{L}[\hat{\rho}(t)]=-i[\hat{H}_\mr{S},\hat{\rho}(t)]+\sum_{b\geq 1}\left(\hat{L}_b\hat{\rho}(t)\hat{L}_b^\dag-\frac{1}{2}\{\hat{\rho}(t),\hat{L}_b^\dag\hat{L}_b\}\right).
}
This GKSL equation leads to the formal solution 
\aln{\label{formalGKSL}
\hat{\rho}(t)=e^{\ml{L}t}[\hat{\rho}(0)],
}
which is a CPTP map.
The map $\ml{E}_t=e^{\ml{L}t}$ satisfies the Markovian condition $\ml{E}_t\circ\ml{E}_s=\ml{E}_{t+s}$ for $t,s\geq 0$.
It is also known that the generator of a Markovian CPTP map has the GKSL form, where $\hat{H}_\mr{S}$ and $\hat{L}_b$ can be time-dependent in general~\cite{lindblad1976generators, rivas2012open}.

\subsubsection{Examples}\label{qtex}
Here, we first discuss an example of the spontaneous emission of a two-level atom.
For simplicity, we consider a toy model where the two-level system is coupled to a single-mode photon\footnote{
This is a simple toy model. 
In reality, we need to consider various modes of the photon, but we neglect them.
}, which serves as a meter.
The interaction Hamiltonian, after the rotating-wave approximation~\cite{shore1993jaynes}, reads
\aln{
\hat{V}_\mr{int}= g(\hat{\sigma}^-\hat{b}^\dag+\hat{\sigma}^+\hat{b}),
}
where $\hat{\sigma}^{+/-}$ denotes the raising/lowering operator for the atom and $\hat{b}$ is the annihilation operator of the photon.
We assume that there is no photon initially, $\ket{0}_\mr{ph}$.

Consider a measurement of the photon number after a short time interval $dt$ with $\hat{U}_\mathrm{int}=e^{-i\hat{V}_\mathrm{int}dt}$.
Then, the measurement operator [cf. Eq.~\eqref{qbusi}] corresponding to photon detection ($b=1$) reads
\aln{\label{eq:Kraus-operator_atom-photon}
\hat{M}_1=_\mr{ph}\braket{1|\hat{U}_\mr{int}|0}_\mr{ph}\simeq{-i} g\hat{\sigma}^-dt={-i}\sqrt{\gamma dt}\hat{\sigma}^-,
}
where $\gamma=g^2dt$.
Comparing Eq.~\eqref{eq:Kraus-operator_atom-photon} with Eq.~\eqref{jumpmeaskraus}, we find
\aln{\label{singlejump}
\hat{L}=\sqrt{\gamma}\hat{\sigma}^-,
}
where we have omitted the subscript $b$ and dropped the phase factor $-i$ since it does not change Eq.~\eqref{stochasticeq}.
The probability of photon detection is given by
\aln{
p=\mr{Tr}[\hat{\rho}_c\hat{L}^\dag\hat{L}]dt=\gamma\braket{\hat{\sigma}^+\hat{\sigma}^-}_cdt.
}
We assume that the photon is immediately reset to $\ket{0}_\mr{ph}$ upon detection, since it is absorbed into the photo-detector.
Then, this continuously monitored process is understood as follows [see Fig.~\ref{fig3}(a)]:
a photon is detected with a rate $\gamma\braket{\hat{\sigma}^+\hat{\sigma}^-}_c$, causing the state of the atom to abruptly jump to its ground state.
If no photon is detected, the state is evolved by the non-Hermitian Hamiltonian $\hat{H}_\mr{eff}=\hat{H}_\mr{S}-\frac{i\gamma}{2}\hat{\sigma}^+\hat{\sigma}^-$.
The average over all quantum trajectories leads to the GKSL equation
\aln{
\frac{d\hat{\rho}}{dt}=-i[\hat{H}_\mr{S},\hat{\rho}]+\gamma\hat{\sigma}^-\hat{\rho}\hat{\sigma}^+-\frac{\gamma}{2}\{\hat{\rho},\hat{\sigma}^+\hat{\sigma}^-\}.
}

\begin{figure}[!h]
\centering\includegraphics[width=\linewidth]{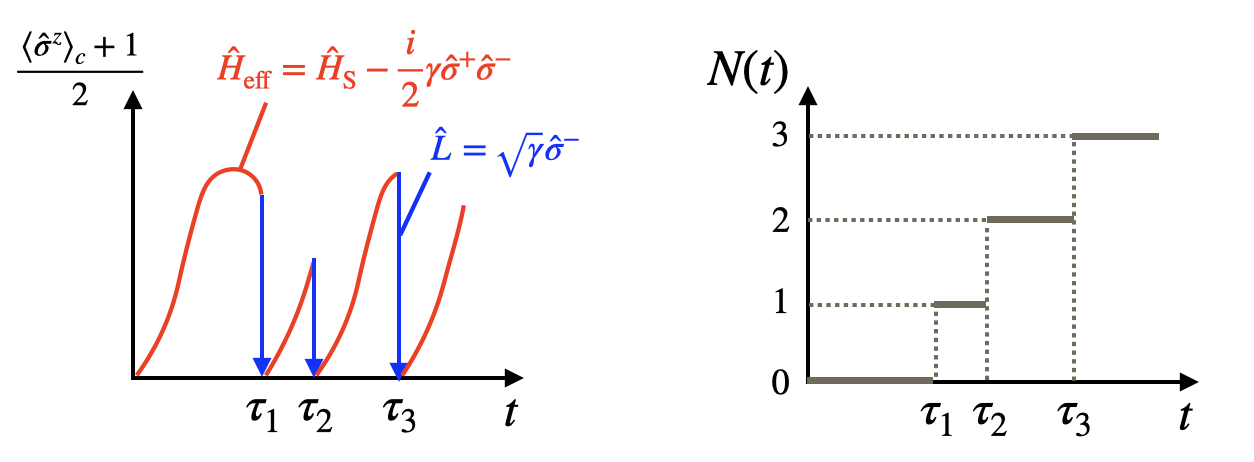}
\caption{
An example of a quantum trajectory for a single atom system with spontaneous emissions. 
(a) If we consider the occupation of the excited state, given as $\frac{\hat{\sigma}^z+1}{2}$, it is described by the continuous time evolution and quantum jumps, which reset the state into the ground state. 
(b) The total number of jumps $N(t)$ until time $t$. It increases by one at times $\tau_1$, $\tau_2$, and $\tau_3$.
}
\label{fig3}
\end{figure}

There are many other types of models that are expected to phenomenologically describe open quantum systems.
For example, if we consider a two-level atom $\hat{H}_\mr{S}=\frac{\omega}{2}\hat{\sigma}^z$ in a finite-temperature bath of photons with an inverse temperature $\beta$, it will exhibit stimulated emission and absorption of photons. 
In this case, the averaged dynamics is often modeled~\cite{wiseman2009quantum} by the GKSL equation with jump operators $\hat{L}_+=\sqrt{\gamma_+}\hat{\sigma}^+$ and 
$\hat{L}_-=\sqrt{\gamma_-}\hat{\sigma}^-$, where the detailed balance condition $\gamma_+/\gamma_-=e^{-\beta\omega}$ is satisfied. 
Under such jump operators, the Gibbs state at temperature $\beta^{-1}$ becomes the stationary state of the GKSL equation.
 
As another example, we can think of many-body systems under continuous measurement.
In particular, let us consider the Bose-Hubbard model realized in a cold atomic system in an optical lattice~\cite{bloch2008many}, whose Hamiltonian is given by
\aln{\label{BH}
\hat{H}_\mr{BH}=-\sum_{l,l'}J_{ll'}(\hat{a}^\dag_l\hat{a}_{l'}+\mr{h.c.})+\frac{U}{2}\sum_l \hat{n}_l(\hat{n}_l-1),
}
where $l$ and $l'$ denote the lattice sites, $\hat{a}_l$ is the annihilation operator of a boson at site $l$, and $\hat{n}_l=\hat{a}_l^\dag\hat{a}_l$ is the number operator at site $l$.
By adding a suitable laser field, atoms will absorb and spontaneously emit photons, causing the incoherent scattering of photons~\cite{daley2014quantum}.
We can determine the positions of atoms by detecting the scattered photons.
If the wavelength of the incident photons is shorter than the lattice spacing, the measurement will be performed with single-site resolution.
The corresponding jump operators are modeled as $\hat{L}_l=\sqrt{\gamma}\hat{n}_l$, i.e., the number operator for each site.
The microscopic deviation of this process was discussed in Refs.~\cite{pichler2010nonequilibrium, daley2014quantum}.
Note that such photon scattering from cold atoms has in fact been experimentally realized~\cite{patil2015measurement, luschen2017signatures, bouganne2020anomalous}, while continuous measurement with single-site resolution has yet to be achieved.

\subsection{Physical quantities characterizing quantum trajectories}
Let us next discuss how to characterize quantum trajectories through physically relevant quantities.
Here we focus on a pure state in the continuous-time case, i.e., $\ket{\psi_c(t)}$ in Eq.~\eqref{conttraj}, although a similar discussion applies to the discrete-time case in Eq.~\eqref{qtpure}.

\subsubsection{Nonlinear observables and postselection cost}
\label{sec:nonlinear-observables}
One may naively consider the expectation value of an observable $\hat{A}$ with respect to $\ket{\psi_c(t)}$, 
\aln{
\braket{\hat{A}}_{c;t}=\braket{\psi_c(t)|\hat{A}|\psi_c(t)},
}
to characterize quantum trajectories.
However, its ensemble average (average over the quantum trajectories) can also be obtained from the GKSL equation as 
\aln{\label{eq:expectation-value}
\mbb{E}[\braket{\hat{A}}_{c;t}]=\braket{\hat{A}}_t
}
with
\aln{\label{eq:expectation-value2}
\braket{\hat{A}}_t=\mr{Tr}[{\hat{\rho}(t)\hat{A}}],
}
where $\hat{\rho}(t)$ is given in Eq.~\eqref{formalGKSL} and we have used $\mathbb{E}[\ket{\psi_c(t)}\bra{\psi_c(t)}]=\hat{\rho}(t)$.
Therefore, to highlight physics unique to quantum trajectories, we should consider nonlinear quantities in $\hat{\rho}_c(t)$, such as $\braket{\hat{A}}_{c;t}^2$, whose ensemble average $\mbb{E}[\braket{\hat{A}}_{c;t}^2]$ cannot be accessed via the GKSL dynamics.

Another interesting nonlinear quantity, which plays an especially important role in many-body systems, is the entanglement entropy of the state 
\aln{
S_{c,X}(t)=-\mr{Tr}[\hat{\rho}_{c,X}(t)\ln \hat{\rho}_{c,X}(t)],
}
where 
\aln{
\hat{\rho}_{c,X}(t) =\mr{Tr}_{\bar{X}}[\ket{\psi_c(t)}\bra{\psi_c(t)}]
}
is the reduced density matrix of a subsystem $X$. 
Here, we have divided the total system into two regions, $X$ and its complement $\bar{X}$.
In fact, we find
\aln{
\mbb{E}[S_{c,X}(t)]\neq -\Tr[\hat{\rho}_X(t)\ln \hat{\rho}_X(t)],
}
where $\hat{\rho}_X(t)=\mbb{E}[\hat{\rho}_{c,X}(t)]=\mr{Tr}_{\bar{X}}[\hat{\rho}(t)]$.
One of the interesting phenomena concerning this quantity is the measurement-induced entanglement phase transition~\cite{skinner2019measurement, li2018quantum, chan2019unitary}, which will be detailed in Chapter~\ref{sec:mipt}.
There, we will see that the long-time behavior of $\mbb{E}[S_{c,X}(t)]$ exhibits an intriguing phase transition as the measurement strength is varied, even when $\hat{\rho}(t)$ becomes a trivial maximally mixed state at long times regardless of the measurement strength.

As an important remark, the evaluation of these nonlinear quantities requires the postselection of quantum trajectories, whose experimental cost is exponentially large. 
For example, precisely estimating $\braket{\hat{A}}_{c;t}$ for a single quantum trajectory needs a large number of measurements of an observable $\hat{A}$ with respect to $\ket{\psi_c(t)}$ since we need to suppress intrinsic quantum fluctuations.
However, for this purpose, we must prepare $\ket{\psi_c(t)}$ multiple times.
Since we cannot clone quantum states~\cite{wootters1982single}, we need to repeat the time evolution from the initial state $\ket{\psi(0)}$ and postselect the states where the jumps occurred according to $\{(\tau_k,b_k)\}_{k=1}^K$.
As one would expect, the probability of obtaining a target trajectory is exponentially small with respect to time. 
If we consider many-body systems, it is also exponentially small with respect to the system size.

To quantify the cost of postselection in a simple setup, let us consider the following model of hardcore bosons on a one-dimensional lattice,
\aln{\label{hcboson}
\hat{H}_\mr{HCB}=-\sum_{l=1}^VJ(\hat{b}^\dag_l\hat{b}_{l+1}+\mr{h.c.})+\frac{U}{2}\sum_{l=1}^V \hat{n}_l\hat{n}_{l+1},
}
subject to position measurements $\hat{L}_l=\sqrt{\gamma}\hat{n}_l$ at each site.
Here, $V$ is the system size, $\hat{b}_l$ is the annihilation operator of a hardcore boson, and $\hat{n}_l=\hat{b}_l^\dag\hat{b}_l$.
The number of bosons, $M=\left<\sum_l\hat{n}_l\right>_{c;t}$, is conserved under the dynamics, which is assumed to scale as $M\propto V$.

In this case, the rate of the jump occurring at one of the sites is given by
\aln{
\sum_l\frac{p_l}{dt}=\sum_l\braket{\hat{L}_l^\dag\hat{L}_l}_{c;t}=\braketL{\gamma\sum_l\hat{n}_l}_{c;t}=\gamma M,
}
which is constant (note that $\hat{n}_l^2=\hat{n}_l$).
In other words, the probability of finding no jumps during the time interval $[0,t]$ is given by $e^{-\gamma M t}$.

Now, we roughly estimate the probability of realizing a quantum trajectory characterized by  $\{(\tau_k,b_k)\}_{k=1}^K$, where we allow some uncertainty $\Delta \tau$ for the jump times $\tau_k$~\cite{yamamoto2023localization}.
We require that $\Delta\tau$ is much smaller than the typical interval of jumps $\tau_{k+1}-\tau_k\sim (\gamma M)^{-1}$, i.e., $\epsilon=\gamma M\Delta\tau\ll 1$.
Given an initial state, the probability that there is no jump during $t\in[0,\tau_1-\Delta\tau)$ and a jump occurs during $t\in[\tau_1-\Delta\tau,\tau_1+\Delta\tau]$ becomes
\aln{
P_{\tau_1}=e^{-\gamma M(\tau_1-\Delta\tau)}-e^{-\gamma M(\tau_1+\Delta\tau)}\simeq (2\epsilon)e^{-\gamma M \tau_1}.
}
We can repeat this process for time intervals $\tau_2-\tau_1, \dots, \tau_K-\tau_{K-1}$, which results in the probability 
\aln{
P_{\tau_1}P_{\tau_2-\tau_1} \cdots P_{\tau_K-\tau_{K-1}}=(2\epsilon)^Ke^{-\gamma M\tau_K}.
}
Finally, multiplying the probability that there is no jump during $(\tau_K,t]$, we find
\aln{
P_t(\{\tau_k\}_{k=1}^K)=(2\epsilon)^Ke^{-\gamma Mt}
}

Moreover, for each jump, we have $V$ different types of jumps. As a rough estimate, we assume that they occur with equal probabilities.
Then, we find that the probability of realizing a quantum trajectory characterized by $\{(\tau_k,b_k)\}_{k=1}^K$ is 
\aln{
P_t(\{(\tau_k,b_k)\}_{k=1}^K)=\lrs{\frac{1}{V}}^K(2\epsilon)^Ke^{-\gamma Mt}\sim \lrs{\frac{V e}{2\epsilon}}^{-\gamma Mt},
}
where we have used $K\sim \gamma M t$ for a typical trajectory.
Therefore, the cost of postselection exponentially increases in time as
\aln{\label{cost_postselection}
\lrs{P_t(\{(\tau_k,b_k)\}_{k=1}^K)}^{-1}\sim  \lrs{\frac{V e}{2\epsilon}}^{\gamma Mt},
}
whose growth rate scales as $\sim V\ln (V/\epsilon)$.

Note that, for the present setup, we can avoid the above postselection and effectively simulate nonlinear quantities by performing suitable experiments in unitary quantum systems~\cite{yamamoto2023localization}. 
This method reduces the cost to $\sim e^{\gamma Mt}$, which is better than Eq.~\eqref{cost_postselection} by the factor $(V/\epsilon)^{\gamma Mt}$.
While this cost is still exponentially large, we can circumvent the factor concerning the jump-time accuracy $\epsilon$.
We also note that the postselection cost is approximately estimated as $e^{\ml{O}(Vt)}$ for discrete-time quantum circuits described by Eq.~\eqref{qtpure}, although there are several attempts to avoid the postselection~\cite{Gullans20PRL, Hoke23, Agrawal24, Kamakari24, Barratt22, Dehghani23, Garratt23, Garratt24, Yamamoto25, Iadecola22, Buchhold22}.

\subsubsection{Statistics of quantum jumps}
Another important quantity, which is unique to stochastic dynamics, is obtained from the counting variable $N_b(t)$ of quantum jumps.
Since $N_b(t)$ is just given by the total number of jumps of type $b\:(\geq 1)$ within $[0,t]$, it is directly observable from a single quantum trajectory without postselection.
To be specific, we can consider~\cite{landi2024current}
\aln{
N_\mu(t)=\sum_b\mu_b N_b(t)
}
and its derivative, which is often called the current:
\aln{
I_\mu(t)=\frac{dN_\mu(t)}{dt}=\sum_b\mu_b \frac{dN_b(t)}{dt}
}
where $\mu=\{\mu_b\}_{b}\in\mbb{R}^B$.

If we take the ensemble average of $I_\mu(t)$, Eq. (\ref{eq:muzukasi}) leads to
\aln{\label{aveInu}
J_\mu(t)=\mbb{E}[I_\mu(t)]=\sum_b\mu_b\frac{\mbb{E}[dN_b(t)]}{dt}=\sum_b\mu_b\braket{\hat{L}_b^\dag\hat{L}_b}_t,
}
where $\braket{\hat{L}_b^\dag\hat{L}_b}_t$ is the expectation value of $\hat{L}_b^\dag\hat{L}_b$ with respect to $\hat{\rho}(t)$ obtained from the GKSL equation in Eq.~\eqref{eq:expectation-value2}.

A more nontrivial question is how the current is correlated in time, which is characterized by 
\aln{
C_\mu(t,t+\tau)=\mbb{E}[I_\mu(t)I_\mu(t+\tau)]-\mbb{E}[I_\mu(t)]\mbb{E}[I_\mu(t+\tau)].
}
We can assume $\tau\geq 0$ since $\mbb{E}[I_\mu(t)I_\mu(t+\tau)]=\mbb{E}[I_\mu(t+\tau)I_\mu(t)]$.
For $\tau>0$, we have
\aln{
&\mbb{E}[I_\mu(t)I_\mu(t+\tau)]\nonumber\\
&=\frac{1}{dt^2}\mbb{E}[dN_\mu(t)dN_\mu(t+\tau)]\nonumber\\
&=\frac{1}{dt^2}\sum_{b,b'}\mu_b\mu_{b'}\:\mr{Prob}[dN_b(t+\tau)=1,dN_{b'}(t)=1]\nonumber\\
&=\frac{1}{dt^2}\sum_{b,b'}\mu_b\mu_{b'}\:\mbb{E}_{[0,t)}[\mr{Prob}_{[t+\tau,t]}[dN_b(t+\tau)=1|dN_{b'}(t)=1,\hat{\rho}_c(t)]\cdot \mr{Prob}_t[dN_{b'}(t)=1|\hat{\rho}_c(t)]],
}
where $\mr{Prob}_{[t+\tau,t]}[dN_b(t+\tau)=1|dN_{b'}(t)=1,\hat{\rho}_c(t)]$ is the conditional probability that a type-$b$ jump occurs at time $t+\tau$ given that a type-$b'$ jump occurs at time $t$ for $\hat{\rho}_c(t)$, and  
\aln{
\mr{Prob}_t[dN_{b'}(t)=1|\hat{\rho}_c(t)]=p_{b'}^c(t)=\mr{Tr}[\hat{L}_{b'}^\dag\hat{L}_{b'}\hat{\rho}_c(t)]dt.
}

Now, $\mr{Prob}_{[t+\tau,t]}[dN_b(t+\tau)=1|dN_{b'}(t)=1,\hat{\rho}_c(t)]$ is obtained as follows.
If we have a state $\hat{\rho}_c(t)$ at time $t$ before the jump, the occurrence of the jump $dN_{b'}(t)=1$ means that the post-jump state is 
\aln{
\hat{\rho}_c'(t)=\frac{\hat{L}_{b'}\hat{\rho}_c(t)\hat{L}_{b'}^\dag}{\braket{\hat{L}_{b'}^\dag\hat{L}_{b'}}_{c;t}}.
}
Then, we do not care about the jumps during $(t,t+\tau)$, so that on average we have
\aln{
\hat{\rho}_c'(t+\tau)=e^{\ml{L}\tau}[\hat{\rho}_c'(t)]
}
at time $t+\tau$ before the jump.
Finally, we consider the probability of the jump, $\mr{Tr}[\hat{L}_{b}^\dag\hat{L}_{b}\hat{\rho}_c'(t+\tau)]dt$. Combining the above steps, we find 
\aln{
\mbb{E}[I_\mu(t)I_\mu(t+\tau)]&=\sum_{bb'}\mu_b\mu_{b'}\mbb{E}_{[0,t)}\lrl{\Tr[\hat{L}_{b}^\dag\hat{L}_{b}e^{\ml{L}\tau}[\hat{L}_{b'}\hat{\rho}_c(t)\hat{L}_{b'}^\dag]]}\nonumber\\
&=\Tr[\ml{G}_\mu\circ e^{\ml{L}\tau}\circ\ml{G}_\mu[\hat{\rho}(t)]],
}
where we have introduced the super-operator
\aln{
\ml{G}_\mu[\hat{X}]=\sum_b\mu_b\hat{L}_b\hat{X}\hat{L}_b^\dag.
}
Note that $J_\mu(t)=\Tr[\ml{G}_\mu[\hat{\rho}(t)]]$.

Finally, for $\tau=0$, we find
\aln{
\mbb{E}[I_\mu(t)^2]=\frac{1}{dt^2}\sum_b\mu_b^2\mbb{E}[dN_b(t)]
=\frac{1}{dt}\sum_b\mu_b^2\braket{\hat{L}_b^\dag\hat{L}_b}_t,
}
which can be interpreted as $\delta(\tau)\sum_b\mu_b^2\braket{\hat{L}_b^\dag\hat{L}_b}_t$.

In conclusion, the correlation function is given by
\aln{
C_\mu(t,t+\tau)=\Tr[\ml{G}_\mu\circ e^{\ml{L}\tau}\circ\ml{G}_\mu[\hat{\rho}(t)]]+\delta(\tau)\sum_b\mu_b^2\braket{\hat{L}_b^\dag\hat{L}_b}_t-J_\mu(t)J_\mu(t+\tau).
}
In particular, if we choose the state $\hat{\rho}$ as a stationary state $\hat{\rho}_\mr{ss}$, which satisfies
\aln{
\ml{L}[\hat{\rho}_\mr{ss}]=0,
}
the correlation function reads 
\aln{
C_\mu^\mr{ss}(\tau)=\Tr[\ml{G}_\mu\circ e^{\ml{L}|\tau|}\circ\ml{G}_\mu[\hat{\rho}_\mr{ss}]]+\delta(\tau)\sum_b\mu_b^2\braket{\hat{L}_b^\dag\hat{L}_b}_\mr{ss}-\Tr[\ml{G}_\mu[\hat{\rho}_\mr{ss}]]^2,
}
where we have also allowed $\tau<0$ and defined $\braket{\hat{A}}_\mr{ss}=\Tr[\hat{A}\hat{\rho}_\mr{ss}]$.

From the above results, we can also obtain the statistics of the counting variable $N_\mu(t)$.
For simplicity, let us assume that $\hat{\rho}(t)$ is a stationary state, although the extension to arbitrary initial states is not difficult. Then, the average of $N_\mu(t)$ is given by
\aln{
\mbb{E}[N_\mu(t)]=\int_0^td\tau \sum_b\mu_b\braket{\hat{L}_b^\dag\hat{L}_b}_\mr{ss}=t\sum_b\mu_b\braket{\hat{L}_b^\dag\hat{L}_b}_\mr{ss}.
}
Next, the variance of $N_\mu(t)$ becomes
\aln{
\mbb{V}[N_\mu(t)]&=\mbb{E}[N_\mu(t)^2]-\mbb{E}[N_\mu(t)]^2\nonumber\\
&=\int_0^td\tau\int_0^td\tau' C_\mu(\tau,\tau').
}
Taking the time derivative, we have
\aln{
D(t)=\frac{d\mbb{V}[N_\mu(t)]}{dt}=2\int_0^t d\tau C_\mu(t,\tau)=2\int_0^t d\tau C_\mu(\tau,t)=2\int_0^t d\tau C_\mu^\mr{ss}(t-\tau)=2\int_0^t d\tau C_\mu^\mr{ss}(\tau)
}
and thus
\aln{\label{noiseD}
D(t)=\sum_b\mu_b^2\braket{\hat{L}_b^\dag\hat{L}_b}_\mr{ss}+2\int_0^t d\tau(\Tr[\ml{G}_\mu\circ e^{\ml{L}|\tau|}\circ\ml{G}_\mu[\hat{\rho}_\mr{ss}]]-\Tr[\ml{G}_\mu[\hat{\rho}_\mr{ss}]]^2).
}
If $t$ is large, we often encounter that $D(t)$ rapidly converges to a constant $D(\infty)$. In that case, we have
\aln{
\mbb{V}[N_\mu(t)]\sim D(\infty)t.
}

We can see that both $\mbb{E}[N_\mu(t)]$ and $\mbb{V}[N_\mu(t)]$ increase linearly with $t$.
Notably, however, the variance $\mbb{V}[N_\mu(t)]$ depends on how the jumps occurring at different times are correlated, as highlighted in the second term on the right-hand side of Eq.~\eqref{noiseD}.
Note that, in general, the higher-order fluctuations of $N_\mu(t)$ can be captured, e.g., by its large deviation~\cite{jakvsic2014entropic, van2015sanov}, which we do not explain here. 
See Refs.~\cite{garrahan2018aspects, jack2020ergodicity, landi2024current, touchette2009large} for a review.

\subsubsection{Examples of the statistics of quantum jumps}
As an example, let us consider the spontaneous emission from a single-atom system, whose jump operator is given by Eq.~\eqref{singlejump} (see Fig.~\ref{fig3}).
We assume that the atom's Hamiltonian is simply given by 
\aln{
\hat{H}_S=\frac{g}{2}\hat{\sigma}^x.
}
Furthermore, we assume that the state is prepared in a stationary state $\hat{\rho}_\mr{ss}$.
Solving $\ml{L}[\hat{\rho}_\mr{ss}]$, we find that the stationary state is given by
\aln{
\hat{\rho}_\mr{ss}=\frac{1}{s^2+2}\pmat{1 & -is\\ is & 1+s^2},
}
where $s=\gamma/g$.

We consider $\mu=1$ and ${I}(t)=\frac{dN(t)}{dt}$. Then, its average reads
\aln{
J^\mr{ss}=\Tr[\ml{G}_\mu[\hat{\rho}_\mr{ss}]]=\gamma\Tr[\hat{\sigma}^+\hat{\sigma}^-\hat{\rho}_\mr{ss}]=\frac{\gamma g^2}{\gamma^2+2g^2}.
}
Next, to gain physical insight into the correlation $C^\mr{ss}(\tau)$, we assume a small $\tau\:(>0)$ regime.
Then, we find
\aln{
C^\mr{ss}(\tau)&\simeq\Tr[\ml{G}_\mu^2[\hat{\rho}_\mr{ss}]]+\tau\Tr[\ml{G}_\mu\ml{L}\ml{G}_\mu[\hat{\rho}_\mr{ss}]]+\frac{1}{2}\tau^2\Tr[\ml{G}_\mu\ml{L}^2\ml{G}_\mu[\hat{\rho}_\mr{ss}]]-(J^\mr{ss})^2\nonumber\\
&=\frac{g^4\gamma^2}{(\gamma^2+2g^2)^2}\lrm{\frac{(\gamma^2+2g^2)\tau^2}{4}-1},
}
where we have omitted $\circ$ between the maps for simplicity.
Importantly, this result shows that $C^\mr{ss}(\tau)<0$ for small $\tau$, indicating the anti-correlation of the two jumps, i.e., spontaneous emissions.
In fact, in the present setup, the detection of spontaneous emission means that the atom is suddenly changed to the ground state $\ket{0}$. 
Since it takes some time for the atom's state to have an overlap with the excited state $\ket{1}$, the second spontaneous emission is suppressed, which explains the anti-correlation of the jumps.

As another example, let us consider a quantum many-body system of hardcore bosons on a one-dimensional lattice, whose Hamiltonian $\hat{H}_\mr{HCB}$ is  given in Eq.~\eqref{hcboson},
with the position measurement $\hat{L}_l=\sqrt{\gamma}\hat{n}_l$ at each site.
Here, we assume that the system size $V$ is even and set the conserved number of particles $M$ to $M=V/2$ (i.e., half-filling).

In this setup, let us first consider the total number of jumps in the whole system, which is given by
\aln{
N_\mr{total}(t)=\sum_{l=1}^V N_l(t)
}
with $\mu_l=1$ for all $l$.
If we consider the maximally mixed state $\hat{\rho}_\mr{ss}\propto \mbb{\hat{I}}$ in the subspace of $\mc{H}$ whose particle number is $M$, using $\hat{n}_l^2=\hat{n}_l$ for the hardcore bosons, we immediately find 
\aln{
\mbb{E}[N_\mr{total}(t)]=t\gamma\sum_{l=1}^V\braket{\hat{n}_l}_\mr{ss}=\frac{\gamma Vt}{2}
}
and
\aln{
\mbb{V}[N_\mr{total}(t)]=\frac{\gamma Vt}{2},
}
where we can confirm that the second term in Eq.~\eqref{noiseD} vanishes\footnote{
To see this explicitly, we use the fact that $\sum_{l=1}^V\hat{n}_l=M$, which is regarded as a c-number. 
Then, 
\aln{
\Tr[\ml{G}_\mu\circ e^{\ml{L}|\tau|}\circ\ml{G}_\mu[\hat{\rho}_\mr{ss}]]&=\gamma
\Tr\lrl{\sum_{l=1}^V\hat{n}_l^2\lrs{ e^{\ml{L}|\tau|}\circ\ml{G}_\mu[\hat{\rho}_\mr{ss}]}}\nonumber\\
&=
\gamma M\Tr\lrl{{ e^{\ml{L}|\tau|}\circ\ml{G}_\mu[\hat{\rho}_\mr{ss}]}}=\gamma M\Tr\lrl{{\ml{G}_\mu[\hat{\rho}_\mr{ss}]}}=\gamma^2M^2\Tr\lrl{{\hat{\rho}_\mr{ss}}}=\gamma^2M^2,
}
and, similarly, we have $\Tr[\ml{G}_\mu[\hat{\rho}_\mr{ss}]]^2=\gamma^2M^2$.
}.
This reflects the fact that the statistics of quantum jumps occurring in the total system is Poissonian, i.e., $\mbb{V}[N_\mr{total}(t)]/\mbb{E}[N_\mr{total}(t)]=1$.

However, a more nontrivial feature appears if we instead consider the quantum jumps occurring in the subsystem, say
\aln{
N_\mr{half}(t)=\sum_{l=1}^{V/2} N_l(t),
}
with $\mu_l=1\:(1\leq l\leq V/2)$ and $\mu_l=0\:(V/2+1\leq l\leq V)$.
In fact, Ref.~\cite{Yamamoto25} numerically found that the variance of $N_\mr{half}(t)$ exhibits an anomalous scaling
\aln{
\mbb{V}[N_\mr{half}(t)]=\frac{\gamma Vt}{4}+\gamma^2 c(\gamma)V^\alpha t,
}
while $\mbb{E}[N_\mr{half}(t)]=\gamma Vt/4$ as expected from particle number conservation.
Here, $c(\gamma)$ is an increasing function of $\gamma$ and $\alpha\simeq 2.7$ is a universal exponent independent of $U$.
This means that for sufficiently small $\gamma$, which should be $o(V^0)$, the variance behaves as $\mbb{V}[N_\mr{half}(t)]\simeq \gamma Vt/4$, corresponding to the Poissonian statistics $\mbb{V}[N_\mr{half}(t)]/\mbb{E}[N_\mr{half}(t)]=1$.
In stark contrast, for $\gamma=\ml{O}(V^0)$, the variance behaves as a super-Poissonian statistics $\mbb{V}[N_\mr{half}(t)]\simeq \gamma^2 c(\gamma)V^\alpha t \gg \mbb{E}[N_\mr{half}(t)]$.
As seen from this example, the statistics of jumps in many-body systems~\cite{ates2012dynamical, buvca2014exactly, vznidarivc2014anomalous, vznidarivc2014large, kewming2022diverging, matsumoto2025dissipative, cech2024space} can be nontrivial.

\subsection{Advantage in numerical simulations of open systems}
\subsubsection{Overview}
So far, we have focused on the physical meaning of quantum trajectories in the context of quantum systems under measurements.
Here, we discuss that quantum trajectories are also useful for the numerical simulations of open quantum systems; this was, in fact, one of the motivations for the early works on this concept~\cite{dalibard1992wave, dum1992monte}.
Our aim is to simulate the GKSL equation
\aln{\label{gksl2}
\frac{d\hat{\rho}}{dt}=-i[\hat{H},\hat{\rho}]+\sum_{b\geq 1}\hat{L}_b\hat{\rho}\hat{L}_b^\dag-\frac{1}{2}\{\hat{\rho},\hat{L}_b^\dag\hat{L}_b\}
}
in an efficient way, where we omit the subscript $\mr{S}$ for simplicity in the following.

If we attempt to directly compute the GKSL dynamics without approximation, we need to store $\sim d^2$ elements of $\hat{\rho}$ in the memory of a computer, where $d=\mr{dim}[\ml{H}]$ is the dimension of the Hilbert space.
Moreover, the computational time for the right-hand side is $\sim d^3$.
If we consider many-body systems, where $d$ grows exponentially with respect to the system size $V$, these scalings become problematic.

In contrast, we have seen that a GKSL equation can be obtained from the ensemble average of stochastic quantum trajectories.
Therefore, it can be simulated by sampling a large number of quantum trajectories that reproduce the GKSL equation on average.
This is called the unraveling of the GKSL equation.

Specifically, given Eq.~\eqref{gksl2}, we can instead simulate the dynamics of pure states,
\aln{\label{SSEqt}
d\ket{\psi_c}=\lrs{-i\hat{H}_\mr{eff}+\frac{1}{2}\sum_{b\geq 1}\braket{\hat{L}_b^\dag\hat{L}_b}_c}\ket{\psi_c}dt+\sum_{b\geq 1}\lrs{\frac{\hat{L}_b}{\sqrt{\braket{\hat{L}_b^\dag\hat{L}_b}_c}}-1}\ket{\psi_c}dN_b
}
and take the ensemble average, $\hat{\rho}=\mbb{E}[\ket{\psi_c}\bra{\psi_c}]$.
For this case, we need to store only $\sim d$ elements of $\ket{\psi_c}$ to the memory.
The computational time is $\sim d^2$ for a single trajectory, and the net cost is $\sim d^2 N_\mr{sample}$, where $N_\mr{sample}$ is the number of samples required to obtain $\hat{\rho}$ with a sufficient accuracy.
Therefore, compared with the direct simulation, the quantum-trajectory method has a memory advantage and, if $N_\mr{sample}\ll d$ (which is often the case), an advantage of computational time\footnote{
In practice, we can sample quantum trajectories by direct parallelization, which provides an additional advantage of computational time.
}.

\subsubsection{Concrete method}
Here, we demonstrate a concrete method to perform the simulation outlined above (see also Refs.~\cite{wiseman2009quantum, landi2024current}).
As a preliminary, we notice the following method to efficiently sample a random variable $X$, which takes $X_b\:(0\leq b\leq B)$ with probability $p_b$.
\begin{enumerate}
\item 
Define
\aln{
Q_b=\sum_{b'=0}^b p_{b'},\quad Q_{-1}=0,
}
where $Q_B=1$ from normalization.

\item 
Sample $R$ randomly from the box distribution in the range $R\in [0,1]$.

\item Find $b$ such that 
\aln{
Q_{b-1}<R\leq Q_b.
}

\item 
Output $X_b$.
\end{enumerate}
Why this method works is visually understood from Fig.~\ref{fig4}.

\begin{figure}[!h]
\centering\includegraphics[width=\linewidth]{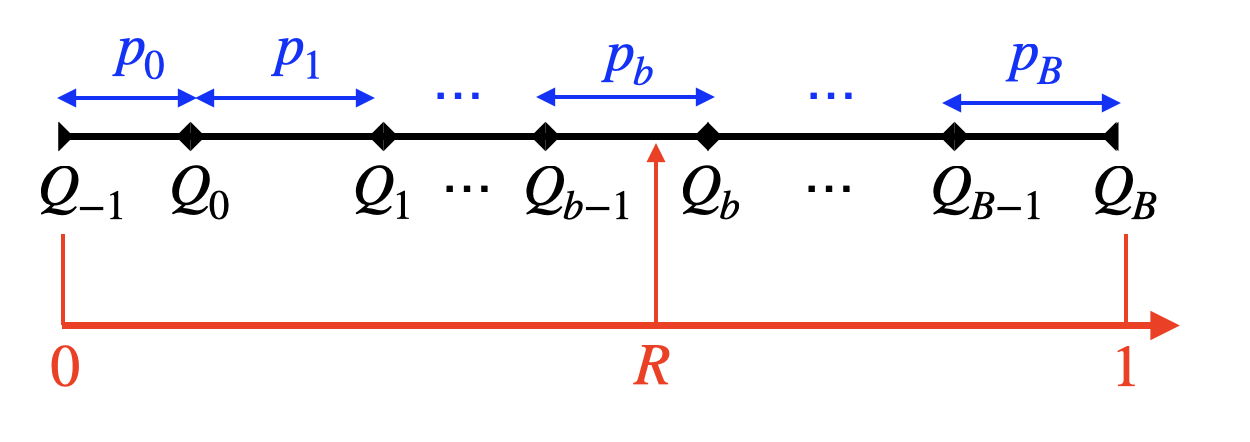}
\caption{
Sampling of a random variable $X$, which takes $X_b\;(0\leq b\leq B)$, with probability $p_b$. 
Considering the cumulative distribution $Q_b$ and choosing $R$ randomly from $R\in [0,1]$, we can sample $X_b$ from the condition $Q_{b-1}<R\leq Q_b$.
}
\label{fig4}
\end{figure}

The continuous-variable version is also understood as follows.
Namely, if we want to sample a random variable $X$, which takes $X_\tau\:(\tau \geq 0)$ with probability distribution $q_\tau$, the following method is efficient.
\begin{enumerate}
\item
Define
\aln{
Q_\tau=\int_0^\tau d\tau' q_{\tau'},
}
where $Q_\infty=1$ from normalization.

\item 
Sample $R$ randomly from the box distribution in the range $R\in [0,1]$.

\item Find $\tau$ such that 
\aln{
Q_\tau =R.
}

\item 
Output $X_\tau$.
\end{enumerate}

Let us now move on to the simulation of Eq.~\eqref{SSEqt}.
A naive method is as follows. 
At each time $t$, we can discretize the dynamics with a small interval $\delta t$.
We have $B+1$ types of $b$, where $p_b(t)=\delta t \braket{\hat{L}_b^\dag\hat{L}_b}_{c;t}$ for $1\leq b\leq B$ and 
$p_0(t)=1-\sum_{b=1}^Bp_b(t)$.
Then, using the above method, we can sample $b$  using a random variable $R(t)$.
According to the sampled $b$, we can update $\ket{\psi_c(t)}$ as 
\aln{
\ket{\psi_c(t+\delta t)}=\frac{\hat{L}_b\ket{\psi_c(t)}}{\|\hat{L}_b\ket{\psi_c(t)}\|}
}
for $1\leq b\leq B$ and
\aln{
\ket{\psi_c(t+\delta t)}=\ket{\psi_c(t)}+\lrs{-i\hat{H}_\mr{eff}+\frac{1}{2}\sum_{b\geq 1}\braket{\hat{L}_b^\dag\hat{L}_b}_{c;t}}\ket{\psi_c(t)}dt
}
for $b=0$. 
However, this method is not efficient, because if $\delta t$ is small, $p_{b\neq 0}(t)$  becomes very small.

\subsubsection{Efficient method}
A more efficient alternative method, which avoids the time discretization, is as follows.
We first determine (I) the jump time and then (II) the type of the jump.

To determine the jump time in (I), we notice that the probability such that no jump occurs during $[t,t+\tau]$ for the state $\ket{\psi_c(t)}$ is given by\footnote{
To see this, we note that $P_0(0)=1$ and 
\aln{
P_0(\tau+\delta \tau)=\lrs{1-\sum_{b\geq 1}\frac{\braket{\psi_c(t)|e^{i\hat{H}_\mr{eff}^\dag\tau}\hat{L}_b^\dag\hat{L}_be^{-i\hat{H}_\mr{eff}\tau}|\psi_c(t)}}{\|e^{-i\hat{H}_\mr{eff}\tau}\ket{\psi_c(t)}\|^2}\delta\tau}P_0(\tau),
}
so that
\aln{
\frac{d\ln P_0(\tau)}{d\tau}=-\sum_{b\geq 1}\frac{\braket{\psi_c(t)|e^{i\hat{H}_\mr{eff}^\dag\tau}\hat{L}_b^\dag\hat{L}_be^{-i\hat{H}_\mr{eff}\tau}|\psi_c(t)}}{\|e^{-i\hat{H}_\mr{eff}\tau}\ket{\psi_c(t)}\|^2}=\frac{d}{d\tau}\ln \|e^{-i\hat{H}_\mr{eff}\tau}\ket{\psi_c(t)}\|^2.
}
}
\aln{
P_{0}(\tau)=\|e^{-i\hat{H}_\mr{eff}\tau}\ket{\psi_c(t)}\|^2={\braket{\psi_c(t)|e^{i\hat{H}_\mr{eff}^\dag\tau}e^{-i\hat{H}_\mr{eff}\tau}|\psi_c(t)}}.
}
Now, consider a probability distribution $q_\tau$ with respect to the jump time $\tau\:(\geq 0)$.
Then, by definition, we have
\aln{
Q_\tau:=\int_0^\tau q_{\tau'}d\tau'=1-P_0(\tau).
}
For this distribution, we can employ the sampling method in the previous subsection.
That is, we determine $\tau$ from the relation
\aln{\label{QPR}
Q_\tau=1-P_0(\tau)=R,
}
where $R$ is sampled uniformly from $[0,1]$.
Once $\tau$ is determined by Eq.~\eqref{QPR}, we can first evolve the state as
\aln{
\ket{\psi_c(t+\tau)}=\frac{e^{-i\hat{H}_\mr{eff}\tau}\ket{\psi_c(t)}}{\|e^{-i\hat{H}_\mr{eff}\tau}\ket{\psi_c(t)}\|}.
}

Next, we let a jump occur. 
The type of the jump in (II) is determined by the conditional probability
\aln{
p_b'=\frac{p_b}{1-p_0}=\frac{\braket{\psi_c(t+\tau)|\hat{L}_b^\dag\hat{L}_b|\psi_c(t+\tau)}}{\sum_{b=1}^B\braket{\psi_c(t+\tau)|\hat{L}_b^\dag\hat{L}_b|\psi_c(t+\tau)}}.
}
To sample $b\:(1\leq b\leq B)$ from this probability distribution, we define $Q_b'=\sum_{b'=1}^bp_{b'}'$ with $Q_0'=0$ and use the above method.
Namely, we generate another random variable $R'$ uniformly chosen from $[0,1]$, and, if $Q_{b-1}'<R'\leq Q_b'$, the jump of type $b$ occurs.
Then, we update the state as
\aln{
\ket{\psi_c(t+\tau)}\ra \frac{\hat{L}_b\ket{\psi_c(t+\tau)}}{\|\hat{L}_b\ket{\psi_c(t+\tau)}\|}.
}
Iterating (I) and (II) produces a quantum trajectory.

\subsection{Quantum diffusion}
\label{sec:Quantum diffusion}
Before ending this section, we mention a different formulation of quantum trajectories in terms of quantum diffusion.
Since there are nice reviews~\cite{wiseman2009quantum, landi2024current} on this topic, we here briefly introduce the basic ideas, along with the governing equation, which will appear later in this review.

We first note that the GKSL equation in Eq.~\eqref{gksl2} is invariant under the following transformations:
\aln{
\begin{split}
\hat{L}_b &\longrightarrow \hat{L}_b'=\hat{L}_b+\alpha_b, \\
\hat{H}_\mr{S} &\longrightarrow \hat{H}_\mr{S}'=\hat{H}_\mr{S}-\frac{i}{2}{\sum_b}(\alpha_b^*\hat{L}_b-\alpha_b\hat{L}_b^\dag),
\end{split}
}
where $\alpha_b=|\alpha_b|e^{i\theta_b}\in\mbb{C}$.
In contrast, the above transformation changes the properties of quantum trajectories.
For example, if we consider the current $I'_\mu(t)$ for quantum trajectories characterized by $\hat{H}_\mr{S}'$ and $\hat{L}_b'$, its average reads
\aln{
J_\mu'(t)=\mathbb{E}[I'_\mu(t)]=J_\mu(t)+\sum_b\mu_b|\alpha_b|\braket{\hat{X}_b}_t+\sum_b\mu_b|\alpha_b|^2\neq J_\mu(t),
}
where
\aln{
\hat{X}_b=e^{-i\theta_b}\hat{L}_b+e^{i\theta_b}\hat{L}_b^\dag
}
and $J_\mu(t)$ is given in Eq.~\eqref{aveInu}.
This fact indicates that there are different unravelings of the GKSL equation, and different unravelings lead to distinct physics at the level of quantum trajectory.

In particular, we are interested in the case where $|\alpha_b|$ is large.
In this case, $J_\mu'(t)$ becomes large, meaning that there are frequent jumps.
It also turns out that the measurement back-action due to each jump becomes small for large $|\alpha_b|$.
Consequently, we can take an appropriate limit where many jumps with small back-action occur.
Specifically, if we count the number of jumps within a small but coarse-grained time, its distribution approximately obeys a Gaussian distribution due to the central limit theorem.

With the above considerations, we obtain the following stochastic equation (see Ref.~\cite{pellegrini2009diffusion} for a rigorous treatment):
\aln{\label{qdfeq}
d\hat{\rho}_c=\lrs{-i[\hat{H}_\mr{S},\hat{\rho}_c]+\sum_b\hat{L}_b\hat{\rho}_c\hat{L}_b^\dag-\frac{1}{2}\{\hat{L}_b^\dag\hat{L}_b,\hat{\rho}_c\}}dt
+\sum_b\lrl{(\hat{L}_b e^{-i\theta_b}\hat{\rho}_c+\mr{h.c.})-\braket{\hat{X}_b}_c\hat{\rho}_c}\cdot dW_b,
}
where $dW_b$ is a Wiener process that satisfies
\aln{
dW_b dW_{b'}=\delta_{bb'}dt, \quad\mbb{E}[dW_b]=0,
}
and ``$\cdot$" before $dW_b$ indicates that we consider the It\^o integral\footnote{Recently, there are proposals to reformulate the dynamics of quantum trajectories in terms of stochastic Hamiltonian (Liouville) dynamics, rather than only using the stochastic master equations~\cite{l18s-9vmh,villanueva2025hamiltonian}.}.
We note that this type of equation is relevant for, e.g., homodyne and heterodyne detections of photons, where we introduce a local oscillator that provides a large number of photons (large $|\alpha|$) to the system.
See, e.g., Ref.~\cite{wiseman2009quantum} for further details.

\section{Spectral properties of CPTP maps and quantum master equations}\label{sec:CPTPspectra}

As discussed in the previous chapters, the dynamics of quantum systems subject to measurements with the outcomes discarded, dissipation and/or decoherence are generally described by CPTP maps.
We are often interested in the infinitely long-time behavior of such systems, which is governed by the steady states of CPTP maps. 
After an overview of the general spectral properties of CPTP maps in Sec.~\ref{sec:GeneralSpecCPTP}, we introduce the two conditions for the existence of a unique steady state, irreducibility and primitivity, in Secs.~\ref{sec:irreducibility} and~\ref{sec:primitivity}, respectively.
In Sec.~\ref{sec:SteadyStateGKSL}, we focus on quantum master equations generated by GKSL superoperators and discuss several conditions known in the literature for the existence of a unique steady state.
Section~\ref{sec:miscellaneous} is devoted to selected topics related to the relaxation dynamics towards steady states and the spectral statistics of random CPTP or GKSL dynamics.

\subsection{General spectral properties of CPTP maps}
\label{sec:GeneralSpecCPTP}

We first focus on the spectral properties of general CPTP maps $\mc{E}: \mbb{B}[\mc{H}] \to \mbb{B}[\mc{H}]$\footnote{
In some literature, $\mc{E}$ is defined as a CPTP map on the set of \emph{trace-class} operators $\hat{T}$ on $\mc{H}$, denoted by $\mbb{T}[\mc{H}]$. 
A trace-class operator $\hat{T}$ is a bounded operator whose trace norm $\| \hat{T} \|_1 := \Tr [\sqrt{\hat{T}^\dagger \hat{T}}]$ is finite. 
Its dual map $\mc{E}^\dagger$ [see Eq.~\eqref{eq:DualMap}] then acts on the set of bounded operators $\hat{X}$ on $\mc{H}$, $\mbb{B}[\mc{H}]$. 
The distinction between $\mbb{T}[\mc{H}]$ and $\mbb{B}[\mc{H}]$ is crucial for ensuring the finiteness of $\Tr[\hat{T} \hat{X}]$ when $\mc{H}$ is infinite-dimensional.
However, when $\mc{H}$ is restricted to be finite-dimensional, there is no need to worry about this distinction. 
This is because any bounded operator $\hat{X} \in \mbb{B}[\mc{H}]$ on $d$-dimensional $\mc{H}$ has a finite trace norm, 
\begin{align*}
\| \hat{X} \|_1 = \sum_{i=1}^d \Lambda_i \leq d \Lambda_1 = d\| \hat{X} \| < \infty.
\end{align*}
Here, $\Lambda_i$ are the singular values of $\hat{X}$ arranged in descending order $\Lambda_1 \geq \cdots \geq \Lambda_d$ and $\| \hat{X} \| := \sup_{\|\ket{\phi}\|=1} \| \hat{X} |\phi \rangle \|$ is the operator norm of $\hat{X}$, which is equal to the largest singular value $\Lambda_1$ of $\hat{X}$. 
This yields $\mbb{B}[\mc{H}] \subseteq \mbb{T}[\mc{H}]$. 
On the other hand, $\mbb{T}[\mc{H}] \subseteq \mbb{B}[\mc{H}]$ by definition, and hence $\mbb{T}[\mc{H}] = \mbb{B}[\mc{H}]$.
}
~\cite{wolf2012quantum, bengtsson2020geometry}.
Let us consider the eigenvalue equation of $\mc{E}$, 
\begin{align} \label{eq:CPTPEigen}
\mc{E}[\hat{X}] = z \hat{X},
\end{align}
where $\hat{X} \in \mbb{B}[\mc{H}]$ is a (right) eigenvector\footnote{
Since $\hat{X}$ is an operator acting on $\mc{H}$, one may prefer to call it an \emph{eigenoperator}.
} of $\mc{E}$ and $z$ is the associated eigenvalue. 
Although the eigenvalue $z$ is in general complex, its distribution is highly restricted due to the complete positivity and the trace-preserving property of the map $\mc{E}$. 
In fact, we can prove that (i) the complex eigenvalues appear in conjugate pairs $z$ and $z^*$, (ii) there exists an eigenvalue $z=1$, and (iii) any eigenvalue lies within the unit disk of the complex plane, $|z| \leq 1$. 
We can also prove the existence of a positive definite eigenvector for the eigenvalue $z=1$, as we will discuss below.

Property (i) can be easily proven.
Since the CPTP map $\mc{E}$ is Hermiticity preserving,
\begin{align}
\mc{E}[\hat{X}^\dagger] = (\mc{E}[\hat{X}])^\dagger,
\end{align}
if $\hat{X}$ satisfies the eigenvalue equation in Eq.~\eqref{eq:CPTPEigen}, we have 
\begin{align}
\mc{E}[\hat{X}^\dagger] = z^* \hat{X}^\dagger.
\end{align}
Thus, if $z$ is an eigenvalue of $\mc{E}$, its complex conjugate $z^*$ must also be an eigenvalue.

In order to prove property (ii), we introduce the dual map $\mc{E}^\dagger: \mbb{B}[\mc{H}] \to \mbb{B}[\mc{H}]$ defined by $\Tr[\hat{X} \mc{E}[\hat{Y}]] = \Tr[\mc{E}^\dagger[\hat{X}] \hat{Y}]$ for all $\hat{X}, \hat{Y} \in \mbb{B}[\mc{H}]$. 
In the Kraus representation, $\mc{E}^\dagger$ is given by
\begin{align} \label{eq:DualMap}
\mc{E}^\dagger[\hat{\rho}] = \sum_b \hat{M}_b^\dagger \hat{\rho} \hat{M}_b.
\end{align}
Compared with the Kraus representation of $\mc{E}$ in Eq.~\eqref{eq:KrausRep}, $\hat{M}_b$ and $\hat{M}_b^\dagger$ are interchanged. 
The dual map $\mc{E}^\dagger$ gives the time evolution of operators in the Heisenberg picture, whereas the CPTP map $\mc{E}$ gives the time evolution of states (density operators) in the Schr\"{o}dinger picture.
One can easily show that the dual map $\mc{E}^\dagger$ has the same eigenvalues as those of $\mc{E}$; in the natural representation, the map $\mc{E}$ and its dual $\mc{E}^\dagger$ can be expressed as
\begin{align} \label{eq:CPTPNaturalRep}
\hat{\mc{E}} = \sum_b \hat{M}_b \otimes \hat{M}_b^*, \quad 
\hat{\mc{E}}^\dagger = \sum_b \hat{M}_b^\dagger \otimes \hat{M}_b^{\sf{T}},
\end{align}
which are related to each other by the transposition of each $\hat{M}_b$, followed by a swap operation $\ket{i} \ket{j} \mapsto \ket{j} \ket{i}$ that is unitary. 
Thus, the eigenvalues of $\mc{E}$ and $\mc{E}^\dagger$ coincide.
Indeed, a right eigenvector of $\mc{E}^\dagger$ is a left eigenvector of $\mc{E}$ and vice versa.
The TP condition for $\mc{E}$ in Eq.~\eqref{eq:TPKraus} implies that the dual map $\mc{E}^\dagger$ is unital, i.e., $\mc{E}^\dagger[\hat{\mbb{I}}] = \hat{\mbb{I}}$. 
Since the dual map $\mc{E}^\dagger$ has an eigenvector $\hat{\mbb{I}}$ with eigenvalue $z=1$, the CPTP map $\mc{E}$ also has the eigenvalue $z=1$. 
This proves property (ii)\footnote{
We note that property (ii) can also be shown by Brouwer's fixed-point theorem, whose detailed statement can be found in Appendix~\ref{app:StationaryState}.
}.

In order to prove property (iii), we first use the inequality for the operator norm 
\begin{align}
\lVert \mc{T}[\hat{X}] \rVert \leq \lVert \mc{T} \rVert \lVert \hat{X} \rVert, \quad
\lVert \mc{T} \rVert := \sup_{\lVert \hat{X} \rVert = 1} \lVert \mc{T}[\hat{X}] \rVert,
\end{align}
which holds for any linear map $\mc{T}: \mbb{B}[\mc{H}] \to \mbb{B}[\mc{H}]$ and any $\hat{X} \in \mbb{B}[\mc{H}]$.
Substituting the eigenvalue equation $\mc{T}[\hat{X}] = z \hat{X}$, we find
\begin{align}
\lVert \mc{T}[\hat{X}] \rVert = |z| \lVert \hat{X} \rVert \leq \lVert \mc{T} \rVert \lVert \hat{X} \rVert
\end{align}
and thus
\begin{align}
|z| \leq \lVert \mc{T} \rVert.
\end{align}
By the Russo-Dye theorem, for which we give a proof in Appendix~\ref{app:RussoDye}, any positive unital map $\mc{T}$ satisfies $\lVert \mc{T} \rVert = 1$ and hence $|z| \leq 1$.
Since the dual map $\mc{E}^\dagger$ is a CP unital map, its eigenvalues satisfy $|z| \leq 1$. 
Hence, the eigenvalues of the CPTP map $\mc{E}$ also satisfy $|z| \leq 1$.

If a CPTP map $\mc{E}$ is diagonalizable, there exist biorthogonal bases $\{ \hat{\mathsf{R}}_a \}_{a=1}^{d^2}$ and $\{ \hat{\mathsf{L}}_a \}_{a=1}^{d^2}$ such that
\begin{align}\label{eigenvalueCPTP}
\mc{E}[\hat{\mathsf{R}}_a] = z_a \hat{\mathsf{R}}_a, \quad
\mc{E}^\dagger[\hat{\mathsf{L}}_a] = z_a^* \hat{\mathsf{L}}_a, \quad
\Tr[\hat{\mathsf{L}}_a^\dagger \hat{\mathsf{R}}_b] \propto \delta_{ab}.
\end{align}
Here, $\hat{\mathsf{R}}_a$ ($\hat{\mathsf{L}}_a$) is a right (left) eigenvector of $\mc{E}$ with eigenvalue $z_a$.
Then, the CPTP map $\mc{E}$ can be diagonalized as 
\begin{align}
\mc{E}[\hat{X}] = \sum_{a=1}^{d^2} z_a \frac{\Tr [\hat{\mathsf{L}}_a^\dagger \hat{X}]}{\Tr [\hat{\mathsf{L}}_a^\dagger \hat{\mathsf{R}}_a]} \hat{\mathsf{R}}_a.
\end{align}

Although a CPTP map $\mc{E}$ is generally not diagonalizable, we can still prove the triviality of Jordan blocks corresponding to the eigenvalues $z$ with unit modulus $|z|=1$.
Let us consider $\mc{E}$ as a $d^2 \times d^2$ matrix $\hat{\mc{E}}$, as given by Eq.~\eqref{eq:CPTPNaturalRep} in the natural representation. 
Then, there exists a $d^2 \times d^2$ invertible matrix $\hat{\mc{X}}$ that brings $\hat{\mc{E}}$ into the Jordan normal form, 
\begin{align} \label{eq:CPTPJordanForm}
\hat{\mc{X}}^{-1} \hat{\mc{E}} \hat{\mc{X}} = \bigoplus_{a=1}^{N_J} J_{D_a}(z_a),
\end{align}
where $\sum_{a=1}^{N_J} D_a = d^2$ and $J_{D_a}(z)$ is a $D_a \times D_a$ upper triangular matrix (Jordan block) of the form
\begin{align}
J_1(z) = z, \ 
J_2(z) = \begin{pmatrix} z & 1 \\ 0 & z \end{pmatrix}, \ 
J_3(z) = \begin{pmatrix} z & 1 & 0 \\ 0 & z & 1 \\ 0 & 0 & z \end{pmatrix}, \ \cdots.
\end{align}
Since $\| \mc{E} \| = 1$ for any CPTP map $\mc{E}$ by the Russo-Dye theorem (see Appendix~\ref{app:RussoDye}), we have $\| \mc{E}[\hat{X}] \| \leq \| \mc{E} \| \| \hat{X} \| = \| \hat{X} \|$ for any $\hat{X} \in \mbb{B}[\mc{H}]$.
The same applies to $\mc{E}^n$, implying $\| \mc{E}^n[\hat{X}] \| \leq \| \hat{X} \|$ for any $n \in \mbb{N}$.
However, $\hat{\mc{X}}^{-1} \hat{\mc{E}}^n \hat{\mc{X}}$ contains the $n$th power of the Jordan blocks $J^n_{D_a}(z_a)$, whose first row reads $[J^n_{D_a}(z_a)]_{1j} = \binom{n}{j-1} z_a^{n-j+1}$ for $j=1,\ldots,D_a$.
These components diverge as $n \to \infty$ when $|z_a|=1$ and $D_a \geq 2$, which contradicts the boundedness of $\mc{E}^n$.
Hence, the Jordan blocks for the eigenvalues $z$ with unit modulus $|z|=1$ must be one-dimensional. 

Property (ii) guarantees the existence of an eigenvector $\hat{X}_0 \in \mbb{B}[\mc{H}]$ corresponding to the eigenvalue $z=1$.
Furthermore, as detailed in Appendix~\ref{app:StationaryState}, the eigenvector $\hat{X}_0$ can be chosen to be a density operator $\hat{\varrho}_0 \succeq 0$. 
Since $\mc{E}[\hat{\varrho}_0]=\hat{\varrho}_0$, it does not change under repeated applications of the CPTP map $\mc{E}$ and thus constitutes a stationary state.
When $\mc{E}$ admits a Kraus decomposition in terms of Hermitian operators $\hat{M}_b^\dagger = \hat{M}_b$, there is an immediate application of this property; since the CPTP map $\mc{E}$ is unital, namely $\mc{E}[\hat{\mbb{I}}] = \hat{\mbb{I}}$, the maximally mixed state $\hat{\varrho}_0 = \hat{\mbb{I}}/d$ becomes a stationary state.
However, the general spectral properties are not enough to ensure the existence of a unique steady state or the ergodicity of physical quantities for quantum trajectories after a sufficiently long time (see Chapters~\ref{sec:linear-quantity_purification} and~\ref{sec:nonlear-quantity_Lyapunov-spectrum} for the ergodicity of quantum trajectories).
In general, the eigenvalue $z=1$ can be degenerate and there can be multiple stationary states. 
Since any linear combination of the stationary states is also a stationary state, there exist infinitely many steady states and the resulting asymptotic dynamics depends on the initial state.
In the following two sections, we discuss the conditions required for the CPTP map $\mc{E}$ to have a nondegenerate eigenvalue $z=1$.

\subsection{Irreducibility of CPTP maps}
\label{sec:irreducibility}

Here we introduce the \emph{irreducibility} condition for a CPTP map $\mc{E}$, which ensures that $\mc{E}$ has a nondegenerate eigenvalue $z=1$.
There are several equivalent definitions of irreducibility~\cite{evans1978spectral, farenick1996irreducible, schrader2000perron, wolf2012quantum, Burgarth2013ergodic, carbone2016open, carbone2016irreducible}, some of which are summarized below. 
When a CPTP map $\mc{E}$ satisfies one of the following equivalent properties, $\mc{E}$ is said to be irreducible\footnote{
In some literature, these conditions are also called Davies-irreducible or ergodic~\cite{carbone2016open, carbone2016irreducible, zhang2024criteria}.
}.
\begin{proposition}[Irreducibility] \label{prop:irreducibility}
Let $\mc{E}: \mbb{B}[\mc{H}] \to \mbb{B}[\mc{H}]$ be a CPTP map. 
The following statements are equivalent.
\begin{enumerate}
\item There exists no orthogonal projection operator $\hat{P} \in \mbb{B}[\mc{H}]$ such that $\hat{P} \notin \{ 0, \hat{\mbb{I}} \}$ and $\mc{E}[\hat{P} \mbb{B}[\mc{H}] \hat{P}] \subseteq \hat{P} \mbb{B}[\mc{H}] \hat{P}$.
\item There exists no orthogonal projection operator $\hat{P} \in \mbb{B}[\mc{H}]$ such that $\hat{P} \notin \{ 0, \hat{\mbb{I}} \}$ and $\mc{E}^\dagger[\hat{P}] \succeq \hat{P}$.
\item For any nonzero positive semidefinite operator $\hat{\rho} \in \mbb{B}[\mc{H}]$ and for any $s > 0$, we have $\exp(s \mc{E})[\hat{\rho}] \succ 0$.
\end{enumerate}
\end{proposition}

\begin{proof}
See Appendix~\ref{app:ProofIrreducibility}.
\end{proof}

The irreducibility of the CPTP map $\mc{E}$ can also be characterized by its Kraus representation in Eq.~\eqref{eq:KrausRep}. 
Let $\mbb{K} \subseteq \mbb{B}[\mc{H}]$ be the complex linear span of all monomials of the Kraus operators $\hat{M}_b$,
\begin{align}
\mbb{K} := \textrm{span} \left( \{ \hat{\mbb{I}} \} \cup \bigcup_{n=1}^\infty \{ \hat{M}_{b_n} \cdots \hat{M}_{b_2} \hat{M}_{b_1} \} \right).
\end{align}
In other words, $\mbb{K}$ is the algebra generated by $\hat{M}_b$ and $\hat{\mbb{I}}$; here an algebra means a subset of $\mbb{B}[\mc{H}]$ closed under scalar multiplication, addition, and multiplication. 
We then have the following results~\cite{farenick1996irreducible, schrader2000perron, jaksic2014entropic, carbone2016open}.
\begin{theorem}[Irreducibility in Kraus representation] \label{thm:IrreducibilityByKraus}
Let $\mc{E}: \mbb{B}[\mc{H}] \to \mbb{B}[\mc{H}]$ be a CPTP map with the Kraus representation in Eq.~\eqref{eq:KrausRep}. 
The following statements are equivalent.
\begin{enumerate}
\item $\mc{E}$ is irreducible.
\item For any nonzero $\ket{\psi} \in \mc{H}$, $\mbb{K} \ket{\psi} = \mc{H}$.
\item $\mbb{K} = \mbb{B}[\mc{H}]$.
\end{enumerate}
\end{theorem}

\begin{proof}
The proof for (1) $\Leftrightarrow$ (2) follows Refs.~\cite{schrader2000perron, jaksic2014entropic}.

(1) $\Leftrightarrow$ (2): 
Let $\ket{\psi}, \ket{\phi} \in \mc{H}$ be nonzero and $s>0$. 
Expanding $e^{s\mc{E}}$ in powers of $\mc{E}$, we have
\begin{align}
\langle \phi | e^{s\mc{E}}[|\psi \rangle \langle \psi|] | \phi \rangle = |\langle \phi | \psi \rangle|^2 + \sum_{n=1}^\infty \frac{s^n}{n!} \sum_{b_1, b_2, \ldots, b_n} | \langle \phi | \hat{M}_{b_n} \cdots \hat{M}_{b_2} \hat{M}_{b_1} | \psi \rangle |^2.
\end{align}
This vanishes if and only if $\ket{\phi}$ is orthogonal to $\mbb{K} \ket{\psi}$.
Equivalently, this becomes nonzero for all $\ket{\phi}\in\mathcal{H}$ if and only if $\mbb{K} \ket{\psi} = \mc{H}$.
Since $\mc{E}$ is irreducible if and only if $\langle \phi | e^{s\mc{E}}[|\psi \rangle \langle \psi|] | \phi \rangle > 0$ for any nonzero $\ket{\psi}, \ket{\phi} \in \mc{H}$, according to (3) in Proposition~\ref{prop:irreducibility}, this proves the claim.

(2) $\Leftrightarrow$ (3): 
Since we consider finite-dimensional $\mc{H}$, (2) equivalently means that the only subspaces of $\mc{H}$ invariant under the action of $\mbb{K}$ are $\{ 0 \}$ and $\mc{H}$.
If $\mbb{K}$ satisfies the latter property, the algebra $\mbb{K}$ is said to be irreducible (not to be confused with the irreducibility of $\mc{E}$ we just introduced). 
By Burnside's theorem on matrix algebras (see Appendix~\ref{app:BurnsideTheorem}), the algebra $\mbb{K} \subseteq \mbb{B}[\mc{H}]$ is irreducible if and only if $\mbb{K} = \mbb{B}[\mc{H}]$.
Thus, (2) is equivalent to $\mbb{K} = \mbb{B}[\mc{H}]$.
\end{proof}

We can relate the irreducibility of a CPTP map $\mc{E}$ to its spectral properties~\cite{evans1978spectral, wolf2012quantum}.

\begin{theorem}[Irreducibility from spectral properties]
\label{thm:IrreducibilityBySpectrum}
Let $\mc{E}: \mbb{B}[\mc{H}] \to \mbb{B}[\mc{H}]$ be a CPTP map.
The following statements are equivalent.
\begin{enumerate}
\item $\mc{E}$ is irreducible.
\item $\mc{E}$ has a nondegenerate eigenvalue $1$ and the corresponding left and right eigenvectors are positive definite.
\end{enumerate}
\end{theorem}

\begin{proof}
The proof is inspired by Refs.~\cite{yoshida2024uniqueness, wolf2012quantum}.

(1) $\Rightarrow$ (2): 
According to property (ii) discussed in Sec.~\ref{sec:GeneralSpecCPTP}, any CPTP map has an eigenvalue $1$. 
Let $\hat{X} \in \mbb{B}[\mc{H}]$ be any right eigenvector of $\mc{E}$ with eigenvalue $1$.
Since $\hat{X}^\dagger$ is also a right eigenvector with eigenvalue $1$, we can assume that $\hat{X}$ is Hermitian without loss of generality\footnote{
By the argument leading to property (i) in Sec.~\ref{sec:GeneralSpecCPTP}, if $\hat{X}$ is an eigenvector corresponding to a real eigenvalue, $\hat{X}^\dagger$ is also an eigenvector with the same eigenvalue.
If we decompose the eigenvector $\hat{X}$ as $\hat{X} = \hat{X}_1 + i\hat{X}_2$ with $\hat{X}_1, \hat{X}_2$ being Hermitian, both $\hat{X}_1$ and $\hat{X}_2$ become eigenvectors. 
Thus, the eigenspace for any real eigenvalue is spanned by Hermitian operators in $\mbb{B}[\mc{H}]$.
}.
As detailed in Appendix~\ref{app:StationaryState}, any CPTP map $\mc{E}$ for finite-dimensional $\mc{H}$ has at least one stationary state $\hat{\varrho}_0 \in \mbb{B}[\mc{H}]$, which is a density operator satisfying $\mc{E}[\hat{\varrho}_0] = \hat{\varrho}_0$.
It remains to show its uniqueness and positive definiteness.

Suppose that $\mathcal{E}$ is irreducible in the sense of (3) in Proposition~\ref{prop:irreducibility}. 
Then, a stationary state $\hat{\varrho}_0$ satisfies $e^{s\mc{E}}[\hat{\varrho}_0] = e^{s} \hat{\varrho}_0 \succ 0$ for any $s>0$, which implies that the stationary state $\hat{\varrho}_0$ must be positive definite.
Assume that there exist two linearly independent stationary states $\hat{\varrho}_0 \succ 0$ and $\hat{\varrho}_0' \succ 0$.
Then, a Hermitian operator $\hat{X}(r) = (1-r) \hat{\varrho}_0 -r \hat{\varrho}_0'$ defined for $0 \leq r \leq 1$ is an eigenvector of $\mc{E}$ with eigenvalue $1$.
Since $\hat{X}(0)$ is positive definite while $\hat{X}(1)$ is negative definite, there exists $r=r^*$ at which $\hat{X}(r)$ has the lowest eigenvalue $0$ and thus becomes singular and positive semidefinite. 
Since $\Tr[\hat{X}(r^*)]>0$, one can construct a singular density operator $\hat{\varrho}^* = \hat{X}(r^*)/\Tr[\hat{X}(r^*)]$, which contradicts the fact that any stationary state must be positive definite.
Thus, the positive definite eigenvector of the eigenvalue $1$ must be unique.

Let $\hat{\varrho}_0 \succ 0$ be such a stationary state.
Assume that there exists a Hermitian operator $\hat{X}_0 \in \mbb{B}[\mc{H}]$ that has at least one negative eigenvalue and satisfies $\mc{E}[\hat{X}_0] = \hat{X}_0$. 
Then, a Hermitian operator $\hat{Y}(r) = (1-r) \hat{\varrho}_0 + r\hat{X}_0$ defined for $0 \leq r \leq 1$ is also an eigenvector of $\mc{E}$ with eigenvalue $1$.
Since $\hat{Y}(0)$ is positive definite while $\hat{Y}(1)$ has at least one negative eigenvalue, there exists $r=r^*$ at which $\hat{Y}(r)$ has the lowest eigenvalue $0$ and thus becomes singular and positive semidefinite.
This again contradicts the fact that any stationary state must be positive definite\footnote{
There is an alternative proof: for a sufficiently small real number $\epsilon>0$, $\hat{X}'_0(\epsilon) = \hat{\varrho}_0 + \epsilon \hat{X}_0$ becomes a positive definite eigenvector of eigenvalue 1 that is linearly independent of $\hat{\varrho}_0$, but this contradicts the uniqueness of the positive definite eigenvector as proven above. 
We thank Hironobu Yoshida for pointing this out.
}.
Therefore, the irreducible CPTP map $\mc{E}$ has a unique positive-definite right eigenvector with eigenvalue $1$.

The left eigenvector of $\mc{E}$ with eigenvalue $1$ is the right eigenvector of its dual map $\mc{E}^\dagger$ with the same eigenvalue. 
Since $\mc{E}^\dagger$ is unital, $\hat{\mbb{I}}$ is precisely such an eigenvector, which is obviously positive definite.
Since $\mc{E}$ and $\mc{E}^\dagger$ have the same spectrum, the uniqueness of the eigenvector is ensured.

(2) $\Rightarrow$ (1): 
We proceed by contradiction.
Assume that $\mc{E}$ has a nondegenerate eigenvalue $1$ and the corresponding eigenvector $\hat{\varrho}_0$ is positive definite.
If $\mc{E}$ is not irreducible in the sense of (1) in Proposition~\ref{prop:irreducibility}, there exists a nontrivial orthogonal projection operator $\hat{P}$ such that $\mc{E}[\hat{P} \mbb{B}[\mc{H}] \hat{P}] \subseteq \hat{P} \mbb{B}[\mc{H}] \hat{P}$. 
If we restrict the space of density operators to $\hat{P} \mbb{B}[\mc{H}] \hat{P}$, we can apply the arguments in Appendix~\ref{app:StationaryState} to show the existence of a density operator $\hat{\varrho}_1 \in \hat{P} \mbb{B}[\mc{H}] \hat{P}$ such that $\mc{E}[\hat{\varrho}_1] = \hat{\varrho}_1$.
Since $\hat{P} \notin \{0 , \hat{\mbb{I}} \}$, $\hat{\varrho}_1$ must have at least one eigenvalue $0$ and cannot be positive definite.
Thus, we have two linearly independent eigenvectors $\hat{\varrho}_0$ and $\hat{\varrho}_1$, which contradict the assumption.
\end{proof}

As the simplest example, we consider a two-level system and a CPTP map generated by the two Kraus operators 
\begin{align} \label{eq:IrreducibleCPTPEx1}
\hat{M}_1 = \ket{0} \bra{1}, \quad
\hat{M}_2 = \ket{1} \bra{0}.
\end{align}
The corresponding CPTP map $\mc{E}$ is irreducible according to Theorem~\ref{thm:IrreducibilityByKraus} since $\{ \hat{M}_1, \hat{M}_2, \hat{M}_1 \hat{M}_2, \hat{M}_2 \hat{M}_1 \}$ spans the entire space of $2 \times 2$ matrices.
In the natural representation [see Eq.~\eqref{eq:CPTPNaturalRep}], $\mc{E}$ can be expressed as
\begin{align}
\hat{\mc{E}} = \begin{pmatrix} 0 & 0 & 0 & 1 \\ 0 & 0 & 0 & 0 \\ 0 & 0 & 0 & 0 \\ 1 
& 0 & 0 & 0\end{pmatrix},
\end{align}
which can be easily diagonalized to find the eigenvalues $z= 0, 0, \pm 1$. 
In particular, the eigenvectors corresponding to $z=\pm1$ are given by $\hat{X}_\pm = \ket{0} \bra{0} \pm \ket{1} \bra{1}$.
This confirms that the eigenvalue $1$ is nondegenerate and the corresponding eigenvector is positive definite, as stated by Theorem~\ref{thm:IrreducibilityBySpectrum}.

As seen from this example, the irreducibility of $\mc{E}$ generally does not rule out the existence of other eigenvalues with unit modulus. 
Indeed, the irreducible map $\mc{E}$ can have a nondegenerate peripheral spectrum, which lies on the unit circle $|z|=1$ in the complex plane. 
In order to see this, let us consider a $d$-level system and denote its basis states by $\{\ket{i} \}_{i=0}^{d-1}$. 
We then introduce a CPTP map generated by the Kraus operators
\begin{align} \label{eq:KrausIrreducibleEx}
\hat{M}_j = \begin{cases} \ket{j-1}\bra{j} & (1 \leq j \leq d-1) \\ \ket{d-1}\bra{0} & (j=d) \end{cases}.
\end{align}
Since an arbitrary basis state $\ket{i}$ can be transformed into any basis state $\ket{i'}$ by applying a string of Kraus operators $\hat{M}_{b_\ell} \cdots \hat{M}_{b_2} \hat{M}_{b_1}$ with length $\ell \leq d-1$, the corresponding CPTP map $\mc{E}$ is irreducible according to (2) of Theorem~\ref{thm:IrreducibilityByKraus}. 
We can easily diagonalize $\mc{E}$ to find that it has nondegenerate peripheral eigenvalues $\gamma_k = e^{2\pi i k/d}$ ($k=0,\ldots,d-1$) while the other eigenvalues are all zero. 
The eigenvector corresponding to $\gamma_k$ is given by $\hat{X}_k = \sum_j e^{2\pi i jk/d} \ket{j} \bra{j}$, and thus $\hat{X}_0$ is positive definite.

In Fig.~\ref{fig:CPTPEigenvalues}(a), we illustrate the complex eigenvalues of a modified CPTP map for $d=6$, 
\begin{align}
& \hat{M}_1 = \sqrt{1-p-q} \ket{0} \bra{1}, \quad 
\hat{M}_2 = \sqrt{1-p-q-r} \ket{1} \bra{2}, \quad
\hat{M}_3 = \sqrt{1-p-r} \ket{2} \bra{3}, \quad \nonumber \\
& \hat{M}_4 = \sqrt{1-p-r} \ket{3} \bra{4}, \quad
\hat{M}_5 = \sqrt{1-p} \ket{4} \bra{5}, \quad
\hat{M}_6 = \sqrt{1-p} \ket{5} \bra{0}, \quad \nonumber \\
& \hat{M}_7 = \sqrt{p} (\ket{0} \bra{1} + \ket{1} \bra{2} + \ket{2} \bra{3} +\ket{3} \bra{4} + \ket{4} \bra{5} + \ket{5} \bra{0}), \nonumber \\
& \hat{M}_8 = \sqrt{q} (\ket{0} \bra{1} + \ket{1} \bra{2}), \quad 
\hat{M}_9 = \sqrt{r} (\ket{1} \bra{2} +\ket{2} \bra{3} + \ket{3} \bra{4}),
\end{align}
with $p=q=r=0.3$. 
This map reduces to the CPTP map given by Eq.~\eqref{eq:KrausIrreducibleEx} for $p=q=r=0$, while finite $p,q,r$ can partially lift the degeneracy of the eigenvalue $z=0$ without altering the irreducibility of the map and thereby its peripheral spectrum.

\begin{figure}[tb]
\centering\includegraphics[width=\linewidth]{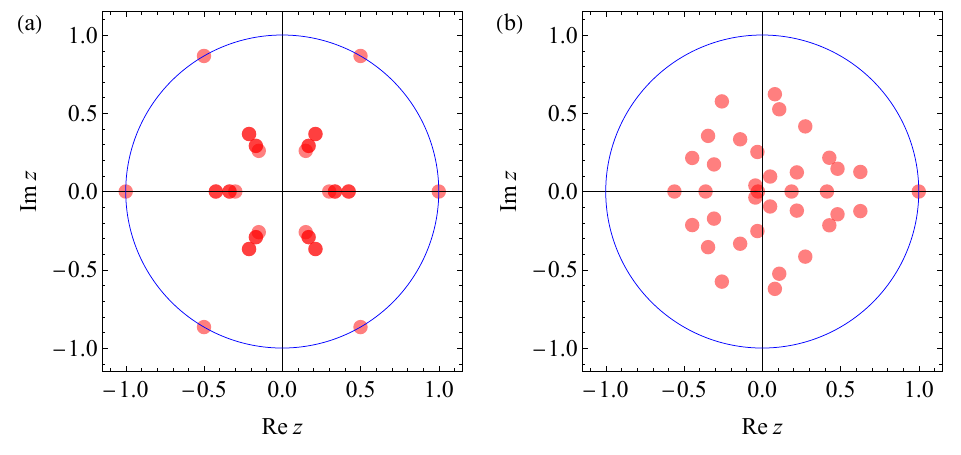}
\caption{(a) When a CPTP map $\mc{E}$ is irreducible, it has a nondegenerate peripheral spectrum $z \in \{ e^{2\pi i k/m} \}_{k=0}^{m-1}$ ($m=6$ in this example). 
(b) When a CPTP map $\mc{E}$ is primitive, $z=1$ is the only eigenvalue with unit modulus.}
\label{fig:CPTPEigenvalues}
\end{figure}

The above observation can be made rigorous through the following theorem~\cite{evans1978spectral, fagnola2009irreducible, wolf2012quantum}.

\begin{theorem}[Peripheral spectrum of irreducible CPTP maps]
\label{thm:CPTPPeripheral}
Let $\mc{E}: \mbb{B}[\mc{H}] \to \mbb{B}[\mc{H}]$ be an irreducible CPTP map. 
Denote by $S = \textrm{spec}(\mc{E}) \cap \exp(i\mbb{R})$ the peripheral spectrum of $\mc{E}$. 
Then, the following statements hold. 
\begin{enumerate}
\item There is an integer $1 \leq m \leq d^2$ such that $S = \{ e^{2\pi ik/m}\}_{k=0}^{m-1}$.
\item All eigenvalues in $S$ are nondegenerate.
\item There is a unitary operator $\hat{U}$ such that $\hat{U}^m=\mbb{\hat{I}}$ and $\mc{E}^\dagger [\hat{U}^k] = e^{2\pi ik/m} \hat{U}^k$.
\item $\hat{U}$ has the spectral decomposition $\hat{U} = \sum_{k=0}^{m-1} e^{2\pi ik/m} \hat{P}_k$ where the spectral projections $\hat{P}_k$ satisfy $\mc{E}^\dagger (\hat{P}_{k+1}) = \hat{P}_k$ (indices are taken modulo $m$).
\end{enumerate}
\end{theorem}

\begin{proof}
The proof is based on Refs.~\cite{fagnola2009irreducible, wolf2012quantum}.

Suppose that $e^{i\theta}$ ($\theta \in \mbb{R}$) is an eigenvalue of $\mc{E}$.
Let $\hat{X} \in \mbb{B}[\mc{H}]$ be nonzero and satisfy $\mc{E}^\dagger[\hat{X}] = e^{i\theta} \hat{X}$. 
Since $\mc{E}^\dagger$ is Hermiticity preserving, we also have $\mc{E}^\dagger[\hat{X}^\dagger] = e^{-i\theta} \hat{X}^\dagger$.
Since $\mc{E}^\dagger$ is a CP unital map, it satisfies the Kadison-Schwarz inequality $\mc{E}^\dagger[\hat{X}^\dagger \hat{X}] \succeq \mc{E}^\dagger[\hat{X}^\dagger] \mc{E}^\dagger[\hat{X}]$, as shown in Appendix~\ref{app:SchwarzMap}.
Since $\mc{E}$ is irreducible, there exists a unique positive-definite density operator $\hat{\varrho}_0 \in \mbb{B}[\mc{H}]$ such that $\mc{E}[\hat{\varrho}_0] = \hat{\varrho}_0$.
Then, we find
\begin{align}
0 
\leq \Tr [\hat{\varrho}_0 (\mc{E}^\dagger[\hat{X}^\dagger \hat{X}] - \mc{E}^\dagger[\hat{X}^\dagger] \mc{E}^\dagger[\hat{X}])]
= \Tr [\mc{E}[\hat{\varrho}_0] \hat{X}^\dagger \hat{X} -\hat{\varrho}_0 e^{-i\theta} \hat{X}^\dagger e^{i\theta} \hat{X}]
= 0
\end{align}
Since $\hat{\varrho}_0 \succ 0$, this implies $\mc{E}^\dagger[\hat{X}^\dagger \hat{X}] = \mc{E}^\dagger[\hat{X}^\dagger] \mc{E}^\dagger[\hat{X}]$ and similarly $\mc{E}^\dagger[\hat{X} \hat{X}^\dagger] = \mc{E}^\dagger[\hat{X}] \mc{E}^\dagger[\hat{X}^\dagger]$.

Let us define $\hat{D}(\hat{X}_1, \hat{X}_2) := \mc{E}^\dagger[\hat{X}_1^\dagger \hat{X}_2] - \mc{E}^\dagger[\hat{X}_1^\dagger] \mc{E}^\dagger[\hat{X}_2]$ for $\hat{X}_1, \hat{X}_2 \in \mbb{B}[\mc{H}]$.
We now claim that $\hat{D}(\hat{X}, \hat{X}) = \hat{D}(\hat{X}^\dagger, \hat{X}^\dagger)=0$ implies $\hat{D}(\hat{X}, \hat{A}) = \hat{D}(\hat{X}^\dagger, \hat{A}) = 0$ for any $\hat{A} \in \mbb{B}[\mc{H}]$.
To prove this, let $z \in \mbb{C}$ and $\hat{Y} = z \hat{X} + \hat{A}$. Using the Kadison-Schwarz inequality $\hat{D}(\hat{Y}, \hat{Y}) \succeq 0$, we have
\begin{align}
0 
&\preceq \hat{D}(\hat{Y}, \hat{Y}) 
= |z|^2 \hat{D}(\hat{X}, \hat{X}) + z\hat{D}(\hat{A}, \hat{X}) + z^* \hat{D}(\hat{X}, \hat{A}) +\hat{D}(\hat{A}, \hat{A}) \nonumber \\
&= z \hat{D}(\hat{A}, \hat{X}) + [z \hat{D}(\hat{A}, \hat{X})]^\dagger + \hat{D}(\hat{A}, \hat{A}).
\end{align}
This must hold for any $z \in \mbb{C}$ [note $\hat{D}(\hat{A}, \hat{A}) \succeq 0$].
Considering the limit $|z| \to \infty$ with appropriate phases\footnote{
When $z \to +\infty  \, (-\infty)$, $\hat{D}(\hat{Y}, \hat{Y}) \succeq 0$ holds if and only if $\hat{D}(\hat{A}, \hat{X}) + \hat{D}(\hat{X}, \hat{A})^\dagger$ does not have negative (positive) eigenvalues.
This implies $\hat{D}(\hat{A}, \hat{X}) = -\hat{D}(\hat{A}, \hat{X})^\dagger$.
When $z \to  +i \infty \, (-i \infty)$, $\hat{D}(\hat{Y}, \hat{Y}) \succeq 0$ holds if and only if $i[\hat{D}(\hat{A}, \hat{X}) - \hat{D}(\hat{A}, \hat{X})^\dagger]$ does not have negative (positive) eigenvalues.
This implies $\hat{D}(\hat{A}, \hat{X}) = \hat{D}(\hat{A}, \hat{X})^\dagger$.
We then find $\hat{D}(\hat{A}, \hat{X}) = -\hat{D}(\hat{A}, \hat{X})^\dagger = -\hat{D}(\hat{A}, \hat{X})$ and thus $\hat{D}(\hat{X}, \hat{A}) = \hat{D}(\hat{A}, \hat{X})^\dagger =0$.
},
we can show that this inequality holds if and only if $\hat{D}(\hat{X}, \hat{A})=0$.
We can similarly show $\hat{D}(\hat{X}^\dagger, \hat{A}) = 0$.
This proves the claim.

Therefore, $\mc{E}^\dagger[\hat{X} \hat{A}] = \mc{E}^\dagger[\hat{X}] \mc{E}^\dagger[\hat{A}]$ and $\mc{E}^\dagger[\hat{X}^\dagger \hat{A}] = \mc{E}^\dagger[\hat{X}^\dagger] \mc{E}^\dagger[\hat{A}]$ hold for any $\hat{A} \in \mbb{B}[\mc{H}]$.
We then have
\begin{align}
\mc{E}^\dagger[\hat{X}^2] = (\mc{E}^\dagger[\hat{X}])^2 = e^{2i\theta} \hat{X}^2, \ \mc{E}^\dagger[\hat{X}^3] = (\mc{E}^\dagger[\hat{X}])^3 = e^{3i\theta} \hat{X}^3, \ \ldots,
\end{align}
and inductively $\mc{E}^\dagger[\hat{X}^k] = e^{ik\theta} \hat{X}^k$ for $k \in \mbb{N}$.
Thus, $e^{ik\theta}$ are eigenvalues of $\mc{E}^\dagger$ and hence of $\mc{E}$.
Since there are at most $d^2$ elements in $S$, there exists an integer $1 \leq m \leq d^2$ such that $\theta = 2\pi/m$.
This proves (1).

Let us consider the property of an eigenvector $\hat{X}$ satisfying $\mc{E}^\dagger[\hat{X}] = e^{i\theta} \hat{X}$.
Since $\mc{E}$ is irreducible, $\hat{\mbb{I}}$ is the only eigenvector of $\mc{E}^\dagger$ with eigenvalue $1$.
On the other hand, the above argument gives $\mc{E}^\dagger[\hat{X}^\dagger \hat{X}] = \mc{E}^\dagger[\hat{X}^\dagger] \mc{E}^\dagger[\hat{X}] = \hat{X}^\dagger \hat{X}$.
Thus, $\hat{X}^\dagger \hat{X}$ must be a scalar multiple of $\hat{\mbb{I}}$, i.e., $\hat{X}^\dagger \hat{X} = c \hat{\mbb{I}}$ for some $c \neq 0$.
We can then choose the eigenvector $\hat{X}$ to be a unitary operator by rescaling $\hat{U}=\hat{X}/\sqrt{c}$.
It satisfies $\hat{U}^m = \hat{\mbb{I}}$ and $\mc{E}^\dagger[\hat{U}^k] = e^{ik\theta} \hat{U}^k = e^{2\pi ik/m} \hat{U}^k$, which proves (3).

If there are two unitaries $\hat{U}_1$ and $\hat{U}_2$ such that $\mc{E}^\dagger[\hat{U}_1] = e^{i\theta} \hat{U}_1$ and $\mc{E}^\dagger[\hat{U}_2] = e^{i\theta} \hat{U}_2$, then $\mc{E}^\dagger[\hat{U}_2^\dagger \hat{U}_1] = \mc{E}^\dagger[\hat{U}_2^\dagger] \mc{E}^\dagger[\hat{U}_1] = \hat{U}_2^\dagger \hat{U}_1$. 
Since $\mc{E}$ is irreducible, $\hat{U}_2^\dagger \hat{U}_1 = c' \hat{\mbb{I}}$ for some $c' \neq 0$.
This means that $\hat{U}_1 = c' \hat{U}_2$, or equivalently, $\hat{U}_1$ and $\hat{U}_2$ are linearly dependent. 
Thus, the eigenvalue $e^{i\theta} \in S$ is nondegenerate, and (2) is proved.

Since $\hat{U}$ is unitary and satisfies $\hat{U}^m = \hat{\mbb{I}}$, it has the spectral decomposition $\hat{U} = \sum_{k=0}^{m-1} e^{2\pi ik/m} \hat{P}_k$, where $\hat{P}_k$ is an orthogonal projection operator onto the eigenspace of $\hat{U}$ with eigenvalue $e^{2\pi ik/m}$.
Since $\hat{U}^n = \sum_{k=0}^{m-1} e^{2\pi ikn/m} \hat{P}_k$, its Fourier transform is $\hat{P}_k = (1/m) \sum_{n=0}^{m-1} e^{-2\pi ikn/m} \hat{U}^n$.
Then, we find
\begin{align}
\mc{E}^\dagger[\hat{P}_k] 
= \frac{1}{m} \sum_{n=0}^{m-1} e^{-2\pi ikn/m} \mc{E}^\dagger [\hat{U}^n] 
= \frac{1}{m} \sum_{n=0}^{m-1} e^{-2\pi i(k-1)n/m} \hat{U}^n 
= \hat{P}_{k-1}.
\end{align}
This proves (4).
\end{proof}

An immediate consequence of the irreducibility of the CPTP map $\mc{E}$ is the existence of a unique steady state in the sense of the time average~\cite{wolf2012quantum}.
\begin{theorem}[Irreducibility from steady state]
Let $\mc{E}: \mbb{B}[\mc{H}] \to \mbb{B}[\mc{H}]$ be a CPTP map. 
The following statements are equivalent.
\begin{enumerate}
\item $\mc{E}$ is irreducible.
\item There exists a unique steady state $\hat{\rho}_\mathrm{ss} \succ 0$ such that for any density operator $\hat{\rho} \in \mbb{B}[\mc{H}]$ we have
\begin{align} \label{eq:UniqueSSAverage}
\lim_{N \to \infty} \frac{1}{N} \sum_{n=0}^{N-1} \mc{E}^n[\hat{\rho}] = \hat{\rho}_\mathrm{ss}.
\end{align}
\end{enumerate}
\end{theorem}

\begin{proof}
As detailed in Appendix~\ref{app:StationaryState}, we can show that the left-hand side of Eq.~\eqref{eq:UniqueSSAverage} is a stationary state of $\mc{E}$.
Then Theorem~\ref{thm:IrreducibilityBySpectrum} immediately implies the equivalence between (1) and (2).
\end{proof}

This means that any density matrix $\hat{\rho}$ converges to a unique state $\hat{\rho}_\mathrm{ss}$, which is the unique eigenvector of $\mc{E}$ with eigenvalue $z=1$, upon time averaging over a sufficiently long time.
We note that the uniqueness of the long-time limit $\lim_{n \to \infty} \mc{E}^n[\hat{\rho}]$ for any $\hat{\rho} \in \mbb{B}[\mc{H}]$ implies the uniqueness of the long-time average $\lim_{N \to \infty} (1/N) \sum_{n=1}^N \mc{E}^n [\hat{\rho}]$, but the converse does not hold in general.
In fact, the uniqueness of $\lim_{n \to \infty} \mc{E}^n[\hat{\rho}]$ requires a more stringent condition on $\mc{E}$ than irreducibility, as we will discuss in the next section.

\subsection{Primitivity of CPTP maps}
\label{sec:primitivity}

Irreducibility ensures that a CPTP map has a nondegenerate eigenvalue $z=1$, but it does not ensure that $z=1$ is the only eigenvalue of unit modulus.
The latter is instead ensured by the \emph{primitivity} condition for $\mc{E}$, which is stated in the following equivalent ways~\cite{evans1978spectral, sanz2010quantum, wolf2012quantum}.

\begin{proposition}[Primitivity]
\label{prop:primitivity}
Let $\mc{E}: \mbb{B}[\mc{H}] \to \mbb{B}[\mc{H}]$ be a CPTP map. 
The following statements are equivalent.
\begin{enumerate}
\item There exists an $n \in \mbb{N}$ such that for any nonzero $\hat{\rho} \succeq 0$ in $\mbb{B}[\mc{H}]$ we have $\mc{E}^n[\hat{\rho}] {\succ} 0$.
\item $\mc{E}^k$ is irreducible for every $k \in \mbb{N}$.
\item $\mc{E}$ has a nondegenerate eigenvalue $1$, which is the only eigenvalue of unit modulus, and the corresponding eigenvector is positive definite.
\item There exists a unique steady state $\hat{\rho}_\mathrm{ss} \succ 0$ such that for any density matrix $\hat{\rho} \in \mbb{B}[\mc{H}]$ we have
\begin{align}
\lim_{n \to \infty} \mc{E}^n[\hat{\rho}] = \hat{\rho}_\mathrm{ss}.
\end{align}
\end{enumerate}
\end{proposition}

\begin{proof}
The proof is based on Refs.~\cite{evans1978spectral, wolf2012quantum}.

(1) $\Rightarrow$ (2): 
We proceed by contradiction. 
Suppose that $\mc{E}^k$ is not irreducible for some $k \in \mbb{N}$. 
Then, there exists a nontrivial orthogonal projection operator $\hat{P} \in \mbb{B}[\mc{H}]$ such that $\mc{E}^k[\hat{P} \mbb{B}[\mc{H}] \hat{P}] \subseteq \hat{P} \mbb{B}[\mc{H}] \hat{P}$. 
Following the arguments in Appendix~\ref{app:StationaryState}, there exists a density operator $\hat{\varrho}_0 \in \hat{P} \mbb{B}[\mc{H}] \hat{P}$ such that $\mc{E}^k[\hat{\varrho}_0] = \hat{\varrho}_0$. 
Since $\hat{P} \notin \{ 0, \hat{\mbb{I}} \}$, $\hat{\varrho}_0$ must have at least one eigenvalue $0$ and cannot be positive definite.
Obviously, $\mc{E}^{km}[\hat{\varrho}_0] = \hat{\varrho}_0$ for any $m \in \mbb{N}$.
On the other hand, by assumption (1), there exists $n \in \mbb{N}$ such that $\mc{E}^n[\hat{\rho}] \succ 0$ for any $\hat{\rho} \succeq 0$ in $\mbb{B}[\mc{H}]$.
This further implies $\mc{E}^{nl}[\hat{\rho}] \succ 0$ for any $l \in \mbb{N}$. 
Since the relations $\mc{E}^{km}[\hat{\varrho}_0] = \hat{\varrho}_0$ and $\mc{E}^{nl}[\hat{\rho}] \succ 0$ hold for arbitrary integers $m$ and $l$, we can specifically choose $m=n$ and $l=k$.
With this choice, setting $\hat{\rho} = \hat{\varrho}_0$, we obtain $\mc{E}^{nk}[\hat{\varrho}_0] = \hat{\varrho}_0$ from the former and $\mc{E}^{nk}[\hat{\varrho}_0] \succ 0$ from the latter. 
This contradicts the fact that $\hat{\varrho}_0$ is not positive definite.

(2) $\Rightarrow$ (3): 
We proceed by contradiction.
Suppose that $\mc{E}$ has a nontrivial peripheral spectrum $S = \{ e^{2\pi ik/m} \}_{k=0}^{m-1}$ with some $m \geq 2$.
Since $\mc{E}^k$ is irreducible for every $k \in \mbb{N}$, $\mc{E}$ itself is irreducible, and by (4) of Theorem~\ref{thm:CPTPPeripheral} we have nontrivial orthogonal projection operators $\hat{P}_k$ such that $\mc{E}^\dagger[\hat{P}_{k+1}] = \hat{P}_k$.
This implies $(\mc{E}^\dagger)^m [\hat{P}_k] = \hat{P}_k$ and thus $\mc{E}^m$ is not irreducible according to (2) of Proposition~\ref{prop:irreducibility}, which leads to a contradiction\footnote{
There is an alternative proof: denoting by $\hat{X}_k$ right eigenvectors corresponding to the eigenvalues $\lambda_k = e^{2\pi ik/m}$, we have $\mc{E}^m[\hat{X}_k] = \lambda_k^m \hat{X}_k = \hat{X}_k$, but this contradicts the irreducibility of $\mc{E}^m$ as it must have a nondegenerate eigenvalue 1 according to Theorem~\ref{thm:IrreducibilityBySpectrum}.
We thank Hironobu Yoshida for pointing this out.
}.
Thus, $S = \{ 1 \}$ and the corresponding eigenvector is positive definite according to Theorem~\ref{thm:IrreducibilityBySpectrum}.

(3) $\Rightarrow$ (4): 
Consider a $d^2 \times d^2$ matrix representation $\hat{\mc{E}}$ of $\mc{E}$ and its Jordan normal form in Eq.~\eqref{eq:CPTPJordanForm}.
The $n$th power of the Jordan block $J^n_{d_a}(z_a)$ vanishes as $n \to \infty$ for all eigenvalues satisfying $|z_a| < 1$. 
Since the peripheral spectrum of $\mc{E}$ contains only $1$ and is nondegenerate, the limit $\lim_{n \to \infty}\mc{E}^n = \mc{E}_\infty$ exists, and $\mc{E}_\infty$ becomes the projection operator onto the eigenspace of $\mc{E}$ with eigenvalue $z_a=1$. 
This means $\mc{E}_\infty[\hat{\rho}] = \hat{\rho}_\mathrm{ss}$ for any density operator $\hat{\rho} \in \mbb{B}[\mc{H}]$, where $\hat{\rho}_\mathrm{ss}$ is the eigenvector for $z_a=1$, which is positive definite.

(4) $\Rightarrow$ (1):
We first note that any operator $\hat{X} \in \mbb{B}[\mc{H}]$ can be expanded as $\hat{X} = c_1 \hat{\rho}_1 -c_2 \hat{\rho}_2 +ic_3 \hat{\rho}_3 -ic_4 \hat{\rho}_4$ with four density operators $\hat{\rho}_j \in \mbb{B}[\mc{H}]$ and $c_j \geq 0$.
Then, (4) implies that the limit $\lim_{n \to \infty} \mc{E}^n[\hat{X}] = (c_1-c_2+ic_3-ic_4) \hat{\rho}_\mathrm{ss} = \Tr[\hat{X}] \hat{\rho}_\mathrm{ss} = \mc{E}_\infty[\hat{X}]$ exists, i.e., the sequence of CPTP maps $\{ \mc{E}^n \}_n$ converges pointwise to $\mc{E}_\infty$.
Since $\mc{H}$ is finite-dimensional, the sequence $\{ \mc{E}^n \}_n$ actually converges uniformly to $\mc{E}_\infty$, and thus $\lim_{n \to \infty} \| \mc{E}_\infty -\mc{E}^n \| = 0$\footnote{
Let $\{ \hat{e}_{ij} \}$ be $d \times d$ matrices whose entries are all 0 except the $(i,j)$th entry, which is 1.
If $\{\mc{E}^n\}_n$ converges pointwise to $\mc{E}_\infty$, we have $\lim_{n \to \infty} \| (\mc{E}_\infty -\mc{E}^n)[\hat{e}_{ij}] \| = 0$ for any $\hat{e}_{ij}$.
Since any $\hat{X} \in \mbb{B}[\mc{H}]$ can be expanded as $\hat{X} = \sum_{i,j} c_{ij} \hat{e}_{ij}$ with $c_{ij} \in \mbb{C}$, we have $\| (\mc{E}_\infty -\mc{E}^n)[\hat{X}] \| \leq \sum_{i,j} |c_{ij}| \| (\mc{E}_\infty -\mc{E}_n)[\hat{e}_{ij}] \|$.
We then introduce $\epsilon_n = \max_{i,j} \{ \| (\mc{E}_\infty-\mc{E}^n)[\hat{e}_{ij}] \| \}$ to write $\| (\mc{E}_\infty -\mc{E}^n)[\hat{X}] \| \leq \epsilon_n \sum_{i,j} |c_{ij}|$.
If we further suppose $\| \hat{X} \| = 1$, there exists a finite $C>0$ that bounds the $L_1$-norm such that $\sum_{i,j} |c_{ij}| \leq C$, yielding $\| (\mc{E}_\infty -\mc{E}^n)[\hat{X}] \| \leq C \epsilon_n$.
Finally, we find $\| \mc{E}_\infty -\mc{E}^n \| = \sup_{\| \hat{X} \|=1} \| (\mc{E}_\infty-\mc{E}^n)[\hat{X}] \| \leq C \epsilon_n$.
Since $\epsilon_n \to 0$ as $n \to \infty$ when $\mbb{B}[\mc{H}]$ is finite-dimensional, we have proved $\lim_{n \to \infty} \| \mc{E}_\infty - \mc{E}^n \| = 0$.
}.
We then proceed by contradiction. 
Suppose that for any $n \in \mbb{N}$ there exists a density operator $\hat{\rho} \in \mbb{B}[\mc{H}]$ such that $\mc{E}^n[\hat{\rho}]$ is not positive definite.
Let $\ket{\psi} \in \mc{H}$ be an eigenvector of $\mc{E}^n[\hat{\rho}]$ with eigenvalue 0.
Since $\hat{\rho}_\mathrm{ss} \succ 0$, we denote its smallest eigenvalue by $\lambda_\textrm{min}(\hat{\rho}_\mathrm{ss}) > 0$.
Then, we have
\begin{align}
0 
&< \lambda_\textrm{min}(\hat{\rho}_\mathrm{ss}) 
\leq \langle \psi | \hat{\rho}_\mathrm{ss} | \psi \rangle
= \langle \psi | (\hat{\rho}_\mathrm{ss} - \mc{E}^n[\hat{\rho}]) | \psi \rangle 
\leq \| \hat{\rho}_\mathrm{ss} - \mc{E}^n[\hat{\rho}] \| \nonumber \\
&= \| (\mc{E}_\infty - \mc{E}^n)[\hat{\rho}] \|
\leq \| \mc{E}_\infty -\mc{E}^n \| \| \hat{\rho} \| \leq \| \mc{E}_\infty - \mc{E}^n \|.
\end{align}
The convergence of the operator norm $\lim_{n \to \infty} \| \mc{E}_\infty-\mc{E}^n \| = 0$ means that for any $\epsilon > 0$ there exists an $N \in \mbb{N}$ such that $\| \mc{E}_\infty-\mc{E}^n \| < \epsilon$ for all $n \geq N$\footnote{See Refs.~\cite{wolf2012quantum, szehr2015spectral} for more detailed discussions on the convergence of $\| \mc{E}_\infty - \mc{E}^n\|$.}.
Thus, setting $\epsilon = \lambda_\textrm{min}(\hat{\rho}_\mathrm{ss})$ leads to a contradiction for $n \geq N$.
\end{proof}

As stated in (4) of Proposition~\ref{prop:primitivity}, the primitivity of a CPTP map $\mc{E}$ ensures that there exists a unique steady state in the sense of the long-time limit $\hat{\rho}_\mathrm{ss} = \lim_{n \to \infty} \mc{E}^n[\hat{\rho}]$, which is nothing but the eigenvector of $\mc{E}$ with eigenvalue $z=1$.
This can be contrasted with the case of general irreducible CPTP maps, for which a unique steady state only exists in the sense of the long-time average as given in Eq.~\eqref{eq:UniqueSSAverage}.

The primitivity of a CPTP map $\mc{E}$ can also be characterized by its Kraus representation, as given in Eq.~\eqref{eq:KrausRep}.
Let $\mbb{K}_m \subseteq \mbb{B}[\mc{H}]$ be the complex linear span of all degree-$m$ monomials of the Kraus operators $\hat{M}_b$, 
\begin{align}
\mbb{K}_m := \textrm{span} \{ \hat{M}_{b_m} \cdots \hat{M}_{b_2} \hat{M}_{b_1} \}.
\end{align}
We have the following criterion based on the properties of $\mbb{K}_m$~\cite{sanz2010quantum, wolf2012quantum}.

\begin{theorem}[Primitivity in Kraus representation]
\label{thm:PrimitivityByKraus}
Let $\mc{E}: \mbb{B}[\mc{H}] \to \mbb{B}[\mc{H}]$ be a CPTP map with the Kraus representation in Eq.~\eqref{eq:KrausRep}.
The following statements are equivalent.
\begin{enumerate}
\item $\mc{E}$ is primitive.
\item There exists an $n \in \mbb{N}$ such that $\mbb{K}_m \ket{\psi} = \mc{H}$ for any nonzero $\ket{\psi} \in \mc{H}$ and all $m \geq n$.
\item There exists a $q \in \mbb{N}$ such that $\mbb{K}_m = \mbb{B}[\mc{H}]$ for all $m \geq q$.
\end{enumerate}
\end{theorem}

\begin{proof}
The proof is based on Refs.~\cite{wolf2012quantum, sanz2010quantum}.

(1) $\Leftrightarrow$ (2):
Let $\ket{\psi}, \ket{\phi} \in \mc{H}$ be nonzero. 
The primitivity of $\mc{E}$ in the sense of (1) in Proposition~\ref{prop:primitivity} implies that there exists an $n \in \mbb{N}$ such that
\begin{align}
\langle \phi | \mc{E}^n [|\psi \rangle \langle \psi|] | \phi \rangle = \sum_{b_1, b_2, \ldots, b_n} | \langle \phi | \hat{M}_{b_n} \cdots \hat{M}_{b_2} \hat{M}_{b_1} | \psi \rangle |^2 > 0.
\end{align}
This holds true if and only if $\mbb{K}_n \ket{\psi} = \mc{H}$ for all nonzero $\ket{\psi} \in \mc{H}$. 
Then, let us consider $\hat{M}_{b_{n+1}} \hat{M}_{b_n} \cdots \hat{M}_{b_1} \ket{\psi}$.
Since at least one of $\hat{M}_{b_1} \ket{\psi}$ is nonzero, denoting it by $\ket{\psi'}$, we find that $\mbb{K}_{n+1} \ket{\psi}$ contains $\mbb{K}_n \ket{\psi'}$ as a subspace.
Since $\mbb{K}_n \ket{\psi'} = \mc{H}$, we have $\mbb{K}_{n+1} \ket{\psi} = \mc{H}$ and, by induction, $\mbb{K}_m \ket{\psi} = \mc{H}$ for all $m \geq n$.

(1) $\Rightarrow$ (3):
Define $\hat{\mathsf{M}}_{\bm{b};m} := \hat{M}_{b_m} \cdots \hat{M}_{b_2} \hat{M}_{b_1}$ with a collection of indices $\bm{b}_m=( b_1, b_2, \cdots, b_m )$.
Using the linear one-to-one correspondence between $\hat{\mathsf{M}}_{\bm{b};m} \in \mbb{B}[\mc{H}]$ and $\ket{\Phi_{\bm{b};m}} := (\hat{\mathsf{M}}_{\bm{b};m} \otimes \hat{\mbb{I}}) \ket{\Phi} \in \mc{H} \otimes \mc{H}$\footnote{
Consider a map $f: \mbb{B}[\mc{H}] \to \mc{H} \otimes \mc{H}, \, \hat{X} \mapsto (\hat{X} \otimes \hat{\mbb{I}}) \ket{\Phi} =: \ket{\Phi_X}$.
It is obviously a linear map.
Since $\ket{\Phi_X} = (1/\sqrt{d}) \sum_{i=1}^d (\hat{X} \ket{i}) \otimes \ket{i} = (1/\sqrt{d}) \sum_{i,j} X_{ij} \ket{i} \otimes \ket{j}$ with $X_{ij} = \langle i | \hat{X} | j \rangle$, $\ket{\Phi_X} = 0$ implies $X_{ij}=0$ and thus $\hat{X}=0$. 
This means that $f^{-1}(0)=0$ and hence $f$ is injective.
For every $\ket{\Psi} = \sum_{i,j} \psi_{ij} \ket{i} \otimes \ket{j} \in \mc{H} \otimes \mc{H}$, we can find $\hat{X} \in \mbb{B}[\mc{H}]$ such that $f(\hat{X}) = \ket{\Psi}$ by choosing $X_{ij} = \sqrt{d} \psi_{ij}$. 
This means that $f$ is surjective and therefore there is one-to-one correspondence between $\hat{X}$ and $\ket{\Phi_X}$.
}
, where $\ket{\Phi}$ is the maximally entangled state defined in Eq.~\eqref{eq:MaximallyEntangledState}, we find that $\mbb{K}_m = \mbb{B}[\mc{H}]$ and $\textrm{span} \{ \ket{\Phi_{\bm{b};m}} \} = \mc{H} \otimes \mc{H}$ are equivalent.
Furthermore, using the one-to-one correspondence between $\mc{E}^m: \mbb{B}[\mc{H}] \to \mbb{B}[\mc{H}]$ and its Choi-Jamiołkowski representation $\hat{E}^m \in \mbb{B}[\mc{H} \otimes \mc{H}]$ [see Eq.~\eqref{choi}], 
\begin{align}
\hat{E}^m 
:= (\mc{E}^m \otimes \mc{I})[| \Phi \rangle \langle \Phi |] 
= \sum_{\bm{b}_m} (\hat{\mathsf{M}}_{\bm{b};m} \otimes \hat{\mbb{I}}) | \Phi \rangle \langle \Phi | (\hat{\mathsf{M}}_{\bm{b};m}^\dagger \otimes \hat{\mbb{I}}) 
= \sum_{\bm{b}_m} | \Phi_{\bm{b};m} \rangle \langle \Phi_{\bm{b};m} |,
\end{align}
we find that $\textrm{span} \{ \ket{\Phi_{\bm{b};m}} \} = \mc{H} \otimes \mc{H}$ and $\hat{E}^m \succ 0$ are equivalent.

We then proceed by contradiction. 
Suppose that $\mbb{K}_m \neq \mbb{B}[\mc{H}]$ and thus $\hat{E}^m$ is not positive definite for all $m \in \mbb{N}$.
Let $\ket{\Psi_m} \in \mc{H} \otimes \mc{H}$ be an eigenstate of $\hat{E}^m$ with eigenvalue $0$.
Since $\mc{E}$ is primitive, the limit $\mc{E}_\infty = \lim_{n \to \infty} \mc{E}^n$ exists so that $(\mc{E}_\infty \otimes \mc{I})[| \Phi \rangle \langle \Phi |] = \hat{\rho}_\mathrm{ss} \otimes (\hat{\mbb{I}}/d)$\footnote{
Since $| \Phi \rangle \langle \Phi | = (1/d) \sum_{ij} |i \rangle \langle j | \otimes |i \rangle \langle j |$, we have $(\mc{E}_\infty \otimes \mc{I})[| \Phi \rangle \langle \Phi |] 
= (1/d) \sum_{ij} \mc{E}_\infty[| i \rangle \langle j |] \otimes | i \rangle \langle j |
= (1/d) \sum_{ij} \Tr[| i \rangle \langle j |] \hat{\rho}_\mathrm{ss} \otimes | i \rangle \langle j | 
= (1/d) \sum_{ij} \delta_{ij} \hat{\rho}_\mathrm{ss} \otimes | i \rangle \langle j |
= (1/d) \hat{\rho}_\mathrm{ss} \otimes \hat{\mbb{I}}$.
} with a positive definite density operator $\hat{\rho}_\mathrm{ss} \in \mbb{B}[\mc{H}]$.
Denoting the smallest eigenvalue of $\hat{\rho}_\mathrm{ss}$ by $\lambda_\textrm{min}(\hat{\rho}_\mathrm{ss})$, we find
\begin{align}
0 
&< \lambda_\textrm{min}(\hat{\rho}_\mathrm{ss})
\leq \langle \Psi_m | \hat{\rho}_\mathrm{ss} \otimes \mbb{I} | \Psi_m \rangle 
= \langle \Psi_m | (\hat{\rho}_\mathrm{ss} \otimes \hat{\mbb{I}} - d \hat{E}^m) | \Psi_m \rangle \nonumber \\
&\leq d \| (\mc{E}_\infty \otimes \mc{I} - \mc{E}^m \otimes \mc{I}) [| \Phi \rangle \langle \Phi |] \|
\leq d \| \mc{E}_\infty \otimes \mc{I} - \mc{E}^m \otimes \mc{I} \|.
\end{align}
Since $\lim_{m \to \infty} \| \mc{E}_\infty \otimes \mc{I} - \mc{E}^m \otimes \mc{I} \| = 0$, this inequality leads to a contradiction for a sufficiently large but finite $m$ [see the proof of (4) $\Rightarrow$ (1) in Proposition~\ref{prop:primitivity}].

(3) $\Rightarrow$ (2): This is obvious.
For any nonzero $\ket{\psi}, \ket{\phi} \in \mc{H}$, we have $\hat{X} = |\phi \rangle \langle \psi | / \langle \psi | \psi \rangle \in \mbb{K}_n = \mbb{B}[\mc{H}]$ so that $\ket{\phi} = \hat{X} \ket{\psi}$.
Thus, $\mbb{K}_n \ket{\psi} = \mc{H}$ for any nonzero $\ket{\psi} \in \mc{H}$.
\end{proof}

Let us revisit the CPTP map given by Eq.~\eqref{eq:IrreducibleCPTPEx1} for a two-level system.
We find
\begin{align}
\mbb{K}_1 &= \mbb{K}_3 = \cdots = \mbb{K}_{2n-1} = \textrm{span} \{ \hat{M}_1, \hat{M}_2 \}, \\
\mbb{K}_2 &= \mbb{K}_4 = \cdots = \mbb{K}_{2n} = \textrm{span} \{ \hat{M}_1 \hat{M}_2, \hat{M}_2 \hat{M}_1 \}.
\end{align}
Since none of these subspaces satisfy $\mbb{K}_m \ket{\psi} = \mc{H}$ for an arbitrary $\ket{\psi} \in \mc{H}$\footnote{
Let $\ket{\psi} = (a, b)^T \in \mc{H}$ with $(a,b) \neq (0,0)$.
For $\mbb{K}_{2n-1} \ket{\psi} = \mc{H}$ to be satisfied, $\hat{M}_1 \ket{\psi}$ and $\hat{M}_2 \ket{\psi}$ must be linearly independent.
However, this does not hold when $ab = 0$.
The same result applies to $\mbb{K}_{2n}$.
}, this CPTP map is irreducible but not primitive.
Consequently, the map possesses multiple eigenvalues of unit modulus, specifically $z=\pm1$.

We next consider a CPTP map $\mc{E}$ generated by the following set of Kraus operators
\begin{align}
\hat{M}_1 = \frac{1}{\sqrt{2}} \ket{0} \bra{1}, \quad
\hat{M}_2 = \frac{1}{\sqrt{2}} \ket{1} \bra{0}, \quad
\hat{M}_3 = \frac{1}{\sqrt{2}} (\ket{0} \bra{0} -\ket{1} \bra{1}).
\end{align}
Since their linear combinations generate the set of Pauli matrices $\hat{\sigma}_x = \sqrt{2}(\hat{M}_1+\hat{M}_2)$, $\hat{\sigma}_y = -i\sqrt{2}(\hat{M}_1-\hat{M}_2)$, and $\hat{\sigma}_z = \sqrt{2} \hat{M}_3$, we find $\mbb{K}_1 = \textrm{span} \{ \hat{\sigma}_x, \hat{\sigma}_y, \hat{\sigma}_z \}$, which fulfills (2) of Theorem~\ref{thm:PrimitivityByKraus}.
If we consider strings of length $m=2$, we find $\mathbb{K}_2 = \textrm{span} \{ \hat{\mbb{I}}_2, \hat{\sigma}_x, \hat{\sigma}_y, \hat{\sigma}_z \}$, which now fulfills both (2) and (3) of Theorem~\ref{thm:PrimitivityByKraus}.
Thus, the corresponding CPTP map $\mc{E}$ is primitive.
Indeed, we find that the eigenvalues of $\mc{E}$ are $z = 1, 0, -1/2, -1/2$, and the eigenvector for $z=1$ is $\hat{X} = \ket{0} \bra{0} + \ket{1} \bra{1}$, which is positive definite.

As another example of a primitive CPTP map, we consider CPTP maps generated by random matrices.
In Fig.~\ref{fig:CPTPEigenvalues}(b), we illustrate the complex eigenvalues of a CPTP map generated through random Kraus operators~\cite{bruzda2009random, sa2020spectral},
\begin{align}
\hat{M}_0 = \sqrt{1-p} \hat{U}, \quad
\hat{M}_1 = \sqrt{p} \hat{V}_{11}, \quad
\hat{M}_2 = \sqrt{p} \hat{V}_{21}.
\end{align}
Here, $\hat{U}$ is a $d \times d$ random unitary matrix and $\hat{V}_{ij}$ are $d \times d$ matrices that are blocks of a $2d \times 2d$ random unitary matrix $\hat{V} = \begin{pmatrix} \hat{V}_{11} & \hat{V}_{12} \\ \hat{V}_{21} & \hat{V}_{22} \end{pmatrix}$.
We then set $p=0.9$ and $d=6$.
Strings of such random Kraus operators, $\hat{M}_{b_\ell} \cdots \hat{M}_{b_2} \hat{M}_{b_1}$, span $\mbb{B}[\mc{H}]$ for a sufficiently large but fixed $\ell$, rendering the corresponding map primitive. 
As a result, the map has a nondegenerate eigenvalue $z=1$.

\subsection{Steady states of quantum master equations}
\label{sec:SteadyStateGKSL}

While we have considered discrete CPTP maps so far, the uniqueness of steady states is equally relevant for quantum master equations described by the GKSL equation.
As discussed in Sec.~\ref{sec:AveragedDynamics}, a linear map $\mc{L}: \mbb{B}[\mc{H}] \to \mbb{B}[\mc{H}]$ of the form 
\begin{align} \label{eq:GKSLSuperop}
\mc{L}[\hat{X}] = -i[\hat{H}, \hat{X}] + \sum_b \left( \hat{L}_b \hat{X} \hat{L}_b^\dagger - \frac{1}{2} \{ \hat{X}, \hat{L}_b^\dagger \hat{L}_b \} \right)
\end{align}
with $\hat{H}=\hat{H}^\dagger \in \mbb{B}[\mc{H}]$ and $\hat{L}_b \in \mbb{B}[\mc{H}]$ generates a family of CPTP maps $\mc{E}_t = e^{\mc{L} t}$ parametrized by $t \geq 0$ that satisfy the Markovian condition $\mc{E}_t \circ \mc{E}_s = \mc{E}_{t+s}$ for $t,s \geq 0$.
It is often convenient to express the GKSL superoperator $\mc{L}$ in one of the following forms:
\begin{align}
\label{eq:GKSLwithK}
\mc{L}[\hat{X}] 
&= \sum_b \hat{L}_b \hat{X} \hat{L}_b^\dagger -\hat{K} \hat{X} -\hat{X} \hat{K}^\dagger \\
\label{kmat}
&= -i[\hat{H}', \hat{X}] + \sum_{k,l=1}^{d^2-1} \frac{C_{k,l}}{2} \left( [\hat{F}_l, \hat{X} \hat{F}_k^\dagger] + [\hat{F}_l \hat{X}, \hat{F}_k^\dagger] \right).
\end{align}
Here, $\hat{K} \in \mbb{B}[\mc{H}]$ satisfies $\hat{K} + \hat{K}^\dagger = \sum_b \hat{L}_b^\dagger \hat{L}_b$; $\hat{H}'$ is defined as $\hat{H}' = \hat{H} +\frac{i}{2d} \sum_b ( \Tr[\hat{L}_b^\dagger] \hat{L}_b -\Tr[\hat{L}_b] \hat{L}_b^\dagger)$; $C$ is a $(d^2-1) \times (d^2-1)$ positive semidefinite matrix called the Kossakowski matrix; and $\{ \hat{F}_k \}_{k=1}^{d^2-1}$ is a complete orthonormal basis for the space of traceless operators in $\mbb{B}[\mc{H}]$.
The equivalence between Eqs.~\eqref{eq:GKSLSuperop} and~\eqref{eq:GKSLwithK} is readily established by setting
\begin{align}
\hat{K} = i\hat{H} + \frac{1}{2} \sum_b \hat{L}_b^\dagger \hat{L}_b=i\hat{H}_\mr{eff}.
\end{align}
Equation~\eqref{kmat} is also equivalent to Eq.~\eqref{eq:GKSLSuperop}.
We observe that Eq.~\eqref{eq:GKSLSuperop} is invariant under the transformation $\hat{L}_b \to \hat{L}_b + c_b \hat{\mbb{I}}$ and $\hat{H} \to \hat{H} -\frac{i}{2} \sum_b (c_b^* \hat{L}_b -c_b \hat{L}_b^\dagger)$.
This allows us to make $\hat{L}_b$ traceless by choosing $c_b = -\Tr[\hat{L}_b]/d$, expand them as $\hat{L}_b = \sum_{k=1}^{d^2-1} B_{b,k} \hat{F}_k$, and set $C = B^\dagger B$.
We note that Eqs.~\eqref{eq:GKSLSuperop} and~\eqref{eq:GKSLwithK} were originally derived by Lindblad in Ref.~\cite{lindblad1976generators} while Eq.~\eqref{kmat} was derived by Gorini, Kossakowski, and Sudarshan in Ref.~\cite{gorini1976completely}.

Here, we consider the spectral properties of Markovian CPTP maps $\mc{E}_t$ generated by a GKSL superoperator $\mc{L}$.
Since $\mc{E}_t = e^{\mc{L}t}$ constitutes a family of CPTP maps, we can relate the spectral properties of $\mc{E}_t$ to those of $\mc{L}$.
Specifically, if we denote by $\{ \lambda_a \}_{a=1}^{d^2}$ the eigenvalues of $\mc{L}$, it follows that (i') the complex eigenvalues appear in conjugate pairs $\lambda_a$ and $\lambda_a^*$, (ii') there exists an eigenvalue $\lambda_a =0$\footnote{
Since $\mc{E}_t = e^{\mc{L}t}$ is a CPTP map, $\mc{L}$ must have an eigenvalue $\lambda$ satisfying $e^{\lambda t}=1$.
Since this holds for any $t>0$, such a $\lambda$ must be $0$.
}, and (iii') any eigenvalue has a non-positive real part $\textrm{Re}[\lambda_a] \leq 0$.
In particular, the existence of a stationary state $\hat{\varrho}_0 \in \mbb{B}[\mc{H}]$, i.e., a state satisfying $\mc{E}_t[\hat{\varrho}_0] = \hat{\varrho}_0$, implies the existence of a nontrivial kernel for $\mc{L}$ such that $\mc{L}[\hat{\varrho}_0] = 0$.
For example, when the operators $\hat{L}_b$ are Hermitian, the maximally mixed state $\hat{\varrho}_0 = \hat{\mbb{I}}/d$ is a stationary state since $\mc{L}[\hat{\mbb{I}}]=0$.
However, the general spectral properties of $\mc{L}$ alone do not ensure the existence of a unique stationary state, and we need additional conditions analogous to the irreducibility or primitivity of CPTP maps discussed in the previous sections.

To proceed, we first argue that for the continuous family $\mc{E}_t = e^{\mc{L}t}$ generated by $\mathcal{L}$, irreducibility is equivalent to primitivity~\cite{wolf2012quantum}.
Let us suppose that $\mc{E}_{t_0}$ is irreducible for some $t_0>0$.
Then, according to Theorem~\ref{thm:IrreducibilityBySpectrum}, $\mc{E}_{t_0}$ has a nondegenerate eigenvalue $1$ and the corresponding eigenvector is positive definite.
Furthermore, according to Theorem~\ref{thm:CPTPPeripheral}, $\mc{E}_{t_0}$ has a peripheral spectrum $S = \{ e^{2\pi ik/m}  \}_{k=0}^{m-1}$ with some $m \in \mbb{N}$.
Now suppose $m>1$.
For any irrational number $\alpha>1$, $\mc{E}_{\alpha t_0}$ also has a nondegenerate eigenvalue $1$ with the same eigenvector and thus is irreducible.
On the other hand, $\mc{E}_{\alpha t_0}$ has eigenvalues of unit modulus $e^{2\pi i\alpha k/m}$.
Since $\alpha$ is irrational, these are not roots of unity for $k \neq 0$, which contradicts the structure of the peripheral spectrum required by Theorem~\ref{thm:CPTPPeripheral}.
Therefore, we must have $m=1$, and thus $\mc{E}_{t_0}$ is primitive.
This implies that $\lambda=0$ is the only eigenvalue of $\mc{L}$ with a vanishing real part; consequently, $\mc{E}_t$ is primitive for all $t>0$.

We are now ready to characterize the irreducibility of Markovian CPTP maps~\cite{wolf2012quantum}.

\begin{proposition}[Irreducibility of Markovian CPTP maps]
\label{prop:IrreducibleMarkovCPTP}
Let $\mc{E}_t = e^{\mc{L}t}: \mbb{B}[\mc{H}] \to \mbb{B}[\mc{H}]$ with $t \geq 0$ be a family of CPTP maps generated by a GKSL superoperator $\mathcal{L}$ as defined in Eq.~\eqref{eq:GKSLwithK}.
Then the following statements are equivalent.
\begin{enumerate}
\item $\mc{E}_t$ is primitive for all $t>0$.
\item Any density operator $\hat{\varrho}_0 \in \mbb{B}[\mc{H}]$ that satisfies $\mc{L}[\hat{\varrho}_0]=0$ is positive definite.
\item There exists a unique steady state $\hat{\rho}_\mathrm{ss} \succ 0$ such that for any density operator $\hat{\rho} \in \mbb{B}[\mc{H}]$ we have $\lim_{t \to \infty} \mc{E}_t [\hat{\rho}] = \hat{\rho}_\mathrm{ss}$.
\item There exists no orthogonal projection operator $\hat{P} \in \mbb{B}[\mc{H}]$ such that $\hat{P} \notin \{ 0, \hat{\mbb{I}} \}$ and $\mc{L}[\hat{P}\mbb{B}[\mc{H}]\hat{P}] \subseteq \hat{P} \mbb{B}[\mc{H}] \hat{P}$.
\item There exists no orthogonal projection operator $\hat{P} \in \mbb{B}[\mc{H}]$ such that $\hat{P} \notin \{ 0, \hat{\mbb{I}} \}$ and $(\hat{\mbb{I}}-\hat{P}) \hat{L}_b \hat{P} = (\hat{\mbb{I}}-\hat{P}) \hat{K} \hat{P} = 0$.
\end{enumerate}
\end{proposition}

\begin{proof}
(1) $\Rightarrow$ (2) follows directly from the previous discussion.
To prove (2) $\Rightarrow$ (1), note that if every density operator $\hat{\varrho}_0 \in \mbb{B}[\mc{H}]$ satisfying $\mc{L}[\hat{\varrho}_0]=0$ is positive definite, it must also be unique; see the proof of Theorem~\ref{thm:IrreducibilityBySpectrum}.
Thus, $\mc{E}_t$ is irreducible and therefore primitive for all $t \geq 0$.
(1) $\Leftrightarrow$ (3) follows from (4) of Proposition~\ref{prop:primitivity} for primitive CPTP maps.

(1) $\Leftrightarrow$ (4): 
Suppose that $\mc{E}_t$ is not primitive. 
Then $\mc{E}_t$ is not irreducible, and there exists a nontrivial orthogonal projection operator $\hat{P} \notin \{ 0, \hat{\mbb{I}} \}$ such that for any $\hat{X} \in \hat{P} \mbb{B}[\mc{H}] \hat{P}$ we have $\mc{E}_t [\hat{X}] \in \hat{P} \mbb{B}[\mc{H}] \hat{P}$. 
Since $\mc{L}[\hat{X}] = \lim_{t \to 0^+} (\mc{E}_t[\hat{X}] -\hat{X})/t$, we also have $\mc{L}[\hat{X}] \in \hat{P} \mbb{B}[\mc{H}] \hat{P}$. 
Conversely, suppose that there exists a nontrivial $\hat{P}$ such that $\mc{L}[\hat{X}] \in \hat{P} \mbb{B}[\mc{H}] \hat{P}$ for any $\hat{X} \in \hat{P} \mbb{B}[\mc{H}] \hat{P}$. 
Then, we have $\mc{E}_t[\hat{X}] \in \hat{P} \mbb{B}[\mc{H}] \hat{P}$ since $\mc{E}_t[\hat{X}] = \hat{X} + \sum_{n=1}^\infty (t^n/n!) \mc{L}^n[\hat{X}]$.
This means that $\mc{E}_t$ is not irreducible and, therefore, is not primitive.

(5) $\Leftrightarrow$ (4): 
Suppose that there exists a nontrivial $\hat{P} \notin \{ 0, \hat{\mathbb{I}} \}$ such that $(\hat{\mbb{I}}-\hat{P}) \hat{L}_b \hat{P} = (\hat{\mbb{I}}-\hat{P}) \hat{K} \hat{P} =0$.
Then, we have $\hat{L}_b \hat{P} = \hat{P} \hat{L}_b \hat{P}$ and $\hat{K} \hat{P} = \hat{P} \hat{K} \hat{P}$.
This implies that for any $\hat{X} \in \mbb{B}[\mc{H}]$ we have $\mc{L}[\hat{P} \hat{X} \hat{P}] = \hat{P} \mc{L}[\hat{P} \hat{X} \hat{P}] \hat{P}$ and thus $\mc{L}[\hat{P} \mbb{B}[\mc{H}] \hat{P}] \subseteq \hat{P} \mbb{B}[\mc{H}] \hat{P}$. 
For the converse, suppose that there exists a nontrivial $\hat{P}$ such that $\mc{L}[\hat{P} \mbb{B}[\mc{H}] \hat{P}] \subseteq \hat{P} \mbb{B}[\mc{H}] \hat{P}$. 
Then, $\mc{L}[\hat{P}] \in \hat{P} \mbb{B}[\mc{H}] \hat{P}$ and thus
\begin{align} \label{eq:TriangularL}
(\hat{\mbb{I}} -\hat{P}) \mc{L}[\hat{P}] = \sum_b (\hat{\mbb{I}}-\hat{P}) \hat{L}_b \hat{P} \hat{L}_b^\dagger -(\hat{\mbb{I}} -\hat{P}) \hat{K} \hat{P} = 0.
\end{align}
Multiplying this equation by $\hat{\mbb{I}} -\hat{P}$ from the right yields $\sum_b \hat{Y}_b \hat{Y}_b^\dagger = 0$ with $\hat{Y}_b = (\hat{\mbb{I}} -\hat{P}) \hat{L}_b \hat{P}$.
Since $\hat{Y}_b \hat{Y}_b^\dagger \succeq 0$, this implies $\hat{Y}_b \hat{Y}_b^\dagger = 0$ and hence $\hat{Y}_b = 0$\footnote{
If $\hat{X} \hat{X}^\dagger = 0$ for $\hat{X} \in \mbb{B}[\mc{H}]$, then $\hat{X}=0$. 
\textit{Proof.} Since $\hat{X} \hat{X}^\dagger = 0$, $\Tr [\hat{X} \hat{X}^\dagger] = 0$.
Expanding in an orthonormal basis $\{ |i\rangle \}_{i=1}^d$, we have $\Tr[\hat{X} \hat{X}^\dagger] = \sum_{i,j} |\langle i | \hat{X} | j \rangle|^2$ = 0, which gives $\langle i | \hat{X} | j \rangle = 0$ and thus $\hat{X}=0$.
}.
Substituting $(\hat{\mbb{I}}-\hat{P}) \hat{L}_b \hat{P} =0$ back into Eq.~\eqref{eq:TriangularL} yields $(\hat{\mbb{I}} -\hat{P}) \hat{K} \hat{P}=0$.
\end{proof}

Early discussions on the irreducibility of Markovian CPTP maps can be found in Refs.~\cite{davies1970quantum, evans1977irreducible, frigerio1977quantum, frigerio1978stationary, spohn1980kinetic}.
Condition (5) in Proposition~\ref{prop:IrreducibleMarkovCPTP} was proven in Ref.~\cite{fagnola2002subharmonic}, covering even the cases of infinite-dimensional $\mc{H}$\footnote{The irreducibility of Markovian CPTP maps for infinite-dimensional Hilbert spaces has been summarized in Ref.~\cite{fagnola2025irreducibility}.}.
For finite-dimensional $\mc{H}$, condition (5) immediately yields the following algebraic condition for $\{ \hat{L}_b, \hat{K} \}$~\cite{wolf2012quantum}, whose applications have recently been discussed in Refs.~\cite{yoshida2024uniqueness, zhang2024criteria}.

\begin{theorem}[Wolf~\cite{wolf2012quantum}, Yoshida~\cite{yoshida2024uniqueness}, Zhang-Barthel~\cite{zhang2024criteria}]
\label{thm:YoshidaTheorem}
Let $\mc{E}_t = e^{\mc{L}t} : \mbb{B}[\mc{H}] \to \mbb{B}[\mc{H}]$ be a Markovian CPTP map generated by a GKSL superoperator $\mc{L}$ as defined in Eq.~\eqref{eq:GKSLwithK}.
Then $\mc{E}_t$ is irreducible if and only if the algebra generated by $\{ \hat{L}_b, \hat{K} \}$ coincides with $\mbb{B}[\mc{H}]$.
\end{theorem}

\begin{proof}
Suppose that $\mc{E}_t$ is irreducible in the sense of (5) in Proposition~\ref{prop:IrreducibleMarkovCPTP}; that is, there exists no nontrivial orthogonal projection operator $\hat{P} \notin \{ 0, \hat{\mathbb{I}} \}$ such that $(\hat{\mbb{I}}-\hat{P}) \hat{L}_b \hat{P} = (\hat{\mbb{I}}-\hat{P}) \hat{K} \hat{P} =0$.
This equivalently means that $\hat{P} \mc{H}$ is a subspace of $\mc{H}$ invariant under the actions of $\hat{L}_b$ and $\hat{K}$, which can only be $\{0\} $ or $\mc{H}$ itself.
We now consider the algebra $\mbb{A}$ (a subset of $\mbb{B}[\mc{H}]$ closed under scalar multiplication, addition, and multiplication) generated by $\{ \hat{L}_b, \hat{K} \}$.
Then, the only subspaces of $\mc{H}$ invariant under the action of $\mbb{A}$ are $\{0\}$ and $\mc{H}$.
This means that $\mbb{A}$ is irreducible, and by Burnside's theorem on matrix algebras (see Appendix~\ref{app:BurnsideTheorem}), this is equivalent to $\mbb{A} = \mbb{B}[\mc{H}]$.
\end{proof}

Reference~\cite{jaksic2014entropic} derived a similar sufficient condition, which states that $\mc{E}_t$ is irreducible if the co GKSL superoperator $\mc{L}$ is written with an irreducible CP map $\mc{T}$ in the form of $\mc{L}[\hat{X}] = -i[\hat{H}, \hat{X}] -\frac{1}{2} \{ \hat{X}, \mc{T}^\dagger [\hat{\mbb{I}}] \} +\mc{T}[\hat{X}]$.
As the CP map $\mc{T}$ admits the Kraus representation $\mc{T}[\hat{X}] = \sum_b \hat{L}_b \hat{X} \hat{L}_b^\dagger$, this amounts to the requirement that the algebra generated by $\{ \hat{L}_b \}$ coincides with $\mbb{B}[\mc{H}]$, which satisfies the condition required by Theorem~\ref{thm:YoshidaTheorem}.

We can also derive other known sufficient conditions for the existence of a unique steady state.

\begin{theorem}[Wolf~\cite{wolf2012quantum}]
\label{thm:Kossakowski}
Let $\mc{E}_t = e^{\mc{L}t} : \mbb{B}[\mc{H}] \to \mbb{B}[\mc{H}]$ be a Markovian CPTP map generated by a GKSL superoperator $\mc{L}$ as defined in Eq.~\eqref{kmat}.
Then $\mc{E}_t$ is irreducible if the Kossakowski matrix $C$ satisfies
\begin{align}
\textrm{rank}(C) > d^2-d.
\end{align}
\end{theorem}

\begin{proof}
We proceed by contraposition. 
Suppose that $\mc{E}_t$ is not irreducible. 
Then, according to (5) of Proposition~\ref{prop:IrreducibleMarkovCPTP}, there exists a nontrivial orthogonal projection operator $\hat{P} \notin \{ 0, \hat{\mbb{I}} \}$ such that $(\hat{\mbb{I}}-\hat{P}) \hat{L}_b \hat{P} = (\hat{\mbb{I}}-\hat{P}) \hat{K} \hat{P} = 0$.
This implies that $\hat{L}_b$ and $\hat{K}$ assume a block upper-triangular form in an appropriate basis\footnote{
Let $r = \textrm{rank}(\hat{P})$.
We can choose a basis such that $\hat{P}$ has a block-diagonal form $\hat{P} = \begin{pmatrix} \mbb{I}_r & 0 \\ 0 & 0 \end{pmatrix}$ with $\mbb{I}_r$ being the $r \times r$ identity matrix.
In this basis, $\hat{L}_b$ must have a block upper-triangular form $\hat{L}_b = \begin{pmatrix} L_{AA} & L_{AB} \\ 0 & L_{BB} \end{pmatrix}$ to satisfy $(\hat{\mbb{I}}-\hat{P}) \hat{L}_b \hat{P} = 0$.
}.
The subspace of traceless operators that are block upper-triangular with respect to a projection of rank $r$ has dimension $d^2-1 - r(d-r)$.
Since $1 \leq r \leq d-1$, the minimum number of zero entries in the lower-left block is $d-1$, which implies that the dimension of this subspace is at most $d^2-1 - (d-1) = d^2-d$.
Consequently, the operators $\{ \hat{L}_b \}$ span a space of dimension at most $d^2-d$, which implies $\textrm{rank}(C) \leq d^2-d$.
\end{proof}

In Ref.~\cite{spohn1976approach}, Spohn proved that $\mc{E}_t$ has a unique steady state if $\dim (\textrm{ker}(C)) < d/2$. 
Since the latter condition is equivalent to $\textrm{rank}(C) > d^2 -d/2-1$, Theorem~\ref{thm:Kossakowski} gives a more general condition for the irreducibility of $\mc{E}_t$.

Another condition for the existence of a unique steady state of $\mc{E}_t$ was first derived by Spohn~\cite{spohn1977algebraic} and subsequently extended to the cases of infinite-dimensional $\mc{H}$ by Frigerio~\cite{frigerio1978stationary}.
To state Spohn's theorem for finite-dimensional $\mc{H}$, we need to introduce several notions regarding sets of bounded operators.
A set $\mbb{A} \subseteq \mbb{B}[\mc{H}]$ is said to be self-adjoint if for every $\hat{X} \in \mbb{A}$, $\hat{X}^\dagger$ is also an element of $\mbb{A}$.
The set $\mbb{A}' \subseteq \mbb{B}[\mc{H}]$ denotes the commutant of $\mbb{A}$, which is defined by
\begin{align}
\mbb{A}' := \{ \hat{X} \in \mbb{B}[\mc{H}] : [\hat{X}, \hat{A}]=0 \ \textrm{for any} \ \hat{A} \in \mbb{A} \}.
\end{align}
We are now ready to state Spohn's theorem.

\begin{theorem}[Spohn~\cite{spohn1977algebraic}]
\label{thm:SpohnTheorem}
Let $\mc{E}_t = e^{\mc{L}t} : \mbb{B}[\mc{H}] \to \mbb{B}[\mc{H}]$ be a Markovian CPTP map generated by a GKSL superoperator as defined in Eq.~\eqref{eq:GKSLSuperop}. 
Then $\mc{E}_t$ is irreducible if the complex linear span of $\{ \hat{L}_b \}$ is a self-adjoint set and $\{ \hat{L}_b \}' = \mbb{C} \hat{\mbb{I}}$. 
\end{theorem}

\begin{proof}
If an algebra $\mbb{A} \subseteq \mbb{B}[\mc{H}]$ is self-adjoint, then by Schur's lemma (see Appendix~\ref{app:SchurLemma}), $\mbb{A}' = \mbb{C} \hat{\mbb{I}}$ if and only if $\mbb{A}$ is irreducible.
By Burnside's theorem on matrix algebras (see Appendix~\ref{app:BurnsideTheorem}), $\mbb{A}$ is irreducible if and only if $\mbb{A} = \mbb{B}[\mc{H}]$.
Now, let $\mbb{A}$ be the algebra generated by $\{ \hat{L}_b \}$.
Since the complex linear span of $\{ \hat{L}_b \}$ is a self-adjoint set, $\mbb{A}$ is also a self-adjoint set\footnote{
The complex linear span of $\{ \hat{L}_b \}$ is a self-adjoint set if for every $\hat{X} = \sum_b x_b \hat{L}_b$ with $x_b \in \mbb{C}$, there exist $y_b \in \mbb{C}$ such that $\hat{X}^\dagger = \sum_b y_b \hat{L}_b$.
Thus, there exist $y_{bb'} \in \mbb{C}$ such that $\hat{L}_b^\dagger = \sum_{b'} y_{bb'} \hat{L}_b$ for every $\hat{L}_b$.
Any element $\hat{A} \in \mbb{A}$ can be written as $\hat{A} = \sum_{b_1,\cdots,b_n} a_{b_1,\cdots,b_n} \hat{L}_{b_1} \cdots \hat{L}_{b_n}$.
Then, we have $\hat{A}^\dagger = \sum_{b_1,\cdots,b_n} a_{b_1,\cdots,b_n}^* \hat{L}_{b_n}^\dagger \cdots \hat{L}_{b_1}^\dagger = \sum_{b_1,\cdots,b_n} \sum_{b'_1,\cdots,b'_n} a_{b_1,\cdots,b_n}^* y_{b_1 b'_1} \cdots y_{b_n b'_n} \hat{L}_{b'_n} \cdots \hat{L}_{b'_1} \in \mbb{A}$.
Hence, $\mbb{A}$ is a self-adjoint set.
}.
Furthermore, since $\{ \hat{L}_b \}' = \mbb{C} \hat{\mbb{I}}$, we have $\mbb{A}' = \mbb{C} \hat{\mbb{I}}$, and thus $\mbb{A} = \mbb{B}[\mc{H}]$. 
Then the algebra generated by $\{ \hat{L}_b, \hat{K} \}$ coincides with the whole space of $\mbb{B}[\mc{H}]$ since it obviously contains $\mbb{A}$ as a subalgebra.
By Theorem~\ref{thm:YoshidaTheorem}, this implies that $\mc{E}_t$ is irreducible.
\end{proof}

Spohn's theorem provides a sufficient condition for Theorem~\ref{thm:YoshidaTheorem} to hold.
As implied by its proof, this theorem can be generalized by replacing $\{ \hat{L}_b \}$ with $\{ \hat{L}_b, \hat{K} \}$~\cite{zhang2024criteria}.
This offers a useful guiding principle for finding open many-body quantum systems with a unique steady state. 
For instance, consider a nonintegrable (or ``chaotic'') Hamiltonian, such as a mixed-field Ising chain, 
\begin{align}
\hat{H} = -J \sum_{l=1}^V \hat{\sigma}^z_l \hat{\sigma}^z_{l+1} - h_x \sum_{l=1}^V \hat{\sigma}^x_l -h_z \sum_{l=1}^V \hat{\sigma}^z_l,
\end{align}
whose only local conserved quantity is energy~\cite{banuls2011strong, chiba2024proof}.
In such cases, adding a self-adjoint set of jump operators $\{ \hat{L}_b \}$ with $[\hat{H},\hat{L}_b]\neq0$ is expected to yield $\{ \hat{L}_b, \hat{K} \}' = \mbb{C} \hat{\mbb{I}}$ in general.
Of course, this approach is far from rigorous, and one may need to employ other criteria, such as Theorem~\ref{thm:YoshidaTheorem}, to rigorously establish the uniqueness of a steady state (see also Ref.~\cite{seltmann2025uniqueness}).

Frigerio derived a similar sufficient condition for the uniqueness of steady states~\cite{frigerio1977quantum}.
While this condition extends to the cases of infinite-dimensional $\mc{H}$, it relies on a nontrivial assumption regarding the existence of a full-rank stationary state.
Here, we provide its finite-dimensional version.

\begin{theorem}[Frigerio~\cite{frigerio1977quantum}]
\label{thm:FrigerioTheorem}
Let $\mc{E}_t = e^{\mc{L}t}: \mbb{B}[\mc{H}] \to \mbb{B}[\mc{H}]$ be a Markovian CPTP map generated by a GKSL superoperator as defined in Eq.~\eqref{eq:GKSLSuperop}. 
Then $\mc{E}_t$ is irreducible if $\{ \hat{L}_b, \hat{L}_b^\dagger, \hat{H} \}' = \mbb{C} \hat{\mbb{I}}$ and there exists a positive definite density operator $\hat{\varrho}_0 \in \mbb{B}[\mc{H}]$ such that $\mc{L}[\hat{\varrho}_0]=0$.
\end{theorem}

\begin{proof}
The proof is based on Ref.~\cite{wolf2012quantum}. 
We first show that if there exists a positive definite density operator $\hat{\varrho}_0 \succ 0$ such that $\mc{L}[\hat{\varrho}_0]=0$, we have $\{ \hat{L}_b, \hat{L}_b^\dagger, \hat{H} \}' = \ker(\mc{L}^\dagger) := \{ \hat{A} \in \mbb{B}[\mc{H}] : \mc{L}^\dagger[\hat{A}] = 0 \}$, where $\mc{L}^\dagger: \mbb{B}[\mc{H}] \to \mbb{B}[\mc{H}]$ is the dual map of $\mc{L}$ governing the Heisenberg evolution.
It is easy to show that $\{ \hat{L}_b, \hat{L}_b^\dagger, \hat{H} \}' \subseteq \ker(\mc{L}^\dagger)$; for any $\hat{X} \in \{ \hat{L}_b, \hat{L}_b^\dagger, \hat{H} \}'$, we have 
\begin{align}
\mc{L}^\dagger[\hat{X}] = i[\hat{H}, \hat{X}] + \sum_b \left( \hat{L}_b^\dagger \hat{X} \hat{L}_b -\frac{1}{2} \{ \hat{X}, \hat{L}_b^\dagger \hat{L}_b \} \right) = \sum_b \left( \hat{X} \hat{L}_b^\dagger \hat{L}_b -\hat{X} \hat{L}_b^\dagger \hat{L}_b \right) = 0.
\end{align}
To prove the reverse inclusion $\ker(\mc{L}^\dagger) \subseteq \{ \hat{L}_b, \hat{L}_b^\dagger, \hat{H} \}'$, note that $\mc{E}_t^\dagger = e^{\mc{L}^\dagger t}$ is a CP unital map and thus satisfies the Kadison-Schwarz inequality $\mc{E}_t^\dagger[\hat{A}^\dagger \hat{A}] \succeq \mc{E}_t^\dagger[\hat{A}^\dagger] \mc{E}_t^\dagger[\hat{A}]$ for $\hat{A} \in \mbb{B}[\mc{H}]$ (see Appendix~\ref{app:SchwarzMap}).
Let $\hat{A} \in \ker(\mc{L}^\dagger)$. 
Using $\mc{E}_t^\dagger[\hat{A}] = \hat{A}$, $\mc{E}_t^\dagger[\hat{A}^\dagger] = \hat{A}^\dagger$, and $\mc{E}_t[\hat{\varrho}_0] = \hat{\varrho}_0$, we find
\begin{align}
0 \leq \Tr[(\mc{E}_t^\dagger[\hat{A}^\dagger \hat{A}] - \mc{E}_t^\dagger[\hat{A}^\dagger] \mc{E}_t^\dagger[\hat{A}]) \hat{\varrho}_0] = \Tr[\hat{A}^\dagger \hat{A} \mc{E}_t[\hat{\varrho}_0] - \hat{A}^\dagger \hat{A} \hat{\varrho}_0] = 0.
\end{align}
Since $\hat{\varrho}_0 \succ 0$, we must have the equality $\mc{E}_t^\dagger[\hat{A}^\dagger \hat{A}] = \mc{E}_t^\dagger[\hat{A}^\dagger] \mc{E}_t^\dagger[\hat{A}] = \hat{A}^\dagger \hat{A}$ and thus $\hat{A}^\dagger \hat{A} \in \ker(\mc{L}^\dagger)$.
We then find
\begin{align}
\sum_b [\hat{A}, \hat{L}_b]^\dagger [\hat{A}, \hat{L}_b] 
&= \sum_b \left( \hat{L}_b^\dagger \hat{A}^\dagger \hat{A} \hat{L}_b +\hat{A}^\dagger \hat{L}_b^\dagger \hat{L}_b \hat{A} - \hat{L}_b^\dagger \hat{A}^\dagger \hat{L}_b \hat{A} - \hat{A}^\dagger \hat{L}_b^\dagger \hat{A} \hat{L}_b \right) \nonumber \\
&= \mc{L}^\dagger[\hat{A}^\dagger \hat{A}] -\hat{A}^\dagger \mc{L}^\dagger[\hat{A}] -\mc{L}^\dagger[\hat{A}^\dagger] \hat{A} 
= 0.
\end{align}
Since the left-hand side is a sum of positive semidefinite operators, we must have $[\hat{A}, \hat{L}_b] = 0$. 
Similarly, we can also find $[\hat{A}^\dagger, \hat{L}_b] = 0$ and thus $[\hat{L}_b^\dagger, \hat{A}]=0$.
Substituting these into $\mc{L}^\dagger[\hat{A}]=0$ yields $[\hat{H}, \hat{A}] = 0$. 
This proves $\ker(\mc{L}^\dagger) \subseteq \{ \hat{L}_b^\dagger, \hat{L}_b, \hat{H} \}'$ and hence $\ker(\mc{L}^\dagger) = \{ \hat{L}_b, \hat{L}_b^\dagger, \hat{H} \}'$.

Now, if $\{ \hat{L}_b, \hat{L}_b^\dagger, \hat{H} \}' = \mbb{C} \hat{\mbb{I}}$, we have $\ker(\mc{L}^\dagger) = \mbb{C} \hat{\mbb{I}}$, implying that $\hat{\mbb{I}}$ is the only eigenvector of $\mc{L}^\dagger$ with eigenvalue $0$.
This implies that $\hat{\varrho}_0$ is the only eigenvector of $\mc{L}$ with eigenvalue $0$.
Since $\hat{\varrho}_0 \succ 0$, $\mc{E}_t$ is irreducible according to (2) of Proposition~\ref{prop:IrreducibleMarkovCPTP}.
\end{proof}

We provide several examples of irreducible Markovian CPTP maps in two-level systems.
The first example is defined by a GKSL superoperator $\mc{L}$ in the form of Eq.~\eqref{eq:GKSLSuperop} with~\cite{yoshida2024uniqueness}
\begin{align}
\hat{H} = 0, \quad 
\hat{L}_1 = | 0 \rangle \langle 1 |, \quad 
\hat{L}_2 = | 1 \rangle \langle 0 |.
\end{align}
Since $\{ \hat{L}_1, \hat{L}_2 \}$ is a self-adjoint set and $\{ \hat{L}_1, \hat{L}_2 \}' = \mathbb{C} \hat{\mbb{I}}$, the Markovian CPTP map $e^{\mc{L}t}$ is irreducible according to Spohn's theorem (Theorem~\ref{thm:SpohnTheorem}).
Indeed, the eigenvalues of $\mc{L}$ are $\lambda = 0$, $-1$, $-1$, $-2$ and the eigenvector corresponding to $\lambda = 0$ is $\hat{X} = | 0 \rangle \langle 0 | + | 1 \rangle \langle 1 |$, which is positive definite.

The second example is taken from Ref.~\cite{zhang2024criteria}:
\begin{align}
\hat{H} = 0, \quad 
\hat{L}_1 = | 0 \rangle \langle 0 | + | 0 \rangle \langle 1 | + | 1 \rangle \langle 1 |.
\end{align}
Here, $\{ \hat{L}_1 \}$ is not self-adjoint, and thus Spohn's theorem does not apply.
On the other hand, we have $2\hat{K} = \hat{L}_1^\dagger \hat{L}_1 = | 0 \rangle \langle 0 | + | 0 \rangle \langle 1 | + | 1 \rangle \langle 0 | +2 | 1 \rangle \langle 1 |$ and the algebra generated by $\{ \hat{L}_1, \hat{K} \}$ coincides with $\mbb{B}[\mc{H}]$. 
This is evident from the relations $| 0 \rangle \langle 1 | = \hat{L}_1^2 -\hat{L}_1$ and $| 1 \rangle \langle0 | = 2\hat{L}_1 \hat{K} -\hat{L}_1 -\hat{L}_1^2$, which allow us to generate all basis oparators.
Thus, the corresponding Markovian CPTP map $e^{\mc{L}t}$ is irreducible according to Theorem~\ref{thm:YoshidaTheorem}.
We find that the eigenvalues of $\mc{L}$ are $\lambda = 0$, $-1/2$, $(-3 \pm i \sqrt{15})/4$ and the eigenvector corresponding to $\lambda=0$ is $\hat{X} = 2 | 0 \rangle \langle 0 | - | 0 \rangle \langle 1 | - | 1 \rangle \langle 0 | + | 1 \rangle \langle 1 |$, which is positive definite.

The third example is also from Ref.~\cite{zhang2024criteria}:
\begin{align}
\hat{H} = (| 0 \rangle \langle 1 | + | 1 \rangle \langle 0 |)/2, \quad 
\hat{L}_1 = | 0 \rangle \langle 1 |.
\end{align}
As before, $\{ \hat{L}_1 \}$ is not self-adjoint, so Spohn's theorem does not apply.
On the other hand, we have $2\hat{K} = 2i\hat{H} +\hat{L}_1^\dagger \hat{L}_1 = i| 0 \rangle \langle 1 | + i| 1 \rangle \langle 0 | + | 1 \rangle \langle 1 |$ and the algebra generated by $\{ \hat{L}_1, \hat{K} \}$ coincides with $\mbb{B}[\mc{H}]$. 
This can be verified by observing that $| 0 \rangle \langle 1 | = \hat{L}_1$ and $| 1 \rangle \langle 0 | = -4i \hat{K}^2 -2 \hat{L}_1 \hat{K}$.
The spectrum of $\mc{L}$ is the same as in the second example and the eigenvector corresponding to $\lambda = 0$ is $\hat{X} = 2 | 0 \rangle \langle 0 | +i | 0 \rangle \langle 1 | -i | 1 \rangle \langle 0 | + | 1 \rangle \langle 1 |$, which is positive definite.

Various other examples of irreducible and reducible Markovian CPTP maps, including many-body systems, have been discussed in Refs.~\cite{yoshida2024uniqueness, zhang2024criteria}. 
In particular, Ref.~\cite{zhang2024criteria} pointed out that some previous literature incorrectly used Frigerio's theorem (Theorem~\ref{thm:FrigerioTheorem}) to predict the presence of a unique steady state, by relying solely on the condition $\{ \hat{L}_b, \hat{L}_b^\dagger, \hat{H} \}' = \mathbb{C} \hat{\mathbb{I}}$ while overlooking the assumption of the existence of a positive definite stationary state.
As discussed in the proof of Theorem~\ref{thm:FrigerioTheorem}, if $\{ \hat{L}_b, \hat{L}_b^\dagger, \hat{H} \}'$ is nontrivial (i.e., $\mathbb{C} \hat{\mbb{I}} \subset \{ \hat{L}_b, \hat{L}_b^\dagger, \hat{H} \}'$), then $\mc{L}^\dagger$ must have degenerate zero eigenvalues, implying the existence of multiple steady states.
However, the converse is not true; the condition $\{ \hat{L}_b, \hat{L}_b^\dagger, \hat{H} \}' = \mbb{C} \hat{\mbb{I}}$ alone does not imply the uniqueness of a steady state.

Before closing this section, we mention some further progress on the uniqueness of steady states for the GKSL equation in recent years.
That is, the discussion on  uniqueness has been extended to \textit{time-dependent} GKSL equations~\cite{menczel2019limit,di2024asymptotic,yoshida2026theory,wolff2026contractivity}, where the Hamiltonian $\hat{H}_t$ and jump operators $\hat{L}_{b,t}$ can explicitly depend on $t$.
For example, Ref.~\cite{yoshida2026theory} demonstrated the necessary and sufficient condition for the time-quasiperiodic GKSL equation with Hermitian jump operators ($\hat{L}_{b,t}=\hat{L}_{b,t}^\dag$).
In this case, the uniqueness is equivalent to $\{\hat{\tilde{L}}_{b,t}\}'=\mathbb{C}\hat{\mathbb{I}}$ for all $t$, where $\hat{\tilde{L}}_{b,t}=U_t^\dag\hat{L}_{b,t}U_t$ is the jump operator at the interaction picture ($U_t=\mathrm{T}e^{-i\int_0^t\hat{H}_\tau d\tau}$ with $\mathrm{T}$ denoting time-ordering).
Alternatively, if the Liouvillian is analytic in time, the criterion is also equivalent to the following algebraic one:
let $\mathbb{A}_t$ be an algebra generated by
\aln{
\{\hat{\mathbb{I}},\:\hat{L}_{b,t},\mathrm{ad} _t[\hat{L}_{b,t}],\:\mathrm{ad}_t^2[\hat{L}_{b,t}],\:\cdots \}.
}
If and only if $\mathbb{A}_t=\mathbb{B}[\mathcal{H}]$ for some single time $t$, the steady state is unique, which is given by the maximally mixed state  because we assume that the jump operators are Hermitian.
Here, $\mathrm{ad}_t[\hat{A}_t]=i[\hat{H}_t,\hat{A}_t]+\partial_t\hat{A}_t$ is the adjoint operation.

\subsection{Miscellaneous topics on the spectra beyond the steady state}
\label{sec:miscellaneous}

In the previous sections, we have detailed the conditions under which CPTP maps and quantum master equations (i.e., the GKSL equations) have unique stationary states. 
In this section, we discuss other spectral properties, such as the spectral gap and statistics, which are less understood than the uniqueness property of a stationary state both physically and mathematically.
While the content in this section is basically not used in the next chapters, we here try to introduce a brief (biased) overview of some relevant topics so that readers can grasp some recent developments in the field. 

We first introduce the general structure of eigenvalues and eigenvectors for CPTP maps and GKSL generators.
For simplicity, we assume that full diagonalization is possible, unless stated otherwise.
As defined in Sec.~\ref{sec:GeneralSpecCPTP}, the $a$th right and left eigenvectors of a CPTP map $\ml{E}$ satisfy
\aln{
\ml{E}[\hat{\mathsf{R}}_a]&=z_a\hat{ \mathsf{R}}_a,\\
\ml{E}^\dag[\hat{\mathsf{L}}_a]&=z_a^*\hat{\mathsf{L}}_a.
}
Here, $\hat{\mathsf{R}}_a$ and $\hat{\mathsf{L}}_a$ are right and left eigenvectors, respectively, which can be taken to satisfy the biorthogonality condition
\aln{
\Tr[\hat{\mathsf{L}}_a^\dag \hat{\mathsf{R}}_b]\propto\delta_{ab}.
}
The eigenvalues $z_a\;(a=1,\ldots, d^2)$ with $d=\mathrm{dim}[\ml{H}]$ are arranged as
\aln{
1=z_1\geq |z_2|\geq\cdots \geq|z_{d^2}|.
}
Similarly, we can consider the eigenvalues of a GKSL superoperator as
\aln{
\ml{L}[\hat{r}_a]&=\lambda_a\hat{ r}_a,\\
\ml{L}^\dag[\hat{l}_a]&=\lambda_a^*\hat{l}_a,
}
where $\hat{{r}}_a$ and $\hat{{l}}_a$ are right and left eigenvectors, respectively, satisfying
\aln{
\Tr[\hat{l}_a^\dag \hat{r}_b]\propto\delta_{ab}.
}
The eigenvalues $\lambda_a\;(a=1,\ldots, d^2)$ are arranged as
\aln{
0=\lambda_1\geq \mr{Re}[\lambda_2]\geq \cdots \geq \mr{Re}[\lambda_{d^2}].
}
For the CPTP dynamics generated by the time-independent GKSL equation, $\ml{E}=e^{\ml{L}t}$, we have $\hat{\mathsf{R}}_a=\hat{r}_a$, $\hat{\mathsf{L}}_a=\hat{l}_a$, and $z_a=e^{\lambda_a t}$.

Using these eigenvalues and eigenvectors, we can evaluate the state after repeated applications of the CPTP map as 
\aln{\label{sddcptp}
\ml{E}^n[\hat{\rho}_0]=\sum_{a=1}^{d^2}C_az_a^n\hat{\mathsf{R}}_a,
}
where
\aln{\label{cacptp}
C_a=\frac{\Tr[\hat{\mathsf{L}}_a^\dag\hat{\rho}_0]}{\Tr[\hat{\mathsf{L}}_a^\dag\hat{\mathsf{R}}_a]}
}
is the overlap between the $a$th eigenvector and the initial state $\hat{\rho}_0$.
Similarly, for the continuous GKSL dynamics, we have
\aln{
e^{\ml{L}t}[\hat{\rho}_0]=\sum_{a=1}^{d^2}c_ae^{\lambda_at}\hat{r}_a
}
with 
\aln{\label{cadef}
c_a=\frac{\Tr[\hat{l}_a^\dag\hat{\rho}_0]}{\Tr[\hat{l}_a^\dag\hat{r}_a]}.
}

\subsubsection{Spectral properties near stationary states and relaxation timescales}
We first give an overview of the relevant topics concerning the spectral gap and relaxation.
Let us assume that the CPTP map is primitive, i.e., the eigenvalue of unit modulus is unique.
Then, Eq.~\eqref{sddcptp} can be rewritten as
\aln{
\ml{E}^n[\hat{\rho}_0]=\hat{\rho}_\mr{ss}+\sum_{a=2}^{d^2}C_az_a^n\hat{\mathsf{R}}_a,
}
where $\hat{\rho}_\mr{ss}=\hat{\mathsf{R}}_1/\mathrm{Tr}(\hat{\mathsf{R}}_1)$ is the unique stationary state of $\ml{E}$ and we have used $C_1=1/\Tr(\hat{\sf{R}}_1)$ because $\hat{\mathsf{L}}_1\propto \hat{\mathbb{I}}$ in Eq.~\eqref{cacptp}.
Since each term in the sum on the right-hand side decays exponentially as $z_a^n$ $(z_a<1)$ as $n$ increases, the asymptotic decay rate of the dynamics is governed by the longest-lived mode.
This rate is given by the (logarithm of) spectral gap
\aln{
\Delta_{\mr{asy}}=-\ln |z_2|
}
as long as $C_2\neq 0$.
This quantity means that the state approaches the stationary state with an exponential decay $e^{-\Delta_\mr{asy}t}$\footnote{
Note that the decay may be accompanied by an oscillation whose frequency is given by $\mr{arg}[z_2]$.
} for asymptotically large $t$ with finite $d^2$ and $C_2\neq 0$.

Similarly, for the GKSL case, we have the expansion
\aln{\label{expgksl}
e^{\ml{L}t}[\hat{\rho}_0]=\hat{\rho}_\mr{ss}+\sum_{a=2}^{d^2}c_ae^{\lambda_at}\hat{r}_a
}
with $\hat{\rho}_\mr{ss}=\hat{r}_1/\mathrm{Tr}(\hat{r}_1)$.
Here, we have used $c_1=1/\Tr(\hat{r}_1)$ because $\hat{l}_1\propto \hat{\mathbb{I}}$ in Eq.~\eqref{cadef}.
From this expression, we can define 
\aln{
\Delta_\mr{asy}=-\mr{Re}[\lambda_2]
}
as long as $c_2\neq 0$, which is called the Liouvillian gap.

One might naively expect that the inverse of $\Delta_\mr{asy}$, $\tau_2=\Delta_\mr{asy}^{-1}$, provides the relaxation timescale $\tau_\mr{relax}$ of the dynamics\footnote{
While the definition of relaxation time may not be unique, we can take it to be, e.g., the mixing time, which is defined as the minimal time $\tau$ such that $d(\tau)<\epsilon$. 
Here, $d(t)=\max_{\hat{\rho}(0)}\|\hat{\rho}(t)-\hat{\rho}_\mr{ss}\|_1/2$ and $\epsilon$ is a small but constant cutoff. 
Note that $\|\hat{A}\|_1=\mr{Tr}[\sqrt{\hat{A}^\dag\hat{A}}]$.
}.
This is indeed true in many situations, and relaxation times are often evaluated from $\Delta_\mr{asy}$~\cite{vznidarivc2015relaxation}.
More generally, the timescale at which the contribution of the $a$th mode becomes negligible may naively be associated with $\tau_a=\Delta_a^{-1}$, where
\aln{\label{deltacptp}
\Delta_a=-\ln |z_a|
}
for the CPTP case and
\aln{\label{deltagksl}
\Delta_a=-\mr{Re}[\lambda_a]
}
for the GKSL case. Note that $\Delta_\mr{asy}=\Delta_2$.

In the following, we introduce several topics concerning the spectral gap and relaxation dynamics. 
We assume the above relation between the eigenvalues and the inverse of the relaxation timescales for the first three topics and then discuss the caveats for this identification at the end of this section.

\vskip\baselineskip
\noindent\textbf{Mpemba effect}

If we choose an initial state with $c_2=\cdots=c_{a_*-1}= 0$ but $c_{a_*}\neq 0$ for some $a_*\geq 3$, the relaxation time will become
\aln{
\tau_\mr{relax}\sim \tau_{a_*}<\tau_2
}
instead of $\tau_\mr{relax}\sim \tau_2$, where $\Delta_2<\Delta_{a_*}$ is assumed.
This means that the relaxation becomes faster for such initial states than for those with $c_2\neq 0$.
This fact has recently been applied to analyze the ``Mpemba effect," which originally represents the counter-intuitive observation that initially hotter water can freeze faster than colder water~\cite{mpemba1969cool}.
The effect has gathered renewed attention in the last decade from the community of non-equilibrium statistical mechanics~\cite{bechhoefer2021fresh, lu2017nonequilibrium, klich2019mpemba, kumar2020exponentially, ares2023entanglement, chatterjee2023quantum, chatterjee2024multiple, liu2024symmetry, yamashika2024entanglement, yamashika2025quenching, joshi2024observing, rylands2024microscopic, aharony2024inverse}; in this context, the (classical or quantum) Mpemba effect often indicates a phenomenon that a state initially far from the stationary state relaxes faster than one initially close to it.
In Markovian open quantum systems described by the GKSL equation, a general mechanism to engineer the quantum Mpemba effect has been proposed~\cite{carollo2021exponentially}.
Specifically, let us consider a state $\hat{\rho}$ with $c_2\neq 0$ and define $\hat{\rho}_0=\hat{U}\hat{\rho}\hat{U}^\dag$ for some unitary matrix $\hat{U}$.
If $\hat{U}$ is chosen such that $c_2=\cdots=c_{a_*-1}= 0$ and some distance between $e^{\ml{L}t}[\hat{\rho}_0]$ and $\hat{\rho}_\mr{ss}$ is larger than the distance between $e^{\ml{L}t}[\hat{\rho}]$ and $\hat{\rho}_\mr{ss}$ for short times, the quantum Mpemba effect occurs.
This scenario has been exemplified in certain situations~\cite{carollo2021exponentially, moroder2024thermodynamics}.

\vskip\baselineskip
\noindent\textbf{Metastability and emergent decoherence-free subspace}

Let us consider a situation where $\Delta_2,\cdots, \Delta_{a_*-1}\ll \Delta_{a_*}$ for some $a_*\geq 3$.
In this case, for a timescale $t$ with $\tau_{a_*}\ll t\ll \tau_{a_*-1}$, the evolving state effectively stays in a metastable manifold, which is spanned by the eigenmodes $\hat{r}_1,\cdots, \hat{r}_{a_*-1}$~\cite{macieszczak2016towards}.
An important situation is where the dissipation term $\ml{L}_\mr{d}[\hat{\rho}]=\sum_b(\hat{L}_b\hat{\rho}\hat{L}_b^\dag-\frac{1}{2}\{\hat{\rho},\hat{L}_b^\dag\hat{L}_b\})$ is dominant compared to the unitary term $\ml{L}_H[\hat{\rho}]=-i[\hat{H},\hat{\rho}]$.
If we neglect the unitary part, we can define a stationary (stable) manifold as the subspace spanned by the full set of zero modes of $\ml{L}_\mr{d}$. 
Specifically, we assume that there are $a_*-1$ zero modes of $\ml{L}_\mr{d}$ created by the existence of a decoherence-free subspace (DFS)~\cite{lidar1998decoherence, beige2000quantum} and denote the corresponding projection operator by $\hat{P}_\mr{DFS}$; see around Eq.~\eqref{DFS} for a more detailed discussion.
These zero modes are separated from the longest-living decaying modes by a gap $\Delta_{\mr{d},a_*}$. 
Here, $\Delta_{\mr{d},a_*}$ is defined from the eigenvalues of $\ml{L}_\mr{d}$ and is expected to be of the order of $\sim \|\hat{L}_b\|^2$. 

When a small unitary contribution is present, many zero modes of $\ml{L}_\mr{d}$ cease to be exact zero modes of $\ml{L}$.
Nevertheless, when $J\ll \Delta_{\mr{d},a_*}$ with $J$ being the characteristic strength of the Hamiltonian\footnote{\label{Foot1}
There is a subtle point in the choice of $J$. 
For few-level systems, $J$ can be chosen as $J=\|\hat{H}\|$~\cite{gong2020error}. 
However, for many-body systems, this choice is inappropriate since $\|\hat{H}\|\ll \Delta_{\mr{d},a_*}$ does not hold for a large system size $V$ because $\|\hat{H}\|\propto V$. 
In such cases, it would be natural to take $J$ as the microscopic coefficients of the Hamiltonian, e.g., hopping amplitude or chemical potential, although it is not easy to justify this choice rigorously.
}, we expect that $0\simeq \Delta_{a_*-1}\ll \Delta_{a_*}\simeq\Delta_{\mr{d},a_*} $ is satisfied, where $\simeq$ allows some correction of the order of $J$.
Then, there exists a timescale $t$ with $\tau_{a_*}\ll t\ll \tau_{a_*-1}$ where the dynamics is effectively constrained within the emergent DFS.
That is, the dynamics is well approximated by the unitary dynamics generated by $\hat{P}_\mr{DFS}\hat{H}\hat{P}_\mr{DFS}$, if it is non-vanishing.
This phenomenon, reminiscent of the quantum Zeno effect~\cite{misra1977zeno, beige2000quantum, facchi2002quantum, zanardi2014coherent}, has been rigorously justified using the Schrieffer-Wolf transformation~\cite{schrieffer1966relation} for few-level systems under certain assumptions~\cite{gong2020error}.
If we assume that the discussion holds for many-body systems\footref{Foot1}, the emergent DFS can be utilized to engineer~\cite{harrington2022engineered} exotic many-body dynamics, such as kinetically constrained quantum dynamics~\cite{stannigel2014constrained, maity2024kinetically}.

\vskip\baselineskip
\noindent\textbf{Dissipative phase transition}

The spectral analysis can be relevant for understanding dissipative phase transitions (DPTs) in open quantum many-body systems with system size $V$~\cite{fazio2025many}.
Let us consider the GKSL dynamics parameterized by $g$.
One of the definitions of the DPT is that, in the thermodynamic limit $V\ra\infty$, an order parameter at the stationary state or its derivatives exhibit a singularity when $g$ is varied.
For example, Ref.~\cite{minganti2018spectral} discussed a second-order transition with $\mathbb{Z}_2$ symmetry breaking\footnote{
Here, we consider a weak $\mathbb{Z}_2$ symmetry~\cite{buvca2012note, albert2014symmetries, Lieu2020symmetry} satisfying $\hat{G}\ml{{L}}[\hat{G}\hat{\rho}\hat{G}]\hat{G}=\ml{L}[\hat{\rho}]$.
}.
In the symmetry unbroken phase ($g<g_c$), the unique stationary state is invariant under a $\mathbb{Z}_2$ operation $\hat{G}$ such that $\hat{G}^2 = \hat{\mbb{I}}$ and $\hat{G}\hat{\rho}_\mr{ss}\hat{G}=\hat{\rho}_\mr{ss}$.
In the symmetry broken phase ($g>g_c$), there are two independent stationary states $\hat{\rho}_\mr{ss}^\pm$ satisfying $\hat{\rho}_\mr{ss}^\mp=\hat{G}\hat{\rho}_\mr{ss}^\pm\hat{G}$ in the thermodynamic limit $V\rightarrow\infty$.
However, if we consider finite-size systems, $\hat{\rho}_\mr{ss}^\pm$ cannot be the true stationary states.
For $g<g_c$, the eigenmodes $\hat{r}_1=\hat{\rho}_\mr{ss}$ and $\hat{r}_2$ are separated by a gap that does not vanish for $V\ra\infty$.
In contrast, for $g>g_c$, the eigenmodes read $\hat{r}_1=\hat{\rho}_\mr{ss}\propto\hat{\rho}_\mr{ss}^++\hat{\rho}_\mr{ss}^-$ and $\hat{r}_2\propto \hat{\rho}_\mr{ss}^+-\hat{\rho}_\mr{ss}^-$\footnote{
Note that $\Tr[{\hat{r}_a}]=0\;(a\geq 2)$ in general due to the trace-preserving property of the dynamics. 
We also note that $\hat{r}_1$ and $\hat{r}_2$ actually respect the $\mbb{Z}_2$ symmetry, as $\hat{G}\hat{r}_1\hat{G}=\hat{r}_1$ and $\hat{G}\hat{r}_2\hat{G}=-\hat{r}_2$.
}.
In this case, $\lambda_2$ is nonzero for finite-size systems but vanishes ($\lambda_2\ra0$) in the thermodynamic limit.
This means that the symmetry broken modes $\hat{\rho}_\mr{ss}^\pm$ or their linear combinations can appear as metastable states depending on the initial state [$\sim (1+c_2)\hat{\rho}_\mr{ss}^++(1-c_2)\hat{\rho}_\mr{ss}^-$ from Eq.~\eqref{expgksl}].
The lifetime of the metastable state is expected to be $\Delta_\mr{asy}^{-1}$ and grows with $V$.
This picture is consistent with the case in the thermodynamic limit, if we take the limit $\lim_{t\ra\infty}\lim_{V\ra\infty}$\footnote{
If we reverse the order of the two limits, $\lim_{V\ra\infty}\lim_{t\ra\infty}$, we have a symmetry unbroken unique stationary state even for $g>g_c$.
}.

Another interesting possibility of the DPT is the dissipative version of time crystals~\cite{gong2018discrete, iemini2018boundary, gambetta2019discrete, kessler2021observation, kongkhambut2022observation}, which exhibits a persistent oscillation of an order parameter in the thermodynamic limit with the breaking of time-translation symmetry~\cite{wilczek2012quantum, sacha2017time, zaletel2023colloquium}.
For finite systems, the time-crystalline phase has eigenvalues $\lambda_2,\lambda_3,\ldots$ whose real parts vanish but imaginary parts remain finite for $V\ra\infty$.
The corresponding eigenmodes contribute to the persistent oscillation, consistent with the behavior in the thermodynamic limit.
Note that we can have a similar discussion for the CPTP map.
It was found in Ref.~\cite{gong2018discrete} that a periodically driven dissipative Dicke model has a discrete time-crystalline phase~\cite{zaletel2023colloquium}; 
the corresponding CPTP map has an eigenvalue $z_2$, which converges to $-1$ in the thermodynamic limit and leads to the period doubling behavior.

\vskip\baselineskip
\noindent\textbf{Discrepancy between the inverse gap and relaxation timescale}

So far, we have assumed that the relaxation timescale $\tau_\mr{relax}$ of systems governed by the GKSL equation is identified as $\tau_2$, the inverse of the asymptotic decay rate $\Delta_\mr{asy}=-\mr{Re}[\lambda_2]$ as long as $c_2\neq 0$ in Eq.~\eqref{expgksl}.
However, it has recently been widely recognized\footnote{
The distinction between $\tau_\mr{relax}$ and $\Delta_\mr{asy}^{-1}$ was known before in both classical~\cite{levin2017markov} and quantum~\cite{kastoryano2012cutoff, kastoryano2013rapid} systems.
} that this identification is not always true~\cite{haga2021liouvillian, mori2020resolving, gong2022bounds, lee2023anomalously}.
This is because the coefficient $c_a$, say $c_2$, can be exponentially large with respect to the system size $V$, due to the exponential smallness of $\Tr[\hat{l}_2^\dag\hat{r}_2]$ in Eq.~\eqref{cadef}\footnote{
Here, we assume the normalization $\Tr[\hat{r}_a^\dag\hat{r}_a]=\Tr[\hat{l}_a^\dag\hat{l}_a]=1$.
}.
In this case, the amplitude for $a=2$, $c_2e^{-\Delta_\mr{asy}t}$, does not become small until $t\sim \ml{O}(V)$, even when $\Delta_\mr{asy}=\ml{O}(V^0)$.
Namely, $\tau_\mr{relax}=\ml{O}(V)$ diverges even for the gapped spectrum.

One of the simplest models for this to occur is a one-dimensional single-particle model with asymmetric dissipative hopping under the open boundary condition~\cite{haga2021liouvillian}.
Here, the jump operators are assumed to be $\hat{L}_{Ri}=\sqrt{\gamma_R}\hat{b}_{i+1}^\dag\hat{b}_i\:(1\leq i\leq V-1)$ and $\hat{L}_{Li}=\sqrt{\gamma_L}\hat{b}_{i-1}^\dag\hat{b}_i\:(2\leq i\leq V)$~\cite{temme2012stochastic}, where $\hat{b}_i$ is the annihilation operator of the particle at site $i$, and there is no Hamiltonian.
We can easily show that the Liouvillian gap is given by $\Delta_\mr{asy}=(\sqrt{\gamma_R}-\sqrt{\gamma_L})^2$.
Now, let us start from an initial state where a single particle is placed at the left end and the rest is empty. 
If we focus on the number of the particle at the right end, $\hat{n}_V=\hat{b}^\dag_V\hat{b}_V$, as an observable, we can rigorously show that $\braket{\hat{n}_V}_t$ is exponentially suppressed as $\braket{\hat{n}_V}_t=e^{-\mathcal{O}(V)}$ for times $t=\ml{O}(V^0)$, and that the relaxation timescale should be of the order of $V$ or larger~\cite{hamazaki2022speed, sawada2024role}.
Namely, the relaxation time diverges with respect to $V$, despite the finite Liouvillian gap.
In this case, the left and right eigenmodes for $a=2$ are exponentially localized at the opposite edges, meaning that their overlap becomes exponentially small as $\mr{Tr}[\hat{l}_2^\dag\hat{r}_2]=e^{-\ml{O}(V)}$.
Therefore, the exponentially large coefficient $c_2$ physically comes from the localization of the eigenmodes at the edges of the system, which is known as the Liouvillian skin effect~\cite{haga2021liouvillian}\footnote{
More generally, the localization of eigenmodes of non-Hermitian systems is called the non-Hermitian skin effect~\cite{gong2018topological, song2019non, okuma2020topological, zhang2020correspondence}.
}.

While we have considered systems described by the GKSL equation, a similar discrepancy between the spectrum and the timescale also appears in other contexts~\cite{mori2021metastability, bensa2021fastest, bensa2022two}.
The general message is that the Liouvillian spectrum is not always reliable to evaluate the relaxation timescale.
Instead, it has been proposed~\cite{vznidarivc2022solvable, rakovszky2024defining, vznidarivc2023phantom} that the timescale is rather characterized by the pseudo-spectrum~\cite{trefethen2020spectra}, i.e., a set of spectra of slightly perturbed matrices.
As a related fact, let us consider the periodic boundary condition for the above asymmetric hopping model (i.e., the ends in the case with the open boundary condition are perturbed).
Then, the Liouvillian gap vanishes with respect to $L$, where such a sudden change of the spectrum due to the change of the boundary condition is a consequence of the skin effect~\cite{gong2018topological}.
This gapless feature of the new spectrum is consistent with the divergent relaxation timescale.

\subsubsection{Universality of spectral statistics and random matrices}
In the previous subsection, we considered spectral properties near the stationary state, such as the spectral gap, and their relation to relaxation timescales.
In this subsection, we instead consider the spectral statistics of the entire spectrum and provide an overview of some selected topics, especially focusing on the relevance of spectral statistics to random matrix theory.

\vskip\baselineskip
\noindent\textbf{Entire shape of the spectrum}

Let us first consider the entire spectral shape of ``typical'' GKSL generators.
Before that, we explain the spectral structure of typical non-Hermitian matrices.
According to non-Hermitian random matrix theory, the so-called circular law~\cite{ginibre1965statistical, girko1985circular, bai1997circular} holds for an ensemble of matrices whose entries are independent and identically distributed (i.i.d.) with zero mean and finite variance~\cite{tao2010random, tao2012topics}.
Namely, the typical contour of the spectrum on the complex plane asymptotically becomes a circle centered at the origin, and the spectral density becomes uniform within the circle as the size of the matrix is increased.
In Fig.~\ref{fig:LemonPlot}~(a), we show the complex spectrum of a $d^2 \times d^2$ complex Ginibre matrix $G$\footnote{
A complex Ginibre matrix is a non-Hermitian random matrix whose entries are i.i.d. complex Gaussian variables.
}, normalized by $1/d$ such that the spectrum is concentrated within the unit circle.
The circular law demonstrates one of the famous examples of the universality of non-Hermitian random matrix theory, i.e., the detailed form of the i.i.d. distribution does not matter for its appearance.

\begin{figure}[tb]
\centering\includegraphics[width=\linewidth]{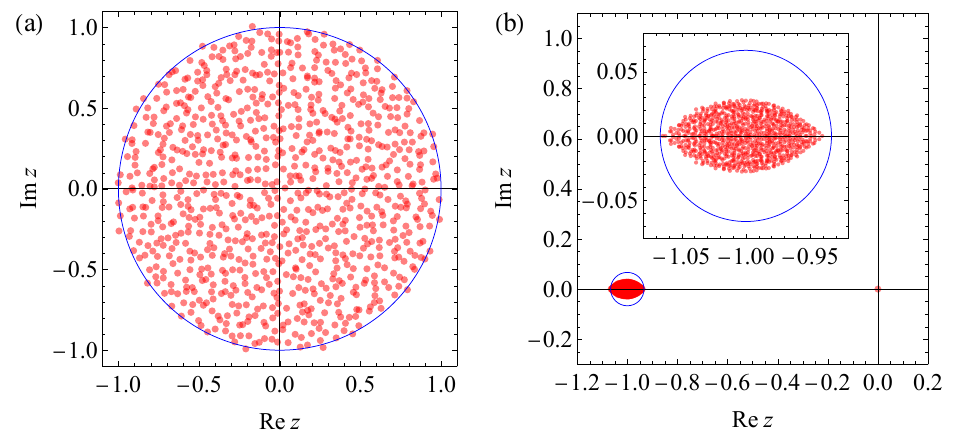}
\caption{
(a) Complex spectrum of a $d^2 \times d^2$ Ginibre matrix $G/d$ for $d=30$. 
(b) Complex spectrum of the GKSL operator $\mc{L}$ generated by a $(d^2-1) \times (d^2-1)$ random Kossakowski matrix for $d=30$. 
The blue circle in (b) has a radius of $2/d$.}
\label{fig:LemonPlot}
\end{figure}

In contrast, if we consider a random GKSL generator $\mathcal{L}$, additional structures lead to a spectrum distinct from the circular law\footnote{
See, e.g., Ref.~\cite{lancien2024limiting} for results on random channels, where the Kraus operators are randomly sampled.
}, even though $\mathcal{L}$ is represented as a non-Hermitian matrix.
To see this, let us consider the GKSL generator in the form of Eq.~\eqref{kmat}.
While it is not trivial how to sample random GKSL generators, Ref.~\cite{denisov2019universal} considered sampling the Kossakowski matrix $C$ randomly.
Since the Kossakowski matrix is positive semidefinite, one natural way of sampling is to sample a Ginibre matrix ${G}$ and define $C=d GG^\dag/\mathrm{Tr}[GG^\dag]$\footnote{
The coefficient of $GG^\dag$ only controls the scale of the spectrum and is not important for the lemon-shape contour, as long as $\hat{H}'=0$.
}.
As shown in Fig.~\ref{fig:LemonPlot}~(b), if there is no unitary part $(\hat{H}'=0)$, the spectral shape becomes a lemon-like contour inside a circle with a radius $\sim 2/d$ centered at $(-1,0)$\footnote{
This means that the lemon-like contour is gapped from the origin, which corresponds to the stationary-state eigenvalue.
}. 
Inside the lemon-like contour, the spectral density is non-uniform.
A similar lemon-like contour is obtained for other types of sampling of $C$ in the limit of large system sizes, demonstrating the universality of the spectral shape to some extent\footnote{
It has been reported that the lemon-type contour appears in a more physical setting, such as the Lindbladian dynamics of the Sachdev-Ye-Kitaev model~\cite{kulkarni2022lindbladian}.
}.
Reference~\cite{denisov2019universal} also considered the case where $\hat{H}'\neq 0$; in this case, the lemon shape is deformed and the contour approaches an ellipse as the strength of $\hat{H}'$ is increased.

However, if we consider additional constraints on the structure of the GKSL equation, it has been found that the universality of the above spectral shape no longer holds.
References~\cite{wang2020hierarchy, sommer2021many, li2022random} assumed that $\hat{F}_i$ are restricted to few-body operators and sampled random $C$ appropriately without the unitary part.
In that case, the spectrum is decoupled into distinct clusters whose decay timescales are different, rather than forming the universal lemon-like contour.
Several studies~\cite{haga2023quasiparticles, hartmann2024fate} considered how the decomposed spectral clusters change due to the unitary effect. 
For example, Ref.~\cite{haga2023quasiparticles} investigated physical models (e.g., hardcore bosons) with dissipation and without randomness, finding the decoupled clusters of the GKSL eigenvalues if the dissipation (assumed to be local dephasing) is strong.
When the dissipation is weakened compared with the unitary dynamics, they found that the decoupled clusters touch one another, manifesting a transition of the spectral shape.
It has been shown that this transition physically alters the dynamics of coherence in the system.

\vskip\baselineskip
\noindent\textbf{Universality of spectral statistics}

As discussed above, the entire spectral shape does not show strong universality, i.e., it significantly depends on the structure of many-body systems.
In contrast, certain spectral statistics exhibit the universality predicted by random matrix theory if the system considered is sufficiently ``complicated.''
To explain this in more detail, let us first review some well-known results for isolated quantum systems and their relation to Hermitian random matrix theory.
We especially focus on the eigenvalue-spacing distribution $P_H(s)$, which is the distribution of the properly normalized difference between neighboring energy eigenvalues of the Hamiltonian, $s_\alpha\propto E_{\alpha+1}-E_\alpha$, in the middle of the spectrum.
If we take a Gaussian random matrix as a Hamiltonian, $P_H(s)$ becomes close to the Wigner-Dyson distribution, $P_H(s)\simeq a_\beta s^\beta e^{-b_\beta s^2}$, where $\beta$ takes one of the three values $\beta=1,2,$ or $4$ depending on the time-reversal symmetry~\cite{dyson1962threefold, haake1991quantum} of the matrices, and $a_\beta,\,b_\beta$ are constants.
This is in contrast with the Poisson-type spacing distribution $P_H(s)=e^{-s}$ obtained from random sequences.
Importantly, sufficiently complicated Hamiltonians, which may have structures such as few-body or local interactions and may not even be random, can have $P_H(s)$ described by random matrix theory with the corresponding time-reversal symmetry.
This has been confirmed in, e.g., the excitation spectrum of nuclei~\cite{wigner1951statistical}, semiclassical systems whose classical limit is chaotic~\cite{bohigas1984characterization}, and non-integrable many-body systems without classical counterparts~\cite{santos2010onset}\footnote{\label{foot_chaos}
Semiclassical systems whose classical limits are chaotic are called quantum chaotic systems. 
Non-integrable many-body systems without classical counterparts, whose spectral statistics obey the random matrix theory, are often called quantum many-body chaotic systems.
}.
Instead, if the Hamiltonian is integrable, $P_H(s)$ can display the Poisson statistics~\cite{berry1977level, santos2010onset}.
Note that a similar strong universality appears in other spectral statistics~\cite{beenakker1997random, guhr1998random, haake1991quantum} and eigenstate statistics~\cite{brody1981random, d2016quantum}.
While the universality has been confirmed mainly numerically, it has also been analytically demonstrated for several physical models~\cite{muller2004semiclassical, muller2005periodic, bertini2018exact}.

Recently, there have appeared various studies aiming to extend the correspondence between physical systems and random matrix theory to dissipative quantum dynamics~\cite{jaiswal2019universality, hamazaki2019non, akemann2019universal, hamazaki2020universality, sa2020complex, huang2020anderson, tzortzakakis2020non, mudute2020non, luo2021universality, sa2021integrable, rubio2022integrability, prasad2022dissipative, garcia2022symmetry, ghosh2022spectral, ray2024ergodic, gupta2024quantum, akemann2025two, pawar2025comparative}.
That is, certain statistics in the middle of the spectrum relevant for sufficiently complicated dynamics have been found to be described by non-Hermitian random matrix theory.
This correspondence was first discussed in Ref.~\cite{grobe1988quantum}, which argued that the eigenvalue-spacing statistics $P(s)$ for the CPTP map describing the periodic kicked top with damping are described by the universal distribution of non-Hermitian random matrices (the Poisson distribution) when the dynamics is chaotic (integrable) in the classical limit.
Here, we note that the eigenvalue spacing $s_a$ is defined through the minimal distance of eigenvalues on the complex plane (e.g., $s_a\propto\min_{b}|z_a-z_b|$ for the case of the CPTP map), since the eigenvalues are no longer real.
Later, similar universal eigenvalue-spacing statistics for non-Hermitian random matrices were found to appear in non-integrable many-body systems described by non-Hermitian Hamiltonians~\cite{markum1999non, hamazaki2019non} and GKSL generators~\cite{akemann2019universal}.
It has also been found that symmetry plays an important role in random matrix theory for open quantum systems\footnote{
While there are 38 symmetry classes for general non-Hermitian matrices~\cite{kawabata2019symmetry} (in contrast with 10 symmetry classes for Hermitian matrices~\cite{altland1997nonstandard}), a smaller number of symmetry classes exist for GKSL generators and CPTP maps due to additional constraints~\cite{lieu2020tenfold, sa2023symmetry, kawabata2023symmetry, nakagawa2025topology}.
}.
For example, Ref.~\cite{hamazaki2020universality} claimed that, among 38 non-Hermitian symmetry classes~\cite{kawabata2019symmetry}, the eigenvalue-spacing distributions $P(s)$ take only three different universality classes depending on the transposition symmetries (e.g., the symmetry given as $H=H^\mathsf{T}$), instead of the time-reversal symmetry (e.g., the symmetry described as $H=H^*$)\footnote{
Reference~\cite{hamazaki2020universality} numerically verified three distinct universal classes of $P(s)$ (see also Ref.~\cite{jaiswal2019universality}) with arguments using degenerate perturbation theory. 
Recently, the existence of the three distinct classes has been studied analytically~\cite{akemann2025complex, kulkarni2025non, forrester2025dualities} for different eigenvalue statistics.
}.

The universality of non-Hermitian random matrix theory has also been discussed beyond eigenvalue-spacing distributions $P(s)$.
Various statistics, such as the complex-spacing ratio~\cite{jaiswal2019universality, sa2020complex}, spectral rigidity~\cite{huang2020spectral}, dissipative spectral form factor~\cite{li2021spectral}, singular-value statistics~\cite{kawabata2023singular, roccati2024diagnosing, nandy2025probing}, and eigenstate statistics~\cite{hamazaki2022lindbladian, cipolloni2023entanglement, ghosh2023eigenvector, cipolloni2024non, singha2025unveiling, almeida2025universality, ferrari2025chaos} have been found to exhibit universality.
These facts motivate researchers to characterize the dissipative version of quantum chaos\footref{foot_chaos} through the universal statistics of non-Hermitian random matrix theory.
However, recent studies on dissipative systems with well-defined classical counterparts\footnote{
Interestingly, the models considered in Ref.~\cite{grobe1988quantum} have also been revisited~\cite{xhwm-jpdf} and it has been found that the correspondence principle fails.
} have found a caveat to this correspondence: the universal random-matrix statistics in the middle of the spectrum can appear even when the corresponding classical systems do not have a chaotic attractor at long times~\cite{villasenor2024breakdown,xhwm-jpdf}.
Indeed, since eigenmodes in the middle of the spectrum decay more rapidly than the longest-lived mode for the CPTP or GKSL dynamics, i.e., $\Delta_a\gg  \Delta_\mathrm{asy}$ with $\Delta_a$ defined in Eqs.~\eqref{deltacptp} and~\eqref{deltagksl}, the universal statistics are rather expected to be related to transient dynamics~\cite{hamazaki2022lindbladian, ferrari2023steady, mondal2025transient}.

\section{Typical properties of quantum trajectories: ergodicity and purification}
\label{sec:linear-quantity_purification}
Keeping in mind the properties of the averaged dynamics, i.e., the CPTP maps or quantum master equations discussed in the previous chapter, we here discuss the properties of quantum trajectories.
Since quantum trajectories are determined stochastically, we will seek properties common to \textit{typical} quantum trajectories, instead of all possible trajectories.

While analyzing the typical properties of quantum trajectories is complicated because we should treat probability measures determined by the Born probability rules, various developments have been made during the last decade~\cite{attal2015central, benoist2019invariant, carollo2019unraveling, bernard2021can, benoist2021invariant, benoist2023limit, tindall2023generality, girotti2023concentration, benoist2024quantum} by many researchers, including the community of mathematical physics.
In this chapter, we especially discuss two pioneering works by K\"ummerer and Maassen, i.e., the ergodicity of linear observables~\cite{kummerer2004pathwise, kummerer2005quantum} and the purification~\cite{maassen2006purification} of quantum trajectories.
Our aim in this chapter is to provide an overview of their results, which were formulated in a mathematical language, in a physicist-friendly (and less rigorous) way.
In Chapter~\ref{sec:nonlear-quantity_Lyapunov-spectrum}, we will present more sophisticated recent results, such as the ergodicity of nonlinear quantities and the convergence of the Lyapunov exponents for typical quantum trajectories.

In the following, we basically discuss quantum trajectories for discrete times on the basis of Eq.~\eqref{qtmixed}, unless stated otherwise.

\subsection{Ergodicity of quantum trajectories and linear quantities}
\label{sec:ergodicity_linear-observable}
While the meaning of ergodicity depends on the context, the ergodicity of quantum trajectories usually refers to the equivalence between the long-time average of a single quantum trajectory and the long-time ensemble average over all trajectories.
For the simplest case, it is roughly formulated as
\aln{\label{linerg1}
\overline{\hat{\rho}_{\bm{b};n}}= \hat{\rho}_\mr{ss}
=\overline{\mathbb{E}\left[\hat{\rho}_{\bm{b};n}\right]}
}
for almost all (i.e., typical) trajectories characterized by measurement outcomes $\bm{b}$, where
\aln{\label{eq:time-average}
\overline{f_n}=\lim_{N\ra \infty}\frac{1}{N}\sum_{n=0}^{N-1}f_n
}
is the long-time average of $f_n$ and $\hat{\rho}_\mr{ss}$ is the $\bm{b}$-independent stationary state of the CPTP map $\ml{E}$.
Equation~\eqref{linerg1} indicates the ergodicity of linear observables, i.e., 
\aln{\label{linerg2}
\overline{\Tr[\hat{\rho}_{\bm{b};n}\hat{A}]}=\Tr[\hat{\rho}_\mr{ss}\hat{A}]
=\overline{\mathbb{E}\left[\Tr(\hat{\rho}_{\bm{b};n}\hat{A})\right]}}
for almost all trajectories.
Equation~\eqref{linerg1} (and its continuous-time version) was first proven in Ref.~\cite{kummerer2004pathwise} under some assumptions.
Its quantum-diffusion-equation version [cf. Eq.~\eqref{qdfeq}] has also been discussed in Ref.~\cite{benoist2021invariant}.

While the above formulation is for linear quantities in $\hat{\rho}_{\bm{b};n}$, the existence of an invariant measure, which can be used to prove the ergodicity of nonlinear quantities in $\hat{\rho}_{\bm{b};n}$, has also been derived in Ref.~\cite{benoist2019invariant} under suitable assumptions (see Secs.~\ref{sec:invariant-measure_ergodicity_outcomes} and~\ref{sec:ergodicity_nonlinear-quantity}).
Moreover, the ergodicity of jump statistics [e.g., $N_\mu(t)$ in the continuous-time case], instead of functions of $\hat{\rho}_{\bm{b};n}$, has also been proven in Refs.~\cite{cresser2001ergodicity, kummerer2003ergodic}.
Furthermore, beyond ergodicity, fluctuation properties, such as the central limit theorem around the average, have been discussed in Refs.~\cite{attal2015central, benoist2023limit, benoist2025quantum}.

\subsubsection{Result by K\"ummerer and  Maassen}\label{KMerg}
Now, let us present the statement of the ergodicity of linear quantities in $\hat{\rho}_{\bm{b};n}$ more formally.
We consider quantum trajectories, which are characterized by a sequence of measurement outcomes $\bm{b}=(b_1,b_2,\ldots)$ for a time-independent CP-instrument $\{\ml{E}_b\}$.
The corresponding CPTP map is given by $\ml{E}=\sum_b\ml{E}_b$.
The state at time step $n$ is given by Eq.~\eqref{qtmixed}, where the initial state is set to $\hat{\rho}_0$.
Then, the following statements hold~\cite{kummerer2005quantum}.
\begin{enumerate}
\item 
Under the above setup, 
\aln{
\hat{\rho}_{\bm{b}}^\mr{\infty}:=\overline{\hat{\rho}_{\bm{b};n}}
}
exists for any initial state $\hat{\rho}_0$ almost surely\footnote{\label{foot:TypicalVSAlmostAll}
A sequence of random variables $X_1,X_2,\cdots$ is said to converge to $X$ almost surely if
\aln{\label{almostsure}
\mr{Prob}\lrl{\lim_{n\ra\infty}X_n=X}=1.
}
The almost-sure convergence is a stronger condition than the convergence in probability, which states that 
\aln{
\lim_{n\ra\infty}\mr{Prob}\lrl{|X_n-X|\leq\epsilon}=1
}
for all $\epsilon>0$.
As an example~\cite{tao2012topics} that highlights these two notions, let us pick a real number $r$ uniformly from $[0,1]$.
We define $S_n\:(n\geq 1)$ as the interval where the decimal expansion of $r$ starts with the digits of $n$ (e.g., $S_{418}=[0.418,0.419)$). 
Now, $X_n$ is defined as 
\aln{
X_n=
\left\{
\begin{array}{ll}
1 & (r\in S_n)\\
0 & (r\notin S_n)
\end{array}
\right..
}
Then, $X_n$ converges to $0$ in probability.
However, since we have infinitely many $n$ such that $X_n=1$ for every $r$, $X_n$ does not show almost-sure convergence to $0$.
As mentioned in the introduction, we often use the term ``typical'' or ``almost all'' to discuss the behaviors of quantum trajectories, and they basically mean almost-sure convergence.
We sometimes explicitly indicate it to stress the exact meaning.
} with respect to the probability measure of quantum trajectories.

\item 
$\hat{\rho}_{\bm{b}}^\mr{\infty}$ is a random variable that satisfies
\aln{
\ml{E}[\hat{\rho}_{\bm{b}}^\mr{\infty}]&=\hat{\rho}_{\bm{b}}^\mr{\infty}
}
and
\aln{
\mbb{E}[\hat{\rho}_{\bm{b}}^\mr{\infty}]=\overline{\ml{E}^n}[\hat{\rho}_0].
}

\item
In particular, if $\ml{E}$ has a unique stationary state $\hat{\rho}_\mr{ss}$, ergodicity holds, i.e.,
\aln{
\hat{\rho}_{\bm{b}}^\mr{\infty}=\hat{\rho}_\mr{ss}={\mathbb{E}\left[\hat{\rho}_{\bm{b}}^\mr{\infty}\right]}
\label{eq:ergodicity_rho}
}
for any initial state $\hat{\rho}_0$ almost surely with respect to the probability measure of quantum trajectories.
\end{enumerate}

Let us explain intuitive meanings of these results, as schematically shown in Fig.~\ref{fig5}. 
Statement~(1) means that while $\Tr[\hat{\rho}_{\bm{b};n}\hat{A}]$ temporally fluctuates, its long-time average converges for almost all trajectories.
Statement~(3) means that if there is a unique stationary state $\hat{\rho}_\mr{ss}$ for the CPTP map $\ml{E}$, for which sufficient conditions have already been detailed in Chapter~\ref{sec:CPTPspectra}, the long-time average $\hat{\rho}_{\bm{b}}^\mr{\infty}$ is given by $\hat{\rho}_\mr{ss}$; i.e., the ergodicity in Eq.~\eqref{eq:ergodicity_rho} [or Eq.~\eqref{linerg2}] holds; see Fig.~\ref{fig5}(a). 
Note that non-decaying oscillatory eigenmodes corresponding to nontrivial peripheral spectrum with $|z_a|=1$ and $z_a\neq1$ [see Eq.~\eqref{eigenvalueCPTP}], which may exist in CPTP maps that are not primitive, vanish in the long-time average on the right-hand side of Eq.~\eqref{eq:ergodicity_rho}.
If the CPTP map is primitive, such modes are absent, and thus Eq.~\eqref{eq:ergodicity_rho} becomes 
\begin{align}
\hat{\rho}_{\bm{b}}^\mr{\infty}=\hat{\rho}_\mathrm{ss}=\lim_{n\rightarrow\infty}\mathbb{E}\left[\hat{\rho}_{\bm{b};n}\right].
\end{align} 
Statement~(2) is related to the situation where $\ml{E}$ has multiple stationary states and ergodicity breaks down, as shown in Fig.~\ref{fig5}(b).
Instead, the long-time average of each quantum trajectory is probabilistically determined and becomes a fixed point of the CPTP map.
If we further average these long-time averages over all measurement outcomes, we obtain a stationary state of the CPTP map that depends on the initial state $\hat{\rho}_0$.

\begin{figure}[!h]
\centering\includegraphics[width=\linewidth]{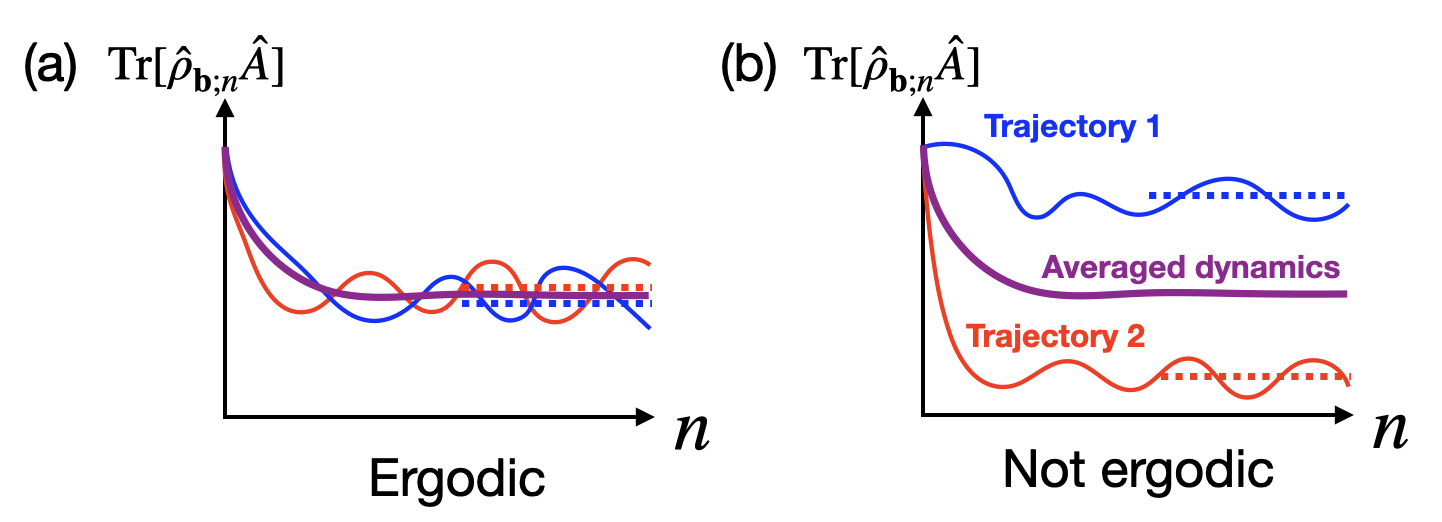}
\caption{
Schematic figures representing the (non-)ergodicity of quantum trajectories for linear observables. 
The expectation value of an observable $\hat{A}$ for each quantum trajectory $\hat{\rho}_{\bm{b};n}$ (blue and red) temporally fluctuates. 
However, the long-time average of $\hat{\rho}_{\bm{b};n}$ exists for almost all trajectories. 
(a) For ergodic cases, the long-time average for almost all trajectories coincides with the long-time ensemble average. 
(b) If ergodicity is broken, the long-time average becomes a random variable that is, in general, different from the ensemble-averaged one.
}
\label{fig5}
\end{figure}

The proof of the ergodicity explained above is given in Appendix~\ref{app:erg}.
There, we basically follow Ref.~\cite{kummerer2005quantum} but try to illustrate it in a physicist-friendly way, at the cost of mathematical rigor.
The mathematical trick to show almost-sure convergence is to employ the martingale convergence theorem~\cite{hall2014martingale, roldan2023martingales}.

While the above result is for the discrete-time case, a similar result holds for quantum trajectories under continuous-time measurement~\cite{kummerer2004pathwise}.
In this case, the ergodicity of quantum trajectories requires a unique  stationary state of the GKSL generator $\ml{L}$.

\subsubsection{Examples}
\label{exampleerg}
We present several examples to provide an intuitive understanding of ergodicity and its breakdown.
The first two examples are simple toy examples, and the third example is the ergodicity breaking of quantum diffusion recently demonstrated in Ref.~\cite{schmolke2024measurement}.

\textbf{Example 1.} Let us consider a four-level system, where the levels are denoted by $\ket{0},\ket{1},\ket{2},\ket{3}$.
We assume that the measurement operators are given by
\aln{\label{example1}
\begin{split}
\hat{M}_0&=\ket{0}\bra{0},\\
\hat{M}_1&=\ket{1}\bra{1},\\
\hat{M}_2&=\ket{1}\bra{3}+\ket{0}\bra{2}.
\end{split}}
We can easily confirm that $\sum_b\hat{M}_b^\dag\hat{M}_b=\hat{\mbb{I}}$ and that $\ml{E}$ has multiple stationary states, namely
\aln{
\ml{E}[c\ket{0}\bra{0}+(1-c)\ket{1}\bra{1}]=c\ket{0}\bra{0}+(1-c)\ket{1}\bra{1}
}
for arbitrary $c\in[0,1]$.
Indeed, the corresponding CPTP map is not irreducible (see Theorem~\ref{thm:IrreducibilityByKraus}), as the algebra $\mathbb{K}$ generated by the Kraus operators does not include, e.g., $\ketbra{2}$ or $\ketbra{3}$ and thus $\mathbb{K}\neq\mathbb{B}[\mathcal{H}]$.

\begin{figure}[!h]
\centering\includegraphics[width=\linewidth]{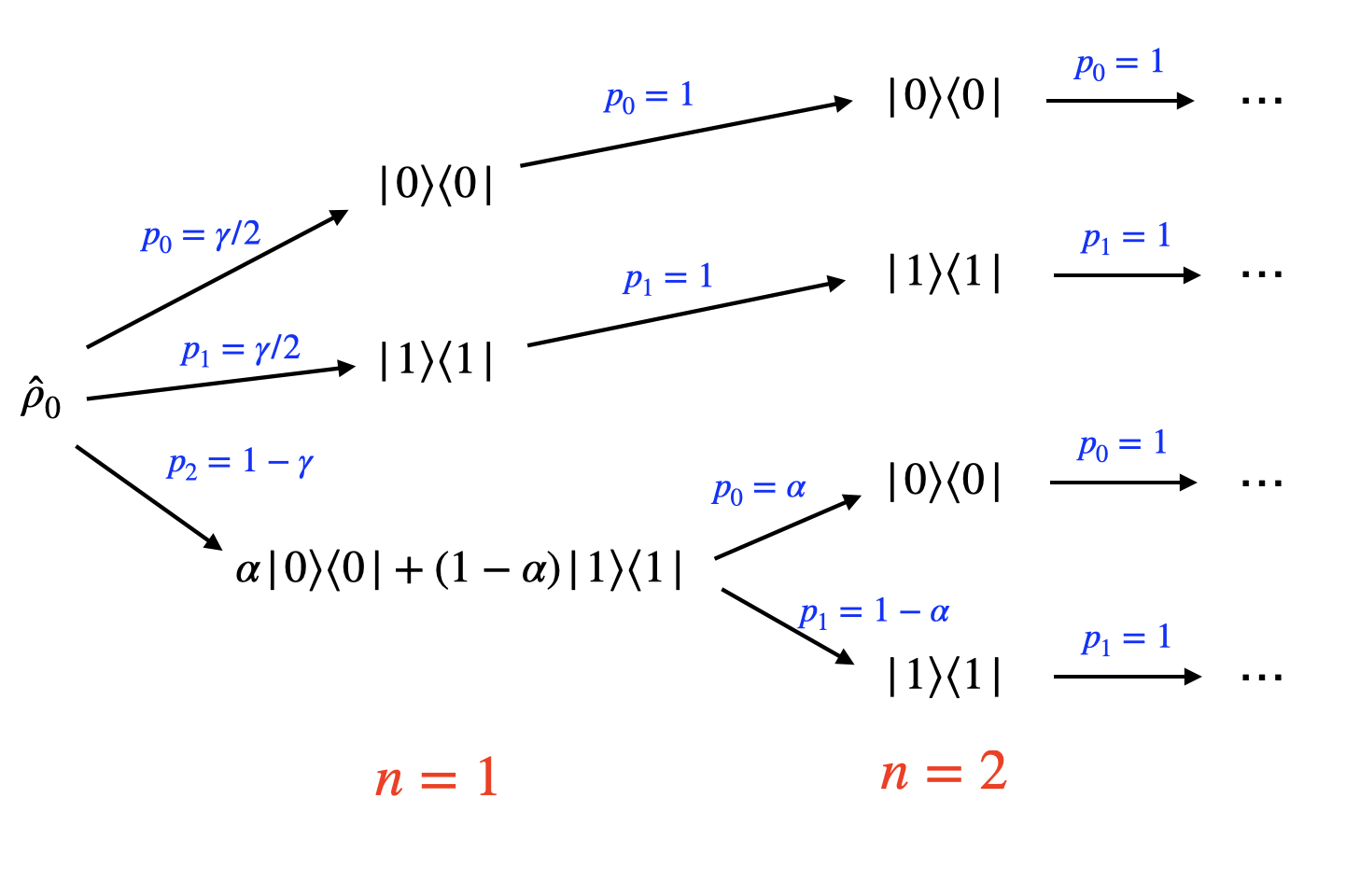}
\caption{
Transitions of the state $\hat{\rho}_0$ for \textbf{Example 1}.
The state becomes $\ketbra{0}$ and $\ketbra{1}$ with probabilities $\gamma/2+(1-\gamma)\alpha$ and $\gamma/2+(1-\gamma)(1-\alpha)$, respectively.
This demonstrates that ergodicity is broken, while purification occurs.
}
\label{fig6}
\end{figure}

Let us consider quantum trajectories starting from an initial state
\aln{\label{inifour}
\hat{\rho}_0=\gamma\ket{+}\bra{+}+(1-\gamma)(\alpha\ket{2}\bra{2}+(1-\alpha)\ket{3}\bra{3}),
}
where $\gamma,\alpha\in [0,1]$ and $\ket{+}=\frac{\ket{0}+\ket{1}}{\sqrt{2}}$.
The transitions and their probabilities are illustrated in Fig.~\ref{fig6}.
After the first measurement, $\hat{\rho}_0$ is mapped to $\ketbra{0}\: (b_1=0), \ketbra{1} \:(b_1=1),$ or $\alpha\ketbra{0}+(1-\alpha)\ketbra{1}\:(b_1=2)$ with probabilities $\gamma/2$, $\gamma/2$, and $1-\gamma$, respectively.
The states $\ketbra{0}$ and $\ketbra{1}$ remain unchanged after subsequent measurements.
On the other hand, the state $\alpha\ketbra{0}+(1-\alpha)\ketbra{1}$ becomes $\ketbra{0}\:(b_2=0)$ or $\ketbra{1}\:(b_2=1)$ with probabilities $\alpha$ and $1-\alpha$, respectively, after the second measurement.

Therefore, we find that for $n \geq 2$,
\aln{
\hat{\rho}_{\bm{b};n}=\hat{\rho}_{\bm{b}}^\infty=
\left\{
\begin{array}{ll}
\ketbra{0} &\quad \mr{Prob}:\frac{\gamma}{2}+(1-\gamma)\alpha, \\
\ketbra{1} &\quad \mr{Prob}:\frac{\gamma}{2}+(1-\gamma)(1-\alpha).
\end{array}
\right.
}
We also find that
\aln{
\av{\ml{E}^n}[\hat{\rho}_0]=\lim_{n\ra\infty}\ml{E}^n[\hat{\rho}_0]
=\lrs{\frac{\gamma}{2}+(1-\gamma)\alpha}\ketbra{0}+\lrs{\frac{\gamma}{2}+(1-\gamma)(1-\alpha)}{\ketbra{1},
}}
which means that ergodicity breaks down due to the existence of multiple stationary states of $\ml{E}$.
In contrast, we can confirm the statement (2) above, i.e., $\ml{E}[\hat{\rho}_{\bm{b}}^\infty]=\hat{\rho}_{\bm{b}}^\infty$ and 
$\mbb{E}[\hat{\rho}_{\bm{b}}^\infty]=\av{\ml{E}^n}[\hat{\rho}_0]$.

Finally, we note that the purification of the quantum state, which will be detailed in Sec.~\ref{sec:pur}, occurs in this case.
That is, while $\hat{\rho}_0$ is initially a mixed state, $\hat{\rho}_{\bm{b}}^\infty$ converges to a pure state.
\newline

\textbf{Example 2.}
We next consider a four-level system with the measurement operators 
\aln{\label{example2}
\begin{split}
\hat{M}_0&=\frac{1}{\sqrt{2}}\ket{0}\bra{0},\\
\hat{M}_1&=\frac{1}{\sqrt{2}}\ket{1}\bra{1},\\
\hat{M}_2&=\frac{1}{\sqrt{2}}(\ket{1}\bra{0}+\ket{0}\bra{1}),\\
\hat{M}_3&=\ket{1}\bra{3}+\ket{0}\bra{2},
\end{split}}
which satisfy $\sum_b\hat{M}_b^\dag\hat{M}_b=\hat{\mbb{I}}$.
In this case, the CPTP map $\ml{E}$ has a unique stationary state,
\aln{
\hat{\rho}_\mr{ss}=\frac{1}{2}(\ketbra{0}+\ketbra{1}).
}
Note that the corresponding CPTP map is not irreducible in this case.
Indeed, while the stationary state is unique, it is not positive definite.

Let us again start from the initial state given in Eq.~\eqref{inifour}.
We first observe that in the long run, measurement outcomes $b=0$ or $b=1$ are obtained, after which the state always becomes either $\ketbra{0}$ or $\ketbra{1}$.
If the state is $\ketbra{0}$ ($\ketbra{1}$), a subsequent measurement leads to the state $\ketbra{0}$ ($\ketbra{1}$) with $b=0$ ($b=1$) or the state
$\ketbra{1}$ ($\ketbra{0}$) with $b=2$.
These possibilities occur with the same probability, as schematically shown in Fig.~\ref{fig7}.

\begin{figure}[!h]
\centering\includegraphics[width=\linewidth]{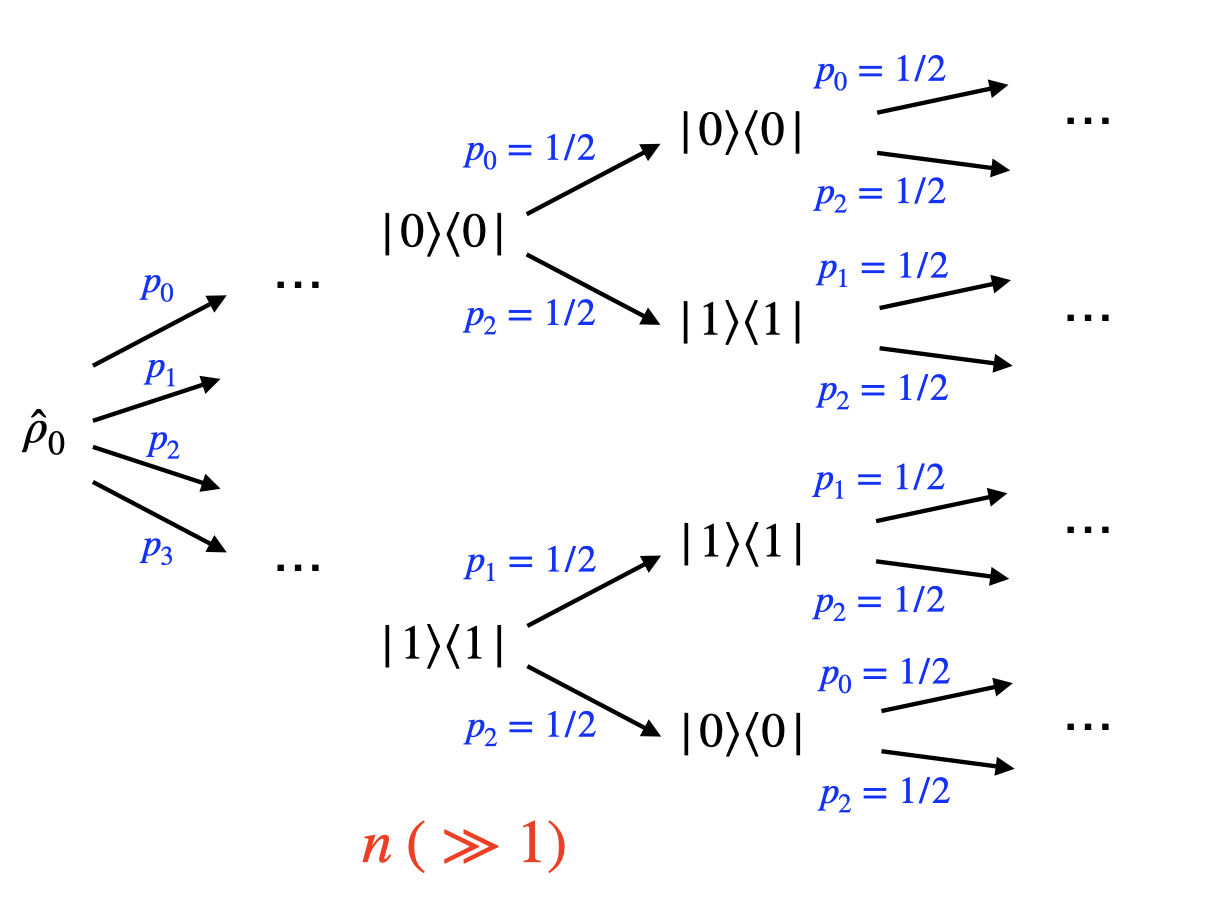}
\caption{
Transitions of the state $\hat{\rho}_0$ for \textbf{Example 2}.
After long times, typical trajectories become pure and they fluctuate between $\ketbra{0}$ and $\ketbra{1}$ with equal probabilities.
In this case, purification and ergodicity for almost all trajectories hold true.
}
\label{fig7}
\end{figure}

Consequently, after sufficiently long times, we find that $\ketbra{0}$ and $\ketbra{1}$ appear with the same frequency for typical trajectories. 
Therefore, the long-time average becomes
\aln{
\hat{\rho}_{\bm{b}}^\infty=\frac{1}{2}\lrs{\ketbra{0}+\ketbra{1}}=\hat{\rho}_\mr{ss},
}
meaning that ergodicity holds true.
We can also see that purification occurs for typical trajectories, since $\ketbra{0}$ and $\ketbra{1}$ are pure states.

We note that the long-time average of rare trajectories can have different properties.
For example, if we find $b=0$ for all measurement outcomes, then $\hat{\rho}_{\bm{b}}^\infty=\ketbra{0}\neq \hat{\rho}_\mr{ss}$.
As another example, if we find $b=2$ for all outcomes, we find 
$\hat{\rho}_{\bm{b}}^\infty=\ketbra{+}\neq \hat{\rho}_\mr{ss}$.
\newline

\textbf{Example 3.}
The final example in this subsection is the ergodicity breakdown for continuous-time quantum trajectories discussed in Ref.~\cite{schmolke2024measurement}.
We especially focus on the diffusion-type stochastic equation in Eq.~\eqref{qdfeq} for the case of an initial pure state, with $\hat{L}=\hat{L}^\dag$ (only one jump type) and $\theta=0$ for simplicity.
Then, it turns out that the corresponding stochastic equation for the pure state is given by
\aln{
d\ket{\psi_c}=\lrs{-i\hat{H}-\frac{\hat{L}^2}{2}+\hat{L}\braket{\hat{L}}_c-\frac{\braket{\hat{L}}_c^2}{2}}\ket{\psi_c}dt+\lrs{\hat{L}-\braket{\hat{L}}_c}\ket{\psi_c}dW.
}

Now, consider the eigenvalue equation for $\hat{L}$,
\aln{
\hat{L}\ket{q_{k,m}}=q_k\ket{q_{k,m}},
}
where $m$ is the label for possible degeneracy.
We then define
\aln{\label{DFS}
\ml{H}_\mr{DFS}^{(k)}=\mr{Span}\{\ket{q_{k,m}}\}_m.
}
If $\hat{H}\ket{\psi}\in \ml{H}_\mr{DFS}^{(k)}$ for all $\ket{\psi}\in\ml{H}_\mr{DFS}^{(k)}$,
$\ml{H}_\mr{DFS}^{(k)}$ is called a decoherence-free subspace (DFS)~\cite{lidar1998decoherence, beige2000quantum}.
In this case, for $\ket{\psi_c(t)}\in \ml{H}_\mr{DFS}^{(k)}$, we find $(\hat{L}-\braket{\hat{L}}_{c;t})\ket{\psi_c(t)}=0$ and thus
\aln{
d\ket{\psi_c(t)}=-i\hat{H}\ket{\psi_c(t)}dt.
}
Namely, the effect of dissipation vanishes and the state essentially evolves via unitary dynamics within $\ml{H}_\mr{DFS}^{(k)}$.

In this situation, the GKSL equation obtained by averaging over measurement outcomes,
\aln{
\frac{d\hat{\rho}}{dt}=-i[\hat{H},\hat{\rho}]+\hat{L}\hat{\rho}\hat{L}-\frac{1}{2}\{\hat{L}^2,\hat{\rho}\},
}
has multiple stationary states.
To see this, let us consider the Hamiltonian projected onto the DFS, $\hat{H}_\mr{DFS}^{(k)}=\hat{P}_\mr{DFS}^{(k)}\hat{H}\hat{P}_\mr{DFS}^{(k)}$, where $\hat{P}_\mr{DFS}^{(k)}$ is the projection operator onto $\ml{H}_\mr{DFS}^{(k)}$.
Then, it is easy to confirm that $\ket{E_\mr{DFS}^{(k)}}\bra{E_\mr{DFS}^{(k)}}$ is a stationary state of the GKSL equation, where $\ket{E_\mr{DFS}^{(k)}}\in \ml{H}_\mr{DFS}^{(k)}$ is an eigenstate of $\hat{H}_\mr{DFS}^{(k)}$.
In contrast, we also have a trivial stationary state $\frac{\hat{\mbb{I}}}{\mr{dim}[\ml{H}]}$, which acts on the entire Hilbert space and is thus different from $\ket{E_\mr{DFS}^{(k)}}\bra{E_\mr{DFS}^{(k)}}$.

The existence of multiple stationary states indicates that the ergodicity of quantum trajectories breaks down for this system, as we will explicitly demonstrate below.
To analyze the time evolution of $\ket{\psi_c(t)}$, we first decompose it as~\cite{schmolke2024measurement}
\aln{
\ket{\psi_c(t)}=\sum_{\ket{q_m}\in\ml{H}_\mr{DFS}^{(k)}}c_m(t)\ket{q_m}+
\sum_{\ket{p_l}\in\ml{H}\backslash\ml{H}_\mr{DFS}^{(k)}}d_l(t)\ket{p_l}
}
Here, since we will consider a fixed DFS in the following, we simply omit the subscript $k$ for quantities such as $\ket{q_m}$.

If we define 
\aln{\label{linearct}
|c(t)|^2=\braket{\psi_c(t)|\lrs{\sum_{\ket{q_m}\in\ml{H}_\mr{DFS}^{(k)}}\ket{q_m}\bra{q_m}}|\psi_c(t)}=\sum_{\ket{q_m}\in\ml{H}_\mr{DFS}^{(k)}}|c_m(t)|^2=\sum_{\ket{q_m}\in\ml{H}_\mr{DFS}^{(k)}}|\braket{q_m|\psi_c(t)}|^2,
}
we find
\aln{
d(|c(t)|^2)&=\sum_{\ket{q_m}\in\ml{H}_\mr{DFS}^{(k)}}\braket{q_m|\lrs{\ket{d\psi_c(t)}\bra{\psi_c(t)}+\ket{\psi_c(t)}\bra{d\psi_c(t)}+\ket{d\psi_c(t)}\bra{d\psi_c(t)}}|q_m}\nonumber\\
&=\lrs{\sum_{\ket{q_m}\in\ml{H}_\mr{DFS}^{(k)}} c_m^*(t)\braket{q_m|\hat{L}-\braket{\hat{L}}_{c;t}|\psi_c(t)}+\mr{c.c.}}dW\nonumber\\
&=\lrl{2|c(t)|^2\lrs{q(1-|c(t)|^2)-\sum_{\ket{p_l},\ket{p_{l'}}\in\ml{H}\backslash\ml{H}_\mr{DFS}^{(k)}}d_{l'}^*(t)d_l(t)\braket{p_{l'}|\hat{L}|p_l}}}dW,
}
where we have used $dtdW=0$, $dW^2=dt$, and the fact that $\braket{q_m|\hat{H}|p_l}=0$ since $\hat{H}\ket{q_m}\in\ml{H}_\mr{DFS}^{(k)}$.
Note that we write $q_k$ simply as $q$, i.e., $\hat{L}\ket{q_m}=q\ket{q_m}$.

Now, we ask what the stationary distribution for $|c(t)|^2$ is, starting from an initial state
\aln{
\ket{\psi_0}=\sum_{\ket{q_m}\in\ml{H}_\mr{DFS}^{(k)}}c_m(0)\ket{q_m}+
\sum_{\ket{p_l}\in\ml{H}\backslash\ml{H}_\mr{DFS}^{(k)}}d_l(0)\ket{p_l}.
}
For this purpose, we first notice that
\aln{
\mbb{E}\lrl{d|c(t)|^2}\propto \mbb{E}[dW]=0, 
}
from which we can conclude that
\aln{\label{initialcondE}
\mbb{E}\lrl{|c(t)|^2}=\mbb{E}\lrl{|c(0)|^2}=|c(0)|^2.
}
Next, we notice that there are two solutions for $d|c(t)|^2=0$: $|c(t)|^2=1$ and $|c(t)|^2=0$, where we have used the normalization condition
\aln{
|c(t)|^2+\sum_{\ket{p_l}\in\ml{H}\backslash\ml{H}_\mr{DFS}^{(k)}}|d_l(t)|^2=1.
}
If we take generic $\hat{H}$ and $\hat{L}$ within our setting, we expect that there are no other solutions for $d|c(t)|^2=0$, so we assume this in the following.

Then, the stationary probability distribution for $|c(t)|^2$ with respect to the ensemble of quantum trajectories\footnote{
We first notice that $|c(t)|^2$ converges to a stationary value because of the martingale convergence theorem~\cite{hall2014martingale} since $\mathbb{E}[d|c(t)|^2]=0$ and $\mathbb{E}[|c(t)|^2]$ is bounded. 
Then, by the assumption, $|c(t)|^2$ takes a stationary value 0 or 1; i.e., the stationary distribution is given by a sum of $\delta(|c(t)|^2-1)$ and $\delta(|c(t)|^2)$. 
Their weights are determined from the condition in Eq.~\eqref{initialcondE}.
} should take the form
\aln{
(1-|c(0)|^2)\delta(|c(t)|^2)+ |c(0)|^2\delta(|c(t)|^2-1),
}
i.e., $|c(t)|^2$ becomes 0 and 1 with probabilities $1-|c(0)|^2$ and $|c(0)|^2$, respectively.
Now, let us assume $0<|c(0)|<1$. 
In this case, since
\aln{
\av{|c(t)|^2}=(1 \:\:\text{or}\:\:0)
\neq |c(0)|^2 = \av{\mbb{E}\lrl{|c(t)|^2}},
}
we find that the ergodicity of quantum trajectories actually breaks down\footnote{
Note that $|c(t)|^2$ is a linear observable with respect to $\ket{\psi_c(t)}\bra{\psi_c(t)}$, which is understood from Eq.~\eqref{linearct}.
}.

Intuitively, trajectories with $|c(t)|^2\ra 0$ are noisy trajectories due to $dW$, since they are out of the DFS. 
In contrast, trajectories with $|c(t)|^2\ra 1$ are noiseless trajectories, which may exhibit coherent dynamics, since they are fully in $\ml{H}_\mr{DFS}$. 
These trajectories are schematically shown in Fig.~\ref{fig8}.

\begin{figure}[!h]
\centering\includegraphics[width=\linewidth]{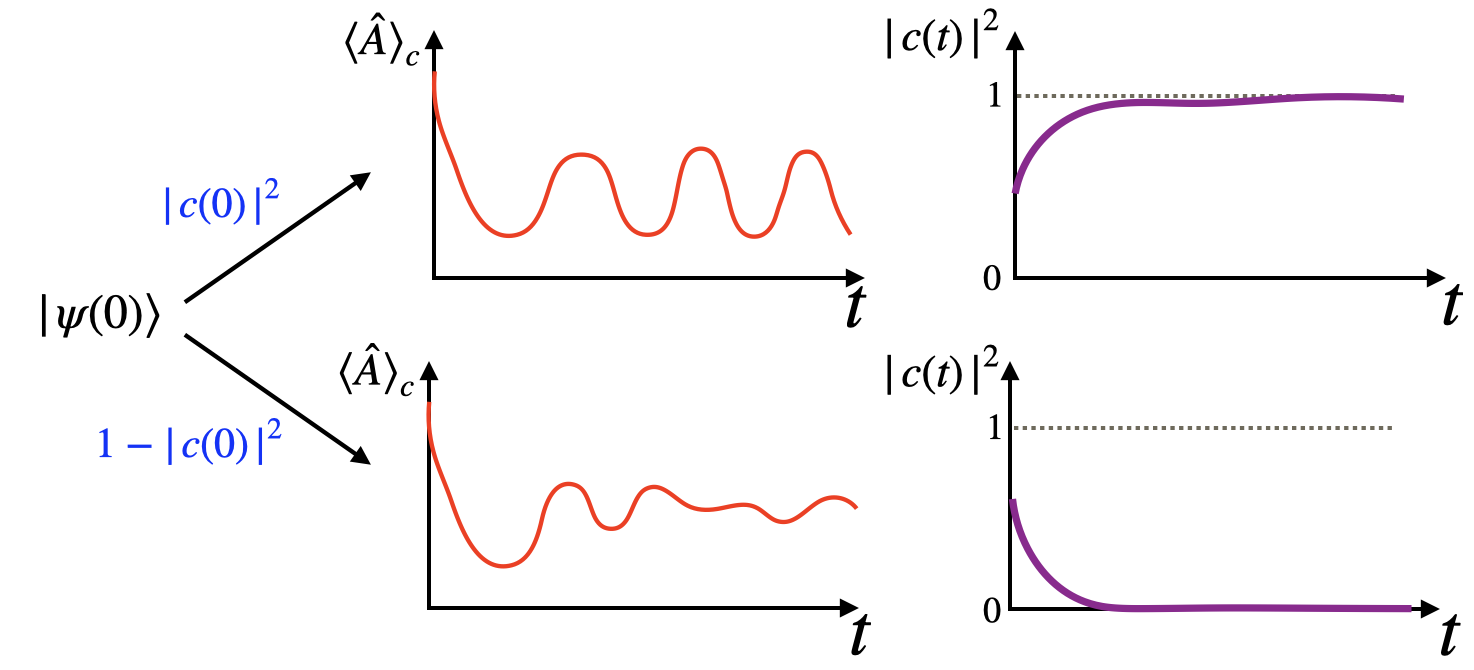}
\caption{
Schematic illustration of the evolution of a quantum trajectory in \textbf{Example 3} (see Fig.~1 of Ref.~\cite{schmolke2024measurement} for actual time evolution). 
With a probability $|c(0)|^2\:(1-|c(0)|^2)$, the long-time trajectory exhibits coherent (incoherent) dynamics, where $|c(t)|^2$ approaches one (zero).
}
\label{fig8}
\end{figure}

In Ref.~\cite{schmolke2024measurement}, the authors considered concrete models showing such breakdown of ergodicity.
For example, they considered a one-dimensional XY model with a magnetic field and on-site dephasing.
For certain system sizes, DFSs appear, in which a local spin coherently synchronizes (or anti-synchronizes) with other local spins~\cite{schmolke2022noise}.
If we start from an initial state with $0<|c(0)|^2<1$, one actually finds that some trajectories (anti-)synchronize because they fall into the DFS, while there also exist trajectories that exhibit only noisy dynamics.
This also demonstrates the breakdown of ergodicity in this system.
Here, we note that the DFS is one example of subspaces that arise if the averaged CPTP dynamics admits several stationary states: a space spanned by eigenmodes of one stationary state corresponding to nonzero eigenvalues becomes a subspace from which states cannot escape under the dynamics \cite{baumgartner2012structures}.
While the above example focuses on the relation between the ergodicity breaking and the DFS, more general discussions about the behaviors of quantum trajectories have been given in the case where there are various such subspaces, not restricted to the DFS.
For instance, the probability and time for a quantum trajectory to converge to one of the subspaces have been discussed in Ref.~\cite{benoist2024exponentially}.
The asymptotic behaviors of quantum trajectories, e.g., complete and incomplete localizations in the subspaces, have also been thoroughly explored in Ref.~\cite{schmolke2025asymptotic}.

\subsection{Purification of quantum trajectories}\label{sec:pur}
As another important property of quantum trajectories, we next discuss their purification.
Let us take an initially mixed state $\hat{\rho}_0$.
If we perform a projective measurement with a rank-one operator, we immediately obtain a pure state, $\hat{P}_\eta\hat{\rho}_0\hat{P}_\eta{\propto}\ket{\eta}\bra{\eta}$.
While indirect measurement does not immediately turn an initial mixed state into a pure state, the purity $\Tr[\hat{\rho}^2]$ of the state $\hat{\rho}$ is expected to be non-decreasing after the measurement. 
This can be mathematically understood from an inequality by Nielsen~\cite{nielsen2001characterizing}\footnote{
While this inequality is intuitive, the proof in Ref.~\cite{nielsen2001characterizing} is rather complicated, as the majorization technique is used. 
Here, instead of a general proof, we give a simple proof available only for $k=2$. 
Using the Cauchy-Schwarz inequality twice, we have
\aln{
\Tr[\hat{\rho}^2]&=\sum_{b,b'}\Tr[\hat{M}_b^\dag\hat{M}_b\hat{\rho}\hat{M}_{b'}^\dag\hat{M}_{b'}\hat{\rho}]\leq \sum_{b,b'}|\Tr[\sqrt{\hat{\rho}}\hat{M}_b^\dag\hat{M}_b\sqrt{\hat{\rho}}\sqrt{\hat{\rho}}\hat{M}_{b'}^\dag\hat{M}_{b'}\sqrt{\hat{\rho}}]|
\nonumber\\
&\leq \sum_{b,b'}\sqrt{\Tr[\hat{M}_b^\dag\hat{M}_b\hat{\rho}\hat{M}_{b}^\dag\hat{M}_{b}\hat{\rho}]\Tr[\hat{M}_{b'}^\dag\hat{M}_{b'}\hat{\rho}\hat{M}_{b'}^\dag\hat{M}_{b'}\hat{\rho}]}\nonumber\\
&=\lrs{\sum_b\sqrt{p_b}\sqrt{\frac{\Tr[\hat{M}_b^\dag\hat{M}_b\hat{\rho}\hat{M}_{b}^\dag\hat{M}_{b}\hat{\rho}]}{p_b}}}^2 
\leq \sum_b\frac{\Tr[\hat{M}_b^\dag\hat{M}_b\hat{\rho}\hat{M}_{b}^\dag\hat{M}_{b}\hat{\rho}]}{p_b}=\sum_b p_b\Tr[(\hat{\rho}_b')^2].
}
},
\aln{
\sum_bp_b\Tr[(\hat{\rho}_b')^k]\geq \Tr[(\hat{\rho})^k],
}
where $k\in \mbb{N}$ and $p_b,\hat{\rho}_b'$ are given in Eq.~\eqref{krauspure}.
This means that the average purity $(k=2)$ of post-measurement states is greater than or equal to the original purity.

Therefore, we expect that repeating the measurement will eventually lead to purification; i.e.,
\aln{\label{eq:purification}
\lim_{n\ra\infty}\hat{\rho}_{\bm{b};n}\;\;\text{becomes a pure state,}
}
or it is rephrased as
\aln{
\lim_{n\ra\infty}\Tr\lrl{\lrs{\hat{\rho}_{\bm{b};n}}^2}=1.
}
Such a purification property has recently attracted much attention in the context of measurement-induced phase transitions~\cite{gullans2020dynamical}.

We are especially interested in whether quantum trajectories become purified almost surely after sufficiently long times, focusing on finite-dimensional systems.
One of the equivalent conditions for the almost sure purification of the trajectories is the one discussed in Refs.~\cite{maassen2006purification, benoist2019invariant}, which is stated as follows:
a quantum trajectory $\hat{\rho}_{\bm{b};n}$ purifies almost surely if and only if there exists no orthogonal projector $\hat{\mathcal{Q}}$ with $\mr{rank}[\hat{\mc{Q}}] \geq 2$ such that for any $n$ and all $\bm{b}_n$\footnote{
While $\bm{b}$ represents a sequence of measurement outcomes in most of this review, the meaning of $\bm{b}$ will be extended in Chapter \ref{sec:mipt}.
There, $\bm{b}$ involves not only discrete measurement outcomes but also continuous variables, such as random unitaries chosen from a classical probability measure.
In this case, ``all $\bm{b}_n$" should be rephrased as ``almost all $\bm{b}_n$".
While Ref.~\cite{maassen2006purification} only considers discrete $\bm{b}$, Ref.~\cite{benoist2019invariant} also considers continuous $\bm{b}$.
}, there is a constant $\zeta_{\bm{b}_n}$ satisfying 
\begin{align}
\hat{\mathcal{Q}}\hat{\mathsf{M}}_{\bm{b};n}^\dagger\hat{\mathsf{M}}_{\bm{b};n}\hat{\mathcal{Q}}=\zeta_{\bm{b}_n}\hat{\mathcal{Q}}.
\label{eq:Benoist-purification}
\end{align}

To understand one direction of the equivalence, let us assume that a projector $\hat{\mc{Q}}$ with $\mathrm{rank}[\hat{\mc{Q}}]\geq 2$ exists such that Eq.~\eqref{eq:Benoist-purification} holds true for any $n$ and all $\bm{b}_n$.
Then, if we consider an initial state $\hat{\rho}_0=\hat{\mc{Q}}/\mr{Tr}[\hat{\mc{Q}}]$, the realizable state at step $n$ becomes $\hat{\rho}_{\bm{b};n}={\hat{\mathsf{M}}_{\bm{b};n}\hat{\mc{Q}}\hat{\mathsf{M}}_{\bm{b};n}^\dag}/\mr{Tr}[\hat{\mathsf{M}}_{\bm{b};n}\hat{\mc{Q}}\hat{\mathsf{M}}_{\bm{b};n}^\dag]$, where $\bm{b}$ is restricted to outcomes for which $\zeta_{\bm{b}_n}\neq 0$.
Then, the purity at time-step $n$ reads
\aln{
\Tr\lrl{\lrs{\hat{\rho}_{\bm{b};n}}^2}=\frac{\mr{Tr}[\hat{\mathsf{M}}_{\bm{b};n}\hat{\mc{Q}}\hat{\mathsf{M}}_{\bm{b};n}^\dag\hat{\mathsf{M}}_{\bm{b};n}\hat{\mc{Q}}\hat{\mathsf{M}}_{\bm{b};n}^\dag]}{\mr{Tr}[\hat{\mathsf{M}}_{\bm{b};n}\hat{\mc{Q}}\hat{\mathsf{M}}_{\bm{b};n}^\dag]^2}=\frac{1}{\Tr[\mc{\hat{Q}}]}<1,
}
where we have used Eq.~\eqref{eq:Benoist-purification}.
Therefore, purification indeed fails to occur in this case.

While we skip the complete proof of the other direction of the equivalence, we discuss the one-step and two-step versions of the condition in detail in Sec.~\ref{purMK}. 
The discussions in that section correspond to sufficient conditions under which purification occurs for almost all quantum trajectories.
The proof for the one-step version is given in Appendix~\ref{app:pur}.
We then provide concrete examples regarding the presence or absence of purification in Sec.~\ref{purex}.
 
We also note that the condition in Eq.~\eqref{eq:Benoist-purification} is related to the matrix rank of $\hat{\mathsf{M}}_{\bm{b};n}$, while we do not explain the mathematical details here. 
Instead, in Sec.~\ref{sec:rank-M_purification}, we explain the relation between the purification of $\hat{\rho}_{\bm{b};n}$ and the rank of $\hat{\mathsf{M}}_{\bm{b};n}$ in a physicist-friendly manner.

\subsubsection{Sufficient condition of purification}\label{purMK}
Here, we discuss some sufficient conditions for the purification of typical quantum trajectories on the basis of Ref.~\cite{maassen2006purification} by Maassen and K\"ummerer.
By considering the $n=1$ case of Eq.~\eqref{eq:Benoist-purification}, one finds that a quantum trajectory $\hat{\rho}_{\bm{b};n}$ purifies almost surely with respect to the probability measure of quantum trajectories, \textit{unless} the following situation occurs:
\begin{itemize}    
\item[($\star$)]
There exists an orthogonal projector $\hat{\ml{Q}}$ with $\mr{rank}[\hat{\ml{Q}}]\geq 2$ such that for all $b$, there is a constant $\zeta_b \geq 0$ satisfying
\aln{\label{purjoken}
\hat{\ml{Q}}\hat{M}_b^\dag\hat{M}_b\hat{\ml{Q}}=\zeta_b\hat{\ml{Q}}.
}
\end{itemize}
In other words, if we consider the space projected by $\hat{\ml{Q}}$, $\hat{M}_b^\dag\hat{M}_b$ acts as the identity (times $\zeta_b$). 
The proof is given in Appendix~\ref{app:pur}.

As a corollary of this result, we can say the following:
if $\mr{dim}[\ml{H}]=2$, a quantum trajectory $\hat{\rho}_{\bm{b};n}$ purifies almost surely, unless $\hat{M}_b$ is proportional to a unitary operator for all $b$.
We note that ($\star$) above is a necessary but not sufficient condition for the breakdown of purification (see the discussion below)\footnote{
Note that the statements that ``($\star$) is a necessary condition for the breakdown of purification for almost all trajectories'' and ``purification for almost all trajectories holds unless ($\star$) occurs'' are equivalent (contrapositives of each other).
}.

While the above result is for the discrete-time case, a similar property can be discussed for quantum trajectories under continuous-time measurement~\cite{barchielli2003asymptotic, benoist2021invariant}.
For example, for the case of the quantum jump process in Eq.~\eqref{stochasticeq}, the purification of quantum trajectories occurs almost surely unless a condition similar to ($\star$) above, where $\hat{M}_b$ is replaced by $\hat{L}_b$, holds.

As a stronger version of the above statement, we can consider the two-step version ($n=2$) of Eq.~\eqref{eq:Benoist-purification} and obtain the following consequence:
a quantum trajectory $\hat{\rho}_{\bm{b};n}$ purifies almost surely with respect to the probability measure of quantum trajectories, \textit{unless} the following situation occurs:
\begin{itemize}    
\item[($\star\star$)]
There exists an orthogonal projector $\hat{\ml{Q}}$ with $\mr{rank}[\hat{\ml{Q}}]\geq 2$ such that for all $b,b'$, there is a constant $\zeta_{b,b'}\geq0$ satisfying
\aln{\label{purjoken2}
\hat{\ml{Q}}\hat{M}_{b'}^\dag\hat{M}_b^\dag\hat{M}_b\hat{M}_{b'}\hat{\ml{Q}}=\zeta_{b,b'}\hat{\ml{Q}}.
}
\end{itemize}
Again, this is a necessary but not sufficient condition for the breakdown of purification.

\subsubsection{Examples}\label{purex}
Let us discuss some examples. 
In Sec.~\ref{exampleerg}, we already saw that \textbf{Examples 1} and \textbf{2} show purification for almost all trajectories, i.e., the time-evolved state approaches $\ketbra{0}$ or $\ketbra{1}$.
For these examples, we find that $(\star\star)$ does not hold, whereas $(\star)$ holds.

Let us especially revisit \textbf{Example 1}, which is given in Eq.~\eqref{example1}.
In this case, while purification occurs, we can find that $(\star)$ holds.
Indeed, if we take $\hat{\mc{Q}}=\ketbra{2}+\ketbra{3}$ with $\mr{rank}[\mc{\hat{Q}}]=2$, Eq.~\eqref{purjoken} holds with $\zeta_0=\zeta_1=0$ and $\zeta_2=1$.
This fact illustrates that $(\star)$ is not a sufficient condition for the breakdown of purification.
However, we also find that $(\star\star)$ is not satisfied in this case;
for instance, if we consider $b=1$ and $b'=2$, $\hat{\mc{Q}}=\ketbra{2}+\ketbra{3}$ no longer satisfies Eq.~\eqref{purjoken2}.
We can then conclude that no projector $\hat{\mc{Q}}$ with $\mr{rank}[\hat{\mc{Q}}]\geq 2$ exists satisfying $(\star \star)$, and therefore purification for typical quantum trajectories is justified.
A similar discussion holds for \textbf{Example 2}.

Next, we introduce two additional examples where purification does not occur.

\textbf{Example 4.}
Let us again consider a four-level system, where each level is denoted by $\ket{0},\ket{1},\ket{2},$ and $\ket{3}$.
We assume that the measurement operators are given by
\aln{
\begin{split}
\hat{M}_0&=\frac{1}{\sqrt{2}}(\ket{0}\bra{0}-\ketbra{1}),\\
\hat{M}_1&=\frac{1}{\sqrt{2}}(\ket{0}\bra{0}+\ket{1}\bra{1}),\\
\hat{M}_2&=\ket{1}\bra{3}+\ket{0}\bra{2}.
\end{split}}
We can easily confirm that $\sum_b\hat{M}_b^\dag\hat{M}_b=\hat{\mbb{I}}$ and that $\ml{E}$ has multiple stationary states, namely
\aln{
\ml{E}[c\ket{0}\bra{0}+(1-c)\ket{1}\bra{1}]=c\ket{0}\bra{0}+(1-c)\ket{1}\bra{1}
}
for arbitrary $c\in[0,1]$.

Let us consider quantum trajectories starting from an initial state
\aln{
\hat{\rho}_0=\gamma\ket{+}\bra{+}+(1-\gamma)(\alpha\ket{2}\bra{2}+(1-\alpha)\ket{3}\bra{3}),
}
where $\gamma,\alpha\in [0,1]$ and $\ket{\pm}=\frac{\ket{0}\pm\ket{1}}{\sqrt{2}}$.
The transition paths are illustrated in Fig.~\ref{fig9}.
After the first step, the state becomes $\ketbra{-}$, $\ketbra{+}$, and ${\alpha\ketbra{0}+(1-\alpha)\ketbra{1}}$ for the outcomes $b_1=0,1,$ and $2$, respectively.
The corresponding probabilities read $p_0=p_1=\gamma/2$ and $p_2=1-\gamma$.
After that, we will measure either $b=0$ or $b=1$ with equal probabilities at each time step.
For the case with $b_1=0$ or $1$, the measurement with $b=1$ keeps the state invariant, while $b=0$ changes $\ketbra{\pm}$ into $\ketbra{\mp}$.
In this case, $\ketbra{\pm}$ and $\ketbra{\mp}$ appear with the same frequency in a quantum trajectory in the long run.
In contrast, for the case with $b_1=2$, the state is kept invariant irrespective of subsequent measurement outcomes $b=0$ or $1$.
Therefore, we find that purification for almost all trajectories does not hold when $\alpha\neq 0,1$ and $\gamma\neq 1$, i.e., the state for the case with $b_1=2$ remains mixed.
Note that the absence of purification is understood from the fact that there exists a projection, e.g., $\ml{Q}=\ketbra{0}+\ketbra{1}$, with which Eq.~\eqref{eq:Benoist-purification} holds, where $\zeta_{\bm{b}_n}=1/2^{n}$ when $b_l\neq 2\;(1\leq l \leq n)$ and $\zeta_{\bm{b}_n}=0$ otherwise.

Moreover, we find that 
\aln{
\hat{\rho}_{\bm{b}}^\infty=
\left\{
\begin{array}{ll}
\frac{1}{2}(\ketbra{0}+\ketbra{1}) &\quad \mr{Prob}:\gamma, \\
{\alpha\ketbra{0}+(1-\alpha)\ketbra{1}} &\quad \mr{Prob}:1-\gamma.
\end{array}
\right.
}
and
\aln{
\av{\ml{E}^n}[\hat{\rho}_0]=\lim_{n\ra\infty}\ml{E}^n[\hat{\rho}_0]
=\lrs{\frac{\gamma}{2}+(1-\gamma)\alpha}\ketbra{0}+\lrs{\frac{\gamma}{2}+(1-\gamma)(1-\alpha)}{\ketbra{1},
\label{eq:ergodicity-breaking_example5}
}}
which means that ergodicity also breaks down due to the existence of multiple stationary states for $\ml{E}$.
In contrast, we can confirm the statement (2) regarding ergodicity, i.e., $\ml{E}[\hat{\rho}_{\bm{b}}^\infty]=\hat{\rho}_{\bm{b}}^\infty$ and 
$\mbb{E}[\hat{\rho}_{\bm{b}}^\infty]=\av{\ml{E}^n}[\hat{\rho}_0]$. 
We note that there is only trivial peripheral spectrum in the CPTP map considered here, which is the reason why the long-time average can be removed in Eq.~\eqref{eq:ergodicity-breaking_example5}.

\begin{figure}[!h]
\centering\includegraphics[width=14cm]{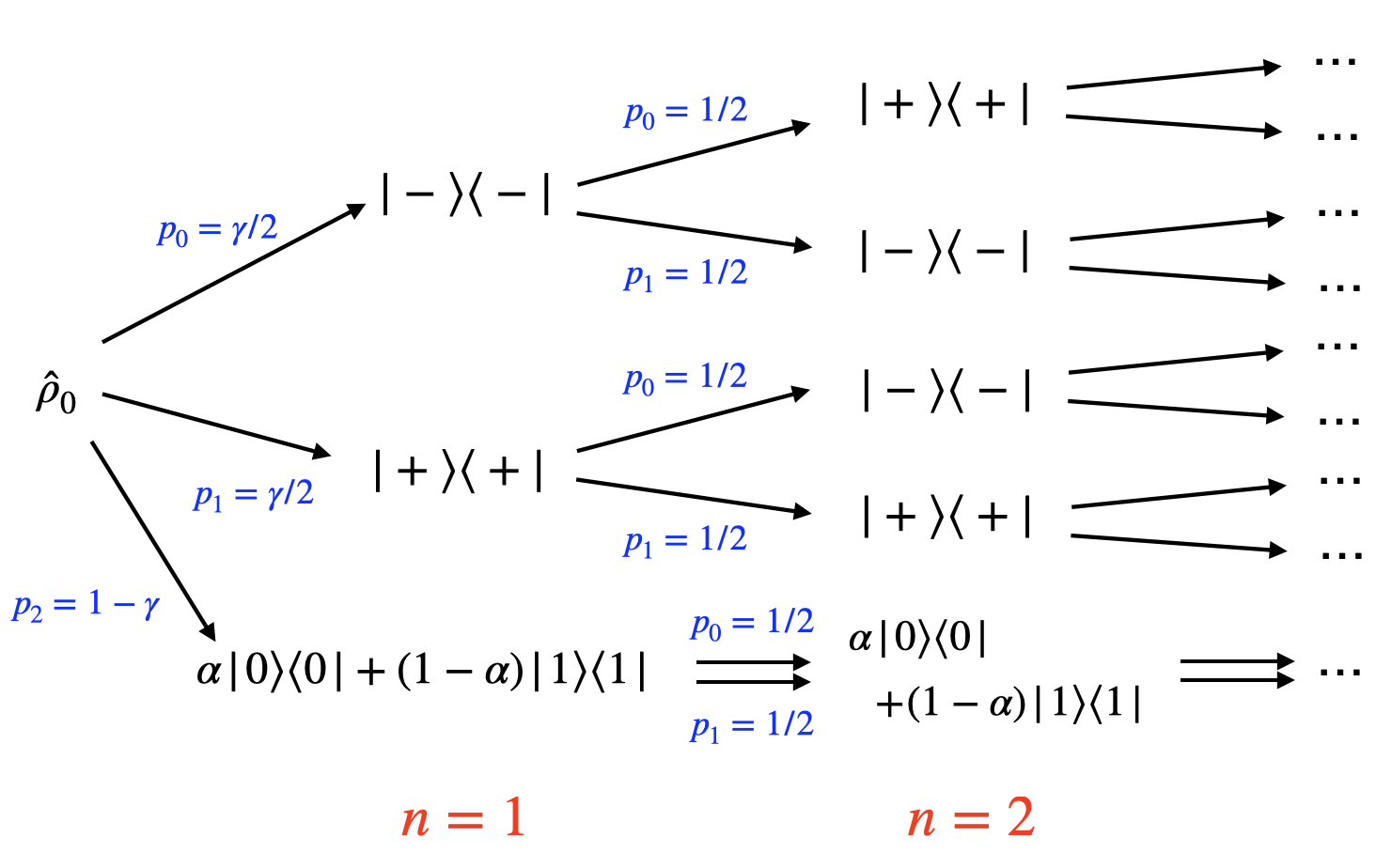}
\caption{
Transitions of the state $\hat{\rho}_0$ for \textbf{Example 4}.
Trajectories transition between $\ketbra{+}$ and $\ketbra{-}$ with equal probabilities for $b_1=0$ and $1$. 
In contrast, the state becomes invariant from $\alpha\ketbra{0}+(1-\alpha)\ketbra{1}$ after subsequent measurements for $b_1=2$.
Consequently, ergodicity and purification for almost all trajectories do not hold true.
}
\label{fig9}
\end{figure}

\textbf{Example 5.}
As the next example, let us consider the four-level system whose measurement operators read
\aln{
\begin{split}
\hat{M}_0&=\frac{1}{\sqrt{2}}(\ket{0}\bra{0}-\ketbra{1}),\\
\hat{M}_1&=\frac{1}{\sqrt{2}}(\ket{0}\bra{1}+\ket{1}\bra{0}),\\
\hat{M}_2&=\ket{1}\bra{3}+\ket{0}\bra{2}.
\end{split}}
In this case, we find that there is only one stationary state,
\aln{
\hat{\rho}_\mr{ss}=\frac{1}{2}(\ketbra{0}+\ketbra{1}).
}

\begin{figure}[!h]
\centering\includegraphics[width=12cm]{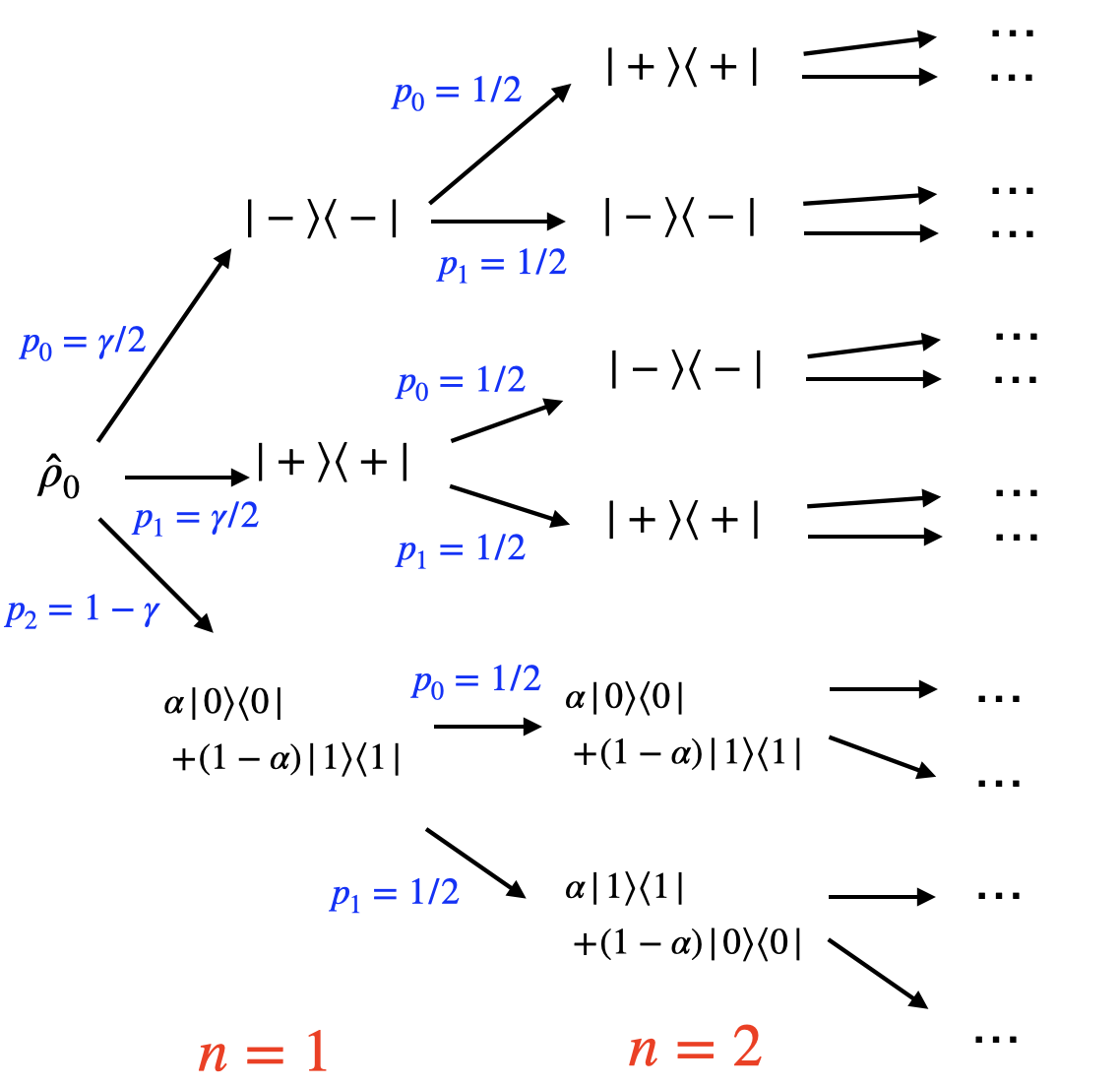}
\caption{
Transitions of the state $\hat{\rho}_0$ for \textbf{Example 5}.
Trajectories transition between $\ketbra{+}$ and $\ketbra{-}$ with equal probabilities for $b_1=0$ and $1$. 
Likewise, they transition between $\alpha\ketbra{0}+(1-\alpha)\ketbra{1}$ and $\alpha\ketbra{1}+(1-\alpha)\ketbra{0}$ with equal probabilities for $b_1=2$. 
Then, ergodicity holds true for almost all quantum trajectories, while purification for $b_1=2$ breaks down.
}
\label{fig10}
\end{figure}

We again start from the same initial state
\aln{
\hat{\rho}_0=\gamma\ket{+}\bra{+}+(1-\gamma)(\alpha\ket{2}\bra{2}+(1-\alpha)\ket{3}\bra{3}).
}
Then, after the first measurement, the state becomes $\ketbra{-}$, $\ketbra{+}$, or ${\alpha\ketbra{0}+(1-\alpha)\ketbra{1}}$ for the outcomes $b_1=0,1,$ and $2$, respectively.
The corresponding probabilities read $p_0=p_1=\gamma/2$ and $p_2=1-\gamma$.
After that, we will measure either $b=0$ or $b=1$ with equal probabilities at each time step.
For the case with $b_1=0$ or $1$, the measurement with $b=1$ keeps the state invariant, while $b=0$ changes $\ketbra{\pm}$ into $\ketbra{\mp}$.
In this case, $\ketbra{\pm}$ and $\ketbra{\mp}$ appear with the same frequency in a quantum trajectory in the long run.
In contrast, for the case with $b_1=2$, the measurement with $b=0$ keeps the state invariant, while $b=1$ changes $\alpha\ketbra{0}+(1-\alpha)\ketbra{1}$ into $\alpha\ketbra{1}+(1-\alpha)\ketbra{0}$ and vice versa.
For this case, $\alpha\ketbra{0}+(1-\alpha)\ketbra{1}$ and $\alpha\ketbra{1}+(1-\alpha)\ketbra{0}$ appear with the same frequency in the long run.
Therefore, for both cases, we have
\aln{
\hat{\rho}_{\bm{b}}^\infty=\frac{1}{2}(\ketbra{0}+\ketbra{1})=\hat{\rho}_\mr{ss}.
}

This means that, while ergodicity holds true, purification does not hold true for almost all trajectories when $\gamma\neq 1$ and $\alpha\neq 0,1$; i.e., the state remains mixed for $b_1=2$ (see Fig.~\ref{fig10}). 
Again, the absence of purification is consistent with the existence of the projection $\ml{Q}=\ketbra{0}+\ketbra{1}$, with which Eq.~\eqref{eq:Benoist-purification} holds.

As seen from \textbf{Examples 1.-5.}, the ergodicity and purification of quantum trajectories are independent concepts in general.
However, we note that purification is often assumed to show the existence of the invariant measure of quantum trajectories~\cite{benoist2019invariant}, which is related to the ergodicity property of nonlinear quantities, as will be detailed in Sec.~\ref{sec:ergodicity_nonlinear-quantity}.

\subsubsection{Relation to the product of Kraus operators}
\label{sec:rank-M_purification}
Here, we discuss the relation between the purification of $\hat{\rho}_{\bm{b};n}$ and the matrix rank of $\hat{\mathsf{M}}_{\bm{b};n}$ on the basis of Ref.~\cite{benoist2019invariant}. 
We try to explain the relation in a physicist-friendly manner. 
See Ref.~\cite{benoist2019invariant} for a rigorous treatment. 
We consider the following operator depending on the sequence of measurement outcomes $\bm{b}$ and perform its singular-value decomposition,
\begin{align}
    \hat{Y}_{\bm{b};n}=\frac{\hat{\mathsf{M}}_{\bm{b};n}}
    {\sqrt{\mathrm{Tr}\left(\hat{\mathsf{M}}_{\bm{b};n}^\dagger\hat{\mathsf{M}}_{\bm{b};n}\right)}}=\sum_i\Lambda_{i,\bm{b};n}   \ket{\Psi_{i,\bm{b};n}}\bra{\Phi_{i,\bm{b};n}},
\end{align}
where $\{\Lambda_{i,\bm{b};n}\}_i$ are the singular values of $\hat{Y}_{\bm{b};n}$ arrayed as $\Lambda_{i,\bm{b};n}\geq\Lambda_{i+1,\bm{b};n}$. 
The states $\ket{\Psi_{i,\bm{b};n}}$ and $\ket{\Phi_{i,\bm{b};n}}$ are eigenstates of $\hat{Y}_{\bm{b};n}\hat{Y}_{\bm{b};n}^\dagger$ and $\hat{Y}_{\bm{b};n}^\dagger\hat{Y}_{\bm{b};n}$, respectively, corresponding to the $i$th eigenvalue $\Lambda_{i,\bm{b};n}^2$. 
These states are orthonormalized as $\langle\Psi_{i,\bm{b};n}|\Psi_{j,\bm{b};n}\rangle=\langle\Phi_{i,\bm{b};n}|\Phi_{j,\bm{b};n}\rangle=\delta_{ij}$. 

We can easily understand that, if $\lim_{n\rightarrow\infty}\mathrm{rank}\left(\hat{Y}_{\bm{b};n}\right)=1$, the purification of $\hat{\rho}_{\bm{b};n}$ occurs, since $\mathrm{rank}\left(\hat{\rho}_{\bm{b};n}\right)=\mathrm{rank}\left(\hat{Y}_{\bm{b};n}\hat{\rho}_0\hat{Y}_{\bm{b};n}^\dagger\right)\leq\mathrm{rank}\left(\hat{Y}_{\bm{b};n}\right)$\footnote{
Note that $\mr{rank}(\hat{A}\hat{B})\leq \mr{rank}(\hat{A})$ and $\mr{rank}(\hat{A}\hat{B})\leq \mr{rank}(\hat{B})$ are always satisfied.
}. 
In particular, if $\hat{Y}_{\bm{b};n}$ typically exhibits $\lim_{n\rightarrow\infty}\mathrm{rank}\left(\hat{Y}_{\bm{b};n}\right)=1$, it is obvious that Eq.~\eqref{eq:purification} is satisfied in typical trajectories. 
To see the converse, i.e., that the typical purification of $\hat{\rho}_{\bm{b};n}$ indicates $\lim_{n\rightarrow\infty}\mathrm{rank}\left(\hat{Y}_{\bm{b};n}\right)=1$ for typical $\hat{Y}_{\bm{b};n}$, we consider the limiting behavior of the matrix 
\begin{align}
    \hat{Z}_{\bm{b};n}=\hat{Y}_{\bm{b};n}^\dagger\hat{Y}_{\bm{b};n}=\frac{\hat{\mathsf{M}}_{\bm{b};n}^\dagger\hat{\mathsf{M}}_{\bm{b};n}}
    {\mathrm{Tr}\left(\hat{\mathsf{M}}_{\bm{b};n}^\dagger\hat{\mathsf{M}}_{\bm{b};n}\right)}
    =\sum_i\Lambda_{i,\bm{b};n}^2\ket{\Phi_{i,\bm{b};n}}\bra{\Phi_{i,\bm{b};n}}.
    \label{eq:Zn}
\end{align}
If we take the initial state as $\hat{\rho}_0=\hat{\mbb{I}}/d$, we can show that the function $\hat{Z}_{\bm{b};n}$ is a martingale~\cite{hall2014martingale}. 
That is, the expectation value of $\hat{Z}_{\bm{b};n+1}$ under the condition that the measurement outcomes from $1$ to $n$ steps are $\bm{b}_n=(b_1,b_2,\ldots,b_n)$ becomes
\begin{align}
    \mathbb{E}_{\mathbb{I}/d}\left(\hat{Z}_{\bm{b};n+1} \middle| \bm{b}_n\right)
    &=\sum_{b_{n+1}}\hat{Z}_{\bm{b};n+1}
    \mathrm{Tr}\left(\hat{M}_{b_{n+1}}
    \hat{\rho}_{\bm{b};n}\hat{M}_{b_{n+1}}^\dagger\right)\nonumber\\
    &=\sum_{b_{n+1}}
    \frac{\hat{\mathsf{M}}_{\bm{b};n}^\dagger\hat{M}_{b_{n+1}}^\dagger\hat{M}_{b_{n+1}}\hat{\mathsf{M}}_{\bm{b};n}}
    {\mathrm{Tr}\left(\hat{\mathsf{M}}_{\bm{b};n+1}^\dagger\hat{\mathsf{M}}_{\bm{b};n+1}\right)}
    \frac{\mathrm{Tr}\left(\hat{\mathsf{M}}_{\bm{b};n+1}\frac{\hat{\mathbb{I}}}{d}
    \hat{\mathsf{M}}_{\bm{b};n+1}^\dagger\right)}
    {\mathrm{Tr}\left(\hat{\mathsf{M}}_{\bm{b};n}\frac{\hat{\mathbb{I}}}{d}\hat{\mathsf{M}}_{\bm{b};n}^\dagger\right)}
    \nonumber\\&=\hat{Z}_{\bm{b};n},
    \label{eq:martingale_Zn}
\end{align}
where $\sum_{b_{n+1}}\hat{M}_{b_{n+1}}^\dagger\hat{M}_{b_{n+1}}=\hat{\mbb{I}}$ is used. 
Since $\hat{Z}_{\bm{b};n}$ is a positive semidefinite matrix satisfying $\mathrm{Tr}\left(\hat{Z}_{\bm{b};n}\right)=1$, the matrix elements of $\hat{Z}_{\bm{b};n}$ satisfy\footnote{
For a positive semidefinite matrix $\hat{A}$, $|A_{ij}|= |\braket{i|\sqrt{\hat{A}}\sqrt{\hat{A}}|j}|\leq \sqrt{A_{ii}A_{jj}}$ is satisfied. 
Combining this with  $A_{ii}\leq \sum_{i}A_{ii}=1$ if $\Tr(\hat{A})=1$, we obtain $|A_{ij}|\leq 1$.
}
\begin{align}
    \left|\left[\hat{Z}_{\bm{b};n}\right]_{ij}\right|\leq1.
    \label{eq:bound_Zn}
\end{align}
Using Eqs.~\eqref{eq:martingale_Zn} and~\eqref{eq:bound_Zn}, we can apply the martingale convergence theorem~\cite{hall2014martingale} to each element of $\hat{Z}_{\bm{b};n}$. 
Then, we find that the limit
\begin{align}
    \lim_{n\rightarrow\infty}\hat{Z}_{\bm{b};n}=\hat{Z}_{\bm{b}}
    \label{eq:limit_Zn}
\end{align}
exists almost surely with respect to $\bm{b}$. 
Equation~\eqref{eq:limit_Zn} means that $\Lambda_{i,\bm{b};n}$ and $\ket{\Phi_{i,\bm{b};n}}$ asymptotically become independent of $n$, and thus $\hat{Y}_{\bm{b};n}$ can be written as
\begin{align}
    \hat{Y}_{\bm{b};n}\simeq\sum_i\Lambda_{i,\bm{b}}\ket{\Psi_{i,\bm{b};n}}\bra{\Phi_{i,\bm{b}}}
    \label{eq:limit_Y}
\end{align}
in the long-time regime, where $\Lambda_{i,\bm{b}}=\lim_{n\rightarrow\infty}\Lambda_{i,\bm{b};n}$ and $\bra{\Phi_{i,\bm{b}}}=\lim_{n\rightarrow\infty}\bra{\Phi_{i,\bm{b};n}}$\footnote{
When $\{\Lambda_{i,\bm{b}}\}_i$ are degenerate, there is some ambiguity on how to determine $\{\ket{\Phi_{i,\bm{b}}}\}_i$. 
However, such ambiguity has no effect on the discussions below.
}. 
Therefore, if $\lim_{n\rightarrow\infty}\mathrm{rank}\left(\hat{Y}_{\bm{b};n}\right)=\tilde{d}\geq2$ and thus $\Lambda_{\tilde{d},\bm{b}}\neq0$ are satisfied in a set $\mathfrak{B} \subseteq \{ \bm{b} \} $ of measurement outcomes, the initially maximally mixed state $\hat{\rho}_0=\hat{\mathbb{I}}/d$ does not purify. 
Indeed, in this case, the matrix rank of $\hat{\rho}_{\bm{b};n}\propto\hat{Y}_{\bm{b};n}\hat{\rho}_0\hat{Y}_{\bm{b};n}^\dagger=\frac{1}{d}\sum_{i=1}^{\tilde{d}}\Lambda_{i,\bm{b}}^2\ket{\Psi_{i,\bm{b};n}}\bra{\Psi_{i,\bm{b};n}}$ is larger than one due to $\Lambda_{\tilde{d},\bm{b}}\neq0$. 

From Eq.~\eqref{eq:limit_Zn}, we can also understand that the breakdown of purification is related to the existence of dark subspaces. 
A subspace $\mathcal{D} \subseteq \mc{H}$ of quantum pure states is referred to as a dark subspace if, for any $\ket{\psi}\in\mathcal{D}$, there exist unitary operators $\{\hat{U}_{\bm{b};m}\}_{\bm{b},m}$ such that 
\begin{align}
    \hat{\mathsf{M}}_{\bm{b};m}\ket{\psi}\propto\hat{U}_{\bm{b};m}\ket{\psi}
\end{align}
is satisfied for arbitrary $\bm{b}$ and $m$. 
If there is a $\tilde{d}$-dimensional dark subspace with $\tilde{d}\geq2$,  Eq.~\eqref{eq:Benoist-purification} is satisfied with $\hat{\mathcal{Q}}=\hat{P}_\mathcal{D}$, which is the projector onto the dark subspace. 
To see the relation between the dark subspace and the breakdown of purification, we compare $\hat{Z}_{\bm{b};n+m}$ and $\hat{Z}_{\bm{b};n}$ for sufficiently large $n$ where Eq. (\ref{eq:limit_Y}) is satisfied. 
Since $\hat{Y}_{\bm{b};n+m}$ satisfies $\hat{Y}_{\bm{b};n+m}\propto\hat{\mathsf{M}}_{\vartheta^n\bm{b};m}\hat{Y}_{\bm{b};n}$, the former becomes
\begin{align}
    \hat{Z}_{\bm{b};n+m}=\hat{Y}_{\bm{b};n+m}^\dagger\hat{Y}_{\bm{b};n+m}&\propto\sum_{ij}\Lambda_{i,\bm{b}}\Lambda_{j,\bm{b}}\ket{\Phi_{i,\bm{b}}}\bra{\Psi_{i,\bm{b};n}}\hat{\mathsf{M}}_{\vartheta^n\bm{b};m}^\dagger\hat{\mathsf{M}}_{\vartheta^n\bm{b};m}\ket{\Psi_{j,\bm{b};n}}\bra{\Phi_{j,\bm{b}}},
    \label{eq:Zn+m}
\end{align}
with $\vartheta^n\bm{b}=(b_{n+1,},b_{n+2},\ldots)$ [see also Eq.~\eqref{eq:shift}]. 
Due to the convergence in Eq.~\eqref{eq:limit_Zn}, $\hat{Z}_{\bm{b};n+m}=\hat{Z}_{\bm{b};n}$ should be satisfied. 
Here, we consider the situation where $\Lambda_{\tilde{d},\bm{b}}\neq0$ is satisfied with $\tilde{d}\geq2$ and $\Lambda_{i>\tilde{d},\bm{b}}=0$.
This means that the quantum trajectory corresponding to $\bm{b}$ does not purify. 
In this case, comparing Eqs.~\eqref{eq:Zn+m} and \eqref{eq:Zn}, with $\Lambda_{i,\bm{b};n}=\Lambda_{i,\bm{b}}$ and $\ket{\Phi_{i,\bm{b};n}}=\ket{\Phi_{i,\bm{b}}}$, we realize that 
\begin{align}
\bra{\Psi_{i,\bm{b};n}}\hat{\mathsf{M}}_{\vartheta^n\bm{b};m}^\dagger\hat{\mathsf{M}}_{\vartheta^n\bm{b};m}\ket{\Psi_{j,\bm{b};n}}\propto\delta_{ij}
\end{align}
should be satisfied for any $m$ and $i,j\leq\tilde{d}$. 
This means that $\ket{\Psi_{i,\bm{b};n}}$ resides in the $\tilde{d}$-dimensional dark subspace.  

We note that the relation between purification and the dark subspace applies to dynamics from arbitrary initial states, while in the discussion above the initial state is chosen to be the maximally mixed state $\hat{\rho}_0=\hat{\mathbb{I}}/d$. 
This is because the maximally mixed state exhibits absolute continuity with respect to any initial state; atypical behaviors of $\hat{Z}_{\bm{b};n}$ not observed in trajectories with $\hat{\rho}_0=\hat{\mathbb{I}}/d$ are also not observed in typical trajectories with arbitrary initial states. 
Explanations about the absolute continuity will be given in Sec.~\ref{sec:typical-convergence_Lyapunov-spectrum}.

\section{Typical behaviors of nonlinear quantities and the Lyapunov spectrum in quantum trajectories}
\label{sec:nonlear-quantity_Lyapunov-spectrum}
In this chapter, we further discuss the typical behaviors of quantum trajectories. 
While Chapter~\ref{sec:linear-quantity_purification} mainly focused on linear observables and purification, we here consider the typical behaviors of nonlinear quantities and the Lyapunov spectrum. 
As will be shown in Chapter~\ref{sec:mipt}, these quantities can reveal intriguing properties of quantum trajectories that may be invisible in the corresponding averaged CPTP dynamics. 
Discussions in this chapter are based on Refs.~\cite{benoist2019invariant, benoist2023limit} by Benoist and coauthors. 
While we focus here on the discrete-time dynamics, discussions on continuous-time dynamics can be found in Ref.~\cite{benoist2021invariant}.

\subsection{Invariant measure and ergodicity of measurement outcomes}
\label{sec:invariant-measure_ergodicity_outcomes}
We first introduce the notions of invariant measure and the ergodicity of measurement outcomes in quantum trajectories. 
These play important roles when we study the time average of nonlinear quantities and the Lyapunov spectrum, which will be discussed in Secs.~\ref{sec:ergodicity_nonlinear-quantity} and~\ref{sec:Lyapunov-analysis}, respectively.

As reviewed in Chapter~\ref{sec:trajectory_master-equation}, a quantum trajectory is described by a sequence of measurement outcomes
\begin{align}
    \bm{b}=(b_1,b_2,\ldots).
\end{align}
The set of the first $n$ components of $\bm{b}$ is denoted by
\begin{align}
    \bm{b}_n=(b_1,b_2,\ldots,b_n),
\end{align}
which is also often used to describe quantum trajectories. 
To explore concepts related to measure theory, we consider a set of outcome sequences $\mathfrak{B}_n \subseteq \{ \bm{b} \}$ where a certain condition is imposed on $\bm{b}_n$, i.e., 
\begin{align}
    \mathfrak{B}_n=\{\bm{b}:\bm{b}_n\in\mathsf{B}_n\},\ \ \mathsf{B}_n \subseteq \{\bm{b}_n\}.
\end{align}
Here, $\{\bm{b}_n\}$ is the set of all possible $\bm{b}_n$ with $n$ being fixed. 
As an example of a set $\mathsf{B}_n$, we can consider $\{\bm{b}_n:b_1=0,b_n=0\}$. 
We also consider shifted sequences of measurement outcomes,
\begin{align}\label{eq:shift}
    \vartheta^m\bm{b}=(b_{m+1},b_{m+2},\ldots),
\end{align}
to describe the dynamics of quantum trajectories.

When we take an initial state $\hat{\rho}_0$, the probability measure that $\mathfrak{B}_n$ is realized is determined through the Born rule, i.e., 
\begin{align}
    P_{\rho_0}(\mathfrak{B}_n)
    =\sum_{\bm{b}_n \in \mathsf{B}_n}\mathrm{Tr}
    \left(\hat{\mathsf{M}}_{\bm{b};n}\hat{\rho}_0\hat{\mathsf{M}}_{\bm{b};n}^\dagger\right),
    \label{eq:born-rule}
\end{align}
where $\hat{\mathsf{M}}_{\bm{b};n}$ is the product of Kraus operators corresponding to $\bm{b}_n$,
\begin{align}
    \hat{\mathsf{M}}_{\bm{b};n}=\hat{M}_{b_n}\hat{M}_{b_{n-1}}\cdots\hat{M}_{b_1},
\end{align}
as defined in Eq.~\eqref{eq:product_Kraus-operators}. 
If the probability measure of $\mathfrak{B}_n$ is the same as that of $\vartheta^{-1}\mathfrak{B}_n$, i.e., 
\begin{align}
    P_{\rho_0}(\mathfrak{B}_n)=P_{\rho_0}(\vartheta^{-1}\mathfrak{B}_n)
    \label{eq:probability_outcomes_shifted}
\end{align}
is satisfied for any $\mathfrak{B}_n$, the probability measure is referred to as an invariant measure. 
Here, $\vartheta^{-m}\mathfrak{B}_n$ is defined as 
\begin{align}
    \vartheta^{-m}\mathfrak{B}_n=\{\bm{b}:\vartheta^m\bm{b} \in \mathfrak{B}_n\}.
\end{align}
In terms of the CPTP dynamics with an initial state $\hat{\rho}_0$, averaged over measurement outcomes, the right-hand side of Eq.~\eqref{eq:probability_outcomes_shifted} is written as 
\begin{align}
    P_{\rho_0}(\vartheta^{-m}\mathfrak{B}_n)
    =P_{\mathcal{E}^m[\rho_0]}(\mathfrak{B}_n)
    \label{eq:shift_probability-measure}
\end{align}
with $m=1$.

If the probability measure of a set $\mathfrak{B}_n$ invariant under the shift becomes $0$ or $1$, the measure is said to be ergodic. 
This condition can be written as
\begin{align}
    P_{\rho_0}(\mathfrak{B}_n) \in \{ 0, 1 \},\ \ \forall \, \mathfrak{B}_n \ \mathrm{s.t.} \ \mathfrak{B}_n=\vartheta^{-1}\mathfrak{B}_n.
    \label{eq:ergodicity_theta}
\end{align}
For example, when we choose a set $\mathfrak{B}_\infty=\{\bm{b}:b_i=0\,\forall\,i\}$, this set is invariant under the shift, i.e., $\vartheta^{-1}\mathfrak{B}_\infty=\mathfrak{B}_\infty$. 
Such a set satisfies $P_{\rho_0}(\mathfrak{B}_\infty)=0$ if the measure of $\mathfrak{B}_\infty$ is ergodic and $\hat{\rho}_0$ satisfies $\mathrm{Tr}(\hat{M}_b\hat{\rho}_0\hat{M}_b^\dagger)\neq0$ with $b\neq0$. 
On the other hand, if we choose $\mathfrak{B}_\infty=\{\bm{b}\}$, this set satisfies $\vartheta^{-1}\mathfrak{B}_\infty=\mathfrak{B}_\infty$ and $P_{\rho_0}(\mathfrak{B}_\infty)=1$. 

The ergodicity of measurement outcomes in Eq.~\eqref{eq:ergodicity_theta} means that a nontrivial invariant set $\mathfrak{B}_n=\vartheta^{-1}\mathfrak{B}_n$ satisfying $0<P_{\rho_0}(\mathfrak{B}_n)<1$ is absent. 
We can also intuitively understand the meaning of the ergodicity of measurement outcomes by considering a set $\mathfrak{B}_n$ that satisfies $0 < P_{\rho_0}(\mathfrak{B}_n) < 1$. 
If the ergodicity in Eq.~\eqref{eq:ergodicity_theta} is satisfied, $\vartheta^{-1}\mathfrak{B}_n$ always deviates from $\mathfrak{B}_n$, i.e., $\vartheta^{-1}\mathfrak{B}_n\neq\mathfrak{B}_n$ holds. 
Since $\vartheta^{-m}\mathfrak{B}_n$ never returns to $\mathfrak{B}_n$, by applying $\vartheta^{-1}$ to $\mathfrak{B}_n$ repeatedly, we find that $\{\vartheta^{-m}\mathfrak{B}_n\}_{m=0}^M$ eventually covers almost all trajectories in $\bm{b}$ with increasing $M$. 
Indeed, if we consider a set $\overline{\mathfrak{B}_n}=\lim_{M\rightarrow\infty}\bigcup_{m=0}^M\vartheta^{-m}\mathfrak{B}_n$ where $0<P_{\rho_0}(\mathfrak{B}_n)<1$ is satisfied, the ergodicity leads to $P_{\rho_0}(\overline{\mathfrak{B}_n})=1$ owing to $\vartheta^{-1}\overline{\mathfrak{B}_n}=\overline{\mathfrak{B}_n}$. 
The schematic picture is given in Fig.~\ref{fig:ergodicity_outcomes}. 

We note that, if the outcome $b_i$ at each step $i$ obeys an independent and identically distributed (i.i.d.) distribution $\{p_b^\mathrm{iid}\}$, $P_\mathrm{iid}(\mathfrak{B}_n)=\sum_{\bm{b}_n \in \mathsf{B}_n}\prod_{i=1}^np_{b_i}^\mathrm{iid}$ becomes an invariant measure that is ergodic. 
In contrast, in quantum trajectories where the probability of measurement outcomes obeys the Born rule in Eq.~\eqref{eq:born-rule}, it is nontrivial whether or not an invariant measure exists and ergodicity is satisfied. 

\begin{figure}[!h]
\centering\includegraphics[width=7cm]{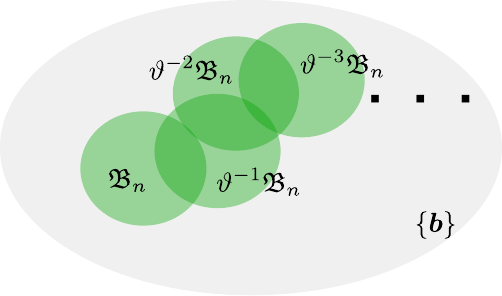}
\caption{
Intuitive picture for the ergodicity of measurement outcomes. 
The gray ellipse represents the entire set $\{\bm{b}\}$. 
The leftmost green circle is a set $\mathfrak{B}_n$ that satisfies $0<P_{\rho_0}(\mathfrak{B}_n)<1$. 
If the ergodicity in Eq.~\eqref{eq:ergodicity_theta} holds, such a set always satisfies $\vartheta^{-1}\mathfrak{B}_n\neq\mathfrak{B}_n$. 
This means that $\{\vartheta^{-m}\mathfrak{B}_n\}_{m=0}^\infty$ covers almost all trajectories in $\{\bm{b}\}$.
}
\label{fig:ergodicity_outcomes}
\end{figure}

If the stationary state $\hat{\rho}_\mathrm{ss}=\mathcal{E}[\hat{\rho}_\mathrm{ss}]$ is unique, we can show that there exists an invariant and ergodic measure of $\mathfrak{B}_n$~\cite{benoist2019invariant}. 
Indeed, if we choose the initial state as the stationary state $\hat{\rho}_\mathrm{ss}$, the probability measure evaluated through $\hat{\rho}_\mathrm{ss}$ becomes an invariant measure,
\begin{align}
    P_{\rho_\mathrm{ss}}(\vartheta^{-1}\mathfrak{B}_n)
    =P_{\rho_\mathrm{ss}}(\mathfrak{B}_n),
    \label{eq:invariant-measure_outcome-sequence}
\end{align}
which can be understood from Eq.~\eqref{eq:shift_probability-measure}. 
To show the ergodicity based on the uniqueness of $\hat{\rho}_\mathrm{ss}$, we calculate the probability measure of the outcome sequences such that $\bm{b} \in \mathfrak{B}_n$ and $\vartheta^m\bm{b} \in \tilde{\mathfrak{B}}_n$ are satisfied. 
When $m \geq n$, such a probability measure becomes
\begin{align}
    P_{\rho_\mathrm{ss}}\left(\mathfrak{B}_n\cap\vartheta^{-m}\tilde{\mathfrak{B}}_n\right)
    =\sum_{\bm{b}_n \in \mathsf{B}_n}\sum_{\tilde{\bm{b}}_n \in \tilde{\mathsf{B}}_n}
    \mathrm{Tr}\left(\hat{\mathsf{M}}_{\tilde{\bm{b}};n}\mathcal{E}^{m-n}
    \left[\hat{\mathsf{M}}_{\bm{b};n}\hat{\rho}_\mathrm{ss}
    \hat{\mathsf{M}}_{\bm{b};n}^\dagger\right]\hat{\mathsf{M}}^\dagger_{\tilde{\bm{b}};n}\right).
    \label{eq:joint-probability}
\end{align}
We consider the time average of $P_{\rho_\mathrm{ss}}\left(\mathfrak{B}_n\cap\vartheta^{-m}\tilde{\mathfrak{B}}_n\right)$ with respect to $m$. 
The time average of the CPTP map in the sum can be evaluated as
\begin{align}
    \lim_{M\rightarrow\infty}\frac{1}{M}\sum_{m=n}^{n+M-1}
    \mathcal{E}^{m-n}\left[\hat{\mathsf{M}}_{\bm{b};n}\hat{\rho}_\mathrm{ss}\hat{\mathsf{M}}_{\bm{b};n}^\dagger\right]=\mathrm{Tr}\left(\hat{\mathsf{M}}_{\bm{b};n}\hat{\rho}_\mathrm{ss}\hat{\mathsf{M}}_{\bm{b};n}^\dagger\right)\hat{\rho}_\mathrm{ss},
    \label{eq:time-average_CPTP-map}
\end{align}
where the uniqueness of the stationary state is used. 
Equation~\eqref{eq:time-average_CPTP-map} leads to
\begin{align}
    \lim_{M\rightarrow\infty}\frac{1}{M}\sum_{m=n}^{n+M-1}P_{\rho_\mathrm{ss}}\left(\mathfrak{B}_n\cap\vartheta^{-m}\tilde{\mathfrak{B}}_n\right)
    &=\sum_{\bm{b}_n \in \mathsf{B}_n}\mathrm{Tr}
    \left(\hat{\mathsf{M}}_{\bm{b};n}\hat{\rho}_\mathrm{ss}
    \hat{\mathsf{M}}_{\bm{b};n}^\dagger\right)
    \sum_{\tilde{\bm{b}}_n \in \tilde{\mathsf{B}}_n}\mathrm{Tr}
    \left(\hat{\mathsf{M}}_{\tilde{\bm{b}};n}\hat{\rho}_\mathrm{ss}
    \hat{\mathsf{M}}^\dagger_{\tilde{\bm{b}};n}\right)\nonumber\\
    &=P_{\rho_\mathrm{ss}}\left(\mathfrak{B}_n\right)
    P_{\rho_\mathrm{ss}}\left(\tilde{\mathfrak{B}}_n\right).
    \label{eq:time-average_joint-probability}
\end{align}
If we choose $\mathfrak{B}_n$ and $\tilde{\mathfrak{B}}_n$ such that $\mathfrak{B}_n=\vartheta^{-1}\mathfrak{B}_n$ and $\mathfrak{B}_n=\tilde{\mathfrak{B}}_n$ are satisfied, the time average on the left-hand side of Eq.~\eqref{eq:time-average_joint-probability} becomes 
\begin{align}
    \lim_{M\rightarrow\infty}\frac{1}{M}\sum_{m=n}^{n+M-1}
    P_{\rho_\mathrm{ss}}\left(\mathfrak{B}_n\cap\vartheta^{-m}\mathfrak{B}_n\right)
    =P_{\rho_\mathrm{ss}}\left(\mathfrak{B}_n\right),
    \label{eq:time-average_joint-probability_invariant}
\end{align}
where $\mathfrak{B}_n\cap\vartheta^{-m}\mathfrak{B}_n=\mathfrak{B}_n\cap\mathfrak{B}_n=\mathfrak{B}_n$ is used.
Thus, from Eqs.~\eqref{eq:time-average_joint-probability} and~\eqref{eq:time-average_joint-probability_invariant}, we can obtain 
\begin{align}
    P_{\rho_\mathrm{ss}}\left(\mathfrak{B}_n\right)
    =P_{\rho_\mathrm{ss}}^2\left(\mathfrak{B}_n\right),\ \ \forall \, \mathfrak{B}_n \ \mathrm{s.t.} \ \mathfrak{B}_n=\vartheta^{-1}\mathfrak{B}_n,
\end{align}
which means $P_{\rho_\mathrm{ss}}\left(\mathfrak{B}_n\right) \in \{ 0, 1\}$. This completes the proof of ergodicity in Eq.~\eqref{eq:ergodicity_theta}.

\subsection{Ergodicity of nonlinear quantities}
\label{sec:ergodicity_nonlinear-quantity}
Here, we consider functions of a pure state $\ketbra{\psi}$, denoted by $f(\psi)$, and their time averages. 
An important example of $f(\psi)$ is the entanglement entropy $-\mathrm{Tr}_X\left[\hat{\rho}_X(\psi)\ln\hat{\rho}_X(\psi)\right]$, where $X$ is a subsystem and $\hat{\rho}_X(\psi)$ is the reduced density matrix with respect to $X$, i.e., $\hat{\rho}_X(\psi)=\mathrm{Tr}_{\bar{X}}(\ket{\psi}\bra{\psi})$ with $\bar{X}$ being the complement of $X$. 
Indeed, the entanglement entropy is often explored in monitored quantum systems to detect measurement-induced transitions, which will be reviewed in Sec.~\ref{sec:entanglement_MIPT}. 
While the time averages of linear observables are determined through the time average of $\hat{\rho}_{\bm{b};n}$ as discussed in Sec.~\ref{sec:ergodicity_linear-observable}, this is not the case when we consider nonlinear quantities $f(\psi)$. 
The main message of this section is Eq.~\eqref{eq:average_nonlinear-quantity_any-measure}, which states that the time average of $f(\psi_{\bm{b};n})$ corresponds to the average of $f(\psi)$ over an invariant measure $\nu_\mathrm{ss}$ for pure states, as will be explained below. 
Monitored quantum systems exhibit this correspondence when typical trajectories purify and the averaged CPTP dynamics exhibits a unique steady state, as summarized in Table~\ref{tab:ergodicity}. 
We note that the conditions for purification and irreducibility have been discussed in Secs.~\ref{sec:pur} and~\ref{sec:irreducibility}, respectively, while the positive definiteness of $\hat{\rho}_\mathrm{ss}$ is not necessary for Eq.~\eqref{eq:average_nonlinear-quantity_any-measure} to be satisfied. 

To explore the averages of $f(\psi)$, we consider a situation where quantum pure states $\{\ket{\psi}\}$ are sampled from a probability measure $\nu$. 
The average of $f(\psi)$ with respect to $\nu$ is given by
\begin{align}
    \mathsf{E}_\nu[f(\psi)]=\int f(\psi)d\nu(\psi).
\end{align}
If we take $f(\psi)=\ket{\psi}\bra{\psi}$, the average becomes the density matrix\footnote{
For example, if pure states $\{\ket{\phi_i}\}$ are sampled with probabilities $\{p_i\}$, the measure is $d\nu(\psi)=\sum_ip_i\delta(\phi_i-\psi)d\psi$ and the corresponding density matrix becomes $\hat{\rho}_\nu=\sum_ip_i\ket{\phi_i}\bra{\phi_i}$.
},
\begin{align}\hat{\rho}_\nu=\mathsf{E}_\nu\left[\ket{\psi}{\bra{\psi}}\right]
    =\int\ket{\psi}\bra{\psi}d\nu(\psi).
    \label{eq:density-matrix_averaged_measure}
\end{align}
We note that there is no one-to-one correspondence between $\hat{\rho}_\nu$ and $\nu$, that is, several measures can lead to the same density matrix; for example, $d\nu(\psi)=\frac{1}{d}\sum_{i=1}^d\delta(\psi-\phi_i)d\psi$ for any orthonormal set $\{\ket{\phi_i}\}$ satisfying $\langle\phi_i|\phi_j\rangle=\delta_{ij}$ leads to the maximally mixed state $\hat{\rho}_\nu=\mathbb{\hat{I}}/d$. 

When a quantum measurement described by $\{\hat{M}_b\}$ is performed, $\nu$ is altered. 
Starting from an initial measure $\nu_0$, the probability measure after $n$ measurements becomes
\begin{align}
    \nu_n(S)=\int\sum_{\bm{b}_n}\bra{\psi_0}\hat{\mathsf{M}}_{\bm{b};n}^\dagger \hat{\mathsf{M}}_{\bm{b};n}\ket{\psi_0}\chi_S(\psi_{\bm{b};n})d\nu_0(\psi_0),
    \label{eq:nu-n}
\end{align}
where $\ket{\psi_{\bm{b};n}}=\hat{\mathsf{M}}_{\bm{b};n}\ket{\psi_0}/\sqrt{\bra{\psi_0}\hat{\mathsf{M}}_{\bm{b};n}^\dagger\hat{\mathsf{M}}_{\bm{b};n}\ket{\psi_0}}$ as defined in Eq.~\eqref{qtpure}, $S$ is a set of pure states, and $\chi_S(\psi)=1\,(0)$ if $\ket{\psi} \in S\,(\notin S)$. 
The average of a function $f(\psi)$ with respect to $\nu_n$ becomes  
\begin{align}\label{eq:average_nu-n}
    \mathsf{E}_{\nu_n}[f(\psi)]=\int\sum_{\bm{b}_n}\bra{\psi_0}\hat{\mathsf{M}}_{\bm{b};n}^\dagger \hat{\mathsf{M}}_{\bm{b};n}\ket{\psi_0}f(\psi_{\bm{b};n})d\nu_0(\psi_0).
\end{align}
In Eqs.~\eqref{eq:nu-n} and~\eqref{eq:average_nu-n}, the states $\{\ket{\psi_0}\}$ sampled with the initial measure $\nu_0$ are transformed to $\{\ket{\psi_{\bm{b};n}}\}_{\bm{b}_n}$, and they are averaged over all possible $\{\bm{b}_n\}$ with weights $\{\bra{\psi_0}\hat{\mathsf{M}}_{\bm{b};n}^\dagger \hat{\mathsf{M}}_{\bm{b};n}\ket{\psi_0}\}_{\bm{b}_n}$. 
For the CPTP dynamics $\hat{\rho}_{n+1}=\sum_b\hat{M}_b\hat{\rho}_{n}\hat{M}_b^\dagger$ as in Eq.~\eqref{eq:KrausRep} starting from the initial state $\hat{\rho}_0=\hat{\rho}_{\nu_0}$, the density matrix at step $n$ becomes
\begin{align}
\hat{\rho}_{n}=\hat{\rho}_{\nu_n},
\label{eq:rho_measure}
\end{align}
which originates from Eq.~\eqref{eq:density-matrix_averaged_measure}. 
If $\nu_n(S)=\nu_0(S)$ is satisfied for any $n$ and $S$, such a measure $\nu_0$ is referred to as an invariant measure. 
It was shown that there exists at least one invariant measure $\nu_\mathrm{ss}$ in Refs.~\cite{benoist2019invariant, benoist2023limit}. 
We note that the invariant measure $\nu_\mathrm{ss}(S)$ for pure states discussed in this section is different from the invariant measure $P_{\rho_\mathrm{ss}}(\mathfrak{B}_n)$ for measurement outcomes discussed in Sec.~\ref{sec:invariant-measure_ergodicity_outcomes}.

In Ref.~\cite{benoist2019invariant}, it was also shown that there exists a \textit{unique} invariant measure $\nu_\mathrm{ss}$, which satisfies
\begin{align}
    \lim_{n\rightarrow\infty}\nu_n=\nu_\mathrm{ss}
    \label{eq:invariant-measure_pure-state}
\end{align}
for any initial measure $\nu_0$, when typical trajectories exhibit purification and the corresponding CPTP map has a unique stationary state $\hat{\rho}_\mathrm{ss}$. 
The unique invariant measure $\nu_\mathrm{ss}$ corresponds to the unique stationary state, $\hat{\rho}_\mathrm{ss}=\hat{\rho}_{\nu_\mathrm{ss}}$, through Eq.~\eqref{eq:rho_measure}.
In such a situation, the time average of $f(\psi)$ corresponds to the average over the invariant measure $\nu_\mathrm{ss}$,
\begin{align}
    \overline{f(\psi_{\bm{b};n})}=\mathsf{E}_{\nu_\mathrm{ss}}[f(\psi)]\ \mathrm{for\ any\ initial\ measure}\ \nu_0,
    \label{eq:average_nonlinear-quantity_any-measure}
\end{align}
almost surely, which was obtained in Ref.~\cite{benoist2023limit}. 
Here, the overline denotes the time average defined in Eq.~\eqref{eq:time-average}.
We note that $f(\psi)$ is assumed to be a continuous function in the following, which is technically important but not detailed here.

\subsubsection{Uniqueness of the invariant measure}
We now present an outline of the proof of the uniqueness and convergence of the invariant measure, which is given in Eq.~\eqref{eq:invariant-measure_pure-state}, based on the uniqueness of the stationary state $\hat{\rho}_\mathrm{ss}$ and the purification discussed in Sec.~\ref{sec:pur}.
Provided that there exists at least one invariant measure, our goal is to show that the existence of more than one invariant measure is inconsistent with a unique $\hat{\rho}_\mathrm{ss}$ if typical trajectories purify. 

If purification occurs in a typical trajectory $\bm{b}$, the matrix rank of $\hat{\mathsf{M}}_{\bm{b};n}$ in Eq.~\eqref{eq:product_Kraus-operators} typically becomes $1$ as $n\rightarrow\infty$. 
In this case, $\hat{\mathsf{M}}_{\bm{b};n}$ can be approximated as 
\begin{align}
    \hat{\mathsf{M}}_{\bm{b};n}\propto\ket{\Psi_{1,\bm{b};n}}\bra{\Phi_{1,\bm{b};n}},
    \label{eq:evolution-operator_rank-1}
\end{align}
for large $n$ (see also Sec.~\ref{sec:rank-M_purification}). 
Here, $\ket{\Psi_{1,\bm{b};n}}$ and $\ket{\Phi_{1,\bm{b};n}}$ are the eigenstates of $\hat{\mathsf{M}}_{\bm{b};n}\hat{\mathsf{M}}_{\bm{b};n}^\dagger$ and $\hat{\mathsf{M}}_{\bm{b};n}^\dagger\hat{\mathsf{M}}_{\bm{b};n}$, respectively, corresponding to the largest eigenvalue. 
We note that $\ket{\Psi_{1,\bm{b};n}}$ in Eq.~\eqref{eq:evolution-operator_rank-1} corresponds to the ground state of an effective Hamiltonian describing the quantum trajectory dynamics, which will be introduced in Sec.~\ref{sec:Lyapunov-analysis}.
Thus, quantum trajectories of pure states $\ket{\psi_{\bm{b};n}}$ approach $\ket{\Psi_{1,\bm{b};n}}$, i.e.,
\begin{align}
    \ket{\psi_{\bm{b};n}}\simeq\ket{\Psi_{1,\bm{b};n}}
\end{align}
is satisfied for sufficiently large $n$.

Then, if we consider $f(\psi_{\bm{b};n})$ averaged over the initial measure $\nu_0$, its long-time limit becomes
\begin{align}
    \lim_{n\rightarrow\infty}
    \mathsf{E}_{\nu_n}\left[f(\psi)\right]&
    =\lim_{n\rightarrow\infty}\sum_{\bm{b}_n}\int\bra{\psi_0}\hat{\mathsf{M}}_{\bm{b};n}^\dagger\hat{\mathsf{M}}_{\bm{b};n}\ket{\psi_0}f(\psi_{\bm{b};n})d\nu_0(\psi_0)\nonumber\\
    &=\lim_{n\rightarrow\infty}\sum_{\bm{b}_n}\int\bra{\psi_0}\hat{\mathsf{M}}_{\bm{b};n}^\dagger\hat{\mathsf{M}}_{\bm{b};n}\ket{\psi_0}d\nu_0(\psi_0)f(\Psi_{1,\bm{b};n})\nonumber\\
    &=\lim_{n\rightarrow\infty}\sum_{\bm{b}_n}\mathrm{Tr}\left(\hat{\mathsf{M}}_{\bm{b};n}\int d\nu_0(\psi_0)\ket{\psi_0}{\bra{\psi_0}}\hat{\mathsf{M}}_{\bm{b};n}^\dagger\right)f(\Psi_{1,\bm{b};n})\nonumber\\
    &=\lim_{n\rightarrow\infty}\mathbb{E}_{\rho_0}\left[f(\Psi_{1,\bm{b};n})\right],
    \label{eq:average_actual-trajectory_ground-state}
\end{align}
where $\hat{\rho}_0=\int d\nu_0(\psi)\ket{\psi}\bra{\psi}$ is the initial state corresponding to the initial measure $\nu_0$. 
In the second and third lines of Eq.~\eqref{eq:average_actual-trajectory_ground-state}, we use the fact that $\ket{\Psi_{1,\bm{b};n}}$ is determined only from $\hat{\mathsf{M}}_{\bm{b};n}$ and is thus independent of the initial pure state $\ket{\psi_0}$. 
In the last line, $\mathbb{E}_{\rho_0}$ represents the average over $\{\bm{b}\}$ in the situation where the initial state is $\hat{\rho}_0$, in the same way as $\mathbb{E}$ in other chapters. 
The initial-state dependence is explicitly written in this chapter since it becomes important which $\hat{\rho}_0$ is chosen, when we discuss the ergodicity of nonlinear quantities and the typical convergence of the Lyapunov spectrum. 

Since the existence of at least one invariant measure has been proven in Refs.~\cite{benoist2019invariant, benoist2023limit}, we take the initial measure $\nu_0$ as an invariant measure $\nu_\mathrm{ss}$. 
Then, Eq.~\eqref{eq:average_actual-trajectory_ground-state} leads to
\begin{align}
    \mathsf{E}_{\nu_\mathrm{ss}}\left[f(\psi)\right]
    =\lim_{n\rightarrow\infty}
    \mathbb{E}_{\rho_\mathrm{ss}}\left[f(\Psi_{1,\bm{b};n})\right],
    \label{eq:average_nonlinear-quantity_invariant-measure_limit}
\end{align}
where $\hat{\rho}_\mathrm{ss}$ is the unique stationary state of the corresponding CPTP dynamics. 
To show the uniqueness of the invariant measure $\nu_\mathrm{ss}$, we suppose that there exist two distinct invariant measures $\nu_\mathrm{ss}^a$ and $\nu_\mathrm{ss}^b$.
From Eq.~\eqref{eq:average_nonlinear-quantity_invariant-measure_limit}, both measures lead to trajectory averages starting from the same unique stationary state $\hat{\rho}_\mathrm{ss}$; hence, these measures satisfy 
\begin{align}
    \mathsf{E}_{\nu_\mathrm{ss}^a}\left[f(\psi)\right]
    =\mathsf{E}_{\nu_\mathrm{ss}^b}\left[f(\psi)\right].
    \label{eq:equivalence_several-measures}
\end{align}
Since $f(\psi)$ is an arbitrary nonlinear function, Eq.~\eqref{eq:equivalence_several-measures} implies $\nu_\mathrm{ss}^a=\nu_\mathrm{ss}^b$, which contradicts the assumption. 
This means that the invariant measure is unique and proves Eq.~\eqref{eq:invariant-measure_pure-state}.

\subsubsection{Coincidence between the time average and ensemble average}
\label{sec:time-average_ensemble-average}
We give an outline of the proof of Eq.~\eqref{eq:average_nonlinear-quantity_any-measure}, assuming the uniqueness of the steady state of the CPTP dynamics and purification in typical trajectories. 
To this end, we first apply Birkhoff's ergodic theorem~\cite{walters2000introduction} to $\overline{f(\Psi_{1,\bm{b};n})}$. 
Second, we consider the dynamics of nonlinear functions and introduce the notion of harmonic functions. 
Combining these, we can see that Eq.~\eqref{eq:average_nonlinear-quantity_any-measure} is satisfied almost surely.

We notice from Eq.~\eqref{eq:evolution-operator_rank-1} that 
\begin{align}
    \overline{f(\psi_{\bm{b};n})}=\overline{f(\Psi_{1,\bm{b};n})}
    \label{eq:time-average_nonlinear-quantity_ground-state}
\end{align}
is satisfied when the trajectory $\bm{b}$ purifies, which is assumed to be typically satisfied. 
Here, the left-hand side is well defined only when the initial state is a pure state. 
On the other hand, the initial state can also be a mixed state when we compute the right-hand side of Eq.~\eqref{eq:time-average_nonlinear-quantity_ground-state} since $\ket{\Psi_{1,\bm{b};n}}$ is independent of $\hat{\rho}_0$. 
To evaluate $\overline{f(\Psi_{1,\bm{b};n})}$, we again consider typical trajectories when the initial state is the unique steady state $\hat{\rho}_\mathrm{ss}$ in the averaged CPTP dynamics. 
In this case, the invariant measure for sequences $\bm{b}$ becomes ergodic, as shown in Sec.~\ref{sec:invariant-measure_ergodicity_outcomes}. 
On the basis of the invariant and ergodic measure of $\bm{b}$, Birkhoff's ergodic theorem ensures
\begin{align}
    \overline{f(\Psi_{1,\bm{b};n})}=C_f\ \mathrm{if\ the\ initial\ state\ is}\ \hat{\rho}_\mathrm{ss}
    \label{eq:average_nonlinear-quantity_invariant-measure}
\end{align}
almost surely, where $C_f$ is a constant independent of $\bm{b}$. 
Details of Birkhoff's ergodic theorem are given in Appendix~\ref{app:Birkhoff-theorem}. 
Taking the average of Eq.~\eqref{eq:average_nonlinear-quantity_invariant-measure} over all possible $\bm{b}$, the constant $C_f$ becomes
\begin{align}
    C_f=\mathbb{E}_{\rho_\mathrm{ss}}(C_f)
    =\mathsf{E}_{\nu_\mathrm{ss}}[f(\psi)],
    \label{eq:average_nonlinear-function_constant}
\end{align}
which can be understood from Eq.~\eqref{eq:average_nonlinear-quantity_invariant-measure_limit}. 

To show Eq.~\eqref{eq:average_nonlinear-quantity_any-measure} on the basis of Eq.~\eqref{eq:average_nonlinear-quantity_invariant-measure}, we consider a continuous function $g(\psi)$ and its evolution governed by 
\begin{align}
    g_n(\psi_0)=\sum_{\bm{b}_n} \bra{\psi_0}\hat{\mathsf{M}}_{\bm{b};n}^\dagger\hat{\mathsf{M}}_{\bm{b};n}\ket{\psi_0}g(\psi_{\bm{b};n}).
    \label{eq:dynamics_function}
\end{align}
If $g(\psi_0)$ is not changed by the trajectory dynamics, i.e., 
\begin{align}
    g_n(\psi_0)=g(\psi_0),
    \label{eq:harmonic-function}
\end{align}
the function $g(\psi_0)$ is referred to as a harmonic function. 
Equations~\eqref{eq:nu-n}, \eqref{eq:average_nu-n}, \eqref{eq:dynamics_function}, and~\eqref{eq:harmonic-function} tell us that any harmonic function satisfies 
\begin{align}
    g(\psi_0)=g_n(\psi_0)
    &=\sum_{\bm{b}_n}\bra{\psi_0}\hat{\mathsf{M}}_{\bm{b};n}^\dagger \hat{\mathsf{M}}_{\bm{b};n}\ket{\psi_0}g(\psi_{\bm{b};n})\nonumber\\
    &=\int\sum_{\bm{b}_n}\bra{\phi}\hat{\mathsf{M}}_{\bm{b};n}^\dagger \hat{\mathsf{M}}_{\bm{b};n}\ket{\phi}g(\psi_{\bm{b};n}) d\nu^{\psi_0}(\phi)=\mathsf{E}_{\nu_n^{\psi_0}}\left[g(\psi)\right],
    \label{eq:sample-average_harmonic-function}
\end{align}
where $\nu^{\psi_0}(S)=\int_{\phi \in S}\delta(\psi_0-\phi)d\phi$ is the measure of the initial pure state $\ket{\psi_0}$.

We can show that harmonic and continuous functions are constants, i.e., such functions are independent of $\ket{\psi_0}$. 
To this end, we consider a probability measure averaged over $n$-step dynamics,
\begin{align}
    \tilde{\nu}_n^{\psi_0}=\frac{1}{n}\sum_{m=0}^{n-1}\nu_m^{\psi_0}.
    \label{eq:time-average_measure}
\end{align}
From Eqs.~\eqref{eq:sample-average_harmonic-function} and~\eqref{eq:time-average_measure}, we notice that any harmonic function $g(\psi_0)$ can be written as 
\begin{align}
    g(\psi_0)=\frac{g(\psi_0)+g_1(\psi_0)+\cdots +g_{n-1}(\psi_0)}{n}=\mathsf{E}_{\tilde{\nu}_n^{\psi_0}}[g(\psi)]
    \label{eq:time-average_harmonic-function}
\end{align}
for arbitrary $n$. 
In addition, since there exists a unique invariant measure $\nu_\mathrm{ss}$, $\tilde{\nu}_n^{\psi_0}$ satisfies
\begin{align}
    \lim_{n\rightarrow\infty}\tilde{\nu}_n^{\psi_0}=\nu_\mathrm{ss}.
    \label{eq:limit_measure}
\end{align}
Equations~\eqref{eq:time-average_harmonic-function} and~\eqref{eq:limit_measure} lead to
\begin{align}
    g(\psi_0)=\mathsf{E}_{\nu_\mathrm{ss}}[g(\psi)],
    \label{eq:time-average_harmonic-function_invariant-measure}
\end{align}
which means that the harmonic and continuous function $g(\psi_0)$ does not depend on $\ket{\psi_0}$. 

Using Eq.~\eqref{eq:time-average_harmonic-function_invariant-measure}, we wish to show Eq.~\eqref{eq:average_nonlinear-quantity_any-measure}. 
To this end, we consider a function defined as\footnote{
The function $G(\psi_0)$ is continuous~\cite{benoist2023limit}, which is technically important but not detailed here.
}
\begin{align}
    G(\psi_0)&=P_{\psi_0}\left(\left\{\bm{b}
    :\lim_{N\rightarrow\infty}
    \frac{1}{N}\sum_{n=0}^{N-1}
    f(\psi_{\bm{b};n})=\mathsf{E}_{\nu_\mathrm{ss}}[f(\psi)]\right\}\right)\nonumber\\
    &=P_{\psi_0}\left(\left\{\bm{b}
    :\lim_{N\rightarrow\infty}
    \frac{1}{N}\sum_{n=0}^{N-1}
    f(\Psi_{1,\bm{b};n})=\mathsf{E}_{\nu_\mathrm{ss}}[f(\psi)]\right\}\right),
    \label{eq:G-psi}
\end{align}
which is the probability of the set of quantum trajectories where $\lim_{N\rightarrow\infty}\frac{1}{N}\sum_{n=0}^{N-1}f(\psi_{\bm{b};n})=\mathsf{E}_{\nu_\mathrm{ss}}[f(\psi)]$ is satisfied with the initial state $\ket{\psi_0}$.
Here, Eqs.~\eqref{eq:average_nonlinear-quantity_invariant-measure} and~\eqref{eq:average_nonlinear-function_constant} mean that $\overline{f(\Psi_{1,\bm{b};n})}=\mathsf{E}_{\nu_\mathrm{ss}}[f(\psi)]$ is satisfied almost surely if we take the initial state as the unique stationary state $\hat{\rho}_\mathrm{ss}=\int\ket{\psi_0}\bra{\psi_0}d\nu_\mathrm{ss}(\psi_0)$. 
Therefore, integrating $G(\psi_0)$ with respect to the unique invariant measure $\nu_\mathrm{ss}$, we can obtain
\begin{align}
    \int G(\psi_0)d\nu_\mathrm{ss}(\psi_0)=P_{\rho_\mathrm{ss}}\left(\left\{\bm{b}
    :\overline{f(\Psi_{1,\bm{b};n})}=\mathsf{E}_{\nu_\mathrm{ss}}[f(\psi)]\right\}\right)=1,
    \label{eq:integral_G-psi}
\end{align}
where $P_\rho(\mathfrak{B})$ is defined in Eq.~\eqref{eq:born-rule}.
This is because $\ket{\Psi_{1,\bm{b};n}}$ is independent of $\ket{\psi_0}$ and thus we can evaluate the probability in the second line of Eq.~\eqref{eq:G-psi} through the measure $P_{\rho_\mathrm{ss}}(\mathfrak{B})$ for outcomes, in the same way as in Eq.~\eqref{eq:average_actual-trajectory_ground-state}. 

We can also show that the function $G(\psi_0)$ becomes a harmonic function as
\begin{align}
    G_1(\psi_0)&=\sum_b
    \bra{\psi_0}\hat{M}_b^\dagger\hat{M}_b\ket{\psi_0}
    P_{\psi_b}\left(\left\{\bm{b}:
    \lim_{N\rightarrow\infty}\frac{1}{N}
    \sum_{n=0}^{N-1}f(\Psi_{1,\bm{b};n})
    =\mathsf{E}_{\nu_\mathrm{ss}}[f(\psi)]
    \right\}\right)\nonumber\\
    &=\sum_b
    \bra{\psi_0}\hat{M}_b^\dagger\hat{M}_b\ket{\psi_0}
    P_{\psi_0}\left(\left\{\bm{b}:
    \lim_{N\rightarrow\infty}\frac{1}{N}
    \sum_{n=0}^{N-1}f(\Psi_{1,\bm{b};{n+1}})
    =\mathsf{E}_{\nu_\mathrm{ss}}[f(\psi)]\middle|b_1=b\right\}\right)\nonumber\\
    &=P_{\psi_0}\left(\left\{\bm{b}:
    \lim_{N\rightarrow\infty}\frac{N+1}{N}\frac{1}{N+1}
    \sum_{n=0}^{N}f(\Psi_{1,\bm{b};n})-\frac{f(\psi_0)}{N}
    =\mathsf{E}_{\nu_\mathrm{ss}}[f(\psi)]\right\}\right)\nonumber\\
    &=G(\psi_0),
\end{align}
where $\ket{\psi_b}=\hat{M}_b\ket{\psi_0}/\sqrt{\bra{\psi_0}\hat{M}_b^\dagger\hat{M}_b\ket{\psi_0}}$. 
Therefore, the harmonic function $G(\psi_0)$ is constant, as shown in Eq.~\eqref{eq:time-average_harmonic-function_invariant-measure}. 
In addition, Eq.~\eqref{eq:integral_G-psi} tells us that the constant value becomes
\begin{align}
    G(\psi_0)=1.
\end{align}
This means that Eq.~\eqref{eq:average_nonlinear-quantity_any-measure} is satisfied almost surely.

\subsection{Lyapunov spectrum of quantum trajectories}
\label{sec:Lyapunov-analysis}
The Lyapunov spectral analysis has been used to study quantum trajectories in monitored systems. 
Indeed, measurement-induced phase transitions, their critical properties, purification timescales, and topological physics have been explored through Lyapunov spectral analysis~\cite{zabalo2022operator, kumar2024boundary, aziz2024critical, chakraborty24charge, bulchandani2024random, mochizuki2025measurement, xiao2024topology, oshima2025topology, mochizuki2025transitions}. 
These studies will be reviewed in Chapter~\ref{sec:mipt}. 
Here, we explain the theoretical aspects and numerical procedures of Lyapunov analysis in quantum systems exposed to indirect measurements. 

In Lyapunov spectral analysis, we consider the effective Hamiltonian 
\begin{align}
    \hat{H}_{\bm{b};n}=-\frac{1}{2n}
    \ln\left(\hat{\mathsf{M}}_{\bm{b};n}\hat{\mathsf{M}}_{\bm{b};n}^\dagger\right),
    \label{eq:effective-Hamiltonian}
\end{align}
where $\hat{\mathsf{M}}_{\bm{b};n}=\hat{M}_{b_n}\hat{M}_{b_{n-1}}\cdots\hat{M}_{b_1}$. 
In a quantum trajectory labeled by $\bm{b}$, the Lyapunov exponents $\{\varepsilon_{i,\bm{b};n}\}_i$ are defined as the eigenvalues of the effective Hamiltonian, i.e., 
\begin{align}
    \hat{H}_{\bm{b};n}\ket{\Psi_{i,\bm{b};n}}
    =\varepsilon_{i,\bm{b};n}\ket{\Psi_{i,\bm{b};n}}.
    \label{eq:Lyapunov-exponent}
\end{align}
In the following, we order the Lyapunov exponents as $\varepsilon_{i,\bm{b};n}\leq\varepsilon_{i+1,\bm{b};n}$. 

The main result is that, if the corresponding CPTP dynamics is irreducible, the Lyapunov exponents converge to values independent of $\bm{b}$ in the long-time regime~\cite{benoist2019invariant, arnold1995random},
\begin{align}
    \varepsilon_i=\lim_{n\rightarrow\infty}\varepsilon_{i,\bm{b};n},
    \label{eq:typical-convergence_Lyapunov-exponent}
\end{align}
almost surely for any initial state. 
In other words, the set of trajectories $\tilde{\mathfrak{B}}_\infty$ such that the Lyapunov exponents depend on the outcomes exhibits zero measure, i.e., $P_{\rho_0}(\tilde{\mathfrak{B}}_\infty)=0$ for any initial state $\hat{\rho}_0$. 
Table~\ref{tab:ergodicity} summarizes the sufficient condition for Eq.~\eqref{eq:typical-convergence_Lyapunov-exponent} to be satisfied. 
Discussions on irreducibility, or equivalently on the unique positive-definite steady state in the averaged CPTP dynamics, were given in Sec.~\ref{sec:irreducibility}. 
If Eq.~\eqref{eq:typical-convergence_Lyapunov-exponent} is satisfied, we can characterize some typical features of quantum trajectories, e.g., the purification timescale, through the Lyapunov spectrum.

\subsubsection{Typical convergence of the Lyapunov spectrum owing to the irreducibility}
\label{sec:typical-convergence_Lyapunov-spectrum}
To show Eq.~\eqref{eq:typical-convergence_Lyapunov-exponent}, we use Kingman's subadditive ergodic theorem, on the basis of the invariant measure and the ergodicity discussed in Sec.~\ref{sec:invariant-measure_ergodicity_outcomes}. 
The theorem is applicable to a sequence of functions $\{f_n(\bm{b})\}_{n=1,2,\ldots}$ that satisfies
\begin{align}
    f_{n+m}(\bm{b}) \leq f_n(\bm{b})+f_m(\vartheta^n\bm{b}).
    \label{eq:subadditive-function}
\end{align}
As detailed in Appendix~\ref{app:proof_Kingman-theorem}, Kingman's subadditive ergodic theorem states the following: if there exists an invariant measure for $\mathfrak{B}\subseteq\{\bm{b}\}$ that is ergodic and if $\{f_n(\bm{b})\}_n$ satisfies Eq.~\eqref{eq:subadditive-function} for any $\bm{b}$, then the subadditive function $f_n(\bm{b})$ divided by $n$ converges to an asymptotic value independent of measurement outcomes, i.e.,
\begin{align}
    \gamma=\lim_{n\rightarrow\infty}
    \frac{f_n(\bm{b})}{n},
    \label{eq:Kingman-theorem}
\end{align}
almost surely with respect to the invariant measure. 
We note that the existence of an invariant and ergodic measure for $\mathfrak{B}$ is ensured if we take the initial state as the unique positive-definite steady state $\hat{\rho}_\mathrm{ss}$ of the corresponding CPTP dynamics that is assumed to be irreducible.

To apply Kingman's subadditive ergodic theorem, we need to identify an appropriate function satisfying the subadditive property.
To this end, we consider the exterior powers of vectors and operators. 
Given a finite-dimensional Hilbert space $\mathcal{H}$, vectors in the space $\land^k\mathcal{H}$ can be written as $\sum_ic_i\left(\ket{\psi_i^1}\land\ket{\psi_i^2}\land\cdots\land\ket{\psi_i^k}\right)$, where $\{c_i\}_i$ are complex numbers and $\{\ket{\psi_i^j}\}_{i,j}$ are vectors in $\mathcal{H}$. 
The elements $\{\ket{\psi_i^1}\land\ket{\psi_i^2}\land\cdots\land\ket{\psi_i^k}\}_i$ satisfy
\begin{align}
    \ket{\psi^{j-1}}\land(\ket{\psi^j}+\ket{\tilde{\psi}^j})\land\ket{\psi^{j+1}}
    &=\ket{\psi^{j-1}}\land\ket{\psi^j}\land\ket{\psi^{j+1}}
    +\ket{\psi^{j-1}}\land\ket{\tilde{\psi}^j}\land\ket{\psi^{j+1}},\\
    \ket{\psi^1}\land\cdots\land c\ket{\psi^j}\land\cdots\land\ket{\psi^k}
    &=c\ket{\psi^1}\land\cdots\land\ket{\psi^j}\land\cdots\land\ket{\psi^k},\\
    \ket{\psi^{\pi(1)}}\land\ket{\psi^{\pi(2)}}\land\cdots\land\ket{\psi^{\pi(k)}}
    &=\mathrm{sgn}(\pi)\ket{\psi^1}\land\ket{\psi^2}\land\cdots\land\ket{\psi^k},
\end{align}
where $\pi$ is a permutation of $(1,2,\ldots,k)$. 
The inner product of two vectors in $\land^k\mathcal{H}$ is given by
\begin{align}
    \left<\ket{\phi^1}\land\cdots\land\ket{\phi^k},\ket{\psi^1}\land\cdots\land\ket{\psi^k}\right>
    =\det(\langle\phi^i|\psi^j\rangle).
\end{align}
An operator $\land^k\hat{A}$ acting on $\land^k\mathcal{H}$ is defined as
\begin{align}
    (\land^k\hat{A})(\ket{\psi^1}\land\cdots\land\ket{\psi^k})
    =(\hat{A}\ket{\psi^1})\land\cdots\land (\hat{A}\ket{\psi^k}),
\end{align}
where $\hat{A}$ is an operator in $\mathbb{B}[\mathcal{H}]$. 
The norm of such an operator is expressed as
\begin{align}
    \left\|\land^k\hat{A}\right\|=\prod_{i=1}^k\Lambda_i,
    \label{eq:norm_exterior-power}
\end{align}
where $\{\Lambda_i\}$ are the singular values of $\hat{A}$ ordered as $\Lambda_i\geq\Lambda_{i+1}$. 
Thus, the norm satisfies the inequality
\begin{align}
    \left\|\land^k(\hat{A}\hat{B})\right\|\leq
    \left\|\land^k\hat{A}\right\|\left\|\land^k\hat{B}\right\|,
    \label{eq:norm-inequality_exterior-power}
\end{align}
which results from the log-majorization of singular values (see, e.g., Corollary~4.3.5 in Ref.~\cite{hiai2010matrix}). 

Here, we focus on the logarithm of the norm for the exterior power of $\hat{\mathsf{M}}_{\bm{b};n}$,
\begin{align}
    f_{n,k}(\bm{b})=\ln\left(\left\|\land^k\hat{\mathsf{M}}_{\bm{b};n}\right\|\right).
    \label{eq:exterior-power_M}
\end{align}
Owing to Eq.~\eqref{eq:norm-inequality_exterior-power}, $f_{n,k}(\bm{b})$ is a subadditive function satisfying
\begin{align}
    f_{n+m,k}(\bm{b}) \leq f_{n,k}(\bm{b})+f_{m,k}(\vartheta^n\bm{b})
    \label{eq:f_subadditive}
\end{align}
for any $k$. 
Then, we can apply Kingman's subadditive ergodic theorem to $f_{n,k}(\bm{b})$ if we take the initial state as the unique positive-definite steady state $\hat{\rho}_\mathrm{ss}$; consequently, there is a limit independent of measurement outcomes,
\begin{align}
    \gamma_k=\lim_{n\rightarrow\infty}
    \frac{f_{n,k}(\bm{b})}{n},
    \label{eq:gamma-k}
\end{align}
almost surely with respect to the invariant measure $P_{\rho_\mathrm{ss}}(\mathfrak{B})$. 
From Eqs.~\eqref{eq:effective-Hamiltonian}, \eqref{eq:Lyapunov-exponent}, and~\eqref{eq:norm_exterior-power}, we can understand that $\gamma_k$ is the sum of the Lyapunov exponents,
\begin{align}
    \gamma_k=-\sum_{j=1}^k\varepsilon_j.
    \label{eq:gamma-k_Lyapunov-exponenets}
\end{align}
Thus, when we take $\hat{\rho}_\mathrm{ss}$ as the initial state, the Lyapunov spectrum typically becomes independent of the sequence of measurement outcomes $\bm{b}$.

Furthermore, based on the positive definiteness of $\hat{\rho}_\mathrm{ss}$, we can also show that the Lyapunov spectrum typically satisfies Eq.~\eqref{eq:typical-convergence_Lyapunov-exponent} for any initial state $\hat{\rho}_0$~\cite{benoist2019invariant}.
This is because the positive definite $\hat{\rho}_\mathrm{ss}$ ensures absolute continuity; any set of outcomes $\mathfrak{B}_n$ that satisfies $P_{\rho_\mathrm{ss}}(\mathfrak{B}_n)= 0$ as $n\rightarrow\infty$ also satisfies $P_{\rho_0}(\mathfrak{B}_n)= 0$ for any $\hat{\rho}_0$\footnote{
To see this, we first note 
$
P_{\rho_\mathrm{ss}}(\mathfrak{B}_n)=\sum_{\bm{b}_n \in \mathsf{B}_n}\mathrm{Tr}\left(\sqrt{\hat{\rho}_\mathrm{ss}}\hat{\mathsf{M}}_{{\bm{b};n}}^\dagger\hat{\mathsf{M}}_{{\bm{b};n}}\sqrt{\hat{\rho}_\mathrm{ss}}\right).
$
This is lower bounded by $\lambda_\mathrm{min}(\hat{\rho}_\mathrm{ss})\sum_{\bm{b}_n \in \mathsf{B}_n}\mathrm{Tr}\left(\hat{\mathsf{M}}_{{\bm{b};n}}^\dagger\hat{\mathsf{M}}_{{\bm{b};n}}\right)$, where $\lambda_\mathrm{min}(\hat{\rho}_\mathrm{ss})$ is the minimum eigenvalue of $\hat{\rho}_\mathrm{ss}$. 
Note that, since $\hat{\rho}_\mathrm{ss}$ is positive definite, $\lambda_\mathrm{min}(\hat{\rho}_\mathrm{ss})>0$.
Now, for a general initial state $\hat{\rho}_0$ we have
$
P_{\rho_0}(\mathfrak{B}_n)=\sum_{\bm{b}_n \in \mathsf{B}_n}\mathrm{Tr}\left(\hat{\rho}_0\hat{\mathsf{M}}_{{\bm{b};n}}^\dagger\hat{\mathsf{M}}_{{\bm{b};n}}\right
)\leq \sum_{\bm{b}_n \in \mathsf{B}_n}\mathrm{Tr}\left(\hat{\mathsf{M}}_{{\bm{b};n}}^\dagger\hat{\mathsf{M}}_{{\bm{b};n}}\right)
$. 
Therefore, we have
\begin{align}\label{eq:absolute}
P_{\rho_0}(\mathfrak{B}_n)\leq \frac{1}{\lambda_\mathrm{min}(\hat{\rho}_\mathrm{ss})}P_{\rho_\mathrm{ss}}(\mathfrak{B}_n).
\end{align}
}.
Consequently, atypical trajectories with the initial state $\hat{\rho}_\mathrm{ss}$, in which the Lyapunov spectrum may depend on measurement outcomes, also become atypical in trajectories starting from another initial state $\hat{\rho}_0$.

\subsubsection{Spectral gap and purification}
\label{sec:spectral_gap_purification}
The spectral gap, obtained in the Lyapunov analysis,
\begin{align}
    \Delta=\varepsilon_2-\varepsilon_1,
    \label{eq:def of Lyapunov gap}
\end{align}
is profoundly related to purification, which was discussed in Sec.~\ref{sec:pur}. 
Here, we assume that there is a unique positive-definite steady state $\hat{\rho}_\mathrm{ss}=\mathcal{E}[\hat{\rho}_\mathrm{ss}]$, and thus Eq.~\eqref{eq:typical-convergence_Lyapunov-exponent} is satisfied in typical trajectories. 
Then, the timescale for purification is determined by the inverse of the spectral gap, $1/\Delta$. 
This can be understood from the singular value decomposition of $\hat{\mathsf{M}}_{\bm{b};n}$,
\begin{align}
    \hat{\mathsf{M}}_{\bm{b};n}&=\sum_{i=1}^de^{-\varepsilon_{i,\bm{b};n}n}
    \ket{\Psi_{i,\bm{b};n}}\bra{\Phi_{i,\bm{b};n}}\nonumber\\
    &\simeq e^{-\varepsilon_1n}\sum_{i=1}^de^{-(\varepsilon_i-\varepsilon_1)n}
    \ket{\Psi_{i,\bm{b};n}}\bra{\Phi_{i,\bm{b};n}},
    \label{eq:singula-value_decomposition}
\end{align}
where $\ket{\Phi_{i,\bm{b};n}}$ is the $i$th eigenstate of $\hat{\mathsf{M}}_{\bm{b};n}^\dagger\hat{\mathsf{M}}_{\bm{b};n}$ that satisfies $\hat{\mathsf{M}}_{\bm{b};n}^\dagger\hat{\mathsf{M}}_{\bm{b};n}\ket{\Phi_{i,\bm{b};n}}=e^{-2\varepsilon_{i,\bm{b};n}n}\ket{\Phi_{i,\bm{b};n}}$. 
Here, the approximation on the right-hand side of Eq.~\eqref{eq:singula-value_decomposition} is valid in the long-time regime where $\varepsilon_{i,\bm{b};n}\simeq\varepsilon_i$ is satisfied. 
When $\varepsilon_1$ is nondegenerate and $n\gg1/\Delta$ is satisfied, terms with $i\geq2$ in the sum can be neglected, and thus we can approximate $\hat{\mathsf{M}}_{\bm{b};n}$ as a rank-$1$ matrix. 
Within this approximation, the matrix rank of the density operator $\hat{\rho}_{\bm{b};n}=\hat{\mathsf{M}}_{\bm{b};n}\hat{\rho}_0\hat{\mathsf{M}}_{\bm{b};n}^\dagger/\mathrm{Tr}\left(\hat{\mathsf{M}}_{\bm{b};n}\hat{\rho}_0\hat{\mathsf{M}}_{\bm{b};n}^\dagger\right)$ also becomes $1$, and the pure state $\ket{\Psi_{1,\bm{b};n}}\bra{\Psi_{1,\bm{b};n}}$ is approximately realized.

It can also be rigorously shown that, if the condition for the purification of typical trajectories is satisfied, there is always a nonzero spectral gap,
\begin{align}
    \Delta>0,
    \label{eq:nonzero-gap}
\end{align}
in finite-dimensional monitored quantum systems~\cite{benoist2019invariant}. 
To see this, we consider the function
\begin{align}
    g_n=\mathbb{E}_{\mathbb{I}/d}\left[d\frac{\left\|\land^2\hat{\mathsf{M}}_{\bm{b};n}\right\|}
    {\mathrm{Tr}\left(\hat{\mathsf{M}}_{\bm{b};n}^\dagger\hat{\mathsf{M}}_{\bm{b};n}\right)}\right]
    =\sum_{\bm{b}_n}\left\|\land^2\hat{\mathsf{M}}_{\bm{b};n}\right\|,
\end{align}
where $\mathbb{E}_\rho[f(\bm{b}_n)]=\sum_{\bm{b}_n}f(\bm{b}_n)\mathrm{Tr}\left[\hat{\mathsf{M}}_{\bm{b};n}\hat{\rho}\hat{\mathsf{M}}_{\bm{b};n}^\dagger\right]$. 
Here, $g_n$ is a submultiplicative function due to the inequality $\left\|\land^2\hat{\mathsf{M}}_{\bm{b};n+m}\right\|\leq\left\|\land^2\hat{\mathsf{M}}_{\bm{b};n}\right\|\left\|\land^2\hat{\mathsf{M}}_{\vartheta^n\bm{b};m}\right\|$. 
Then, $\ln(g_n)$ is subadditive, and we can apply Fekete's subadditive lemma\footnote{
A subadditive function $f_{n+m}\leq f_n+f_m$ always satisfies $f_n/n \leq (qf_k)/(kq+r)+f_r/n \leq f_k/k+\max(f_0,f_1,\cdots,f_r)/n$, where $q$ is an integer, $n=kq+r$, and $0\leq r \leq k-1$. 
This leads to $\limsup_{n\rightarrow\infty}f_n/n\leq\inf_{k\geq1}f_k/k\leq\liminf_{n\rightarrow\infty}f_n/n$ and thus $\lim_{n\rightarrow\infty}f_n/n=\inf_{k\geq1}f_k/k$.
} to it: the limit of $\ln(g_n)/n$ exists and is given by $\lim_{n\rightarrow\infty}\ln(g_n)/n=\inf_{n\geq1}\ln(g_n)/n$. 

As discussed in Sec.~\ref{sec:pur}, if purification typically occurs, then
\begin{align}
    \lim_{n\rightarrow\infty}\frac{\left\|\land^2\hat{\mathsf{M}}_{\bm{b};n}\right\|}
    {\mathrm{Tr}\left(\hat{\mathsf{M}}_{\bm{b};n}^\dagger\hat{\mathsf{M}}_{\bm{b};n}\right)}=\lim_{n\rightarrow\infty}\Lambda_{1,\bm{b};n}\Lambda_{2,\bm{b};n}=0
    \label{eq:limit_gn_integrant}
\end{align}
is satisfied in a typical trajectory $\bm{b}$. 
Here, $\Lambda_{i,\bm{b};n}$ is the $i$th singular value of $\hat{\mathsf{M}}_{\bm{b};n}/\sqrt{\mathrm{Tr}\left(\hat{\mathsf{M}}_{\bm{b};n}^\dagger\hat{\mathsf{M}}_{\bm{b};n}\right)}$, ordered as $\Lambda_{i,\bm{b};n}\geq\Lambda_{i+1,\bm{b};n}$.  
From Eq.~\eqref{eq:limit_gn_integrant}, we can obtain\footnote{
Owing to $\left\|\land^2\hat{\mathsf{M}}_{\bm{b};n}\right\|\leq\left\|\hat{\mathsf{M}}_{\bm{b};n}\right\|^2$, $\left\|\land^2\hat{\mathsf{M}}_{\bm{b};n}\right\|/\mathrm{Tr}\left(\hat{\mathsf{M}}_{\bm{b};n}^\dagger\hat{\mathsf{M}}_{\bm{b};n}\right)\leq1$ is always satisfied, which allows us to apply Lebesgue's dominated convergence theorem to $\lim_{n\rightarrow\infty}g_n$ and to exchange the average and the limit.
}
\begin{align}
    \lim_{n\rightarrow\infty}g_n
    =\mathbb{E}_{\mathbb{I}/d}\left[d\lim_{n\rightarrow\infty}\frac{\left\|\land^2\hat{\mathsf{M}}_{\bm{b};n}\right\|}
    {\mathrm{Tr}\left(\hat{\mathsf{M}}_{\bm{b};n}^\dagger\hat{\mathsf{M}}_{\bm{b};n}\right)}\right]=0.
    \label{eq:limit_gn}
\end{align}
Fekete's subadditive lemma and Eq.~\eqref{eq:limit_gn} lead to
\begin{align}
    \lim_{n\rightarrow\infty}\frac{\ln(g_n)}{n}=\inf_{n\geq1}\frac{\ln(g_n)}{n}<0.
    \label{eq:limit_log-gn_negative}
\end{align}
Equation~\eqref{eq:limit_log-gn_negative} implies that there are constants $\lambda$ and $C$ that satisfy 
\begin{align}
    g_n \leq C\lambda^n,\ \ 0<\lambda<1.
    \label{eq:inequality_gn_exponential-decay}
\end{align}

To show Eq.~\eqref{eq:nonzero-gap} based on Eq.~\eqref{eq:inequality_gn_exponential-decay}, we note that the inequality $\mathrm{Tr}\left(\hat{\mathsf{M}}_{\bm{b};n}^\dagger\hat{\mathsf{M}}_{\bm{b};n}\right) \leq d\left\|\hat{\mathsf{M}}_{\bm{b};n}\right\|^2$ results in
\begin{align}
    \mathbb{E}_{\mathbb{I}/d}\left(\frac{\left\|\land^2\hat{\mathsf{M}}_{\bm{b};n}\right\|}
    {\left\|\hat{\mathsf{M}}_{\bm{b};n}\right\|^2}\right) \leq g_n.
    \label{eq:inequality_gn_norm}
\end{align}
Applying Jensen's inequality to Eq.~\eqref{eq:inequality_gn_norm}, we can obtain 
\begin{align}
    \mathbb{E}_{\mathbb{I}/d}
    \left[\frac{1}{n}\ln\left(\frac{\left\|\land^2\hat{\mathsf{M}}_{\bm{b};n}\right\|}
    {\left\|\hat{\mathsf{M}}_{\bm{b};n}\right\|^2}\right)\right]
    \leq\frac{1}{n}\ln\left[\mathbb{E}_{\mathbb{I}/d}
    \left(\frac{\left\|\land^2\hat{\mathsf{M}}_{\bm{b};n}\right\|}
    {\left\|\hat{\mathsf{M}}_{\bm{b};n}\right\|^2}\right)\right]\leq\frac{1}{n}\ln(g_n).
    \label{eq:inequality_gn_log-norm}
\end{align}
Here, from Eqs.~\eqref{eq:exterior-power_M}, \eqref{eq:gamma-k}, and~\eqref{eq:gamma-k_Lyapunov-exponenets}, we see that $\ln\left(\left\|\land^2\hat{\mathsf{M}}_{\bm{b};n}\right\|\right)/n$ and $\ln\left(\left\|\hat{\mathsf{M}}_{\bm{b};n}\right\|^2\right)/n$ converge almost surely to $-\varepsilon_1-\varepsilon_2$ and $-2\varepsilon_1$, respectively, in the long-time limit. 
Such convergence is guaranteed when the CPTP dynamics averaged over measurement outcomes is irreducible, as discussed in Sec.~\ref{sec:typical-convergence_Lyapunov-spectrum}. 
Thus, taking the limit $n\rightarrow\infty$ in Eq.~\eqref{eq:inequality_gn_log-norm} leads to
\begin{align}
    \lim_{n\rightarrow\infty}\mathbb{E}_{\mathbb{I}/d}
    \left[\frac{1}{n}
    \ln\left(\frac{\left\|\land^2\hat{\mathsf{M}}_{\bm{b};n}\right\|}
    {\left\|\hat{\mathsf{M}}_{\bm{b};n}\right\|^2}\right)\right]
    =\varepsilon_1-\varepsilon_2\leq\lim_{n\rightarrow\infty}\frac{1}{n}\ln(g_n)\leq\ln(\lambda).
    \label{eq:inequality_gap}
\end{align}
Equation~\eqref{eq:inequality_gap} with $0<\lambda<1$ means that there is always a nonzero spectral gap, i.e., Eq.~\eqref{eq:nonzero-gap} is satisfied. 

We note that $\Delta$ can be a decreasing function of $d$, and thus the spectral gap can vanish in the limit $d\rightarrow\infty$, while $\Delta$ is always nonzero in monitored systems with finite $d$ where typical trajectories purify. 
Such a vanishing gap, which causes the purification timescale to diverge, characterizes entangled phases in monitored quantum many-body systems, as will be reviewed in Secs.~\ref{sec:purification_MIPT} and~\ref{sec:Lyapunov-spectrum_MIPT}. 
We also note that the order of the average $\mathbb{E}_{\mathbb{I}/d}$ and the limit $n\rightarrow\infty$ is interchanged in Eq.~\eqref{eq:inequality_gap}.
This interchange is ensured by Lebesgue's dominated convergence theorem when we consider indirectly monitored systems where $\ln\left(\left\|\land^2\hat{\mathsf{M}}_{\bm{b};n}\right\|/\left\|\hat{\mathsf{M}}_{\bm{b};n}\right\|^2\right)/n$ does not diverge. 
More rigorous treatments applicable to quantum systems exposed to direct (projective) measurements, where $\varepsilon_{i\neq1}$ can diverge, can be found in Ref.~\cite{benoist2019invariant}.

\subsubsection{First Lyapunov exponent and probability of typical trajectories}
The first Lyapunov exponent is related to the probability of realizing typical trajectories. 
Indeed, for typical trajectories,
\begin{align}
    \lim_{n\rightarrow\infty}\left[
    \frac{1}{n}\ln\left(\left\|\hat{\mathsf{M}}_{\bm{b};n}\ket{\psi_0}\right\|\right)
    -\frac{1}{n}\ln\left(\left\|\hat{\mathsf{M}}_{\bm{b};n}\right\|\right)\right]=0
    \label{eq:typical-convergence_norm}
\end{align}
is satisfied, where $\ket{\psi_0}$ is an arbitrary initial state. 
In Eq.~\eqref{eq:typical-convergence_norm}, the first term corresponds to the decay rate of the probability that an outcome sequence $\bm{b}_n$ is realized while the second term corresponds to the first Lyapunov exponent $\varepsilon_1$. They coincide in the long-time limit. 

To show Eq.~\eqref{eq:typical-convergence_norm}, we again consider $\hat{Z}_{\bm{b};n}$ defined in Eq.~\eqref{eq:Zn}. 
Using $\hat{Z}_{\bm{b};n}$, the norm of $\hat{\mathsf{M}}_{\bm{b};n}\ket{\psi_0}$ for the initial state $\ket{\psi_0}$ can be written as
\begin{align}
    \left\|\hat{\mathsf{M}}_{\bm{b};n}\ket{\psi_0}\right\|^2
    =\bra{\psi_0}\hat{Z}_{\bm{b};n}\ket{\psi_0}
    \mathrm{Tr}\left(\hat{\mathsf{M}}_{\bm{b};n}^\dagger\hat{\mathsf{M}}_{\bm{b};n}\right).
    \label{eq:norm_Zn}
\end{align}
Since $\left\|\hat{\mathsf{M}}_{\bm{b};n}\right\|^2 \leq \mathrm{Tr}\left(\hat{\mathsf{M}}_{\bm{b};n}^\dagger\hat{\mathsf{M}}_{\bm{b};n}\right) \leq d\left\|\hat{\mathsf{M}}_{\bm{b};n}\right\|^2$ is always satisfied, Eq.~\eqref{eq:norm_Zn} leads to
\begin{align}
    \bra{\psi_0}\hat{Z}_{\bm{b};n}\ket{\psi_0}
    \left\|\hat{\mathsf{M}}_{\bm{b};n}\right\|^2
    \leq\left\|\hat{\mathsf{M}}_{\bm{b};n}\ket{\psi_0}\right\|^2
    \leq\bra{\psi_0}\hat{Z}_{\bm{b};n}\ket{\psi_0}
    d\left\|\hat{\mathsf{M}}_{\bm{b};n}\right\|^2.
    \label{eq:inequality_norm}
\end{align}
In monitored quantum systems, $0<\bra{\psi_0}\hat{Z}_{\bm{b};n}\ket{\psi_0}\leq1$ is satisfied in typical trajectories, owing to the Born rule. 
In addition, Eq.~\eqref{eq:limit_Zn} means that $\bra{\psi_0}\hat{Z}_{\bm{b};n}\ket{\psi_0}$ does not decay exponentially, and thus $\lim_{n\rightarrow\infty}\ln\left[\bra{\psi_0}\hat{Z}_{\bm{b};n}\ket{\psi_0}\right]/n=0$ is satisfied. 
Since finite-dimensional systems are considered here, $\lim_{n\rightarrow\infty}\ln(d)/n=0$ is also satisfied. 
Therefore, taking the logarithm of Eq.~\eqref{eq:inequality_norm}, dividing by $n$, and taking the limit $n\rightarrow\infty$, we can obtain Eq.~\eqref{eq:typical-convergence_norm}.

\subsubsection{Numerical procedure} \label{sec:numerical_procedure}
When we numerically compute the Lyapunov exponents, it is difficult to directly diagonalize $\hat{\mathsf{M}}_{\bm{b};n}\hat{\mathsf{M}}_{\bm{b};n}^\dagger$ since the singular values of $\hat{\mathsf{M}}_{\bm{b};n}$ decay exponentially with $n$ and can quickly fall below the machine precision. 
However, $\{\varepsilon_i\}_i$ and $\{\ket{\Psi_{i,\bm{b};n}}\}_i$ with $i=1,2,\ldots,q$ can be obtained efficiently through the Gram-Schmidt orthonormalization, if $q$ is not so large. 
To this end, we compute the dynamics of a set of states $\ket{\tilde{\Psi}_{i,\bm{b};n}}$ that approach $\ket{\Psi_{i,\bm{b};n}}$ for large $n$. 

First, we prepare $q$ initial states $\ket{\tilde{\Psi}_{i,0}}$ that are orthonormalized as $\langle\tilde{\Psi}_{i,0}|\tilde{\Psi}_{j,0}\rangle=\delta_{ij}$. 
Second, we compute the states $\ket{\tilde{\Psi}_{i,\bm{b};n}}$ by partitioning the time evolution into $m$ blocks, each of which consists of $c$ steps,
\begin{align}
    \ket{\varphi_{i,\bm{b};mc}}
    =\hat{M}_{b_{mc}}\hat{M}_{b_{mc-1}}\cdots\hat{M}_{b_{(m-1)c+1}}
    \ket{\tilde{\Psi}_{i,\bm{b};(m-1)c}},
\end{align}
where $\ket{\tilde{\Psi}_{i,\bm{b};0}}=\ket{\tilde{\Psi}_{i,0}}$. 
If $c$ is not so large, we can avoid numerical underflow since the singular values of $\hat{M}_{b_{mc}}\hat{M}_{b_{mc-1}}\cdots\hat{M}_{b_{(m-1)c+1}}$ can be within the numerical precision. 
Third, we carry out the Gram-Schmidt orthonormalization of $\ket{\varphi_{i,\bm{b};mc}}$,
\begin{align}
    \ket{\phi_{i,\bm{b};mc}}
    =\left(\hat{\mathbb{I}}-\hat{\Pi}_{i,\bm{b};mc}\right)
    \ket{\varphi_{i,\bm{b};mc}},
\end{align}
where $\hat{\Pi}_{1,\bm{b};mc}=0$, $\hat{\Pi}_{i,\bm{b};mc}=\sum_{j=1}^{i-1}\ket{\tilde{\Psi}_{j,\bm{b};mc}}\bra{\tilde{\Psi}_{j,\bm{b};mc}}$ for $i\geq2$, and $\hat{\mathbb{I}}$ is the identity operator. 
Here, $\ket{\tilde{\Psi}_{i,\bm{b};mc}}$ is given by
\begin{align}
    \ket{\tilde{\Psi}_{i,\bm{b};mc}}=\frac{\ket{\phi_{i,\bm{b};mc}}}
    {\sqrt{\langle\phi_{i,\bm{b};mc}|\phi_{i,\bm{b};mc}\rangle}}.
\end{align}

Thus, $\ket{\tilde{\Psi}_{i,\bm{b};mc}}$ is obtained from $\ket{\varphi_{i,\bm{b};mc}}$ and $\{\ket{\tilde{\Psi}_{j,\bm{b};mc}}\}_j$ with $j=1,2,\ldots,i-1$, at each step $m$. 
In the procedure explained above, candidates for the Lyapunov exponents are obtained as
\begin{align}
    \tilde{\varepsilon}_{i,\bm{b};mc}
    =-\frac{1}{mc}\sum_{\ell=1}^m\ln\left(\sqrt{\langle\phi_{i,\bm{b};\ell c}|\phi_{i,\bm{b};\ell c}\rangle}\right).
\end{align}
For sufficiently large $m$, $\tilde{\varepsilon}_{i,\bm{b};mc}$ and $\ket{\tilde{\Psi}_{i,\bm{b};mc}}$ approach the $i$th Lyapunov exponent and the corresponding eigenmode in Eq.~\eqref{eq:Lyapunov-exponent},
\begin{align}
    \tilde{\varepsilon}_{i,\bm{b};mc}\rightarrow\varepsilon_{i,\bm{b};mc},\ \ 
    \ket{\tilde{\Psi}_{i,\bm{b};mc}}\rightarrow\ket{\Psi_{i,\bm{b};mc}},
    \label{eq:candidates}
\end{align}
respectively. 
We do not detail the reason why Eq.~\eqref{eq:candidates} is satisfied for large $m$, which has been elucidated in Refs.~\cite{ershov1998concept, mochizuki2025transitions}. 
We note that the procedure explained above may not be applicable to quantum systems exposed to projective measurements. 
This is because the matrix rank of $\hat{M}_{b_{mc}}\hat{M}_{b_{mc-1}}\cdots\hat{M}_{b_{(m-1)c+1}}$ is smaller than $d$ and thus $\{\ket{\varphi_{i,\bm{b};mc}}\}_i$ with $i\geq2$ can reside in the kernel of the time-evolution operator. 
When this occurs, we cannot perform the Gram-Schmidt orthonormalization, and the method becomes ineffective. 


\section{Measurement-induced phase transitions and their Lyapunov analysis}
\label{sec:mipt}
While the preceding chapters have focused on the general and formal properties of open quantum systems and quantum trajectories, this chapter explores the measurement-induced phase transition (MIPT) as an intriguing phenomenon unique to the quantum trajectories of many-body systems, which has received substantial attention and witnessed rapid progress in recent years.
Although there have already been several reviews on this topic~\cite{Potter2022, fisher2023random, Lunt2022Quantum, Skinner23lecture, HanZeLi2025Measurement}, we here provide an overview of MIPTs from three perspectives—entanglement, purification, and spectrum—with a particular focus on the Lyapunov analysis for quantum trajectories as discussed in the previous chapter.

\subsection{Entanglement transition}
\label{sec:entanglement_MIPT}
\subsubsection{General concepts}
The MIPTs are a class of non-equilibrium quantum phase transitions that primarily arise in many-body systems subject to both unitary evolution and quantum measurements.
Without measurements, a highly entangled state is realized due to the strong scrambling of the state's information by the unitary evolution.
If, however, we measure some local observables for the state and extract the information so frequently during the evolution, the superposition of the state is mostly broken, and the entanglement dies out.
Such competitive actions in the dynamics cause a phase transition in the qualitative behavior of the entanglement, which is the most prototypical example of MIPTs.

These MIPTs are prominently studied in random quantum circuits (RQCs) and systems under Hamiltonian dynamics, where unitary evolution is interspersed with local measurements.
Such dynamics generate stochastic quantum trajectories, reflecting the probabilistic nature of quantum measurements governed by the Born rule.
Because of this inherent randomness, the evolution of a quantum system under repeated measurements yields not a single deterministic state, but an ensemble of quantum trajectories $\{\ket{\psi_{\bm{b};t}}\}_{\bm{b}}$ conditioned on the sequence $\bm{b}$\footnote{\label{f:b-continuous}
While $\bm{b}$ represents a sequence of measurement outcomes in previous chapters, the meaning of $\bm{b}$ is extended in this chapter;
in RQCs treated in this chapter, $\bm{b}$ contains which positions are chosen for local measurements, which outcomes are obtained by the measurements, and what unitaries are applied in the trajectory.
That is, the physical source of randomness may not be due to the measurement alone.
Nevertheless, the other sources of randomness can safely be treated in the framework of Kraus operators as well, so most of the general discussions in the previous chapters still hold;
see the next footnote for some subtle points.
}.

A crucial feature of the MIPTs is that they may not manifest in the ensemble-averaged density matrix, $\mathbb{E}[\hat{\rho}_{\bm{b};t}]=\sum_{\bm{b}_t} p_{\bm{b};t}\ket{\psi_{\bm{b};t}}\bra{\psi_{\bm{b};t}}$\footnote{
The average $\mathbb{E}$ is taken over the three kinds of randomness above, positions of quantum measurements, measurement outcomes, and random unitaries.
If unitary matrices $\{\hat{U}\}$ are sampled from a continuous set, e.g. Haar random unitaries, $\sum_{\bm{b}_t}$ includes the integral $\prod_{s=1}^t\int d\mu(U_s)$ where $\mu$ is the probability measure of $\{\hat{U}\}$. 
We here apply such an abuse of symbols for simplicity, where the average expressed by the discrete sum includes integrals over continuous variables.
With this replacement, discussions in Chapters~\ref{sec:CPTPspectra}, \ref{sec:linear-quantity_purification}, and~\ref{sec:nonlear-quantity_Lyapunov-spectrum} are applicable to situations reviewed in this chapter, where unitary dynamics are interspersed with quantum measurements.
}, where $p_{\bm{b};t}=\bra{\psi_0}\hat{\mathsf{M}}_{\bm{b};t}^\dagger\hat{\mathsf{M}}_{\bm{b};t}\ket{\psi_0}$ is the probability (density) to obtain the sequence $\bm{b}$ and $\ket{\psi_0}$ is an initial state.
Instead, the MIPTs emerge from properties of the trajectory ensemble $\{ \ket{\psi_{\bm{b};t}} \}_{\bm{b}}$ itself.
By tuning the measurement rate or strength $\omega$, one obtains a trajectory ensemble that depends on $\omega$.
However, if the CPTP map corresponding to the monitored system is irreducible and unital, the averaged density matrix approaches the maximally mixed state $\hat{\rho}_\mathrm{ss}=\hat{\mathbb{I}}/d$, irrespective of $\omega$.
Therefore, no distinction arises in any physical quantities evaluated on the averaged state. 
In this case, linear functions of a density matrix cannot probe the properties specific to each trajectory, since the average of any linear function coincides with the expectation value evaluated on the averaged state, as discussed in Sec.~\ref{sec:ergodicity_linear-observable}.
Such an invisibility of MIPTs in the averaged density matrix occurs in various monitored systems. 

In contrast, nonlinear functions averaged over a trajectory ensemble, whose values generally differ from those of the same functions evaluated on the averaged state $\mathbb{E}[\hat{\rho}_{\bm{b};t}]$, can probe the $\omega$-dependent physics (see also Sec.~\ref{sec:nonlinear-observables}).
Such nonlinear functions include the R\'enyi entanglement entropy,
\begin{align}
    S_{X,\alpha}(\bm{b};\!t) = \frac{1}{1-\alpha}\ln\Tr [\hat{\rho}_{\bm{b};t, X}^\alpha]
    \qquad (0\leq\alpha<1\:\:\:\mathrm{or}\:\:\:1<\alpha),
\end{align}
and the von Neumann entanglement entropy, obtained in the limit $\alpha\to1$:
\begin{align}
    S_{X}(\bm{b};\!t) = \lim_{\alpha\to1}S_{X,\alpha}(\bm{b};\!t) =  -\Tr \left(\hat{\rho}_{\bm{b};t, X} \ln \hat{\rho}_{\bm{b};t, X}\right).
\end{align}
Here, $\hat{\rho}_{\bm{b};t, X} = \Tr_{\bar{X}}[\hat{\rho}_{\bm{b};t}]$ is the reduced density matrix of a pure-state trajectory $\hat{\rho}_{\bm{b};t} = \ket{\psi_{\bm{b};t}}\bra{\psi_{\bm{b};t}}$ for a subsystem $X$, with $\bar{X}$ being the complement of $X$.
The MIPTs are revealed in the trajectory average of these nonlinear quantities.
If the corresponding CPTP map has a unique steady state and the purification property, we can instead look at the temporal average of them in a single typical trajectory, owing to the ergodicity discussed in Sec.~\ref{sec:ergodicity_nonlinear-quantity}. 

In early foundational works~\cite{li2018quantum, skinner2019measurement, li19measurement, chan2019unitary}, MIPTs were identified in $(1+1)$-dimensional RQCs where a chain of qubits is evolved by random unitary gates and projective measurements.
Each unitary gate $\hat{U}_{l, l+1}$ is randomly chosen from the uniform Haar measure or Clifford gates and acts on the qubits on two neighboring sites $l$ and $l+1$.
These unitary gates are arranged in a brickwork manner, as depicted in Fig.~\ref{fig:circuit_schematic}(a).
After one layer of unitary gates, the qubits are measured at each site $l$ in a certain basis with a probability $p$.
The corresponding Kraus operators are given, for example, by $\hat{M}_{l,\pm}=(\mathbb{\hat{I}}\pm\hat{\sigma}^z_l)/2$, where $\hat{\sigma}^z_l$ is the Pauli $z$ operator at the site $l$ and $\pm 1$ are possible measurement outcomes.
This layer of measurements contains the randomness in the measurement positions and in the measurement outcomes determined by the Born rule.
Note that the randomness about the unitary gates and measurement positions are controllable in the sense that they can be manually drawn from a given probability distribution or probability measure upon running a circuit, whereas the randomness about the measurement outcomes is intrinsic and uncontrollable\footnote{\label{f:Kraus-unitary_whole}
The quantum and classical random processes originated respectively from the Born rule and the probability $p$ can be unified by redefining the Kraus operators as $\hat{M}_{l,\tilde{b}=0}=\sqrt{1-p}\mathbb{\hat{I}}$ and $\hat{M}_{l,\tilde{b}=\pm}=\sqrt{p}(\mathbb{\hat{I}}\pm\hat{\sigma}^z_l)/2$.
Further unification of measurements and random unitary dynamics by $\hat{U}=\prod_l\hat{U}_{l,l+1}$ is realized by considering the one-step Kraus operators $\hat{M}_b=(\prod_l\hat{M}_{l,\tilde{b}_l})\hat{U}$, where $\{b\}=\{\tilde{b}_1,\cdots,\tilde{b}_L,U\}$ with $\tilde{b}_l\in\{0,+,-\}$.
Thus, even for such unitary/measurement-hybrid random dynamics, one can apply the previous discussion of the ergodicity in Chapters~\ref{sec:linear-quantity_purification} and~\ref{sec:nonlear-quantity_Lyapunov-spectrum} without any issues.
}, which makes analytical treatment of monitored systems complicated.
Hereafter, we consider the chain of qubits and denote the length of the chain by $L$ unless otherwise specified.
\begin{figure}[!h]
\centering\includegraphics[width=\linewidth]{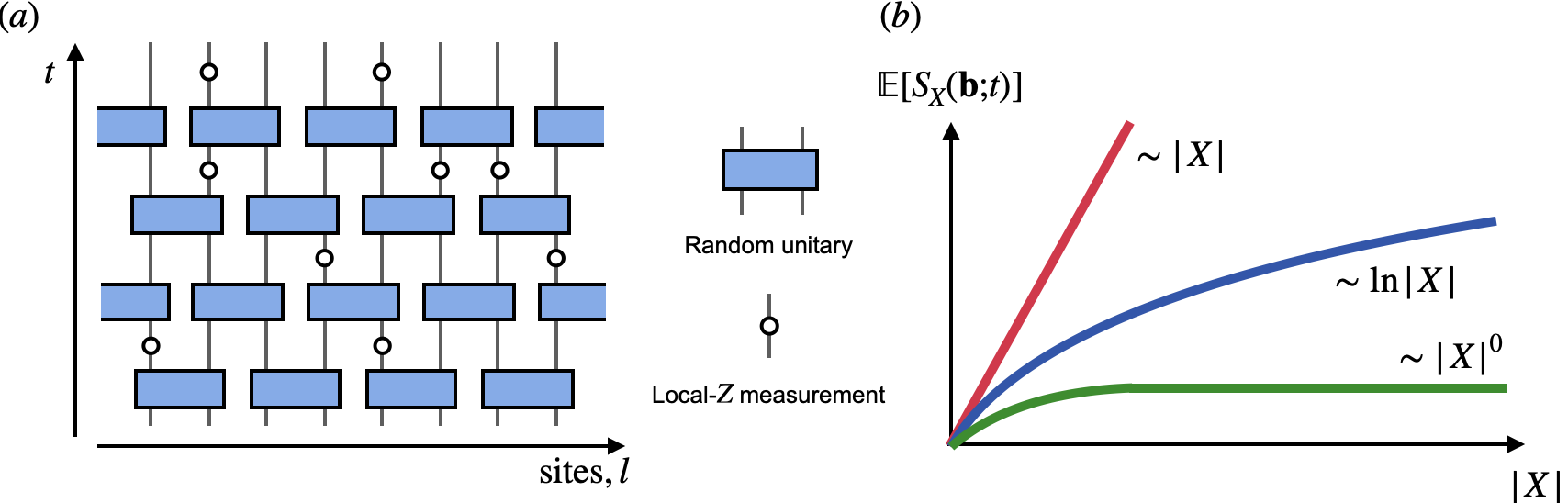}
\caption{
(a) Schematic of a monitored random quantum circuit.
Each unitary gate is randomly chosen in an independent manner, and a local-$Z$ measurement is applied on each qubit with a probability $p$.
(b) Schematic of the long-time values of the averaged entanglement entropies, $\mathbb{E}[S_{X}(\bm{b};\!t)]$.
The entanglement entropy shows a phase transition from a volume-law phase for $p<p_c$ (red curve) to an area-law phase for $p>p_c$ (green curve).
At the transition $p=p_c$, it shows a logarithmic scaling in the subsystem size $|X|$ (blue curve).
The figures were created based on Refs.~\cite{li19measurement, skinner2019measurement}.
}
\label{fig:circuit_schematic}
\end{figure}

After evolution of an initial state through the RQC, we obtain an ensemble of output pure-state trajectories.
By averaging the late-time ($t\gtrsim L$) entanglement entropy over the trajectory ensemble, we can observe a phase transition in its qualitative scaling behaviors:
For low measurement rates ($p < p_c$), the system resides in a highly entangled ``volume-law'' phase, where the entanglement entropy scales extensively with the subsystem size, $\mathbb{E}[S_X(\bm{b};\!t)]=\mathcal{O}(|X|)$, with $|X|$ being the length of a contiguous subsystem $X$.
For high measurement rates ($p > p_c$), the system is in a weakly entangled ``area-law'' phase with constant entropy\footnote{
Note that we consider one-dimensional systems, where the ``area" of a subsystem is a zero-dimensional boundary.
}, $\mathbb{E}[S_X(\bm{b};\!t)]=\mathcal{O}(|X|^0)$.
At the critical point $p = p_c$, the system exhibits universal scaling behavior characteristic of a continuous phase transition.
This includes the logarithmic scaling of the entanglement entropies and the power-law decay of the squared connected correlation functions.
The entanglement scaling across the transition is summarized as
\begin{align}
    \mathbb{E}[S_X(\bm{b};\!t)] \propto
    \left\{
    \begin{array}{ll}
    |X| \qquad & (p < p_c) \\
    \ln |X| \qquad & (p = p_c) \\
    |X|^0 \qquad & (p > p_c)
    \end{array}
\right.,
\end{align}
and schematically shown in Fig.~\ref{fig:circuit_schematic}(b).
This transition, known as the \emph{measurement-induced entanglement transition}, is a nontrivial phenomenon observable only at the level of individual quantum trajectories\footnote{
This monitored system is unital and irreducible, i.e., the maximally maxed state $\hat{\rho}_\mathrm{ss}=\hat{\mathbb{I}}/2^L$ is the unique stationary state in the averaged CPTP dynamics in the whole parameter region. 
Indeed, we can easily confirm that $\mathcal{E}[\hat{\mathbb{I}}]=\hat{\mathbb{I}}$ and (3) of Theorem \ref{thm:IrreducibilityByKraus} are satisfied for any $p$, since we can construct any operators through multiplications and linear combinations of local Haar unitaries $\{\hat{U}_{l,l+1}\}$.
}.

\subsubsection{Classical toy model --- Vertical minimal cut}
\label{sec:VerticalMinimalCut}

To gain an intuitive understanding of the entanglement transition, one can employ a toy model that focuses on the zeroth R\'enyi entanglement entropy, $S_{X,0}(\bm{b};\!t)$, which is called the Hartley entanglement entropy~\cite{skinner2019measurement, Skinner23lecture}.
Note that the Hartley entanglement entropy counts the Schmidt rank (the number of nonzero eigenvalues) of the reduced density matrix $\hat{\rho}_{\bm{b};t,X}$, and thus its value does not care about the measurement outcomes.
For a $(1+1)$-dimensional RQC with Haar-random unitary gates and projective measurements at rate $p$, the dynamics of the Hartley entanglement entropy can be analyzed through the exact mapping to a classical bond percolation problem on the two-dimensional square lattice (see Fig.~\ref{fig:vertical_mincut}).

\begin{figure}[!h]
\centering\includegraphics[width=\linewidth]{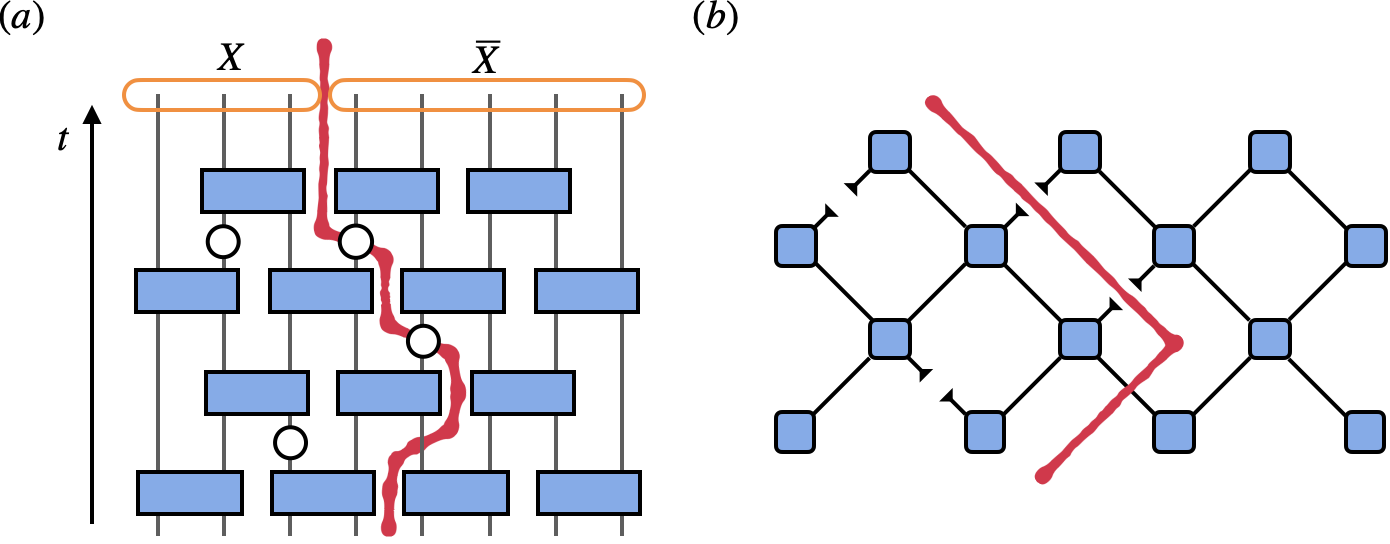}
\caption{
(a) A $(1+1)$-dimensional monitored RQC.
The red curve should be drawn from the final time to the initial time to separate the circuit to two circuits in a way that one contains the state of the $|X|$ qubits at the final time, and another contains the state of the $|\overline{X}|=L-|X|$ qubits at the final time.
In the figure, the red curve is a possible vertical minimum cut passing through one link.
(b) The corresponding bond percolation problem on a square lattice where
blue nodes represent the unitary gates. 
The bonds, corresponding to links in (a), are inactivated by measurements.
The figures were created based on Refs.~\cite{skinner2019measurement, Skinner23lecture}.
}
\label{fig:vertical_mincut}
\end{figure}

First, we draw a vertical path that divides the RQC into two circuits. 
The path starts from an arbitrary bond including boundaries at the initial time and ends at the bond connecting the subsystem $X$ and its complement $\bar{X}$ at the final time. 
Let $N_\mathrm{cut}$ be the number of shared links between the two circuits.
Then, the Hartley entanglement entropy is bounded from above by $N_{\mathrm{cut}}$,
\begin{align}
    S_{X,0}(\bm{b};t) \leq N_{\mathrm{cut}}\times\ln2,
\end{align}
since each link can transmit at most one qubit information. 
The equality holds when $N_\mathrm{cut}$ takes the minimum value, i.e., the number of links that must be cut to separate the two subsystems is minimized over all possible paths and possible starting points at the initial time,
\begin{align}
    S_{X,0}(\bm{b};t) = \min N_{\mathrm{cut}}\times\ln2.
\end{align}

When a measurement occurs at a spacetime point in the RQC, we have full information of the qubit at the spacetime point from the measurement outcome without knowing the pre-measurement state, implying that a measured link does not transmit any information.
Hence, the role of the measurement is a break or inactivation of the link.

Crucially, this minimal cut problem is classical in nature;
its consequence is insensitive to the specific measurement outcomes, depending only on their spatiotemporal locations.
This classical problem can be further mapped to two-dimensional bond percolation, where the measurement rate $p$ corresponds to the probability of a bond being inactive.

The percolation model is exactly solvable and exhibits a phase transition at a critical probability $p^{\mathrm{perc}}_c=1/2$.
This transition is described by a conformal field theory (CFT) with a correlation length exponent $\nu^{\mathrm{perc}}=4/3$, which provides the late-time behavior of the Hartley entanglement entropy across the MIPT for a system with length $L$:
\begin{align}
    \lim_{t\to\infty}\mathbb{E}[S_{X,0}(\bm{b};\!t)] \sim
    \left\{
    \begin{array}{ll}
    |X| \qquad &(p < p^{\mathrm{perc}}_c) \\
    \ln |X| \qquad &(p = p^{\mathrm{perc}}_c) \\
    |X|^0 \qquad &(p > p^{\mathrm{perc}}_c)
    \end{array}
\right..
\end{align}
Thus, the model captures the transition from a volume-law phase to an area-law phase.
These minimal cut and percolation pictures provide powerful tools for a qualitative understanding of the concept of MIPTs.

\subsubsection{Beyond minimal cut}
So far, we have focused on the RQC with projective measurements and introduced the mapping of the Hartley entanglement entropy to the classical percolation problem through the minimal cut picture.
An important remark is that the R\'enyi entanglement entropies with a general index $\alpha>0$ cannot be mapped to the percolation problem.

As the Hartley entanglement entropy only provides an upper bound on general R\'enyi entanglement entropy, $S_{A,0} \geq S_{A,\alpha}$ with $\alpha>0$, its critical point $p^{\mathrm{perc}}_c=1/2$ is located above the critical point $p_c$ measured by $S_{A,\alpha}$ with $\alpha>0$.
In fact, the critical point $p_c$ in the same model was reported to be $p_c=0.26\pm0.08$ from numerical data for the von Neumann entanglement entropy~\cite{skinner2019measurement}.
Namely, the dynamics of the general R\'enyi entanglement entropies in the monitored system are still quantum problems so that direct and rigorous predictions of the behaviors of these quantities await further theoretical development. 
See, e.g., Ref.~\cite{Potter2022}, for an analytical approach for the general R\'enyi entanglement entropies.

Even for the Hartley entanglement entropy, its mapping to the percolation picture is justified only when the measurements are projective.
The scope of generic MIPTs, however, extends beyond the projective measurements (see, e.g., Refs. \cite{Szyniszewski2019entanglement, Fuji2020measurement, Alberton2020entanglement, bao20theory, Turkeshi21measurement, LeGal24entanglement}).
For instance, systems evolving under Hamiltonian dynamics and subject to continuous weak measurements\footnote{
In this chapter, we denote the non-projective measurement scheme realized as, e.g., an indirect measurement described in Sec.~\ref{sec:indirect}, by the weak measurement or the generalized measurement, following the standard terminology in the literature.
} also exhibit similar transitions as functions of the measurement strength.
Common protocols for such continuous monitoring include the quantum jump formalism leading to discrete trajectory updates and the quantum diffusion formalism leading to continuous trajectory updates, as explained in Chapter~\ref{sec:trajectory_master-equation}.

\subsection{Purification transition}
\label{sec:purification_MIPT}
\subsubsection{General concepts}
In addition to the entanglement transition, which concerns pure-state properties, MIPTs can also manifest themselves in the properties of mixed-state trajectories.
This is captured by the purification transition~\cite{gullans2020dynamical} (see also Refs.~\cite{loio2023purification, Sierant22measurement}).

To see this, consider an RQC setup where the system is initialized in the maximally mixed state.
The dynamics involves two competing effects: unitary evolution tends to scramble information and tries to keep the state mixed, while projective measurements (at rate $p$) extract information and thus tend to purify it.
This competition can trigger a phase transition related to the timescale for the state to be purified, although the mixed state subject to the measurements at any finite rate is purified in the long-time limit under a certain assumption as discussed in Sec.~\ref{sec:pur}, as long as $L$ is finite.

\begin{figure}[!h]
\centering\includegraphics[width=\linewidth]{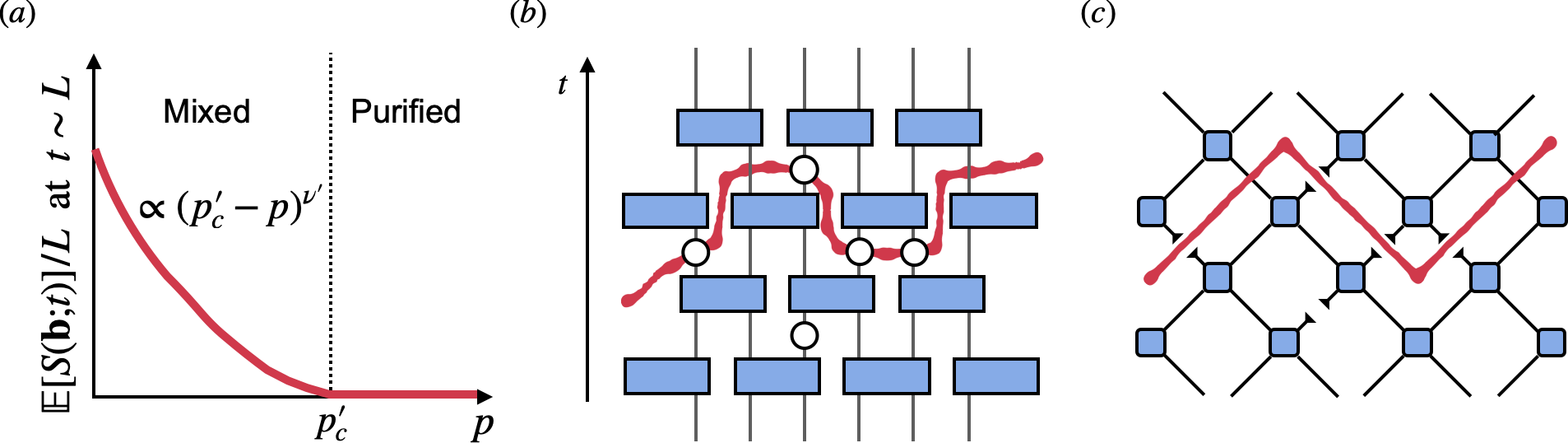}
\caption{
(a) The averaged mixed-state entropy density of the whole system at a time $t\sim L$ against the measurement probability $p$ in an RQC with local measurements starting from the maximally mixed state.
(b) A $(1+1)$-dimensional monitored RQC.
The red curve should be drawn to separate the whole input state and the whole output state.
In the figure, the red curve is a possible horizontal minimum cut.
(c) The corresponding bond percolation problem on a square lattice.
Blue nodes represent the unitary gates and the measured bonds are inactivated.
The figures were created based on Refs.~\cite{gullans2020dynamical, Skinner23lecture, nahum21measurement, bao20theory}.
}
\label{fig:purification_transition}
\end{figure}

Reference~\cite{gullans2020dynamical} numerically confirmed this picture, as summarized in the phase diagram of Fig.~\ref{fig:purification_transition}(a).
For a system of size $L$ starting from the maximally mixed state, the trajectory-averaged entropies of the mixed state,
\begin{align}
    S_{\alpha}(\bm{b};\!t) = \frac{1}{1-\alpha}\ln\Tr [\hat{\rho}_{\bm{b};t}^\alpha],
\end{align}
at times $t\sim L$ exhibit two distinct phases;
the system remains in a mixed phase with volume-law entropy for $p<p'_c$, while it enters a purified phase where the entropy decays exponentially in $L$ for $p>p'_c$.

As we will review in Sec.~\ref{sec:HorizontalMinumalCut}, for the RQC with Haar random unitary gates and projective measurements, this purification problem on the Hartley entropy can be mapped to the same percolation problem as discussed for the entanglement transition.
For general Renyi entropies, Ref.~\cite{gullans2020dynamical} studied a $(1+1)$-dimensional random Clifford circuit\footnote{
The Clifford circuit is a class of quantum circuits that can be efficiently simulated by classical computers.
See, e.g., Ref.~\cite{Aaronson04improved} for details.
} and found that the critical point $p'_c$ and the critical exponent $\nu'$ of the purification transition are numerically identical to those of the entanglement transition ($p_c$ and $\nu$) in the same circuit.
Theoretical frameworks based on replica statistical mechanics models have also shown that both transitions in local circuits can be described by the same effective theory~\cite{bao20theory}.
This correspondence suggests that the entanglement and purification transitions are two facets of the same underlying critical phenomenon.

While it is highly nontrivial to generally establish the equivalence between the entanglement and purification transitions, one can formally view the purification transition as a kind of entanglement transition in the following sense.
Any mixed state of a system can be represented as the reduced state of a larger, pure state by introducing ancilla qubits.
In this picture, the initial maximally mixed state is obtained by tracing out the ancilla qubits from a maximally entangled state between the physical system and the ancilla qubits.
If the monitored circuit acts only on the physical system, its entropy at any time is equal to the entanglement entropy between the system and ancilla.
Therefore, the purification transition of the system can be directly interpreted as an entanglement transition for the system-ancilla partition, although it is still far from proving their equivalence.

The purification transition can be fundamentally characterized by the Lyapunov spectrum of the monitored dynamics.
Consider a $(1+1)$-dimensional RQC of size $L$ without any internal symmetries.
The evolution from time $0$ to $t$ for a given measurement sequence $\bm{b}$ is described by a Kraus operator $\hat{\mathsf{M}}_{\bm{b};t}$, which satisfies the POVM condition $\sum_{\bm{b}_t} \hat{\mathsf{M}}_{\bm{b};t}^\dagger \hat{\mathsf{M}}_{\bm{b};t}  = \hat{\mathbb{I}}$.
If the initial state is maximally mixed, i.e., $\hat{\rho}_{t=0} = \hat{\mathbb{I}}/2^L$, the state at time $t$ is given by
\begin{align}
    \hat{\rho}_{\bm{b};t} = \frac{\hat{\mathsf{M}}_{\bm{b};t} \hat{\mathsf{M}}^\dagger_{\bm{b};t}}{\Tr [\hat{\mathsf{M}}_{\bm{b};t} \hat{\mathsf{M}}^\dagger_{\bm{b};t}]}.
    \label{eq:evolved state from I}
\end{align}
The eigenvalues of $\hat{\mathsf{M}}_{\bm{b};t} \hat{\mathsf{M}}^\dagger_{\bm{b};t}$ are the squares of the singular values of the Kraus operator, $\hat{\mathsf{M}}_{\bm{b};t}$.
Let these singular values be $\Lambda_{j, \bm{b};t} = e^{-\varepsilon_{j, \bm{b};t}t}$.
As discussed in Sec.~\ref{sec:Lyapunov-analysis}, the late-time values $\varepsilon_j=\underset{t\to\infty}{\lim}\varepsilon_{j,\bm{b};t}$ are independent of the trajectory label $\bm{b}$ almost surely, and called the Lyapunov spectrum.
Since measurements cause the state's norm to decrease or remain the same, these exponents are non-negative ($\varepsilon_j\geq0$) for the monitored settings.
The R\'enyi entropies of the state $\hat{\rho}_{\bm{b};t}$ are expressed directly in terms of these singular values:
\begin{align}
    S_\alpha(\bm{b};\!t) = \frac{1}{1-\alpha}\ln\left(\sum_{j=1}^{2^L} \Lambda_{j,\bm{b};t}^{2\alpha}\right).
    \label{eq:Relation bw Entropy and Lyapunov}
\end{align}
Thus, the purification transition is governed by a qualitative change in these mixed-state entropies that are closely connected to the Lyapunov spectrum via Eq.~\eqref{eq:Relation bw Entropy and Lyapunov}.

\subsubsection{Classical toy models --- Minimal-cut approach}
\label{sec:HorizontalMinumalCut}

The classical mapping for the Hartley entanglement entropy, as reviewed in Sec.~\ref{sec:VerticalMinimalCut}, can be extended to the purification transition in a $(1+1)$-dimensional RQC subject to projective measurements.
Analogous to the entanglement transition, which is modeled by the vertical minimal cut separating two spatial regions, the purification transition can be understood through a horizontal minimal cut~\cite{Skinner23lecture, nahum21measurement}.
Importantly, this mapping yields the same classical percolation problem as that for the entanglement transition.
As illustrated in Figs.~\ref{fig:purification_transition}(b) and (c), a horizontal path is drawn to measure the flow of information from the initial time to the final time in the RQC. 
The minimum number of links that must be cut to separate the initial and final states by such a path corresponds to the Hartley mixed-state entropy $S_0(\bm{b};\!t)$ for the maximally mixed initial state;
if there is a horizontal percolating path without any active bonds, the initial state is completely forgotten and the state is rendered to be a pure state.
Even if the measurement probability $p$ is small, there is always a possibility that all sites are measured at the same time and the state is purified, i.e., the entropy ends up with zero after a sufficiently long time.

For a finite time, the entropy decays as $\mathbb{E}[S_0(\bm{b};t)]\sim Le^{-t/\tau_\mathrm{P}^{\mathrm{perc}}(L)}$, but the purification timescale $\tau_\mathrm{P}^{\mathrm{perc}}(L)$ exhibits qualitatively distinct behaviors between the non-percolating phase (mixed phase) and the percolating phase (purified phase),
\begin{align}
    \tau_\mathrm{P}^{\mathrm{perc}}(L)=
    \left\{
    \begin{array}{ll}
    e^{\mathcal{O}(L)} \qquad &(p < p^{\mathrm{perc}}_c) \\
    \mathcal{O}(L^0) \qquad &(p > p^{\mathrm{perc}}_c)
    \end{array}
\right..
\end{align}
This gives the expected change in the entropy density across the purification transition for $t= \mathcal{O}(L)$ as discussed above.

So far, we have focused on the monitored RQCs with local unitary gates acting only on neighboring sites, which exhibit both entanglement and purification transitions.
However, once we consider an RQC with nonlocal unitary gates or a Hamiltonian evolution by long-range interactions, the absence of natural spatial structure obscures the area-law entangled phase and thereby the entanglement transition~\cite{block2022measurement, Minato2022Fate}.
Below, we discuss an example of such a nonlocal monitored system.

\vskip\baselineskip
\noindent\textbf{All-to-all circuit}

\begin{figure}[!h]
\centering\includegraphics[width=\linewidth]{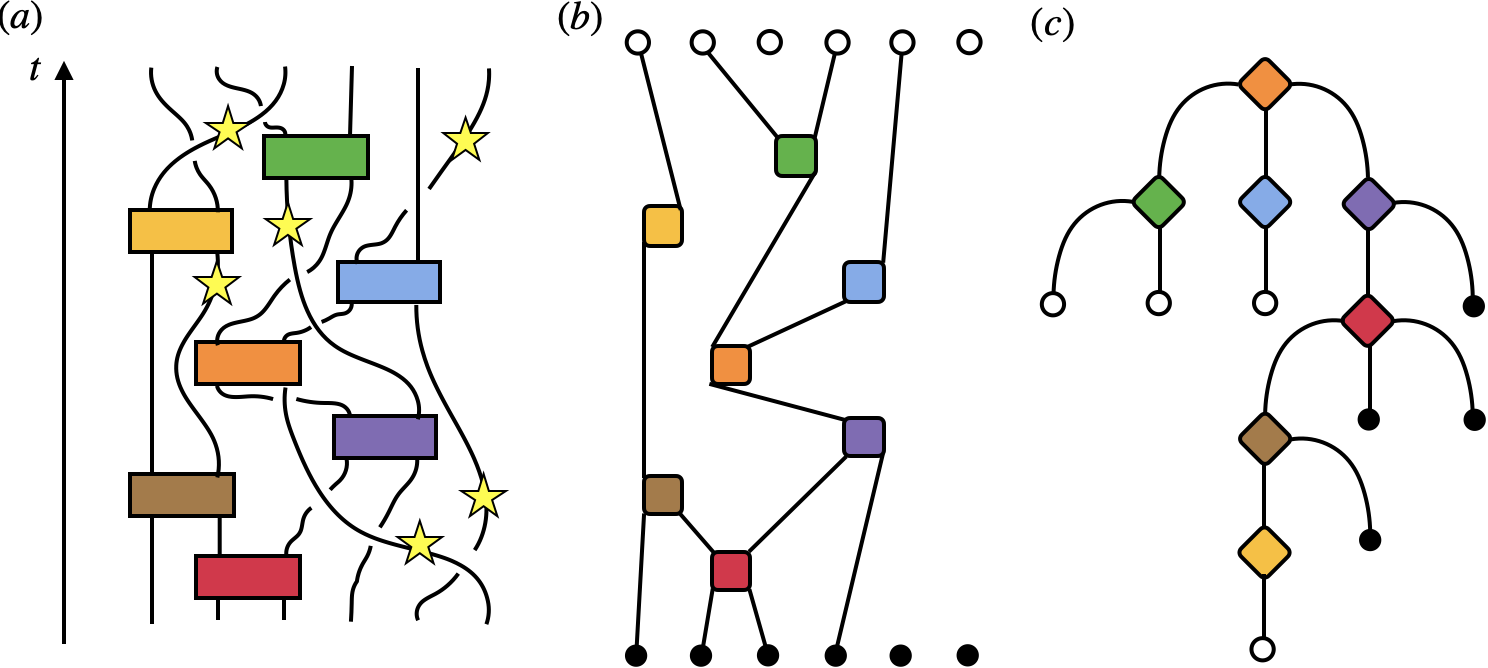}
\caption{(a) Schematic of an all-to-all circuit.
Black worldlines represent the evolution of a particular
spin, with time proceeding vertically.
Colored blocks are two-spin unitary gates, and yellow stars are single-spin measurements.
(b) The corresponding graph, where nodes represent unitaries and edges denote unbroken segments of the worldline.
The color of each node matches the color of the corresponding unitary gate in (a).
Small black (white) circles indicate the initial (final) time associated with each spin.
(c) The classical graph displayed in a tree-like structure, rooted at the orange seed node.
The figures were created based on Ref.~\cite{nahum21measurement}.}
\label{fig:Nahum21_mincut}
\end{figure}

The simplest model of a nonlocal monitored system is a monitored all-to-all circuit of $L$ qubits \cite{nahum21measurement}.
In this model, there are on average $pL$ measurements and $(1-p)L$ unitary gates at random positions in a unit interval of time.
The unitary gates couple randomly-chosen arbitrary pairs of qubits, regardless of their distance.
While there is no simple definition of the entanglement transition in this case, the purification transition, being a mixed-state property, can still occur.

As we have reviewed above, if we focus on the Hartley mixed-state entropy, we can map the problem onto a percolation model on the circuit graph itself. 
We can then analytically find the critical point for the purification transition.
The analysis of the percolation model is regarded as a connectivity problem or branching process, as illustrated in Fig.~\ref{fig:Nahum21_mincut}.
The underlying idea is to imagine that measurements cut the connections (worldlines) between the nodes (gates) of the circuit.
We then ask whether a connected cluster of gates can survive this cutting and grow to an infinite size in the thermodynamics limit $L\rightarrow\infty$.

Specifically, this connectivity problem is modeled as follows:
each gate has four legs (two in-coming, two out-going) potentially connected to other gates.
Each leg can either be cut by a measurement or be successfully connected to a new gate.
When the measurement probability is small, unitary gates form a connected cluster that grows extensively with the system size.
This regime corresponds to the mixed phase.
For a large measurement probability, the unitary gates only form finite-size clusters, leading to a fragmented graph.
The latter regime corresponds to the purified phase.
Since the number of unitary gates is $(1-p)L$ on average in a time slice and each gate has two out-going legs, the effective rate at which a given worldline experiences a unitary action in the time slice is $2(1-p)$.
Thus, the probability $r$ that a given worldline is terminated by a measurement before reaching the next gate becomes $r=p/(p + 2(1-p))=p/(2-p)$.

For a sufficiently large $L$, the local structure of this random graph can be regarded as a tree, as shown in Fig.~\ref{fig:Nahum21_mincut}(c).
Viewed in this way, the connectivity problem is mapped onto a simple branching process:
starting from a seed gate, each unitary node can generate several descendants through its uncut legs, and the growth of the tree represents how quantum information propagates through the circuit.
Each unitary node has three potential downward legs, each of which survives with probability $1-r$.
This gives the average branching number $3(1-r)$.
The percolation transition therefore occurs when the branching number becomes unity, $3(1-r^{\mathrm{perc}}_c)=1$, giving the critical probability $p^{\mathrm{perc}}_c = 4/5$.
This percolation transition corresponds to the purification transition probed by the Hartley mixed-state entropy in the monitored all-to-all circuit.

\subsubsection{Analytical treatments for purification dynamics beyond minimal cut}

In this section, we review several analytical approaches for the purification dynamics of mixed-state trajectories evolved under the competition between unitary dynamics and measurements.

\vskip\baselineskip
\noindent\textbf{Random matrix theory approach for non-local circuits}

The properties of the mixed phase can be analytically treated using the framework of random matrix theory (RMT) for a nonlocal circuit~\cite{bulchandani2024random, DeLuca2025universality}.
Here we review the results of Ref.~\cite{bulchandani2024random}. 
Before analyzing a specific model, we define two key timescales for a general monitored system described by Kraus operators $\hat{\mathsf{M}}_{\bm{b};t}$.
Note that we consider irreducible systems, where the Lyapunov exponents are determined  independently of $\bm{b}$;
see Sec.~\ref{sec:typical-convergence_Lyapunov-spectrum}.

\begin{enumerate}
    \item Rank collapse time $\tau_{\mathrm{RC}}$:
    When starting from the maximally mixed state, measurements tend to reduce the rank of the density matrix.
    The rank collapse time is the time expected for the state to be exactly pure:
    \begin{align} \label{eq:RankCollapseTime}
        \tau_{\mathrm{RC}} = \mathbb{E}[\min\{t:r_{\bm{b};t}=1\}],
    \end{align}
    where $r_{\bm{b};t}=\mathrm{rank}(\hat{\mathsf{M}}_{\bm{b};t}\hat{\mathsf{M}}_{\bm{b};t}^\dagger)$.

    \item Purification time $\tau_\mathrm{P}$:
    If we arrange the singular values of $\hat{\mathsf{M}}_{\bm{b};t}$ as $\Lambda_{1, \bm{b};t}\geq\Lambda_{2, \bm{b};t}\geq\cdots$, the state is dominated by the leading left singular vector $\ket{\Psi_{1, \bm{b};t}}$, which is the ground state of the effective Hamiltonian introduced in Sec.~\ref{sec:Lyapunov-analysis}.
    The approach to a pure state is governed by the ratio $(\Lambda_{2,\bm{b};t}/\Lambda_{1,\bm{b};t})^2$.
    The purification time is defined from the exponential decay of this ratio, which relates to the first two Lyapunov exponents $\varepsilon_1$ and $\varepsilon_2$:
    \begin{align}
        \tau_\mathrm{P}^{-1} = \lim_{t\to\infty}\frac{-\mathbb{E}[\ln(\Lambda_{2,\bm{b};t}/\Lambda_{1,\bm{b};t})^2]}{t} = 2(\varepsilon_2 - \varepsilon_1)
        = 2 \Delta.
        \label{eq:def of purification time}
    \end{align}
    Here, $\Delta$ is the spectral gap introduced in Eq.~\eqref{eq:def of Lyapunov gap}.
\end{enumerate}
We consider $\tau_\mathrm{RC}=\infty$ in the following discussion to avoid the divergence of the second Lyapunov exponent $\varepsilon_2$, which would lead to an ill-defined purification time.

Let us now consider the following non-local quantum circuit on a system of $L$ qubits.
At each timestep, a global Haar-random unitary $\hat{U}\in U(2^L)$ is applied, followed by projective measurements at exactly $pL$ sites chosen randomly.
This measurement projects the system onto a subspace of dimension $2^{(1-p)L}$.
Due to the strong scrambling effect of the global unitaries, this model is always in a mixed phase and does not exhibit a measurement-induced transition.
The rank of the state is fixed to be $r_{\bm{b};t}=2^{(1-p)L}$ for all $t\geq 1$ and almost all unitaries, meaning that the rank collapse time is infinite, $\tau_{\mathrm{RC}}=\infty$.
This allows us to safely study the purification time $\tau_\mathrm{P}$.

We can show that $\tau_\mathrm{P}$ is exponentially long in the system size using RMT.
A key insight is that, due to the Haar randomness of the unitaries, the singular value statistics of the process becomes independent of the specific measurement outcomes.
We can therefore analyze the process for a fixed sequence of projections, effectively bypassing the Born rule average.
The Kraus operators for a $t$-step evolution become equivalent to a product of $t$ random matrices:
\begin{align}
    \hat{\mathsf{M}}_t \sim (\hat{P}\hat{U}_t)(\hat{P}\hat{U}_{t-1})\cdots(\hat{P}\hat{U}_1),
\end{align}
where $\hat{P}$ is a fixed projector onto a $2^{(1-p)L}$-dimensional subspace and each $\hat{U}_t$ is an independent Haar-random unitary.
The Lyapunov exponents $\varepsilon_j$ for such a product of random matrices are known from RMT:
\begin{align}
    \varepsilon_j = \frac{1}{2}[\psi(2^L-j+1) - \psi(2^{(1-p)L}-j+1)],
\end{align}
where $\psi(x)$ is the digamma function\footnote{
The digamma function is the logarithmic derivative of the gamma function $\Gamma(x)$: $\psi(x) = \frac{d}{dx}\ln \Gamma(x) = \Gamma(x)'/\Gamma(x)$.
}.
Using this exact result, we can calculate the purification time:
\begin{align}
    \tau_\mathrm{P} = \frac{1}{2(\varepsilon_2-\varepsilon_1)}
    =\left(\frac{1}{2^{(1-p)L}-1} - \frac{1}{2^L-1}\right)^{-1}
\end{align}
For large $L$, this gives a purification time that is exponentially long in the system size $\tau_\mathrm{P}\sim 2^{(1-p)L}$.

The RMT analysis can be extended to a similar non-local model, but with measurement layers comprising weak measurements.
A single-qubit weak measurement operator can be written as $\hat{M}_{l, \pm}=(\hat{\mathbb{I}}\pm\eta \hat{\sigma}^z_l)/\sqrt{2(1+\eta^2)}$, where $0\leq\eta\leq1$ quantifies the measurement strength.
This ensures $\tau_\mathrm{RC}=\infty$ unless $\eta=1$ and $p=1$.
The Kraus operators corresponding to the whole time evolution $\hat{\mathsf{M}}_{\bm{b};t}$\footnote{
The Kraus operator for one-step evolution is given as explained in the footnote~\ref{f:Kraus-unitary_whole}.
} can be analyzed using similar techniques as the projective case.
Note that since the randomness of unitaries and measurements included in $\bm{b}$ does not affect the singular value statistics of the Kraus operators, we drop the explicit $\bm{b}$-dependence for simplicity.

To analyze the singular value statistics, we focus on the evolution of $\hat{\mathcal{X}}_{t} = 2^{2pLt}\hat{\mathsf{M}}_{t}\hat{\mathsf{M}}^\dagger_{t}$ and its eigenvalues $x_{j}(t)$.
Under an assumption of $pL\eta^2\ll1$, by performing a perturbative expansion in the measurement strength $\eta$ and averaging over the Haar-random unitaries, one obtains a Langevin equation for variables $w_{j}(t)=(\ln x_{j}(t) + \Gamma Nt /4)/2$ with $N=2^L$.
This equation describes the motion of $w_j(t)$ as being driven by random noise while also repelling each other:
\begin{align}
    w_{j}(t+1) - w_{j}(t) = \frac{1}{2}\xi_j(t) + \frac{\Gamma}{4}\sum_{i\neq j}\coth(w_{j}(t) - w_{i}(t)),
    \label{eq:dynamics_w}
\end{align}
where $\xi_j(t)$ is a Gaussian white noise with variance $\langle \xi_i(t_1)\xi_j(t_2)\rangle = \Gamma \delta_{ij}\delta_{t_1t_2}$ and the parameter $\Gamma = e^{-\mathcal{O}(L)}$ for $L\gg1$ sets the strength of both noise and repulsive interaction.

In the continuous time limit of the Langevin equation, we reach a Fokker-Planck equation of the joint probability distribution $P(\vec{w}, s)$ of the variables $\vec{w}(t)=(w_1(t), w_2(t),\ldots)$ at time $s=\Gamma t/8$:
\begin{align}
    \partial_s P(\vec{w},s) = \sum_{j=1}^{N}(-\partial_{w_j}(D_j(\vec{w}) P(\vec{w},s)) + \partial_{w_j}^2 P(\vec{w},s)),
\end{align}
where the drift term is $D_j(\vec{w}) = 2\sum_{i(\neq j)}\coth(w_j-w_i)$.
This equation is a direct analogue of the Dorokhov-Mello-Pereyra-Kumar equation, which describes universal conductance fluctuations in disordered mesoscopic wires~\cite{dorokhov1982transmission, Mello1988macroscopic}.

Remarkably, this Fokker-Planck equation is exactly solvable due to its connection to Calogero-Sutherland models~\cite{le1985isotropic, kulkarni2017emergence, bulchandani2024random}.
The explicit solution for $P(\vec{w}, s)$ is known, allowing for a complete, analytical description of the purification dynamics.
In particular, at late times $t \gg \Gamma^{-1}$, the Lyapunov exponents are given by
\begin{align}
    \varepsilon_j = pL\ln 2 - \frac{\Gamma}{8}(2^L + 2 - 4j),
\end{align}
from which we can read off the purification time as
\begin{align}
    \tau_{\mathrm{P}} = \Gamma^{-1} = e^{\mathcal{O}(L)}
\end{align}
for this model.

\vskip\baselineskip
\noindent\textbf{Universal fluctuations in the mixed phase of fermionic Gaussian systems}

In Ref.~\cite{xiao25universal}, the Fokker-Planck equation for the joint probability distribution of Lyapunov exponents was derived for monitored quantum dynamics of fermionic Gaussian states
\footnote{
“Gaussian states” refer to states whose density matrix can be written as the exponential of a quadratic form of fermionic operators, so that all higher-order correlations factorize according to Wick’s theorem.
See Ref.~\cite{Surace22Fermionic} for a general review and Ref.~\cite{Bravyi05Lagrangian} for measurements that preserve the Gaussianity.
}.
They obtained exact solutions for the Fokker-Planck equation for the mixed phase, which leads to a universal entropy fluctuation.

The system consists of fermions on $L$ sites and starts from the maximally mixed state.
The state is evolved by a unitary operator $\hat{U}_t$, which is generated by a time-dependent random Hamiltonian quadratic in the fermion operators.
Meanwhile, the particle number $\hat{n}_l=\hat{c}_l^\dagger \hat{c}_l$ at each site $l$ is continuously measured where $\hat{c}_l$ ($\hat{c}_l^\dagger$) are fermionic annihilation (creation) operators.
This measurement is described by 
$\hat{P}_t = \exp[\sum_l\{(\hat{n}_l - \langle \hat{n}_l\rangle_{c;t})\sqrt{\gamma} dW_{t,l} - (\hat{n}_l - \langle \hat{n}_l\rangle_{c;t})^2\gamma dt\}]$
where $\gamma$ is the measurement strength and $dW_{t,l}$ is the standard Wiener process (see also Sec.~\ref{sec:Quantum diffusion}).

We now discretize the time as $t=N\Delta t$ so that the unitary dynamics $\hat{U}_{n\Delta t}$ occurs in the interval $((n-1)\Delta t, n\Delta t)$ and the measurement $\hat{P}_{n\Delta t}$ is applied at time $n\Delta t$.
Due to Gaussianity, the information of the state is fully encoded in a single-particle matrix $\mathsf{M}_{\bm{b};t}=P_tU_t\cdots P_{\Delta t}U_{\Delta t}$ where $P_t$ and $U_t$ satisfy $\hat{P}_t \hat{c}_l^\dagger \hat{P}_t^{-1}=\sum_m (P_t)_{lm}\hat{c}_m^\dagger$ and $\hat{U}^\dagger_t \hat{c}_l^\dagger \hat{U}_t=\sum_m (U_t)_{lm} \hat{c}_m^\dagger$, respectively.
Here, the symbols without hats denote the corresponding single-particle matrices, not the many-body operators.
For this dynamics, we construct the effective single-particle Hamiltonian as $H_{\bm{b};t} = -\ln( \mathsf{M}_{\bm{b};t}\mathsf{M}^\dagger_{\bm{b};t})/2t$.
The snapshot single-particle Lyapunov exponents $z_{l, \bm{b};t}$, which are the eigenvalues of $H_{\bm{b};t}$, determine the mixed-state entropy, but we here consider a variable $\zeta_{l, \bm{b};t} = -z_{l, \bm{b};t}t$ instead.

For a simple analytical treatment, the single-particle unitary matrix $U_t$ is modeled as a Haar random $U(L)$ matrix.
We consider the dynamics in an infinitesimal interval $(t, t+\Delta t)$ that renormalizes the probability distribution function $p(\vec{\zeta}_{\bm{b};t},t)$.
Performing a perturbative analysis, we find that the Fokker-Planck equation for $p(\vec{\zeta}_{\bm{b};t},t)$ is derived as
\begin{align} \label{eq:FPKawabata}
    \frac{L+1}{\gamma}\frac{\partial p(\vec{\zeta}_{\bm{b};t},t)}{\partial t} = -\sum_{l=1}^L \frac{\partial [(\mu_{l, \bm{b};t} + \nu_{l, \bm{b};t})p(\vec{\zeta}_{\bm{b};t},t)]}{\partial \zeta_{l, \bm{b};t}} + \frac{1}{2}\sum_{l,m=1}^L\frac{\partial^2[(1+\delta_{lm})p(\vec{\zeta}_{\bm{b};t},t)]}{\partial\zeta_{l, \bm{b};t}\partial\zeta_{m, \bm{b};t}},
\end{align}
where $\mu_{l, \bm{b};t} = \sum_{m\neq l}\coth(\zeta_{l, \bm{b};t}-\zeta_{m, \bm{b};t})$ and $\nu_{l, \bm{b};t} = \sum_m(1+\delta_{lm})\tanh \zeta_{m, \bm{b};t}$.
In order to find its solution, we first need to consider the case of ``forced measurements'' where the probability to obtain a measurement outcome is independent of the pre-measurement state so that $\nu_{l, \bm{b};t}=0$.
In this case, we find the exact solution to the Fokker-Plank equation as
\begin{align}
    p_F(\vec{\zeta}_{\bm{b};t}, t)
    &=\mathcal{N}(t)\left(\prod_{l<m}(\zeta_{l, \bm{b};t}-\zeta_{m, \bm{b};t})\sinh(\zeta_{l, \bm{b};t} - \zeta_{m, \bm{b};t})\right) \nonumber\\
    &\times \exp\left(-\frac{L+1}{2\gamma t}\sum_{l,m}\zeta_{l, \bm{b};t}\left(-\frac{1}{L+1}+\delta_{l m}\right)\zeta_{m, \bm{b};t}\right),
\end{align}
where $\mathcal{N}(t)$ is a normalization constant.
Next, we consider the case where the measurement outcomes obey the Born rule.
Using the fact that the Born probability for the trajectory $\bm{b}_t$ to be realized is proportional to $\prod_{l}\cosh \zeta_{l, \bm{b};t}$~\cite{Cheong2004many}, we can show that the Fokker-Planck equation \eqref{eq:FPKawabata} has the exact solution, 
\begin{align}
    p_B(\vec{\zeta}_{\bm{b};t}, t) = e^{-L\gamma t/2}\left(\prod_{l}\cosh \zeta_{l, \bm{b};t}\right)p_F(\vec{\zeta}_{\bm{b};t}, t).
\end{align}

Next, we consider the R\'enyi mixed-state entropy.
For free fermion systems, it is decomposed into a sum $S_\alpha(\bm{b};t)=\sum_{l=1}^L f_\alpha(\zeta_{l, \bm{b};t})$ with
\begin{align}
    f_\alpha(\zeta)=\frac{1}{1-\alpha}\ln\left[\frac{1}{(1+e^{2\zeta})^\alpha} + \frac{1}{(1+e^{-2\zeta})^\alpha}\right].
\end{align}
Using the exact solution to the Fokker-Planck equation, we can uncover the universal behavior of the entropy in the large-$L$ and short-$t$ limit $1\ll\gamma t\ll L$,
\begin{align}
    \mathbb{E}[S_\alpha(\bm{b};\!t)]\simeq \frac{L}{2\gamma t}\int_{-\infty}^\infty f_{\alpha}(\zeta)d\zeta
    =\frac{\pi^2L}{24\gamma t}\left(1 +\frac{1}{\alpha}\right),
\end{align}
and the universal entropy fluctuations,
\begin{align}
    \mathbb{V}[S_\alpha(\bm{b};\!t)]
    \equiv\mathbb{E}[(S_\alpha(\bm{b};\!t))^2] - \mathbb{E}[(S_\alpha(\bm{b};\!t))]^2
    = \int_{-\infty}^\infty dq \frac{|q|(1-e^{-\pi|q|})}{4\pi^2}
    \left(\int_{-\infty}^\infty f_\alpha(\zeta)e^{-iq\zeta}d\zeta\right)^2.
\end{align}
Specifically, we have $\mathbb{V}[S_2(\bm{b};\!t)]=10\ln2-6\ln\pi$, which was numerically confirmed to appear in other settings, such as those where the unitary evolution is generated by a local Hamiltonian, and those where the measurements are replaced by projective ones~\cite{xiao25universal}.
Reference~\cite{xiao25universal} also argued that the saturating value of the entropy fluctuation is independent of microscopic details, but depends only on fundamental symmetries of the system.

\subsection{Lyapunov spectrum and measurement-induced transitions}
\label{sec:Lyapunov-spectrum_MIPT}
While the previous section focused on analytically tractable models where the Lyapunov spectrum or the mixed-state entropy can be precisely treated, we now shift our focus to the connection between the Lyapunov spectrum and measurement-induced phase transitions, whose analytical treatment remains challenging, primarily through numerical investigations.

This section will explore this connection by addressing two central questions:
(1) How can the universal data of the MIPT be extracted from the Lyapunov spectrum?
(2) How does the qualitative behavior of the Lyapunov spectrum serve as a distinct signature for the measurement-induced phases, and how does this behavior change across the transition?
By addressing these questions, we aim to demonstrate how the Lyapunov spectrum acts as a powerful tool to probe the measurement-induced phases and the underlying universality class.

\subsubsection{Extracting critical data from the Lyapunov spectrum}

The universal properties of an MIPT at its critical point can be revealed by numerically analyzing the Lyapunov spectrum~\cite{zabalo2022operator, chakraborty24charge, aziz2024critical, kumar2024boundary}.
We here particularly focus on its application to the (1+1)-dimensional monitored RQC under the periodic boundary conditions in the spatial direction, as depicted in Fig.~\ref{fig:circuit_schematic}(a).

For an initial state $\hat{\rho}_0$, the probability to obtain a trajectory labeled by $\bm{b}_t$ is given by $p_{\bm{b};t} = \Tr [\hat{\mathsf{M}}_{\bm{b};t} \hat{\rho}_0 \hat{\mathsf{M}}^\dagger_{\bm{b};t}]$.
Using the singular value decomposition of $\hat{\mathsf{M}}_{\bm{b};t}$ in Eq.~\eqref{eq:singula-value_decomposition}, this can be written as 
\begin{align}
p_{\bm{b};t} = \sum_i e^{-2 \varepsilon_{i,\bm{b};t} t} \langle \Phi_{i,\bm{b};t} | \hat{\rho}_0 | \Phi_{i,\bm{b};t} \rangle
\end{align}
where $\varepsilon_{i,\bm{b};t}$ are the Lyapunov exponents and $\ket{\Phi_{i,\bm{b};n}}$ are the corresponding right singular vectors.
As discussed in Sec.~\ref{sec:Lyapunov-analysis}, under the assumption that both irreducibility and purification conditions hold for the CPTP map associated with the monitored RQC, it is dominated by the leading Lyapunov exponent at late times, 
\begin{align}
p_{\bm{b};t} \simeq e^{-2\varepsilon_{1,\bm{b};t}t} \langle \Phi_{1,\bm{b};t} | \hat{\rho}_0 | \Phi_{1,\bm{b};t} \rangle.
\end{align}

The core idea of Ref.~\cite{zabalo2022operator} is to identify each trajectory as a two-dimensional classical statistical mechanics model, which is defined through the ``partition function'' $Z_{\bm{b};t} \equiv p_{\bm{b};t}$\footnote{
Writing the full circuit evolution as $\hat{\mathsf{M}}_{\bm{b};t} = \hat{M}_{b_t} \cdots \hat{M}_{b_2} \hat{M}_{b_1}$, each time slice $\hat{M}_{b_n}$ may be viewed as a transfer matrix for a two-dimensional statistical mechanics model with quenched randomness drawn from $\{ \bm{b} \}$.
The Born probability $p_{\bm{b};t} = \Tr [\hat{\mathsf{M}}_{\bm{b};t} \hat{\rho}_0 \hat{\mathsf{M}}^\dagger_{\bm{b};t}]$ represents two layers of transfer matrix evolution, one with forward evolution by $\hat{M}_{b_n}$ and another with backward evolution by $\hat{M}_{b_n}^\dagger$, which are glued at one side by $\hat{\rho}_0$ while summed over all possible states at the other side.
Such objects are regarded as the partition functions of statistical mechanics models in the literature (see, e.g., Ref.~\cite{Jian2020measurement}).
One may also view $p_{\bm{b};t}$ as a partition function in the Keldysh formalism~\cite{Kamenev2023field}, although the evolution does not preserve the trace of $\hat{\rho}_0$.
}.
The free energy of this classical model is given by 
\begin{align} \label{eq:FreeEnergy}
F_{\bm{b};t} 
= -\ln Z_{\bm{b};t}
\simeq 2\varepsilon_{1,\bm{b};t} t -\ln \langle \Phi_{1,\bm{b};t} | \hat{\rho}_0 | \Phi_{1,\bm{b};t} \rangle.
\end{align}
As discussed in Sec.~\ref{sec:rank-M_purification}, since the right singular vectors $\ket{\Phi_{i,\bm{b};t}}$ asymptotically become $t$-independent, one can safely neglect the second term of Eq.~\eqref{eq:FreeEnergy} for a sufficiently long time.
Averaging the free energy $F_{\bm{b};t}$ over the trajectories yields,
\begin{align}
\mathbb{E}[F_{\bm{b};t}] 
= \sum_{\bm{b}_t}p_{\bm{b};t} F_{\bm{b};t} 
= -\sum_{\bm{b}_t} p_{\bm{b};t} \ln p_{\bm{b};t},
\end{align}
which is nothing but the Shannon entropy of the trajectories and behaves as $\mathbb{E}[F_{\bm{b};t}] \simeq 2\sum_{\bm{b}_t} p_{\bm{b};t} \varepsilon_{1,\bm{b};t} t$ at late times.

According to CFT, the free energy density $f_1(L)\equiv \mathbb{E}[F_{\bm{b};t}]/A$ with an effective spacetime area $A=\alpha L t$, where $\alpha$ is an anisotropy parameter\footnote{
The anisotropy parameter $\alpha$ in the effective area $A$ relates the scales of space and time at the critical point.
It can be determined by finding the time $t_*$ at which spatial and temporal correlations become equal, which leads to the relation $\alpha = (L/\pi t_*)\ln(1+\sqrt{2})$.
A numerical procedure to determine $\alpha$ using the mutual information was proposed in Refs.~\cite{zabalo2022operator, chakraborty24charge}.
}, 
is predicted to obey the following finite-size scaling form at the critical point,
\begin{align}
    f_1(L) = f_1(L=\infty) - \frac{\pi c_{\mathrm{eff}}}{6L^2} + \cdots.
\end{align}
Here, $c_{\mathrm{eff}}$ is a universal number called the effective central charge.
We can also consider generalizations of the free energy density to the higher Lyapunov exponents $\varepsilon_{i,\bm{b};t}$ with $i \geq 2$, which are given by $f_i(L) \equiv 2\mathbb{E}[\varepsilon_{i,\bm{b};t}] / (\alpha L)$.
These free energy densities are expected to obey the finite-size scaling form,
\begin{align}
    f_{i}(L)-f_1(L) = \frac{2\pi x_i^{\mathrm{typ}}}{L^2},
\end{align}
from which we can extract the scaling dimensions $x_i^{\mathrm{typ}}$ of operators in the underlying CFT.

We end this section with several caveats.
When the CPTP map corresponding to the monitored RQC obeys the irreducibility condition, we do not actually need the average over the trajectories to compute the free energy densities; as discussed in Sec.~\ref{sec:typical-convergence_Lyapunov-spectrum}, the Lyapunov exponents converge to values independent of $\bm{b}$ for a typical trajectory.
However, if we consider a finite-size system of length $L$ subject to projective measurements at each site with probability $p$, which is a common setup for the MIPT, there is always the possibility that the Kraus operator $\hat{\mathsf{M}}_{\bm{b};t}$ becomes rank one during the evolution and all higher Lyapunov exponents $\varepsilon_{i,\bm{b};t}$ with $i \geq 2$ diverge thereafter [see also the discussion around Eq.~\eqref{eq:RankCollapseTime}].
A typical situation is that all sites are measured in a single time step, which happens with probability $p^L$.
Although such events may only occur with probability exponentially small in $L$, they can be obstacles in computing higher Lyapunov exponents, as they often require a long convergence time.
In order to successfully perform their finite-size scaling analysis as described above, the higher Lyapunov exponents must converge in a time long enough but much shorter than $e^{\mc{O}(L)}$. 
For these cases, the average over the trajectories may help to improve the convergence.
As we will see in the next section, we need not be worried about the timescale for the convergence when we consider weak measurements.

\subsubsection{Spectral behavior in measurement-induced  phases and transitions}
The Lyapunov spectrum probes not only the universal properties of the MIPT, but also the phases themselves.
Reference~\cite{mochizuki2025measurement} studied the behaviors of the Lyapunov spectrum within the measurement-induced phases for a $(1+1)$-dimensional RQC with local Haar unitary gates and generalized measurements of local qubits, as depicted in Fig.~\ref{fig:Mochizuki25_circuit}(a).
The Kraus operator corresponding to the single-site generalized measurement with an outcome $b=\pm 1$ is given by $\hat{M}_{l,b} = (\hat{\mathbb{I}} + b\eta \hat{\sigma}^z_l)/\sqrt{2(1+\eta^2)}$, where the parameter $0\leq\eta\leq1$ tunes the strength of the measurement.
The Kraus operator representing the whole monitored circuit is denoted by $\hat{\mathsf{M}}_{\bm{b};t}$\footnote{
$\hat{\mathsf{M}}_{\bm{b};t}$ is constructed as in the footnote~\ref{f:Kraus-unitary_whole} with $p=1$.
}, where the $\eta$-dependence is omitted for ease of notation.
We write the singular value decomposition of the Kraus operator as $\hat{\mathsf{M}}_{\bm{b};t} = \sum_{i=1}^{2^L} \Lambda_{i,\bm{b};t} | \Psi_{i,\bm{b};t} \rangle \langle \Phi_{i,\bm{b};t} |$, where $\Lambda_{i,\bm{b};t}$ are the singular values arranged as $\Lambda_{1,\bm{b};t} \geq \Lambda_{2,\bm{b};t} \geq \cdots \geq 0$ and $\ket{\Psi_{i,\bm{b};t}}$ and $\ket{\Phi_{i,\bm{b};t}}$ are the corresponding left and right singular vectors, respectively [see also Eq.~\eqref{eq:singula-value_decomposition}].

In the long-time limit, an initial state $\ket{\psi_0}$ converges to the left singular vector corresponding to the leading singular value $\Lambda_{1,\bm{b};t}$, $\hat{\mathsf{M}}_{\bm{b};t}\ket{\psi_0} \sim \Lambda_{1,\bm{b};t}\ket{\Psi_{1,\bm{b};t}}$, if we assume $\Lambda_{1,\bm{b};t} > \Lambda_{2,\bm{b};t}$\footnote{
We can prove that this monitored circuit satisfies the purification condition discussed in Sec.~\ref{sec:pur} unless $\eta=0$, so that $\Lambda_{1,\bm{b};t} > \Lambda_{2,\bm{b};t}$ holds for $t \to \infty$ according to the results in Sec.~\ref{sec:spectral_gap_purification}.
If we decompose the full Kraus operator as $\hat{\sf{M}}_{\bm{b};t} = \hat{M}_{b_t} \cdots \hat{M}_{b_2} \hat{M}_{b_1}$, each time slice $\{ \hat{M}_{b_i} \}_{b_i}$, which consists of one layer of random two-site unitary gates $\hat{U}_{l,l+1}$ and another layer of generalized measurements $\hat{M}_{l,b}$, is also a set of Kraus operators.
By successively taking the sums over the measurement outcomes at later times, one can reduce the condition in Eq.~\eqref{eq:Benoist-purification} for $\hat{\sf{M}}_{\bm{b};t}$ to that for the first time slice $\hat{M}_{b_1}$.
Now, the time slice $\hat{M}_{b_1}$ is a tensor product of $\hat{X}_{b,b',U_{1,2}}=\hat{M}_{1,b} \hat{M}_{2,b'} \hat{U}_{1,2}$. 
For any $\hat{U}_{1,2}$, $\hat{X}^\dagger_{b,b',U_{1,2}} \hat{X}_{b,b',U_{1,2}}$ can be easily diagonalized, and their eigenvalues are given by $(1 \pm b\eta)^2 (1 \pm b'\eta)^2/(2(1+\eta^2))^2$, with $b,b'\in\{+1,-1\}$.
While each $\hat{M}_{b_1}$ has degenerate eigenvalues, the set of $\hat{M}_{b_1}$ over all possible measurement outcomes has no common eigenspace with dimension greater than one for $0 < \eta \leq 1$.
This implies that an orthogonal projection operator $\hat{\mc{Q}}$ that satisfies $\hat{\mc{Q}} \hat{M}_{b_1}^\dagger \hat{M}_{b_1} \hat{\mc{Q}} = \zeta_{b_1} \hat{\mc{Q}}$ with some $\zeta_{b_1} \geq 0$ for almost all $b_1$ cannot have $\mathrm{rank}[\hat{\mc{Q}}] \geq 2$.
A similar argument also applies to the monitored Kitaev-Majorana circuit discussed in Sec.~\ref{sec:topology_MIPT}.
}.
The long-time state $\ket{\Psi_{1,\bm{b};t}}$ can be regarded as the ``ground state" of the effective Hamiltonian
\begin{align} \label{eq:EffHamMochi}
    \hat{H}_{\bm{b};t} = -\frac{1}{2t}\ln \hat{\mathsf{M}}_{\bm{b};t}\hat{\mathsf{M}}_{\bm{b};t}^\dagger
\end{align}
introduced in Eq.~\eqref{eq:effective-Hamiltonian}, with the ground state energy $\varepsilon_{1,\bm{b};t}=-\ln(\Lambda_{1,\bm{b};t})/t$, which is nothing but the first Lyapunov exponent. 
We note that these monitored systems are irreducible, i.e., (3) of Theorem \ref{thm:IrreducibilityByKraus} is satisfied~\cite{mochizuki2025transitions} for each Haar and Clifford case. 
Therefore, the Lyapunov exponents are independent of $\bm{b}$ and $\ket{\psi_0}$, as discussed in Sec.~\ref{sec:Lyapunov-analysis}.

\begin{figure}[!h]
\centering\includegraphics[width=\linewidth]{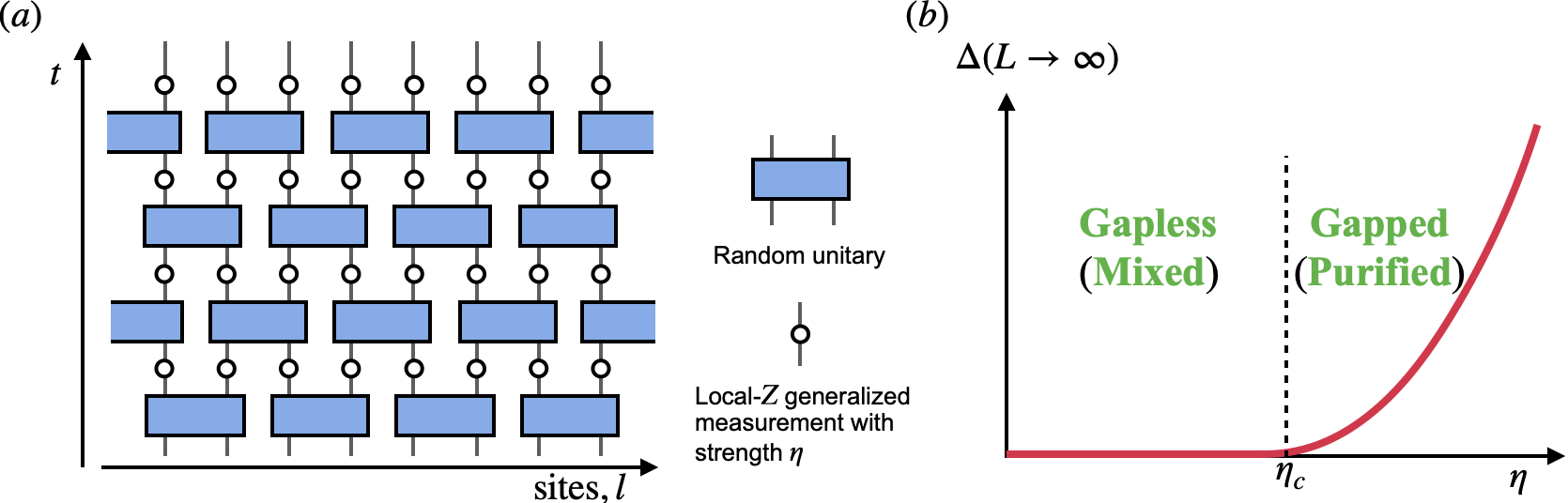}
\caption{
(a) A monitored circuit with generalized measurements.
The two-site Haar random unitary gates are arranged in a brickwork manner.
The generalized measurements with strength $\eta$ are applied at all sites followed by a unitary layer.
(b) Schematic of the long-time values of the first Lyapunov gap $\Delta(L\to\infty)$ against the measurement strength $\eta$.
The location of the purification transition is denoted by $\eta_c$.
The figures were created based on Ref.~\cite{mochizuki2025measurement}.
}
\label{fig:Mochizuki25_circuit}
\end{figure}

Since the rank of $\hat{\mathsf{M}}_{\bm{b};t}$ never becomes one except for the projective case $\eta=1$, the higher Lyapunov exponents $\varepsilon_{i,\bm{b};t}=-\ln(\Lambda_{i,\bm{b};t})/t$ can be safely computed by the numerical procedure in Sec.~\ref{sec:numerical_procedure}.
Since the Lyapunov spectrum converges to values independent of the trajectory $\bm{b}$ for sufficiently long times almost surely, they can, in practice, be computed by running the monitored circuit only once for each $\eta$. 
Hence, we drop the subscript $\bm{b};t$ from $\varepsilon_{i,\bm{b};t}$ hereafter and focus on the converged values of the Lyapunov exponents $\varepsilon_1(L)\leq\varepsilon_2(L)\leq\cdots$ for finite-size systems of length $L$.

Our main focus is the first Lyapunov gap, $\Delta(L) = \varepsilon_2(L) - \varepsilon_1(L)$.
From the numerical results, the Lyapunov gap shows a transition from a gapless phase where the gap exponentially closes in the system size $L$, $\Delta(L)=\exp[-\mathcal{O}(L)]$, to a gapped phase where the gap becomes almost independent of $L$, $\Delta(L)= \mathcal{O}(L^0)$ as $\eta$ increases:
\begin{align}
    \Delta(L)=\begin{cases}
        e^{-\mathcal{O}(L)}&(\eta<\eta_c)\\
        \mathcal{O}(L^0)&(\eta>\eta_c)
    \end{cases}.
\end{align}
The behavior of $\Delta(L\to\infty)$ is schematically shown in Fig.~\ref{fig:Mochizuki25_circuit}(b), and it exhibits the gap closing/opening transition. 
In the gapless phase with $\eta<\eta_c$, the low-lying spectrum also exhibits the exponential decay $\varepsilon_i(L) - \varepsilon_1(L)=\exp[-\mathcal{O}(L)]$, which was numerically confirmed up to $i=10$. 
Regarding the entire spectrum, it was analytically shown that the spectral width always exhibits the bound $\varepsilon_{2^L}(L) - \varepsilon_1(L)\leq\mathcal{O}(L)$, which means that level spacings $\varepsilon_{i+1}(L) - \varepsilon_i(L)$ are exponentially narrow with respect to $L$ for almost all $i$~\cite{mochizuki2025measurement}.

Notably, this spectral transition in fact captures the purification transition.
Remember that different initial states typically yield the same long-time Lyapunov spectrum, owing to the irreducibility. 
This means that the first Lyapunov gap is $\Delta(L)$ even when the initial state is the maximally mixed state
$\hat{\rho}_0 = \hat{\mathbb{I}}/2^L$.
In this case, the gap $\Delta(L)$ gives the purification time by $\tau_\mathrm{P}^{-1}(L)=2\Delta(L)$ as in Eq.~\eqref{eq:def of purification time}.
Then, the spectral transition can be regarded as the transition of the purification time $\tau_\mathrm{P}(L)$ between a mixed phase with $\tau_\mathrm{P}(L)= \exp[\mathcal{O}(L)]$ and a purified phase with $\tau_\mathrm{P}(L)=\mathcal{O}(L^0)$.
Numerically, the location of the spectral transition turns out to be almost the same as that of the entanglement transition, which is consistent with the conjecture that the entanglement transition is the same as the purification transition in the general monitored RQCs.
As pointed out in Ref.~\cite{mochizuki2025measurement}, the coincidence between the entanglement and the spectral transition is analogous to what has been found at the ground-state quantum phase transition in equilibrium systems~\cite{Eisert2010colloquium, laflorencie2016quantum}\footnote{
In the ground state of a Hamiltonian with short-range interactions in one dimension, a gapped phase exhibits a finite energy gap $\Delta(L) = \mc{O}(L^0)$ and area-law entanglement $S_X = \mc{O}(|X|^0)$, whereas a gapless phase typically exhibits a polynomially decaying energy gap $\Delta(L) = \mc{O}[1/\mathrm{poly}(L)]$ and logarithmic violation of the area-law entanglement $S_X = \mc{O}(\ln |X|)$.
The latter should be contrasted with the gapless phase in the monitored systems where $\Delta(L) = \exp[-\mc{O}(L)]$ and $S_X = \mc{O}(|X|)$.
This difference might stem from possible long-range interactions and inhomogeneity in the effective Hamiltonian defined through Eq.~\eqref{eq:EffHamMochi}.
}.
This coincidence was confirmed even when the unitary gates may possess spacetime translation symmetries~\cite{mochizuki2025transitions}.

\subsection{Measurement-induced topology}
\label{sec:topology_MIPT}

Topology is one of the most fundamental concepts for characterizing stable phases of matter~\cite{wen2017colloquium}.
For local Hamiltonians with a finite excitation gap, ground states belonging to different topological phases cannot be adiabatically deformed into one another without closing the gap, a distinction maintained by their discrete topological invariants~\cite{chen2010local}.

A key manifestation of this occurs in symmetry-protected topological (SPT) phases~\cite{gu2009tensor, pollmann2010entanglement, chen2011classification, fidkowski2011topological, lu2012theory, chen2013symmetry}, where nontrivial bulk topology guarantees the presence of gapless boundary states.
These states are robust against any perturbations that preserve the system's symmetry as long as the bulk gap is open.
This phenomenon, known as the bulk-edge correspondence, profoundly influences the physical properties of topological phases and has been intensively researched over the past decades~\cite{hasa2010colloquium, qi2011topological, senthil2015symmetry, witten2016fermion, chiu2016classification}.

Despite significant advancements in equilibrium systems, the role of topology in out-of-equilibrium systems driven by environmental interactions remains not fully understood.
Specifically, how topology manifests itself in monitored quantum systems is still elusive.

Concerning this issue, previous studies have investigated measurement-induced topological phase transitions separating trivial and topological area-law phases, for example, in spin systems with $\mathbb{Z}_2\times \mathbb{Z}_2$ symmetry~\cite{Lavasani2021measurement, klocke2022topological, kuno2023production, morralyepes2023detecting} and in fermion systems with particle-hole symmetry~\cite{Nahum2020entanglement, Merritt2023entanglement, kells2023topological, nehra2025controlling, loio2023purification, Klocke23majorana, fava2023nonlinear, pan2024topological, xiao2024topology, Bhuiyan2025freefermion, oshima2025topology}.
For one-dimensional monitored systems under open boundary conditions (OBC), the existence of a topological phase can be probed by the topological entanglement entropy or the purification dynamics; the latter can be characterized by the exponentially-long-time survival of the entanglement entropy of an ancilla system initially entangled with the main system.
They capture the quantum information encoded in the edge states, if any, that are protected by the bulk topology.

\begin{figure}[!h]
\centering\includegraphics[width=\linewidth]{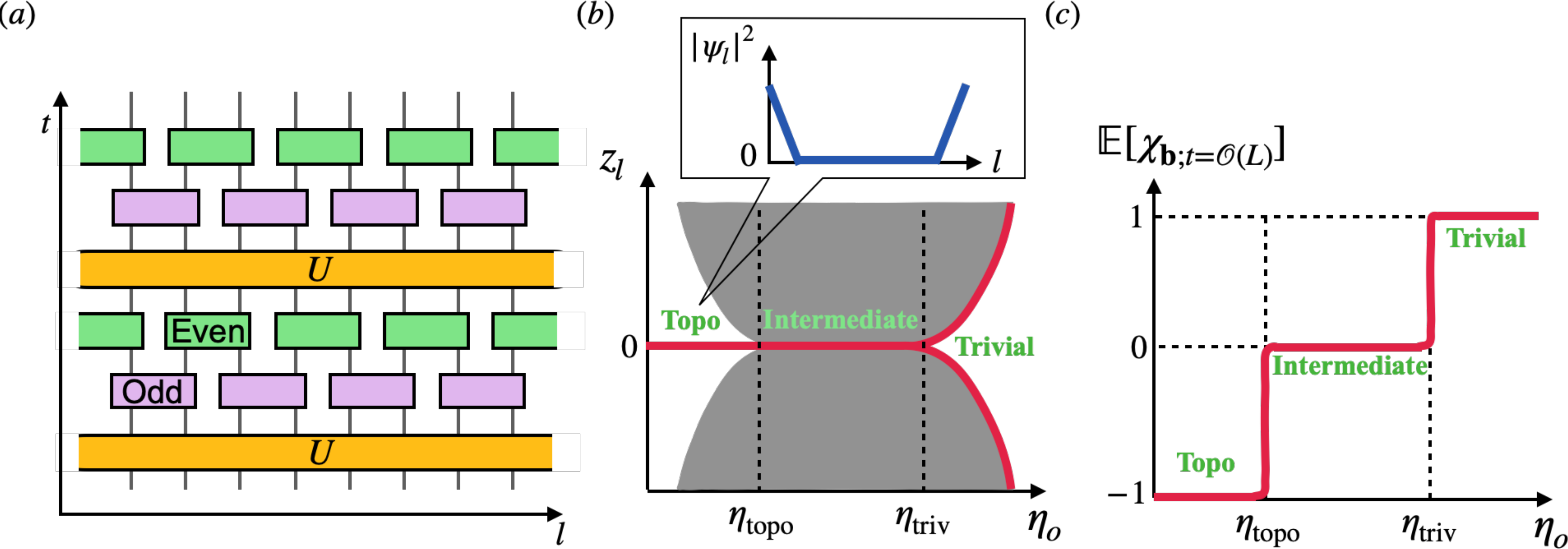}
\caption{
(a) A weakly monitored Kitaev Majorana circuit.
The purple (green) two-site gate represents the generalized measurement of the parity of the neighboring Majorana pair on an odd (even) bond.
The orange gate is the unitary evolution by a local Hamiltonian that preserves particle-hole symmetry, say, the unitary evolution by the Kitaev chain Hamiltonian.
We switch off the unitary evolution when we consider the measurement-only circuit.
(b) The single-particle Lyapunov exponents $z_l$ for the same circuit with unitary evolution by the Kitaev chain Hamiltonian under the OBC.
The red lines highlight the two middle exponents $z_L$ and $z_{L+1}$, and the gray bands correspond to other exponents.
The locations of the transition from a topological area-law phase to an intermediate subvolume-law phase and the one from the intermediate phase to a trivial area-law phase are marked by $\eta_{\mathrm{topo}}$ and $\eta_{\mathrm{triv}}$, respectively.
The inset depicts the average of the squared single-particle wave function amplitude associated with the zero modes $z_L$ and $z_{L+1}$ in the topological phase, showing their localization at edges.
(c) The bulk dynamical topological invariant at $t=\mathcal{O}(L)$ averaged over trajectories for sufficiently large $L$.
The figures were created based on Ref.~\cite{oshima2025topology}.
}
\label{fig:Oshima25_results}
\end{figure}

While the above signatures, e.g., topological entanglement entropy, are suggestive, these are not responsible for the underlying symmetry protection of the topological phase.
Reference~\cite{oshima2025topology} addressed this issue.
It considered a weakly monitored variant of the Kitaev-Majorana chain as shown in Fig.~\ref{fig:Oshima25_results}(a).
The circuit is defined on a one-dimensional chain of $2L$ Majorana fermions $\hat{\gamma}_l = \hat{\gamma}_l^\dagger$, which obey the anticommutation relations $\{\hat{\gamma}_l, \hat{\gamma}_{l'}\}=2\delta_{ll'}$. 
The Majorana fermions are subject to repeated weak measurements of the parity of Majorana pairs on odd and even bonds with the strength $0\leq\eta_o<1$ and $\eta_e=1-\eta_o$, which are described by the Kraus operators,
\begin{align}
    \hat{M}_{l,b}^o = \frac{e^{-ib\theta_o\hat{\gamma}_{2l-1}\hat{\gamma}_{2l}}}{\sqrt{2\cosh2\theta_o}}, \quad
    \hat{M}_{l,b}^e = \frac{e^{-ib\theta_e\hat{\gamma}_{2l}\hat{\gamma}_{2l+1}}}{\sqrt{2\cosh2\theta_e}},
    \qquad\mathrm{with}\quad
    \theta_{o/e}=\tanh^{-1}(\eta_{o/e}),
\end{align}
respectively.
Here, $b=\pm1$ is the measurement outcome randomly determined through the Born probability $p_{l,b}^{o/e} = \left\| \hat{M}_{l,b}^{o/e}\ket{\psi}\right\|^2$ for the normalized pre-measurement state $\ket{\psi}$.
The unitary time evolution in one time step is described by $\hat{U}_{\mathrm{Kitaev}}=e^{i\hat{H}_{\mathrm{Kitaev}}}$ with $\hat{H}_{\mathrm{Kitaev}}=iJ\sum_{l=1}^{2L-1}\hat{\gamma}_l\hat{\gamma}_{l+1}+iJ'\hat{\gamma}_{2L}\hat{\gamma}_1$, where $J$ and $J'$ are real coupling constants.
The open (OBC), periodic (PBC), and antiperiodic (APBC) boundary conditions are specified by $J'=0$, $J$, and $-J$, respectively.
This circuit preserves the particle-hole symmetry $\hat{\gamma}_l \to -\hat{\gamma}_l$, as both measurements and unitary evolution are described by operators bilinear in the Majorana fermions.

Similar monitored circuits with Majorana fermions have been studied in the literature~\cite{Nahum2020entanglement, Merritt2023entanglement, kells2023topological, nehra2025controlling, loio2023purification, Klocke23majorana, fava2023nonlinear}. 
It has been shown that there is a direct transition between two different area-law phases in the absence of unitary evolution: 
one exhibits nontrivial topological signatures and is realized when the measurements on even bonds are dominant.
The other is topologically trivial and realized when the measurements on odd bonds are dominant.
When we add unitary dynamics that preserves the particle-hole symmetry and Gaussianity (see below), such as $\hat{U}_{\mathrm{Kitaev}}$, a subvolume-law entanglement phase where the entanglement entropy $\mathbb{E}[S_{L/2,\alpha}(\bm{b};t)]$ scales as $(\ln L)^2$ emerges between the two area-law phases, as predicted by an effective field theory using a non-linear sigma model~\cite{fava2023nonlinear}.
To summarize, in the circuit model described above, we have three distinct measurement-induced phases; a topological area-law phase ($\eta_o<\eta_{\mathrm{topo}}$), a subvolume-law entanglement phase ($\eta_{\mathrm{topo}}<\eta_o<\eta_{\mathrm{triv}}$), and a trivial area-law phase ($\eta_{\mathrm{triv}}<\eta_o$).

Reference~\cite{oshima2025topology} investigated the measurement-induced topological phase transition in the above model, not only from the entanglement-wise quantities but also from the spectrum, bulk SPT invariant, and Majorana edge modes.
For a trajectory up to time $t$ and the corresponding Kraus operator $\hat{\mathsf{M}}_{\bm{b};t}$, we can define the effective Hamiltonian $\hat{H}_{\bm{b};t}=-\ln(\hat{\mathsf{M}}_{\bm{b};t}\hat{\mathsf{M}}_{\bm{b};t}^\dagger)/2t$, as introduced in Eq.~\eqref{eq:effective-Hamiltonian}.
In the current model, the circuit preserves the Gaussianity so that the full information of the dynamics is encoded in a single-particle matrix\footnote{
The circuit preserves the Gaussianity when it maps a Gaussian state $\hat{\rho} \propto \exp (-i \vec{\gamma}^\mathsf{T} G \vec{\gamma})$ to another Gaussian state $\hat{\rho}' \propto \exp(-i\vec{\gamma}^\mathsf{T} G' \vec{\gamma})$ with real antisymmetric matrices $G$ and $G'$.
In the current model, this is ensured because the corresponding Kraus operator $\hat{\mathsf{M}}_{\bm{b};t}$ is generated by operators bilinear in the Majorana fermions.
}.
Hence, we define the effective single-particle Hamiltonian matrix $H_{\bm{b};t}$ through
\begin{align}
    \hat{H}_{\bm{b};t} = -i\vec{\gamma}^\mathsf{T}H_{\bm{b};t}\vec{\gamma},
\end{align}
where $H_{\bm{b};t}$ is the $2L\times2L$ real antisymmetric matrix written in the basis of Majorana fermions $\vec{\gamma}=(\hat{\gamma}_1, \ldots,\hat{\gamma}_{2L})^{\sf{T}}$.
This enables us to efficiently compute the single-particle Lyapunov spectrum $z_1\leq z_2\leq\cdots\leq z_{2L}$ from the eigenvalues of $H_{\bm{b};t}$.
For the bulk SPT invariant, Ref.~\cite{oshima2025topology} proposed the monitored version of Kitaev's $\mbb{Z}_2$ invariant~\cite{Kitaev01unpaired}:
\begin{align}
    \chi_{\bm{b};t} = P(\hat{H}_{\bm{b};t}^{\mathrm{PBC}})P(\hat{H}_{\bm{b};t}^{\mathrm{APBC}}),
\end{align}
where $P(\hat{H})$ is the ground-state parity of the Hamiltonian $\hat{H}$.
Here, $\hat{H}_{\bm{b};t}^{\mathrm{PBC}/\mathrm{APBC}}$ are the effective Hamiltonians along with the trajectory label $\bm{b}_t$ under the PBC and APBC, respectively, and defined as follows.
First, we evolve the system by circuit $\hat{\mathsf{M}}_{\bm{b};t}$ under the PBC with the Born rule, which gives the effective Hamiltonian $\hat{H}^\mathrm{PBC}_{\bm{b};t}$.
Next, we make an identical copy of the PBC circuit by postselecting all measurement outcomes but flipping the outcomes of the copied circuit only at the boundary bond between $l=1$ and $2L$; this yields the APBC circuit and the corresponding effective Hamiltonian $\hat{H}^\mathrm{APBC}_{\bm{b};t}$.
In this way, we can realize the APBC circuit \textit{with respect to} the given PBC circuit\footnote{
Note that, independent evolution of the two circuits with the PBC and APBC by the Born rule is meaningless since the measurement histories become different between them, yielding the unrelated Hamiltonians; we cannot make a precise meaning and distinction of PBC or APBC in this way.
}.
Then, cutting off the circuits up to time $t$, the dynamical topological invariant is obtained by comparing the ground state parities of $\hat{H}_{\bm{b};t}^{\mathrm{PBC}}$ and $\hat{H}_{\bm{b};t}^{\mathrm{APBC}}$\footnote{
The basic ideas presented here are generally applicable to many-body systems, whereas the numerical calculations in Ref.~\cite{oshima2025topology} employed a slightly different method that exploits the Gaussianity of free fermion systems.
See that reference for details.
}.

Figure~\ref{fig:Oshima25_results}(b) schematically shows that the single-particle Lyapunov spectrum of a typical trajectory under the OBC exhibits a bulk gap in both area-law phases.
In the topological area-law phase, however, two Majorana zero modes as $z_L \sim z_{L+1} \to 0$ for $L \to \infty$ appear.
It was numerically confirmed that these zero modes are localized at the edges, as schematically shown in the inset of Fig.~\ref{fig:Oshima25_results}(b), similarly to the ground state in the topological phase of the static Kitaev chain.
On the other hand, the bulk gap closes in the subvolume-law phase.
For the spectra under the PBC, the numerical results show that the two area-law phases are both gapped with the finite-size gap scaling $\Delta(L)=\mathcal{O}(L^0)$, while the two transition points $\eta_o=\eta_{\mathrm{topo}}$ and $\eta_{\mathrm{triv}}$ both exhibit the critical scaling $\Delta(L)= \mathcal{O}(L^{-1})$.
Interestingly, the gap decays faster than $L^{-1}$ inside the subvolume-law phase, meaning that this phase is gapless but has no straightforward counterpart in the ground states of local Hamiltonians with the particle-hole symmetry alone.

Turning to the topological invariant $\chi_{\bm{b};t}$, it typically converges to $-1$ in the topological area-law phase for $\eta_o<\eta_\mathrm{topo}$ and $+1$ in the trivial area-law phase for $\eta_o>\eta_\mathrm{triv}$ within a short time $t = \mathcal{O}(L)$,
\begin{align}
    \chi_{\bm{b};t}=\begin{cases}
        -1&(\eta_o<\eta_\mathrm{topo})\\
        +1&(\eta_o>\eta_\mathrm{triv})
    \end{cases}.
\end{align}
This numerically ensures the bulk-edge correspondence in the gapped measurement-induced topological phases, i.e., the nontrivial (trivial) topological number $\chi_{\bm{b};t}=-1$ ($\chi_{\bm{b};t}=+1$) leads to the presence (absence) of topological Majorana zero modes.

The gapless subvolume-law phase requires a careful discussion since, at least for finite size systems, the dynamical topological invariant flows towards $\pm1$ in time for general parameters even inside the gapless phase.
Here, we focus on the fact that the convergence time of the dynamical topological invariant $\chi_{\bm{b};t}$ respects the relaxation time scale $\tau_\mathrm{P}(L) \sim 1/\Delta(L)$.
Therefore, the convergence time of $\chi_{\bm{b};t}$ is less than $\mathcal{O}(L)$ for the gapped area-law phases and greater than $\mathcal{O}(L)$ for the gapless subvolume-law phase.
Then, the trajectory-averaged dynamical topological invariant $\mathbb{E}[\chi_{\bm{b};t}]$ at time $t=\mathcal{O}(L)$ is expected to separate the three phases.
In fact, it takes $-1$ in the topological area-law phase for $\eta_o<\eta_{\mathrm{topo}}$, $+1$ in the trivial area-law phase for $\eta_o>\eta_{\mathrm{triv}}$, and $0$ in the subvolume-law phase for $\eta_{\mathrm{topo}}<\eta_o<\eta_{\mathrm{triv}}$ in a sufficiently large system.
The results are schematically shown in Fig.~\ref{fig:Oshima25_results}(c).

The model studied above is a fermionic Gaussian dynamics, and in the convention of Ref.~\cite{xiao2024topology}, it belongs to class D in the tenfold Altland-Zirnbauer (AZ) symmetry class~\cite{altland1997nonstandard}.
The reference~\cite{xiao2024topology} performed the AZ classification for monitored free fermionic systems based on the symmetries of the single-particle time evolution operators.
The classification scheme varies depending on which operators are considered and how the relevant symmetries are defined.
In other works, based on the correspondence between a monitored free fermion dynamics and a static Anderson localization problem~\cite{Jian22criticality, fava2023nonlinear, Poboiko23Thoery}, the AZ classification was instead performed by analyzing the symmetries of the single-particle transfer matrix that describes the corresponding Anderson localization problem~\cite{pan2024topological, Bhuiyan2025freefermion}.
In the latter convention, the above model belongs to class DIII.
Reference~\cite{Bhuiyan2025freefermion} further proposed a general framework to treat symmetries in many-body systems and demonstrated that it consistently reproduces the tenfold AZ symmetry classification in the free fermion limit.
A comprehensive classification of topological phases in interacting monitored dynamics, however, remains elusive.

\section{Conclusion and outlook}
\label{sec:conclusion}
We reviewed the basics and various aspects of monitored quantum systems, focusing on typical behaviors, spectral properties, and many-body phases. 
Regarding the CPTP dynamics averaged over measurement outcomes, we provided pedagogical introductions about various representations, discrete- and continuous-time descriptions, spectral features, and properties of steady states. 
We explained that long-time behaviors of quantum states are profoundly related to the irreducibility and primitivity of the averaged dynamics, and these properties can be examined through the algebra of Kraus operators in the discrete dynamics or jump operators in the continuous dynamics. 
Regarding quantum trajectories, where outcomes of all measurements are recorded, we discussed that the jump statistics, nonlinear observables, ergodicity, purification, and Lyapunov spectrum capture intriguing features of the random trajectories that may not be seen in the averaged CPTP dynamics. 

The aforementioned features of the CPTP dynamics have huge effects on the corresponding quantum trajectories. 
For example, quantum trajectories exhibit the ergodicity of linear observables if the steady state for the averaged dynamics is unique due to, e.g., irreducibility. 
In addition, the irreducibility of the CPTP dynamics, combined with the purification of typical trajectories, leads to the ergodicity of nonlinear observables and the typical convergence of the Lyapunov spectrum. 
We also introduced measurement-induced phase transitions as representative phenomena that occur in quantum trajectories but are usually invisible in the averaged CPTP dynamics; we  explained that the Lyapunov spectrum, purification timescales, and nonlinear observables such as the entanglement entropy, computed only from quantum trajectories, are good indicators to detect measurement-induced phase transitions. 

There are several future directions that should be intriguing. 
One direction is to explore the behaviors of quantum trajectories and the CPTP dynamics in intermediate timescales:
On one hand, the ground states of the effective Hamiltonians describe the behaviors of quantum trajectories in the long-time regime $t \gg 1/\Delta$ where $\Delta$ is the Lyapunov spectral gap, as discussed in Sec.~\ref{sec:Lyapunov-analysis}. 
On the other hand, the steady states of the CPTP dynamics describe the behaviors of averaged quantum states in the long-time regime $t \gg 1/\Delta'$ where $\Delta'$ is the spectral gap of the corresponding CPTP map, as discussed in Chapter~\ref{sec:CPTPspectra}.
However, it is highly nontrivial how other excited states and the corresponding spectrum contribute to their behaviors in intermediate timescales $t \ll 1/\Delta$ or $t \ll 1/\Delta'$.
For example, exploring the relations between the excitation spectrum and the convergent behaviors of time averages of observables should be interesting. 

The second direction is to study the physical consequences of breaking ``nice" conditions, such as the irreducibility, ergodicity, and/or typical purification of quantum trajectories. 
It is known that in reducible monitored quantum systems, the Hilbert space may be decomposed into several orthogonal subspaces, such as the decaying subspace and decoherence-free subspace~\cite{baumgartner2012structures}. 
The dynamics of quantum trajectories become complicated if such subspaces exist~\cite{benoist2024exponentially,schmolke2025asymptotic}. 
Exploring the relations between the spectral features and the decompositions in such reducible quantum systems should be intriguing. 
It should also be interesting to study the dynamical behaviors of the Lyapunov spectrum and nonlinear observables in quantum trajectories of reducible, nonergodic, and/or unpurifying monitored systems.

Specifically, such breakdowns of the ergodic or typical properties of quantum trajectories are often tied to the symmetry inherent to the system, which plays an indispensable role in classifying equilibrium phases.
While we have only briefly touched upon the roles of symmetry in this review, such as dissipative phase transitions and measurement-induced topological transitions, there has been significant recent advancement in the understanding of symmetry in monitored dynamics. 
In open quantum systems, symmetry can be defined either at the level of individual trajectories (strong symmetry) or at the level of averaged mixed states (weak symmetry)~\cite{buvca2012note, albert2014symmetries, Lieu2020symmetry}. 
Particularly in many-body systems, the interplay between these two notions of symmetry leads to nontrivial phenomena never seen in equilibrium systems, such as strong-to-weak symmetry breaking or novel mixed-state topological orders~\cite{Ma2023Average, Lee2023quantum, Ma2025strong, Wang2025intrinsic, Sohal2025noisy, Ellison2025toward}. 
Furthermore, strong symmetry for individual trajectories enriches the landscape of measurement-induced phase transitions, giving rise to phenomena such as spontaneous symmetry breaking and charge-sharpening transitions~\cite{Sang2021measurement, Bao2021symmetry, Agrawal2022entanglement, Barratt2022field, Oshima2023charge, Majidy2023critical}.
This research direction remains a fertile ground for future exploration.

The third direction is to explore how to detect distinctions between quantum trajectories and the corresponding CPTP dynamics. 
If we focus on linear observables, their time average in one quantum trajectory and ensemble average in the CPTP dynamics become the same when the system is ergodic, as discussed in Sec.~\ref{sec:ergodicity_linear-observable}.
This means that it is difficult to find the distinctions on the basis of linear observables.
While we can reveal unique features of quantum trajectories absent in the CPTP dynamics through focusing on nonlinear observables, there is a postselection problem as discussed in Sec.~\ref{sec:nonlinear-observables}.
Finding ways to circumvent these problems should be important in understanding monitored quantum systems and also in probing measurement-induced phase transitions. 
One possible approach is to study statistics of quantum jumps, which is free from the postselection problem. 
Another approach will be to perform feedback operations depending on measurement outcomes\footnote{
Although we have entirely omitted quantum feedback operations in this review, there are several references to this subject; see, e.g, Refs.~\cite{wiseman2009quantum, Dong2010quantum, JingZhang2017quantum, Dong2022quantum, landi2024current, Berberich2024quantum}.
} so that certain properties of quantum trajectories become visible even for the averaged dynamics.

\section*{Acknowledgment}
We would like to deeply thank Naomichi Hatano for encouraging us to write this paper. 
We thank Tristan Benoist, Masaya Nakagawa, Clément Pellegrini, and Hironobu Yoshida for reading our manuscript with valuable comments. 
We also thank Masaaki Tokieda for fruitful discussions. 
R.H. is supported by JSPS KAKENHI Grant No. JP24K16982. 
R. H. and K.M. are supported by JST ERATO Grant Number JPMJER2302, Japan. 
K.M. is supported by JSPS KAKENHI Grant No. JP23K13037. 
H.O. is supported by RIKEN Junior Research Associate Program. 
Y.F. is supported by JSPS KAKENHI Grants No. JP20K14402 and No. JP24K06897.

\appendix

\section{Some important mathematical theorems on positive maps}

\subsection{Russo-Dye theorem}
\label{app:RussoDye}

Here, we provide a proof of the Russo-Dye theorem for a positive unital linear map $\mc{T}: \mbb{B}[\mc{H}] \to \mbb{B}[\mc{H}]$, which is used in Sec.~\ref{sec:GeneralSpecCPTP}. 
The proof follows Ref.~\cite{bhatia2007positive}.
To do so, we first prove the following lemma.
\begin{lemma}
Let $\hat{A} \in \mbb{B}[\mc{H}]$.
Then $\hat{A}$ satisfies $\lVert \hat{A} \rVert \leq 1$ (i.e., $\hat{A}$ is contractive) if and only if $\begin{pmatrix} \hat{\mbb{I}} & \hat{A} \\ \hat{A}^\dagger & \hat{\mbb{I}} \end{pmatrix}$ is positive semidefinite.
\end{lemma}

\begin{proof}
Let $d = \dim [\mc{H}]$. 
Let $\hat{A} = \hat{U} \hat{\Lambda} \hat{V}$ be the singular value decomposition of $\hat{A}$, where $\hat{U}$ and $\hat{V}$ are unitary matrices and $\hat{\Lambda} = \mathrm{diag}(\Lambda_1, \ldots, \Lambda_d)$ is a diagonal matrix with singular values $\Lambda_1 \geq \cdots \Lambda_d \geq 0$.
Then, we have
\begin{align}
\begin{pmatrix} \hat{\mbb{I}} & \hat{A} \\ \hat{A}^\dagger & \hat{\mbb{I}} \end{pmatrix} 
= \begin{pmatrix} \hat{\mbb{I}} & \hat{U} \hat{\Lambda} \hat{V} \\ \hat{V}^\dagger \hat{\Lambda} \hat{U}^\dagger & \hat{\mbb{I}} \end{pmatrix}
= \begin{pmatrix} \hat{U} & 0 \\ 0 & \hat{V}^\dagger \end{pmatrix} \begin{pmatrix} \hat{\mbb{I}} & \hat{\Lambda} \\ \hat{\Lambda} & \hat{\mbb{I}} \end{pmatrix} \begin{pmatrix} \hat{U}^\dagger & 0 \\ 0 & \hat{V} \end{pmatrix}.
\end{align}
This matrix is unitarily equivalent to the following matrix,
\begin{align}
\begin{pmatrix} 1 & \Lambda_1 \\ \Lambda_1 & 1 \end{pmatrix} \oplus \cdots \oplus \begin{pmatrix} 1 & \Lambda_d \\ \Lambda_d & 1 \end{pmatrix}.
\end{align}
Since each $2 \times 2$ block matrix has eigenvalues $\{ 1 \pm \Lambda_j \}$, this matrix is positive semidefinite if and only if $\lVert \hat{A} \rVert = \Lambda_1 \leq 1$.
\end{proof} 

We are now ready to prove the Russo-Dye theorem.
\begin{theorem}
If $\mc{T}$ is a positive unital linear map, then its operator norm satisfies $\lVert \mc{T} \rVert = 1$.
\end{theorem}

\begin{proof}
We first observe that the operator norm $\lVert \mc{T} \rVert$ is equivalently written as 
\begin{align}
\lVert \mc{T} \rVert = \sup_{\lVert \hat{X} \rVert \leq 1} \lVert \mc{T}[\hat{X}] \rVert.
\end{align}
Suppose that $\hat{X} \in \mbb{B}[\mc{H}]$ has $\lVert \hat{X} \rVert \leq 1$.
Let $\hat{X} = \hat{U}' \hat{\Lambda} \hat{V}'$ be the singular value decomposition of $\hat{X}$, where $\hat{U}'$ and $\hat{V}'$ are unitary matrices and $\hat{\Lambda} = \mathrm{diag}(\Lambda_1, \ldots, \Lambda_d)$ is a diagonal matrix with singular values $\Lambda_1 \geq \cdots \geq \Lambda_d \geq 0$.
Since $\Lambda_1 \leq 1$, we can write $\hat{\Lambda}$ as 
\begin{align}
\hat{\Lambda} 
= \begin{pmatrix} \Lambda_1 && \\ & \ddots & \\ && \Lambda_d \end{pmatrix}
= \frac{1}{2} \begin{pmatrix} e^{i\theta_1} && \\ & \ddots & \\ && e^{i\theta_d} \end{pmatrix} + \frac{1}{2} \begin{pmatrix} e^{-i\theta_1} && \\ & \ddots & \\ && e^{-i\theta_d} \end{pmatrix}
\end{align}
with some $\theta_j \in \mbb{R}$. 
Thus, $\hat{X}$ can be written as $\hat{X} = \frac{1}{2}(\hat{U}_+ + \hat{U}_-)$ with unitary matrices $\hat{U}_+$ and $\hat{U}_-$.
We then find
\begin{align}
\lVert \mc{T}[\hat{X}] \rVert 
= \frac{1}{2} \lVert \mc{T}[\hat{U}_+ + \hat{U}_-] \rVert 
\leq \frac{1}{2} \left( \lVert \mc{T}[\hat{U}_+] \rVert + \lVert \mc{T} [\hat{U}_-] \rVert \right).
\end{align}
We now wish to show $\lVert \mc{T}[\hat{U}] \rVert \leq 1$ for unitary $\hat{U}$.
Using the spectral decomposition $\hat{U} = \sum_{j=1}^d e^{i\phi_j} \hat{P}_j$ with $\phi_j \in \mbb{R}$ and rank-one projectors $\hat{P}_j$, we find
\begin{align}
\begin{pmatrix} \hat{\mbb{I}} & \mc{T}[\hat{U}] \\ \mc{T}[\hat{U}]^\dagger & \hat{\mbb{I}} \end{pmatrix}
= \begin{pmatrix} \hat{\mbb{I}} & \sum_j e^{i\phi_j} \mc{T}[\hat{P}_j] \\ \sum_j e^{-i\phi_j} \mc{T}[\hat{P}_j] & \hat{\mbb{I}} \end{pmatrix}
= \sum_{j=1}^d \begin{pmatrix} 1 & e^{i\phi_j} \\ e^{-i\phi_j} & 1 \end{pmatrix} \otimes \mc{T}[\hat{P}_j],
\end{align}
where we have used $\sum_j \mc{T}[\hat{P}_j] = \hat{\mbb{I}}$ at the last equality. 
Since this matrix is positive semidefinite, by the above lemma we obtain $\lVert \mc{T}[\hat{U}] \rVert \leq 1$. 
This implies $\lVert \mc{T}[\hat{X}] \rVert \leq 1$ for $\lVert \hat{X} \rVert \leq 1$ and thus $\lVert \mc{T} \rVert \leq 1$. 
On the other hand, since $\mc{T}$ is unital, we have $\lVert \mc{T}[\hat{\mbb{I}}] \rVert = \lVert \hat{\mbb{I}} \rVert = 1$ and therefore $\lVert \mc{T} \rVert = 1$.
\end{proof}

\subsection{Existence of a stationary state for CPTP maps}
\label{app:StationaryState}

Here, we provide a proof that any CPTP map $\mc{E}: \mbb{B}[\mc{H}] \to \mbb{B}[\mc{H}]$ has at least one stationary state $\hat{\varrho}_0 \in \mbb{B}[\mc{H}]$, i.e., a density operator satisfying $\mc{E}[\hat{\varrho}_0] = \hat{\varrho}_0$.
This can be understood as a consequence of Brouwer's fixed-point theorem, which states that any continuous map $\mc{F}$ from a nonempty compact convex set $\mbb{S} \subset \mbb{R}^n$ to itself has a fixed point $\hat{x}_0 \in \mbb{S}$ so that $\mc{F}[\hat{x}_0] = \hat{x}_0$.
Since we consider a finite-dimensional Hilbert space $\mc{H}$ with $d = \dim[\mc{H}]$, any Hermitian operator $\hat{X} \in \mbb{B}[\mc{H}]$ can be represented as a $d \times d$ Hermitian matrix and is isomorphic to $\mbb{R}^{d^2}$.
The set of density operators is nonempty (since there exists $\hat{\rho}=\hat{\mbb{I}}/d$), convex (since $s \hat{\rho}_1 + (1-s) \hat{\rho}_2$ with $0 \leq s \leq 1$ is a density operator if $\hat{\rho}_1$ and $\hat{\rho}_2$ are density operators), and compact (since boundedness and closedness are ensured by $\Tr(\hat{\rho})=1$ and $\hat{\rho} \succeq 0$).
Since $\mc{E}$ is linear, it is also continuous for finite-dimensional $\mc{H}$.
Thus, the existence of a stationary state is ensured by Brouwer's fixed-point theorem.

However, since we only consider linear CPTP maps and any CPTP map is a contraction for the operator norm, we can explicitly construct a stationary state by adapting the proof of Theorem 4.4.1 in Ref.~\cite{rivas2012open}.

\begin{theorem}
For finite-dimensional $\mc{H}$, any CPTP map $\mc{E}: \mbb{B}[\mc{H}] \to \mbb{B}[\mc{H}]$ has at least one density operator $\hat{\varrho}_0 \in \mbb{B}[\mc{H}]$ such that $\mc{E}[\hat{\varrho}_0] = \hat{\varrho}_0$.
\end{theorem}

\begin{proof}
We explicitly construct such a density operator.
Taking any density operator $\hat{\rho} \in \mbb{B}[\mc{H}]$, we consider its time average,
\begin{align}
\hat{\rho}_\textrm{av} = \lim_{N \to \infty} \frac{1}{N} \sum_{n=0}^{N-1} \mc{E}^n[\hat{\rho}],
\end{align}
where $\mc{E}^0[\hat{\rho}] = \hat{\rho}$.
Since $\mc{E}$ is CPTP, each $\hat{\rho}_n = \mc{E}^n[\hat{\rho}]$ is a density operator and their convex sum $(1/N) \sum_{n=0}^{N-1} \hat{\rho}_n$ for a finite $N$ is also a density operator.
Since the set of density operators is compact, the limit $\hat{\rho}_\textrm{av}$ exists within the set and is a density operator.
Since $\mc{E}$ is a continuous linear map for finite-dimensional $\mc{H}$, we have 
\begin{align}
\mc{E}[\hat{\rho}_\textrm{av}] 
= \lim_{N \to \infty} \frac{1}{N} \sum_{n=0}^{N-1}  \mc{E}[\mc{E}^n[\hat{\rho}]] 
= \lim_{N \to \infty} \frac{1}{N} \sum_{n=1}^{N} \mc{E}^{n}[\hat{\rho}].
\end{align}
We then find
\begin{align}
\mc{E}[\hat{\rho}_\textrm{av}] - \hat{\rho}_\textrm{av} = \lim_{N \to \infty} \frac{\mc{E}^N[\hat{\rho}] - \hat{\rho}}{N}.
\end{align}
Since $\| \mc{E}^N[\hat{\rho}] -\hat{\rho} \| \leq \| \mc{E}^N[\hat{\rho}] \| + \| \hat{\rho} \| \leq \| \mc{E} \|^N \| \hat{\rho} \| + \| \hat{\rho} \|  \leq 2$, we have 
\begin{align}
\lim_{N \to \infty} \frac{\| \mc{E}^N[\hat{\rho}] -\hat{\rho} \|}{N} = 0
\end{align}
and therefore
\begin{align}
\lim_{N \to \infty} \frac{\mc{E}^N[\hat{\rho}] - \hat{\rho}}{N} = 0.
\end{align}
Thus, we find $\mc{E}[\hat{\rho}_\textrm{av}] = \hat{\rho}_\textrm{av}$.
\end{proof}

\subsection{Burnside's theorem on matrix algebras}
\label{app:BurnsideTheorem}

A subset $\mbb{A} \subseteq \mbb{B}[\mc{H}]$ is called an algebra if it is closed under scalar multiplication, addition, and multiplication.
A set $\mbb{S} \subseteq \mbb{B}[\mc{H}]$ is said to be irreducible if the only subspaces of $\mc{H}$ invariant under the action of $\mbb{S}$ are $\{0\}$ and $\mc{H}$.
Then, Burnside's theorem on matrix algebras specialized to $\mbb{A} \subseteq \mbb{B}[\mc{H}]$ is stated as follows.

\begin{theorem}[Burnside's theorem on matrix algebras]
An algebra $\mbb{A} \subseteq \mbb{B}[\mc{H}]$ is irreducible if and only if $\mbb{A} = \mbb{B}[\mc{H}]$.
\end{theorem}

\begin{proof}
Since the if part is trivial, we prove the only if part following Ref.~\cite{Halperin1980burnside}.

We first show that $\mbb{A}$ contains a rank-1 operator.
Let $\hat{T}_0 \in \mbb{A}$ be a nonzero operator with the smallest rank $r$ among all elements in $\mbb{A}$.
Suppose $r>1$. 
Let $\ket{x_1}, \ket{x_2} \in \mc{H}$ be such that $\hat{T}_0 | x_1 \rangle$ and $\hat{T}_0 | x_2 \rangle$ are linearly independent.
Since $\mbb{A} | x \rangle = \mc{H}$ for any nonzero $\ket{x} \in \mc{H}$ and $\hat{T}_0 | x_1 \rangle \neq 0$, there exists $\hat{A} \in \mbb{A}$ such that $\hat{A} \hat{T}_0 | x_1 \rangle = | x_2 \rangle$.
Then, $\hat{T}_0 \hat{A} \hat{T}_0 | x_1 \rangle$ and $\hat{T}_0 | x_1 \rangle$ are linearly independent, so that $\hat{T}_0 \hat{A} \hat{T}_0 - \lambda \hat{T}_0$ is nonzero for all $\lambda \in \mbb{C}$.
When restricted to $\hat{T}_0 \mc{H}$, $\hat{T}_0 \hat{A}$ has an eigenvector $\ket{y_0} \in \hat{T}_0 \mc{H}$ with eigenvalue $\lambda_0$, and hence $\hat{T}_0 \hat{A} -\lambda_0 \hat{\mbb{I}}$ has a rank lower than $\dim [\hat{T}_0 \mc{H}] = r$.
This contradicts the assumption that $\hat{T}_0$ has the lowest rank, and thus $r$ must be 1.

We next show that $\mbb{A}$ contains all rank-1 operators.
Since $\hat{T}_0 \in \mbb{A}$ has rank 1, we can write $\hat{T}_0 = | u_0 \rangle \langle v_0 |$ with some $| u_0 \rangle, | v_0 \rangle \in \mc{H}$.
Since $\mbb{A} | u_0 \rangle = \mbb{A} | v_0 \rangle = \mc{H}$, for any $| u \rangle, | v \rangle \in \mc{H}$ there exist $\hat{B}, \hat{C} \in \mbb{A}$ such that $| u \rangle = \hat{B} | u_0 \rangle$ and $| v \rangle = \hat{C}^\dagger | v_0 \rangle$. 
Since $\mbb{A}$ is an algebra, $\hat{B} \hat{T}_0 \hat{C} \in \mbb{A}$ for any $\hat{B}, \hat{C} \in \mbb{A}$.
In this way, we can construct $\hat{B} \hat{T}_0 \hat{C} = | u \rangle \langle v | \in \mbb{A}$ for any pairs of $\ket{u}, \ket{v} \in \mc{H}$.
Since any $\hat{X} \in \mbb{B}[\mc{H}]$ can be written as a sum of rank-1 operators $| u \rangle \langle v |$, we prove $\mbb{A} = \mbb{B}[\mc{H}]$.
\end{proof}

\subsection{CP unital map is a Schwarz map}
\label{app:SchwarzMap}

Here, we give a proof that a CP unital map $\mc{T}: \mbb{B}[\mc{H}] \to \mbb{B}[\mc{H}]$ is a Schwarz map, i.e., a map satisfying the Kadison-Schwarz inequality $\mc{T}[\hat{A}^\dagger \hat{A}] \succeq \mc{T}[\hat{A}^\dagger] \mc{T}[\hat{A}]$ for any $\hat{A} \in \mbb{B}[\mc{H}]$.
To begin with, we first prove the following lemma:

\begin{lemma}
Let $\hat{A}$, $\hat{B}$, $\hat{C} \in \mbb{B}[\mc{H}]$.
If a block matrix $\hat{M} := \begin{pmatrix} \hat{A} & \hat{B}^\dagger \\ \hat{B} & \hat{C} \end{pmatrix}$ is positive semidefinite and $\hat{C}$ is positive definite, then $\hat{A} \succeq \hat{B}^\dagger \hat{C}^{-1} \hat{B}$.
\end{lemma}

\begin{proof}
Let $\ket{x} \in \mc{H}$. 
Since $\hat{C}$ is invertible, we can define $v = \begin{pmatrix} \ket{x} \\ -\hat{C}^{-1} \hat{B} \ket{x} \end{pmatrix}$. 
Since $\hat{M} \succeq 0$, $v^\dagger \hat{M} v \geq 0$ and thus
\begin{align}
0 \leq v^\dagger \hat{M} v 
= \begin{pmatrix} \bra{x} & -\bra{x} \hat{B}^\dagger \hat{C}^{-1} \end{pmatrix} \begin{pmatrix} \hat{A} & \hat{B}^\dagger \\ \hat{B} & \hat{C} \end{pmatrix} \begin{pmatrix} \ket{x} \\ -\hat{C}^{-1} \hat{B} \ket{x} \end{pmatrix}
= \bra{x} \left( \hat{A} - \hat{B}^\dagger \hat{C}^{-1} \hat{B} \right) \ket{x}.
\end{align}
Since this must hold for any $\ket{x} \in \mc{H}$, we have $\hat{A} \succeq \hat{B}^\dagger \hat{C}^{-1} \hat{B}$.
\end{proof}

We then prove the original proposition.

\begin{theorem}
If a map $\mc{T}: \mbb{B}[\mc{H}] \to \mbb{B}[\mc{H}]$ is CP and unital, $\mc{T}$ is also a Schwarz map.
\end{theorem}

\begin{proof}
Let us consider a block matrix $\hat{X} := \begin{pmatrix} \hat{A}^\dagger \hat{A} & \hat{A}^\dagger \\ \hat{A} & \hat{\mbb{I}} \end{pmatrix}$ with any $\hat{A} \in \mbb{B}[\mc{H}]$. 
This matrix is obviously positive semidefinite as it can be written as $\hat{X} = \begin{pmatrix} \hat{A}^\dagger \\ \hat{\mbb{I}} \end{pmatrix} \begin{pmatrix} \hat{A} & \hat{\mbb{I}} \end{pmatrix}$.
Since $\mc{T}$ is CP, it is also 2-positive\footnote{
As given in Eq.~\eqref{eq:DefCP}, a linear map $\mc{T}: \mbb{B}[\mc{H}] \to \mbb{B}[\mc{H}]$ is CP if $(\mc{T} \otimes \mc{I}_n)[\hat{R}] \succeq 0$ for all $n \in \mbb{N}$ and for $\hat{R} \succeq 0$ on $\mbb{B}[\mc{H} \otimes \mc{H}_n]$, where $\mc{H}_n$ is an auxiliary $n$-dimensional Hilbert space and $\mc{I}_n: \mbb{B}[\mc{H}_n] \to \mbb{B}[\mc{H}_n]$ is the identity map.
When this property holds for a specific $n$, $\mc{T}$ is called $n$-positive.
}
and thus 
\begin{align}
\hat{M} 
:= (\mc{T} \otimes \mc{I}_2)[\hat{X}] 
= \begin{pmatrix} \mc{T} [\hat{A}^\dagger \hat{A}] & \mc{T}[\hat{A}^\dagger] \\
\mc{T}[\hat{A}] & \mc{T}[\hat{\mbb{I}}] \end{pmatrix} 
\succeq 0.
\end{align}
Since $\mc{T}$ is unital, we have
\begin{align}
\hat{M} = \begin{pmatrix} \mc{T}[\hat{A}^\dagger \hat{A}] & (\mc{T}[\hat{A}])^\dagger \\ \mc{T}[\hat{A}] & \hat{\mbb{I}} \end{pmatrix} \succeq 0,
\end{align}
where we have used the Hermiticity preserving property of $\mc{T}$, $\mc{T}[\hat{A}^\dagger] = (\mc{T}[\hat{A}])^\dagger$.
Using the above lemma, we find the Kadison-Schwarz inequality $\mc{T}[\hat{A}^\dagger \hat{A}] \succeq \mc{T}[\hat{A}^\dagger] \mc{T}[\hat{A}]$.
\end{proof}

\subsection{Schur's lemma on self-adjoint sets}
\label{app:SchurLemma}

We call $\mbb{S} \subseteq \mbb{B}[\mc{H}]$ a self-adjoint set if $\hat{S}^\dagger \in \mbb{S}$ for every $\hat{S} \in \mbb{S}$.
We say that $\mbb{S} \subseteq \mbb{B}[\mc{H}]$ is irreducible if the only subspaces of $\mc{H}$ invariant under the action of $\mbb{S}$ are $\{0\}$ and $\mc{H}$.
The commutant of $\mbb{S}$ is defined as a set of operators commuting with all elements of $\mbb{S}$, that is, 
\begin{align}
\mbb{S}' = \{ \hat{T} \in \mbb{B}[\mc{H}] : [\hat{T}, \hat{S}] = 0 \textrm{ for all } \hat{S} \in \mbb{S} \}.
\end{align}
Then, we have the following theorem known as Schur's lemma in the context of operator algebra (see, e.g., Theorem~5.1.6 of Ref.~\cite{deitmar2014harmonic}).

\begin{theorem}[Schur's lemma]
Let $\mbb{S} \subseteq \mbb{B}[\mc{H}]$ be a self-adjoint set. 
Then $\mbb{S}' = \mbb{C} \hat{\mbb{I}}$ if and only if $\mbb{S}$ is irreducible.
\end{theorem}

\begin{proof}
Let $\hat{T} \in \mbb{S}'$.
For any $\hat{S} \in \mbb{S}$, we have $[\hat{T}, \hat{S}]=0$ and $[\hat{T}^\dagger, \hat{S}^\dagger] = [\hat{S}, \hat{T}]^\dagger = 0$. 
Since $\hat{S}^\dagger \in \mbb{S}$, $\hat{T}^\dagger$ is also an element of $\mbb{S}'$ (this means that $\mbb{S}'$ is also self-adjoint).
Since any $\hat{T} \in \mbb{B}[\mc{H}]$ can be decomposed into a sum of two Hermitian operators $\hat{T}_1 = \frac{1}{2}(\hat{T} + \hat{T}^\dagger)$ and $\hat{T}_2 = \frac{1}{2i}(\hat{T} - \hat{T}^\dagger)$ as $\hat{T} = \hat{T}_1 + i\hat{T}_2$, we can assume $\hat{T}^\dagger = \hat{T}$ without loss of generality.

Suppose that $\mbb{S}$ is irreducible.
Let $\hat{T} = \sum_{i=1}^d \lambda_i | \psi_i \rangle \langle \psi_i |$ be the spectral decomposition of $\hat{T}$ with eigenvalues $\lambda_i \in \mbb{R}$ and eigenvectors $\ket{\psi_i} \in \mc{H}$.
Let $\mc{V}_\lambda \subseteq \mc{H}$ be an eigenspace spanned by the eigenvectors of $\hat{T}$ with the eigenvalue $\lambda$.
For any $\ket{\psi} \in \mc{V}_\lambda$ and $\hat{S} \in \mbb{S}$, we have $\hat{T} \hat{S} | \psi \rangle = \hat{S} \hat{T} | \psi \rangle = \lambda \hat{S} | \psi \rangle$.
Since $\hat{S} | \psi \rangle \in \mc{V}_\lambda$, $\mc{V}_\lambda$ is an invariant subspace of $\mc{H}$ under $\mbb{S}$.
However, such a subspace must be $\{0\}$ or $\mc{H}$ by assumption, and thus $\mc{V}_\lambda = \mc{H}$ for any $\lambda$.
This implies $\lambda_i = \lambda$ and hence $\hat{T} = \lambda \hat{\mbb{I}}$. 
This proves $\mbb{S}' = \mbb{C} \hat{\mbb{I}}$.

Suppose that $\mbb{S}' = \mbb{C} \hat{\mbb{I}}$.
Let $\mc{V} \subseteq \mc{H}$ be an invariant subspace of $\mc{H}$ under $\mbb{S}$ and $\mc{V}_\perp$ be its orthogonal complement.
For any $\ket{v} \in \mc{V}$, $\ket{u} \in \mc{V}_\perp$, and $\hat{S} \in \mbb{S}$, we have $\hat{S}^\dagger | v \rangle \in \mc{V}$ and $\langle v | \hat{S} | u \rangle = (\hat{S}^\dagger | v \rangle)^\dagger | u \rangle$ = 0. 
This implies $\hat{S} | u \rangle \in \mc{V}_\perp$ and that $\mc{V}_\perp$ is also an invariant subspace of $\mc{H}$ under $\mbb{S}$.
Then, let $\hat{P} \in \mbb{B}[\mc{H}]$ be an orthogonal projection operator to $\mc{V}$.
Since any $\ket{x} \in \mc{H}$ can be written as $\ket{x} = \ket{v}+\ket{u}$ with some $\ket{v} \in \mc{V}$ and $\ket{u} \in \mc{V}_\perp$, we have $\hat{P} \hat{S} | x \rangle = \hat{S} | v \rangle = \hat{S} \hat{P} | x \rangle$. 
Thus, we have $[\hat{P}, \hat{S}] = 0$ for any $\hat{S} \in \mbb{S}$.
By assumption, $\hat{P}$ must be proportional to the identity operator, i.e. $\hat{P} = \mu \hat{\mbb{I}}$ with some $\mu \in \mbb{C}$. 
Since $\hat{P}^2 = \hat{P}$, we have $\mu=0$ or $1$ and therefore $\mc{V} = \{0\}$ or $\mc{H}$.
\end{proof}

\section{Proof of Proposition~\ref{prop:irreducibility}}
\label{app:ProofIrreducibility}

Here we provide a proof of Proposition~\ref{prop:irreducibility} for the irreducibility of CPTP maps $\mc{E}$.

\begin{proof}
The proof for (1) $\Leftrightarrow$ (2) follows Ref.~\cite{carbone2016open}, while that for (1) $\Leftrightarrow$ (3) follows Refs.~\cite{evans1978spectral, wolf2012quantum}.

(1) $\Rightarrow$ (2): 
Suppose that there exists an orthogonal projection operator $\hat{P} \notin \{ 0, \hat{\mbb{I}}\}$ such that $\hat{P} \preceq \mc{E}^\dagger[\hat{P}]$.
For any positive semidefinite $\hat{\rho} \in \mbb{B}[\mc{H}]$, we have 
\begin{align}
\Tr[\hat{P} \hat{\rho} \hat{P}] 
\leq \Tr[\hat{P} \hat{\rho} \hat{P} \mc{E}^\dagger[\hat{P}]] 
= \Tr[\hat{P} \mc{E}[\hat{P} \hat{\rho} \hat{P}] \hat{P}] 
\leq \Tr[\mc{E}[\hat{P} \hat{\rho} \hat{P}]]
= \Tr[\hat{P} \hat{\rho} \hat{P}],
\end{align}
which yields the equality $\Tr[\hat{P} \mc{E}[\hat{P} \hat{\rho} \hat{P}] \hat{P}] = \Tr[\mc{E}[\hat{P} \hat{\rho} \hat{P}]]$. 
Since $\mc{E}[\hat{P} \hat{\rho} \hat{P}] \succeq 0$,
this implies $\mc{E}[\hat{P} \hat{\rho} \hat{P}] = \hat{P} \mc{E}[\hat{P} \hat{\rho} \hat{P}] \hat{P} \in \hat{P} \mbb{B}[\mc{H}] \hat{P}$\footnote{
Let $\hat{A} \in \mbb{B}[\mc{H}]$ be a positive semidefinite operator and $\hat{P} \in \mbb{B}[\mc{H}]$ be an orthogonal projection operator. 
If $\Tr[\hat{A}] = \Tr[\hat{P} \hat{A}]$, then $\hat{A} = \hat{P} \hat{A} \hat{P}$. 
\textit{Proof}. Let $\hat{A} = \sum_n \lambda_n \ket{\psi_n} \bra{\psi_n}$ be the spectral decomposition of $\hat{A}$ with $\lambda_n \geq 0$, $\ket{\psi_n} \in \mc{H}$, and $\| |\psi_n \rangle \|=1$.
Using $\Tr[\hat{A}]=\Tr[\hat{P} \hat{A}]$, we have $\sum_n \lambda_n (1- \| \hat{P} | \psi_n \rangle \|^2) = \sum_n \lambda_n \| (\hat{\mathbb{I}}-\hat{P}) | \psi_n \rangle \|^2 = 0$, which implies $\lambda_n = 0$ or $(\hat{\mbb{I}}-\hat{P}) \ket{\psi_n}=0$.
Then, we have $(\hat{\mbb{I}}-\hat{P}) \hat{A} = \hat{A} (\hat{\mbb{I}}-\hat{P})=0$ and thus $\hat{A} = \hat{P} \hat{A}=\hat{A} \hat{P} = \hat{P} \hat{A} \hat{P}$.
}.
Moreover, $\mc{E}[\hat{P} \hat{X} \hat{P}] \in \hat{P} \mbb{B}[\mc{H}] \hat{P}$ holds for any (not necessarily positive semidefinite) $\hat{X} \in \mbb{B}[\mc{H}]$, since $\mc{E}$ is linear and $\hat{X}$ can always be decomposed into a linear combination of positive semidefinite operators\footnote{
Any $\hat{X} \in \mbb{B}[\mc{H}]$ can be decomposed as $\hat{X} = \hat{H}_+ + i\hat{K}_+ -(\hat{H}_- + i\hat{K}_-)$ with some positive semidefinite operators $\hat{H}_\pm, \hat{K}_\pm \in \mbb{B}[\mc{H}]$.
}. 
We therefore find $\hat{P} \notin \{ 0, \hat{\mbb{I}} \}$ such that $\mc{E}[\hat{P} \mbb{B}[\mc{H}] \hat{P}] \subseteq \hat{P} \mbb{B}[\mc{H}] \hat{P}$.

(2) $\Rightarrow$ (1):
Suppose that there exists an orthogonal projection operator $\hat{P} \notin \{ 0, \hat{\mbb{I}} \}$ such that $\mc{E}[\hat{P} \mbb{B}[\mc{H}] \hat{P}] \subseteq \hat{P} \mbb{B}[\mc{H}] \hat{P}$.
For any $\hat{X} \in \mbb{B}[\mc{H}]$, the following holds
\begin{align}\label{eqXP}
\Tr[\hat{X} \hat{P}] 
= \Tr[\hat{P} \hat{X} \hat{P}] 
= \Tr[\mc{E}[\hat{P} \hat{X} \hat{P}]]
= \Tr[\hat{P} \mc{E}[\hat{P} \hat{X} \hat{P}] \hat{P}]
= \Tr[\hat{X} \hat{P} \mc{E}^\dagger[\hat{P}] \hat{P}].
\end{align}
Here, the second equality comes from TP of $\mc{E}$ and the third equality from $\mc{E}[\hat{P} \hat{X} \hat{P}] \in \hat{P} \mbb{B}[\mc{H}] \hat{P}$ by the assumption, which implies $\mc{E}[\hat{P} \hat{X} \hat{P}] = \hat{P} \mc{E}[\hat{P} \hat{X} \hat{P}] \hat{P}$.
From Eq.~\eqref{eqXP}, we obtain $\hat{P} = \hat{P} \mc{E}^\dagger[\hat{P}] \hat{P}$. 
Let $\hat{Q} = \hat{\mbb{I}}-\hat{P}$ be the projection operator to the complement of $\hat{P}$.
Since $\mc{E}^\dagger$ is unital, i.e., $\mc{E}^\dagger[\hat{\mbb{I}}]=\hat{\mbb{I}}$, we have $\hat{P} \mc{E}^\dagger[\hat{Q}] \hat{P} = \hat{P} (\hat{\mbb{I}}-\mc{E}^\dagger[\hat{P}]) \hat{P} = 0$.
Since $\mc{E}^\dagger[\hat{Q}] \succeq 0$, 
this implies $\hat{P} \mc{E}^\dagger[\hat{Q}] = \mc{E}^\dagger[\hat{Q}] \hat{P} = 0$\footnote{
Let $\hat{A} \in \mbb{B}[\mc{H}]$ be a positive semidefinite operator and $\hat{P} \in \mbb{B}[\mc{H}]$ be an orthogonal projection operator. 
If $\hat{P} \hat{A} \hat{P}=0$, then $\hat{P} \hat{A} = \hat{A} \hat{P} = 0$. 
\textit{Proof.} Since $\hat{A} \succeq 0$, there is a unique decomposition $\hat{A} = \hat{A}^{1/2} \hat{A}^{1/2}$ by some positive semidefinite operator $\hat{A}^{1/2} = (\hat{A}^{1/2})^\dagger \in \mbb{B}[\mc{H}]$.
Since $\langle \psi | \hat{P} \hat{A} \hat{P} | \psi \rangle = \| \hat{A}^{1/2} \hat{P} \ket{\psi} \|^2 = 0$ holds for any $\ket{\psi} \in \mc{H}$, we have $\hat{A}^{1/2} \hat{P} = 0$ and thus $\hat{A}^{1/2} \hat{A}^{1/2} \hat{P} = \hat{A} \hat{P} = 0$. 
Its Hermitian conjugate gives $(\hat{A} \hat{P})^\dagger = \hat{P} \hat{A} = 0$.
} and thus $\mc{E}^\dagger[\hat{Q}] = (\hat{P}+\hat{Q}) \mc{E}^\dagger[\hat{Q}] (\hat{P}+\hat{Q}) = \hat{Q} \mc{E}^\dagger[\hat{Q}] \hat{Q}$.
Finally, we find that $\hat{P} \notin \{ 0, \hat{\mbb{I}} \}$ satisfies
\begin{align}
\mc{E}^\dagger[\hat{P}] = \hat{\mbb{I}} -\mc{E}^\dagger[\hat{Q}] = \hat{\mbb{I}} -\hat{Q} \mc{E}^\dagger[\hat{Q}] \hat{Q} = \hat{P} + \hat{Q} \mc{E}^\dagger[\hat{P}] \hat{Q} \succeq \hat{P}.
\end{align}

(1) $\Rightarrow$ (3): 
We first show that irreducibility in the sense of (1) implies that $(\mc{I}+\mc{E})^{d-1}[\hat{\rho}] \succ 0$ for any nonzero $\hat{\rho} \succeq 0$ in $\mbb{B}[\mc{H}]$, which gives another definition of irreducibility when $\mc{H}$ has a finite dimension $d$.
Here, $(\mc{I}+\mc{E})[\hat{X}] = \hat{X} + \mc{E}[\hat{X}]$. 
We then show that it further implies (3).

Let $\hat{\rho} \in \mbb{B}[\mc{H}]$ be nonzero and positive semidefinite. 
Assume that (1) holds for $\mc{E}$. 
Since $\mc{E}[\hat{\rho}] \succeq 0$, $\ker((\mc{I}+\mc{E})[\hat{\rho}]) \subseteq \ker(\hat{\rho})$\footnote{
Let $\ket{\psi} \in \ker(\hat{\rho}+\mc{E}[\hat{\rho}])$.
Then, $\langle \psi | (\hat{\rho}+\mc{E}[\hat{\rho}]) | \psi \rangle = 0$ and thus $\langle \psi | \hat{\rho} | \psi \rangle = -\langle \psi | \mc{E}[\hat{\rho}] | \psi \rangle$. 
Since $\hat{\rho} \succeq 0$ and $\mc{E}[\hat{\rho}] \succeq 0$, this implies $\hat{\rho} \ket{\psi} = \hat{\mc{E}}[\hat{\rho}] \ket{\psi} = 0$. 
As this holds for any $\ket{\psi} \in \ker(\hat{\rho}+\mc{E}[\hat{\rho}])$, we have $\ker(\hat{\rho}+\mc{E}[\hat{\rho}]) \subseteq \ker(\hat{\rho}), \ker(\mc{E}[\hat{\rho}])$. 
[We can actually prove $\ker(\hat{\rho} + \mc{E}[\hat{\rho}]) = \ker(\hat{\rho}) \cap \ker(\mc{E}[\hat{\rho}])$].}. 
Suppose that the equality $\ker((\mc{I}+\mc{E})[\hat{\rho}]) = \ker(\hat{\rho})$ holds, that is $\ker(\mc{E}[\hat{\rho}]) \supseteq \ker(\hat{\rho})$.
Let $\hat{P} \in \mbb{B}[\mc{H}]$ be an orthogonal projection operator onto the image of $\hat{\rho}$.
We then introduce a set of positive semidefinite operators whose images lie in the image of $\hat{\rho}$: $\mbb{F}_P := \{ \hat{\sigma} \succeq 0 : \ker(\hat{\sigma}) \supseteq \ker(\hat{P}) \}$. 
Obviously, $\hat{\rho} \in \mbb{F}_P$ and $\mc{E}[\hat{\rho}] \in \mbb{F}_P$.
Since $\hat{\rho}$ is strictly positive ($\hat{\rho} \succ 0$) in the space $\hat{P} \mbb{B}[\mc{H}]\hat{P}$, there exists $c>0$ such that $\hat{\sigma} \preceq c \hat{\rho}$ for any $\hat{\sigma} \in \mbb{F}_P$\footnote{
If $\hat{A} \in \mbb{B}[\mc{H}]$ is positive definite and $\hat{B} \in \mbb{B}[\mc{H}]$ is positive semidefinite, then there exists $c>0$ such that $\hat{B} \preceq c \hat{A}$.
\textit{Proof}. Since $\hat{A}$ is invertible, we can introduce $\hat{B}' = \hat{A}^{-1/2} \hat{B} \hat{A}^{-1/2}$. 
Since $\hat{B}' \succeq 0$ and $\hat{B}' \in \mbb{B}[\mc{H}]$, we have $0 \preceq \hat{B}' \preceq \| \hat{B}' \| \hat{\mbb{I}}$ or equivalently $\langle \psi | \hat{B}' | \psi \rangle \leq \| \hat{B}' \| \langle \psi | \psi \rangle$ for any $\ket{\psi} \in \mc{H}$. 
This can be rewritten by $\ket{\phi} = \hat{A}^{-1/2} \ket{\psi}$ as $\langle \phi | \hat{B} | \phi \rangle \leq \| \hat{B}' \| \langle \phi | \hat{A} | \phi \rangle$. 
Since this holds for any $\ket{\phi} \in \mc{H}$, we have $\hat{B} \preceq \| \hat{B}' \| \hat{A}$. 
Setting $c=\| \hat{B}' \|$ proves the claim.
This obviously holds even when $\mbb{B}[\mc{H}]$ is restricted to a subspace $\hat{P} \mbb{B}[\mc{H}] \hat{P}$.
}. 
Since $\mc{E}$ is positive, this gives $\mc{E}[\hat{\sigma}] \preceq c \mc{E}[\hat{\rho}]$ and thus $\ker(\mc{E}[\hat{\sigma}]) \supseteq \ker(\mc{E}[\hat{\rho}]) \supseteq \ker(\hat{\rho})$\footnote{
Let $\hat{A}, \hat{B} \in \mbb{B}[\mc{H}]$ be positive semidefinite operators. 
If there exists $c>0$ such that $\hat{B} \preceq c \hat{A}$, then $\textrm{ker}(\hat{A}) \subseteq \textrm{ker}(\hat{B})$. 
\textit{Proof}. Let $\ket{\psi} \in \textrm{ker}(\hat{A})$, then $\hat{A} \ket{\psi}=0$. 
Since $\hat{B} \preceq c\hat{A}$, we have $0 \leq \langle \psi | (c\hat{A}-\hat{B}) | \psi \rangle = -\langle \psi | \hat{B} | \psi \rangle$.
Since $\hat{B} \succeq 0$, this implies $\langle \psi | \hat{B} | \psi \rangle = 0$ and thus $\hat{B} \ket{\psi} = 0$.
As this holds for any $\ket{\psi} \in \textrm{ker}(\hat{A})$, we have $\textrm{ker}(\hat{A}) \subseteq \textrm{ker}(\hat{B})$.
}. 
This establishes $\mc{E}(\mbb{F}_P) \subseteq \mbb{F}_P \subseteq \hat{P} \mbb{B}[\mc{H}] \hat{P}$. 
Since $\mc{E}$ is linear and $\hat{P} \hat{X} \hat{P}$ for any $\hat{X} \in \mbb{B}[\mc{H}]$ can be decomposed into a linear combination of operators in $\mbb{F}_P$, we have $\mc{E}[\hat{P} \hat{X} \hat{P}] \in \hat{P} \mbb{B}[\mc{H}] \hat{P}$ and hence $\mc{E}[\hat{P} \mbb{B}[\mc{H}] \hat{P}] \subseteq \hat{P} \mbb{B}[\mc{H}] \hat{P}$. 
However, (1) implies that such $\hat{P}$ must be $\hat{P} = \hat{\mbb{I}}$ and $\hat{\rho} \succ 0$ (i.e., full rank), which leads to $(\mc{I}+\mc{E})^{d-1}[\hat{\rho}] \succ 0$.
In contrast, when $\hat{\rho}$ is not full rank, we must instead have $\ker((\mc{I}+\mc{E})[\hat{\rho}]) \subset \ker(\hat{\rho})$ and thus $\textrm{rank}((\mc{I}+\mc{E})[\hat{\rho}]) > \textrm{rank}(\hat{\rho})$.
Therefore, even in that case, $\hat{\rho}$ becomes full rank at least after $(d-1)$ times applications of $\mc{I}+\mc{E}$, that is, $(\mc{I}+\mc{E})^{d-1}[\hat{\rho}] \succ 0$.

By expanding $(\mc{I}+\mc{E})^{d-1}[\hat{\rho}]$ and $\exp(s\mc{E})[\hat{\rho}]$ for $s>0$ in powers of $\mc{E}$, we find that all terms are positive semidefinite in both expansions.
Since all powers appearing in $(\mc{I}+\mc{E})^{d-1}[\hat{\rho}]$ are also present in $\exp(s\mc{E})[\hat{\rho}]$, there exists some $c>0$ such that $\exp(s\mc{E})[\hat{\rho}] \succeq c(\mc{I}+\mc{E})^{d-1}[\hat{\rho}]$.
This proves $\exp(s\mc{E})[\hat{\rho}] \succ 0$ for any $s>0$.

(3) $\Rightarrow$ (1): 
Suppose that there exists an orthogonal projection operator $\hat{P} \notin \{ 0, \hat{\mbb{I}} \}$ such that $\mc{E}[\hat{P}] = \hat{P} \hat{\sigma} \hat{P}$ for some positive semidefinite operator $\hat{\sigma} \in \mbb{B}[\mc{H}]$.
This implies $\mc{E}[\hat{P}] = \hat{P} \hat{\sigma} \hat{P}\preceq c \hat{P}$ for some $c>0$\footnote{
If $\hat{P} \in \mbb{B}[\mc{H}]$ is an orthogonal projection operator and $\hat{A} \in \mbb{B}[\mc{H}]$ is a positive semidefinite operator, then there exists $c>0$ such that $\hat{P} \hat{A} \hat{P} \preceq c \hat{P}$.
\textit{Proof}. For any $\ket{\psi} \in \mc{H}$, we have $\langle \psi | \hat{P} \hat{A} \hat{P} | \psi \rangle \leq \| \hat{P} |\psi \rangle \| \| \hat{A} \hat{P} |\psi \rangle \|$ by the Cauchy-Schwarz inequality. 
Then, $\| \hat{P} |\psi \rangle \| \| \hat{A} \hat{P} |\psi \rangle \| \leq \| \hat{A} \| \| \hat{P} |\psi \rangle \|^2 = \| \hat{A} \| \langle \psi | \hat{P} | \psi \rangle$. 
Setting $c=\| \hat{A} \|$ yields $\hat{P} \hat{A} \hat{P} \preceq c\hat{P}$.
}
and then $\mc{E}^n[\hat{P}] \preceq c^n \hat{P}$ for $n \in \mbb{N}$. 
We thus find $\exp(s\mc{E})[\hat{P}] \preceq e^{cs} \hat{P}$ for any $s > 0$.
This implies that $\exp(s\mc{E})[\hat{P}]$ is not positive definite for all $s>0$, as we have $\langle \psi | \exp(s\mc{E})[\hat{P}] | \psi \rangle \leq e^{cs} \langle \psi | \hat{P} | \psi \rangle = 0$ for $\ket{\psi} \in \textrm{ker}(\hat{P})$.
\end{proof}

\section{Outline of the proof of the ergodicity explained in Sec.~\ref{KMerg}} 
\label{app:erg}
Here, we show the proof of the ergodicity explained in Sec.~\ref{KMerg}.
We basically follow Ref.~\cite{kummerer2005quantum} but try to illustrate it in a physicist-friendly way, at the cost of mathematical rigor.

We first consider the conditional expectation value
\aln{\label{juyou}
\mbb{E}\lrl{\hat{\rho}_{\bm{b};n}|b_1,\ldots,b_{n-1}}=\sum_{b_{n}}\frac{\ml{E}_{b_{n}}[\hat{\rho}_{\bm{b};n-1}]}{\Tr[\ml{E}_{b_{n}}[\hat{\rho}_{\bm{b};n-1}]]}\cdot \Tr[\ml{E}_{b_{n}}[\hat{\rho}_{\bm{b};n-1}]]=\ml{E}[\hat{\rho}_{\bm{b};n-1}]
}
and define
\aln{
\delta_{\bm{b};n}=\Tr\lrl{\lrs{\hat{\rho}_{\bm{b};n}-\ml{E}[\hat{\rho}_{\bm{b};n-1}]}\hat{A}},
}
which serves as the fluctuation of the trajectory measured by an observable $\hat{A}$, given $\hat{\rho}_{\bm{b};n-1}$.
We have
\aln{
\mbb{E}[\delta_{\bm{b};n}|b_1,\ldots,b_{n-1}]=0.
}

Now, for the weighted sum of fluctuations
\aln{
Y_{\bm{b};n}=\sum_{s=1}^n\frac{\delta_{\bm{b};s}}{s},
}
We find
\aln{
\mbb{E}[Y_{\bm{b};n}|b_1,\ldots,b_{n-1}]=Y_{\bm{b};n-1}+\frac{1}{n}\mbb{E}[\delta_{\bm{b};n}|b_1,\ldots,b_{n-1}]=Y_{\bm{b};n-1}.
}
This means that $Y_{\bm{b};n}$ is a martingale. 
We also note that $Y_{\bm{b};n}$ is bounded as well,
\aln{
\mbb{E}[|Y_{\bm{b};n}|]^2\leq \mbb{E}[Y_{\bm{b};n}^2]
=\sum_{ss'}\frac{1}{ss'}\mbb{E}[\delta_{\bm{b};s}\delta_{\bm{b};s'}]=\sum_s\frac{1}{s^2}
\mbb{E}[\delta_{\bm{b};s}^2]\leq 4\|\hat{A}\|^2\sum_s\frac{1}{s^2}
\leq\frac{2\pi^2\|\hat{A}\|^2}{3},
}
where $\|\hat{A}\|$ is the operator norm for $\hat{A}$.
Therefore, we can use the martingale convergence theorem~\cite{hall2014martingale} to find that
\aln{
Y_{\bm{b};\infty}:=\lim_{n\ra\infty}Y_{\bm{b};n}=\lim_{n\ra\infty}\sum_{s=1}^n\frac{\delta_{\bm{b};s}}{s}
}
exists almost surely with respect to the probability measure for quantum trajectories.

Now, using Kronecker's lemma\footnote{
Kronecker's lemma states the following. Let us consider an infinite sequence $\{x_s\}_{s=1}^\infty$ satisfying $\sum_{s=1}^\infty x_s<\infty$. 
Then, for all $\{g_n\}_n$ such that $0<g_1\leq\cdots\leq g_n$ with $\lim_{n\ra\infty}g_n=\infty$, we find
\aln{
\lim_{n\ra\infty}\frac{1}{g_n}\sum_{s=1}^ng_sx_s=0.
} 
We here consider $g_s=s$ and $x_s=\delta_s/s$.
}, we find
\aln{
\lim_{n\ra\infty}\frac{1}{n}\sum_{s=1}^{n}\delta_{\bm{b};s}=0
}
almost surely.
This means
\aln{
\lim_{n\ra\infty}\frac{1}{n}\sum_{s=1}^n\Tr\lrl{\lrs{\hat{\rho}_{\bm{b};s}-\ml{E}[\hat{\rho}_{\bm{b};s-1}]}\hat{A}}=0.
}
Since $\lim_{n\ra\infty}\frac{1}{n}\Tr[\hat{\rho}_{\bm{b};n}\hat{A}]=\lim_{n\ra\infty}\frac{1}{n}\Tr[\hat{\rho}_{\bm{b};n=0}\hat{A}]=0$, we can rewrite this equation as
\aln{\label{ergproof1}
\lim_{n\ra\infty}\frac{1}{n}\sum_{s=0}^{n-1}\Tr\lrl{\lrs{\hat{\rho}_{\bm{b};s}-\ml{E}[\hat{\rho}_{\bm{b};s}]}\hat{A}}=0.
}

Next, we replace $\hat{A}$ with 
\aln{
\ml{E}^\dag[\hat{A}]=\sum_b\hat{M}_b^\dag\hat{A}\hat{M}_b,
}
obtaining
\aln{\label{ergproof2}
\lim_{n\ra\infty}\frac{1}{n}\sum_{s=0}^{n-1}\Tr\lrl{\lrs{\ml{E}[\hat{\rho}_{\bm{b};s}]-\ml{E}^2[\hat{\rho}_{\bm{b};s}]}\hat{A}}=0.
}
Summing up Eqs.~\eqref{ergproof1} and~\eqref{ergproof2}, we find
\aln{
\lim_{n\ra\infty}\frac{1}{n}\sum_{s=0}^{n-1}\Tr\lrl{\lrs{\hat{\rho}_{\bm{b};s}-\ml{E}^2[\hat{\rho}_{\bm{b};s}]}\hat{A}}=0.
}
We can repeat this procedure $m$ times and take the average over $m$, obtaining
\aln{
\lim_{n\ra\infty}\frac{1}{n}\sum_{s=0}^{n-1}\Tr\lrl{\lrs{\hat{\rho}_{\bm{b};s}-\frac{1}{M}\sum_{m=0}^{M-1}\ml{E}^m[\hat{\rho}_{\bm{b};s}]}\hat{A}}=0.
}
By taking $M\ra\infty$ limit\footnote{
We assume that $\lim_{n\ra\infty}$ and $\lim_{M\ra\infty}$ commute.
}, we finally obtain
\aln{\label{ergproof3}
\lim_{n\ra\infty}\frac{1}{n}\sum_{s=0}^{n-1}\Tr\lrl{\lrs{\hat{\rho}_{\bm{b};s}-\av{\ml{E}^m}[\hat{\rho}_{\bm{b};s}]}\hat{A}}=0.
}

The next step is to notice
\aln{
\mbb{E}\lrl{\Tr\lrl{\hat{A}\av{\ml{E}^m}[\hat{\rho}_{\bm{b};n+1}]}|b_1,\ldots,b_n}=\Tr\lrl{\hat{A}\av{\ml{E}^m}[\hat{\rho}_{\bm{b};n}]}
}
because Eq.~\eqref{juyou} and $\av{\ml{E}^m}\circ \ml{E}=\av{\ml{E}^m}$ are satisfied.
This means that $\Tr\lrl{\hat{A}\av{\ml{E}^m}[\hat{\rho}_{\bm{b};n}]}$ is also martingale, and we can again employ the martingale convergence theorem.
That is,
\aln{
F_{\bm{b}}[\hat{A}]:=\lim_{n\ra\infty}\Tr\lrl{\hat{A}\av{\ml{E}^m}[\hat{\rho}_{\bm{b};n}]}
}
exists almost surely.
Substituting this into Eq.~\eqref{ergproof3}, we find that $\hat{\rho}_{\bm{b}}^\mr{\infty}=\overline{\hat{\rho}_{\bm{b};n}}$ exists and satisfies
\aln{
\Tr[\hat{\rho}_{\bm{b}}^\infty\hat{A}]=F_{\bm{b}}[\hat{A}].
}
Thus, (1) in Sec.~\ref{KMerg} is proved.

As the final step, we note
\aln{
F_{\bm{b}}[\ml{E}^\dag[\hat{A}]]
=\lim_{n\ra\infty}\Tr\lrl{\hat{A}\ml{E}\circ\av{\ml{E}^m}[\hat{\rho}_{\bm{b};n}]}
=\lim_{n\ra\infty}\Tr\lrl{\hat{A}\av{\ml{E}^m}[\hat{\rho}_{\bm{b};n}]}
=F_{\bm{b}}[\hat{A}]
}
and thus
\aln{
\Tr[\hat{\rho}_{\bm{b}}^\infty\hat{A}]=\Tr[\ml{E}[\hat{\rho}_{\bm{b}}^\infty]\hat{A}],
}
where we have used $\Tr[\hat{A}\ml{E}[\hat{B}]]=\Tr[\ml{E}^\dag[\hat{A}]\hat{B}]$.
Since this holds for arbitrary $\hat{A}$, we can conclude
\aln{\label{ergproof4}
\ml{E}[\hat{\rho}_{\bm{b}}^\infty]=\hat{\rho}_{\bm{b}}^\infty
}
almost surely.
We also find
\aln{
\mbb{E}\lrl{F_{\bm{b}}[\hat{A}]}=\lim_{n\ra\infty}\Tr\lrl{\hat{A}\av{\ml{E}^m}[\ml{E}^n[\hat{\rho}_0]]}=\Tr\lrl{\hat{A}\av{\ml{E}^m}[\hat{\rho}_0]}
}
because $\av{\ml{E}^m}\circ \ml{E}^n=\av{\ml{E}^m}$.
This means that
\aln{
\mbb{E}\lrl{\hat{\rho}_{\bm{b}}^\infty}=\av{\ml{E}^m}[\hat{\rho}_0]
}
almost surely.
Therefore, (2) in Sec.~\ref{KMerg} is proved.
Finally, if there is a unique stationary state $\hat{\rho}_\mr{ss}$ of $\ml{E}$, Eq.~\eqref{ergproof4} ensures
\aln{
\hat{\rho}_{\bm{b}}^\infty=\hat{\rho}_\mr{ss}
}
almost surely, which proves (3).

\section{Outline of the proof of the purification explained in Sec.~\ref{purMK}} 
\label{app:pur}
Here, we outline the proof of the sufficient condition for the purification of typical quantum trajectories presented in Ref.~\cite{maassen2006purification}.
First, we note that the statement in Sec.~\ref{purMK} is justified from the fact that either of the following statements holds (i.e., if (1)' is not satisfied, then (2)' holds):
\begin{enumerate}
\renewcommand{\labelenumi}{(\arabic{enumi})'}
\item 
For all $k\in\mbb{N}$,
\aln{
\lim_{n\ra\infty}P_{\bm{b};n}^{(k)}= 1
}    
almost surely, where
\aln{
P_{\bm{b};n}^{(k)}=\Tr\lrl{\lrs{\hat{\rho}_{\bm{b};n}}^k}.
}

\item
There exists a mixed (i.e., not pure) state $\hat{\sigma}$ such that for all $b$, there exists $\zeta_b\:(\geq 0)$ and
\aln{
\hat{M}_b\hat{\sigma}\hat{M}_b^\dag\sim\zeta_b\hat{\sigma},
}
where $\sim$ denotes the unitary equivalence.
\end{enumerate}
In fact, (1)' clearly indicates that a quantum trajectory purifies almost surely by considering $k=2$.
Moreover, as will be explained in the next paragraph, (2)' indicates $(\star)$ in Sec.~\ref{purMK}.
Therefore, if $(\star)$ does not hold, (2)' does not hold, either, which leads to the fact that (1)' and then the purification of almost all trajectories hold.

Showing that (2)' indicates ($\star$) is a bit complicated: suppose that $\hat{\sigma}$ has a support onto which we can define a projection operator  $\hat{\ml{Q}}$.
Let us denote $\mr{det}_+(\hat{A})$ as the product of all strictly positive eigenvalues of a Hermitian operator $\hat{A}$.
Then, by using the polar decomposition of $\hat{M}_b\hat{\ml{Q}}$ as $\hat{M}_b\hat{\ml{Q}}=\hat{v}_b\sqrt{\hat{\ml{Q}}\hat{M}_b^\dag\hat{M}_b\hat{\ml{Q}}}$ ($\hat{v}_b$ is unitary), we have
\aln{
\mr{det}_+({\zeta}_b\hat{\sigma})&=
\mr{det}_+(\hat{M}_b\hat{\sigma}\hat{M}_b^\dag)
=
\mr{det}_+(\hat{M}_b\hat{\ml{Q}}\hat{\sigma}\hat{\ml{Q}}\hat{M}_b^\dag)
=
\mr{det}_+\lrs{\hat{v}_b\sqrt{\hat{\ml{Q}}\hat{M}_b^\dag\hat{M}_b\hat{\ml{Q}}}\hat{\sigma}\sqrt{\hat{\ml{Q}}\hat{M}_b^\dag\hat{M}_b\hat{\ml{Q}}}\hat{v}_b^\dag}\nonumber\\
&=\mr{det}_+\lrs{\sqrt{\hat{\ml{Q}}\hat{M}_b^\dag\hat{M}_b\hat{\ml{Q}}}\hat{\sigma}\sqrt{\hat{\ml{Q}}\hat{M}_b^\dag\hat{M}_b\hat{\ml{Q}}}}=
\mr{det}_{\ml{Q}}\lrs{\sqrt{\hat{\ml{Q}}\hat{M}_b^\dag\hat{M}_b\hat{\ml{Q}}}\hat{\sigma}\sqrt{\hat{\ml{Q}}\hat{M}_b^\dag\hat{M}_b\hat{\ml{Q}}}}\nonumber\\
&=\mr{det}_{\ml{Q}}\lrs{{\hat{\ml{Q}}\hat{M}_b^\dag\hat{M}_b\hat{\ml{Q}}}}\mr{det}_{\ml{Q}}\lrs{\hat{\sigma}}
}
where we have used the unitary invariance of the spectrum and $\det_\mathcal{Q}$ is the standard determinant for the projected space determined by $\hat{\ml{Q}}$.
Since we also have $\mr{det}_+({\zeta}_b\hat{\sigma})=\mr{det}_\ml{Q}({\zeta}_b\hat{\ml{Q}}\hat{\sigma}\hat{\ml{Q}})=\mr{det}_\ml{Q}({\zeta}_b\hat{\ml{Q}})\mr{det}_\ml{Q}(\hat{\sigma})$ and $\mr{det}_\ml{Q}(\hat{\sigma})>0$, we have 
\aln{\label{detjoken}
\mr{det}_\ml{Q}({\zeta}_b\hat{\ml{Q}})=\mr{det}_{\ml{Q}}\lrs{{\hat{\ml{Q}}\hat{M}_b^\dag\hat{M}_b\hat{\ml{Q}}}}. 
}
This indicates 
\aln{\label{tracejoken}
\Tr[\zeta_b{\hat{\ml{Q}}}]\leq \Tr[\hat{\ml{Q}}\hat{M}_b^\dag\hat{M}_b\hat{\ml{Q}}],
}
where the equality condition is achieved only when $\zeta_b{\hat{\ml{Q}}}=\hat{\ml{Q}}\hat{M}_b^\dag\hat{M}_b\hat{\ml{Q}}$\footnote{
This is known by considering the eigenvalues of $\hat{\ml{Q}}\hat{M}_b^\dag\hat{M}_b\hat{\ml{Q}}$, say $x_1,\cdots, x_q$ with $q=\mr{rank}(\hat{\ml{Q}})$. 
We find that Eq.~\eqref{detjoken} means $\zeta_b^q=x_1\cdots x_q\leq \lrs{\frac{1}{q}\sum_{s=1}^qx_s}^q$, where the equality condition is $x_1=\cdots =x_q$. 
This indicates $q\zeta_b\leq \sum_{s=1}^qx_s$.
}.
In contrast, since\footnote{
Note that $\sum_b\zeta_b=\sum_b\Tr[\zeta_b\hat{\sigma}]=\sum_b\Tr[\hat{M}_b\hat{\sigma}\hat{M}_b^\dag]=1.$
} $\sum_b\Tr[\zeta_b{\hat{\ml{Q}}}]=q=\sum_b\Tr[\hat{\ml{Q}}\hat{M}_b^\dag\hat{M}_b\hat{\ml{Q}}]$ and $\Tr[\zeta_b{\hat{\ml{Q}}}]\geq 0$ for each $b$, we actually find the equality in Eq.~\eqref{tracejoken}.
Therefore, we finally have
\aln{
\hat{\ml{Q}}\hat{M}_b^\dag\hat{M}_b\hat{\ml{Q}}=\zeta_b\hat{\ml{Q}},
}
which is the condition ($\star$).

Now, let us show the statement that either (1)' or (2)' holds.
We first define the function $D_k$:
\aln{
D_k[\hat{\rho}]=\sum_bp_b\lrm{\Tr\lrl{\lrs{\frac{\hat{M}_b\hat{\rho}\hat{M}_b^\dag}{p_b}}^k}-\Tr[\hat{\rho}^k]}^2,
}
where $p_b=\Tr[\hat{M}_b^\dag\hat{M}_b\hat{\rho}]$.
By definition, we have
\aln{
D_k[\hat{\rho}_{\bm{b};n}]=\mbb{E}\lrl{(P_{\bm{b};n+1}^{(k)}-P_{\bm{b};n}^{(k)})^2|b_1,\ldots,b_n}.
}
Now, applying the inequality by Nielsen~\cite{nielsen2001characterizing}, we obtain
\aln{
\mbb{E}\lrl{P_{\bm{b};n+1}^{(k)}|b_1,\ldots,b_n}\geq \mbb{E}\lrl{P_{\bm{b};n}^{(k)}|b_1,\ldots,b_n}=P_{\bm{b};n}^{(k)}
}
with $\mbb{E}[|P_{\bm{b};n}^{(k)}|]\leq 1$, which indicates that $P_{\bm{b};n}^{(k)}$ is a positive submartingale bounded by 1.
Therefore, by the convergence theorem of submartingale, we have 
\aln{
P_{\bm{b};\infty}^{(k)}:=\lim_{n\ra\infty}P_{\bm{b};n}^{(k)}
}
almost surely.
Then, we have\footnote{
Note that $\mbb{E}\lrl{P_{\bm{b};n}^{(k)}\lrs{P_{\bm{b};n+1}^{(k)}-P_{\bm{b};n}^{(k)}}}\geq 0$ because $P_{\bm{b};n}^{(k)}$ is the submartingale.
}
\aln{
\sum_{n=0}^\infty\mbb{E}[D_k[\hat{\rho}_{\bm{b};n}]]
&= \sum_{n=0}^\infty\mbb{E}\lrl{{\lrs{P_{\bm{b};n+1}^{(k)}}^2-\lrs{P_{\bm{b};n}^{(k)}}^2}}-2\mbb{E}\lrl{P_{\bm{b};n}^{(k)}\lrs{P_{\bm{b};n+1}^{(k)}-P_{\bm{b};n}^{(k)}}}\nonumber\\
&\leq\sum_{n=0}^\infty\mbb{E}\lrl{\lrs{P_{\bm{b};n+1}^{(k)}}^2}-\mbb{E}\lrl{\lrs{P_{\bm{b};n}^{(k)}}^2}
= \mbb{E}\lrl{\lrs{P_{\bm{b};\infty}^{(k)}}^2}-\mbb{E}\lrl{\lrs{P_{\bm{b};0}^{(k)}}^2}\leq 1.
}
Therefore, we have
\aln{\label{EDk0}
\lim_{n\ra\infty}\sum_{k=1}^d\mbb{E}\lrl{D_k[\hat{\rho}_{\bm{b};n}]}=0,
}
where $d$ is the dimension of the Hilbert space, which is assumed to be finite.

Now, suppose that (1)' does not hold. Then, for some (and thus all) $k\geq 2$, we find 
\aln{
c_k:=\mbb{E}\lrl{P_{\bm{b};\infty}^{(k)}}<1.
}
Employing the inequality by Nielsen, we have
\aln{\label{c_2p_2}
\mbb{E}[P_{\bm{b};n}^{(2)}]\leq \mbb{E}[P_{\bm{b};\infty}^{(2)}]=c_2<1.
}

Let us define a set $S_n$ of quantum trajectories that satisfies $P_{\bm{b};n}^{(2)}\leq \frac{c_2+1}{2}<1$ and its complement $S_n^c$.
Then, using Eq.~\eqref{c_2p_2}, we have
\aln{
c_2\geq \mbb{E}\lrl{P_{\bm{b};n}^{(2)}}\geq \mbb{E}\lrl{P_{\bm{b};n}^{(2)}\chi_{S_n^c}(\bm{b})}
\geq \frac{c_2+1}{2}\mbb{P}[S_n^c]=\frac{c_2+1}{2}\lrs{1-\mbb{P}[S_n]}
}
for all $n$, where $\chi$ is the indicator function and $\mbb{P}[S]$ denotes the probability where quantum trajectories are in a set $S$.
This leads to 
\aln{\label{psnc2}
\mbb{P}[S_n]\geq \frac{1-c_2}{1+c_2}.
}

Next, we note that we can take a quantum trajectory characterized by a sequence of outcomes $\bm{b}^* \in S_n$ such that it satisfies
\aln{
\mbb{P}[S_n]\sum_{k=1}^d{D_k[\hat{\rho}_{\bm{b}^*;n}]}
\leq\mbb{E}\lrl{\chi_{S_n}(\bm{b})\cdot\sum_{k=1}^d{D_k[\hat{\rho}_{\bm{b};n}]}},
}
since the right-hand side denotes the average over $\bm{b}$.
For this quantum trajectory, we have
\aln{\label{c2ineq}
\sum_{k=1}^d{D_k[\hat{\rho}_{\bm{b}^*;n}]}\leq \frac{1+c_2}{1-c_2}\sum_{k=1}^d{\mbb{E}\lrl{D_k[\hat{\rho}_{\bm{b};n}]}}
}
due to Eq.~\eqref{psnc2} and the fact
\aln{
\mbb{E}\lrl{\chi_{S_n}(\bm{b})\cdot\sum_{k=1}^d{D_k[\hat{\rho}_{\bm{b};n}]}}
\leq\sum_{k=1}^d{\mbb{E}\lrl{D_k[\hat{\rho}_{\bm{b};n}]}}.
}
Since the trajectory is in $S_n$, $\hat{\rho}_{\bm{b}^*;n}$ is in the compact set of the quantum state
\aln{\label{setpuri}
\lrm{\hat{\rho}\: : \:\Tr[\hat{\rho}^2]\leq \frac{1+c_2}{2}}
}
for every $n$.

Now, let us consider $n\ra\infty$. 
From Eqs.~\eqref{EDk0} and~\eqref{c2ineq}, we have
\aln{\label{limndk}
\lim_{n\ra\infty}\sum_{k=1}^d{D_k[\hat{\rho}_{\bm{b}^*;n}]}=0
}
and, from Eq.~\eqref{setpuri}, we have
\aln{
\lim_{n\ra\infty}P_{\bm{b}^*;n}^{(2)}\leq \frac{1+c_2}{2}<1, 
}
which indicates that the state remains a mixed state for the trajectory characterized by $\bm{b}^*$.

Finally, Eq.~\eqref{limndk} means that $\lim_{n\ra\infty}{D_k[\hat{\rho}_{\bm{b}^*;n}]}=0$
for all $k=1,\ldots,d$, and this condition is written as
\aln{
\lim_{n\ra\infty}\Tr\lrl{\hat{M}_b\hat{\rho}_{\bm{b}^*;n}\hat{M}_b^\dag}\lrm{\Tr\lrl{\lrs{\frac{\hat{M}_b\hat{\rho}_{\bm{b}^*;n}\hat{M}_b^\dag}{\Tr\lrl{\hat{M}_b\hat{\rho}_{\bm{b}^*;n}\hat{M}_b^\dag}}}^k}-\Tr\lrl{\lrs{\hat{\rho}_{\bm{b}^*;n}}^k}}^2=0
}
for all $k=1,\ldots,d$ and $b$.
This means that, for $n\ra\infty$, either $\Tr\lrl{\hat{M}_b\hat{\rho}_{\bm{b}^*;n}\hat{M}_b^\dag}=0$ or $\Tr\lrl{\lrs{\frac{\hat{M}_b\hat{\rho}_{\bm{b}^*;n}\hat{M}_b^\dag}{\Tr\lrl{\hat{M}_b\hat{\rho}_{\bm{b}^*;n}\hat{M}_b^\dag}}}^k}-\Tr\lrl{\lrs{\hat{\rho}_{\bm{b}^*;n}}^k}=0\:(k=1,\cdots, d)$ holds. 
The former means that (2)' holds with $\zeta_b=0$.
The latter indicates that $\frac{\hat{M}_b\hat{\rho}_{\bm{b}^*;n}\hat{M}_b^\dag}{\Tr\lrl{\hat{M}_b\hat{\rho}_{\bm{b}^*;n}\hat{M}_b^\dag}}$ and $\hat{\rho}_{\bm{b}^*;n}$ should be unitarily equivalent due to the coincidence of the $k$th moments with $k=1,\ldots, d$, meaning that (2)' is proven.

\section{Birkhoff's ergodic theorem used in Sec.~\ref{sec:time-average_ensemble-average}}
\label{app:Birkhoff-theorem}
We explore the time average of a function $f(\Psi_{1,\bm{b};m})$,
\begin{align}
F_n(\bm{b})=\frac{1}{n}\sum_{m=\tilde{m}}^{\tilde{m}+n-1}f(\Psi_{1,\bm{b};m}),
\label{eq:time-average_finite}
\end{align}
where $\tilde{m}$ is an integer such that $\varepsilon_{1,\bm{b};m}$ is nondegenerate and thus $\ket{\Psi_{1,\bm{b};m}}$ is unique for almost all trajectories and for all $m\geq\tilde{m}$. 
In the following, we focus on situations where $m\geq\tilde{m}$ is satisfied. 
Here, we consider the case where the initial state is the stationary state, $\hat{\rho}_0=\hat{\rho}_\mathrm{ss}$, and thus the probability distribution of $\{b\}$ at each step obeys the invariant measure determined from $\hat{\rho}_\mathrm{ss}$. 
On the basis of the invariant measure, Birkhoff's ergodic theorem states that the three equations below are satisfied in typical trajectories:
\begin{align}
F^\pm(\bm{b})&=F^\pm(\vartheta\bm{b}),
\label{eq:theta-invariance_lim-sup_lim-inf}\\
F^+(\bm{b})&=F^-(\bm{b}),
\label{eq:lim-sup=lim-inf}\\
\mathbb{E}_{\rho_\mathrm{ss}}\left[F^\pm(\bm{b})\right]&=\mathbb{E}_{\rho_\mathrm{ss}}\left[f(\Psi_{1,\bm{b};m})\right],
\label{eq:sample-average}
\end{align}
where $F^\pm(\bm{b})$ are defined as
\begin{align}
F^+(\bm{b})=\limsup_{n\rightarrow\infty}F_n(\bm{b}),\ \ 
F^-(\bm{b})=\liminf_{n\rightarrow\infty}F_n(\bm{b}).
\end{align}
The outline of the proof for Eqs.~\eqref{eq:theta-invariance_lim-sup_lim-inf}-\eqref{eq:sample-average} is explained as follows, on the basis of Ref.~\cite{walters2000introduction}. 
Equations~\eqref{eq:theta-invariance_lim-sup_lim-inf} and~\eqref{eq:lim-sup=lim-inf} mean that we can consider the limit $\lim_{n\rightarrow\infty}F_n(\bm{b})$ invariant under the shift of $\bm{b}$.
If the dynamics is ergodic with respect to the shift of $\bm{b}$, as in Eq.~\eqref{eq:ergodicity_theta}, $\{\vartheta^{-n}\mathfrak{B}\}_n$ covers almost all trajectories for an arbitrary set $\mathfrak{B}$ with $P_{\rho_\mathrm{ss}}(\mathfrak{B})\neq0$, and thus $\lim_{n\rightarrow\infty}F_n(\bm{b})=F^\pm(\bm{b})$ typically becomes independent of $\bm{b}$. 
Then, Eq.~\eqref{eq:sample-average} can be written as
\begin{align}
\lim_{n\rightarrow\infty}F_n(\bm{b})
&=\mathbb{E}_{\rho_\mathrm{ss}}\left[f(\Psi_{1,\bm{b};m})\right],
\label{eq:time-average=sample-average}
\end{align}
which means that the time average of $f(\Psi_{1,\bm{b};m})$ (on the left-hand side) corresponds to the ensemble average of $f(\Psi_{1,\bm{b};m})$ (on the right-hand side). 

We first show inequalities
\begin{align}
\mathbb{E}_{\rho_\mathrm{ss}}^{\mathfrak{C}_\alpha^+ \cap \mathfrak{D}}\left[f(\Psi_{1,\bm{b};m})\right]
\geq\alpha P_\mathrm{\rho_\mathrm{ss}}(\mathfrak{C}_\alpha^+ \cap \mathfrak{D}),\ \ \mathfrak{C}_\alpha^+=\left\{\bm{b} :
 \sup_{n\geq1}F_n(\bm{b})>\alpha\right\},
\label{eq:integral-inequality_plus}\\
\mathbb{E}_{\rho_\mathrm{ss}}^{\mathfrak{C}_\beta^- \cap \mathfrak{D}}\left[f(\Psi_{1,\bm{b};m})\right]
\leq\beta P_{\rho_\mathrm{ss}}(\mathfrak{C}_\beta^- \cap \mathfrak{D}),\ \ \mathfrak{C}_\beta^-=\left\{\bm{b} :
 \inf_{n\geq1}F_n(\bm{b})<\beta\right\},
\label{eq:integral-inequality_minus}
\end{align}
where $\alpha$ and $\beta$ are arbitrary real values. 
Here, $\mathfrak{D}$ is an arbitrary invariant set $\mathfrak{D}=\vartheta^{-1}\mathfrak{D}$ and $\mathbb{E}_{\rho_\mathrm{ss}}^\mathfrak{B}\left[f(\Psi_{1,\bm{b};m})\right]=\sum_{\bm{b} \in \mathfrak{B}}\mathrm{Tr}\left(\hat{\mathsf{M}}_{\bm{b};\infty}\hat{\rho}_\mathrm{ss}\hat{\mathsf{M}}_{\bm{b};\infty}^\dagger\right)f(\Psi_{1,\bm{b};m})$ is the average of $f(\Psi_{1,\bm{b};m})$ over $\mathfrak{B}$. 
Birkhoff's ergodic theorem can be derived from Eqs.~\eqref{eq:integral-inequality_plus} and~\eqref{eq:integral-inequality_minus}. 
To show Eqs.~\eqref{eq:integral-inequality_plus} and~\eqref{eq:integral-inequality_minus}, we evaluate functions
\begin{align}
F_N^{+\alpha}(\bm{b})=\sup_{0 \leq n \leq N}n\left[F_n(\bm{b})-\alpha\right],\ \ 
F_N^{-\beta}(\bm{b})=\inf_{0 \leq n \leq N}n\left[F_n(\bm{b})-\beta\right].
\label{eq:max-min_F}
\end{align}
Here, for $n=0$, we set $n\left[F_n(\bm{b})-\gamma\right]=0$ with $\gamma=\alpha,\beta$. 
On the basis of $F_N^{\pm\gamma}(\bm{b})$, we consider two sets of outcomes $\bm{b}$,
\begin{align}
\mathfrak{F}_N^{+\alpha}=\left\{\bm{b} : F_N^{+\alpha}(\bm{b})>0\right\},\ \ 
\mathfrak{F}_N^{-\beta}=\left\{\bm{b} : F_N^{-\beta}(\bm{b})<0\right\}.
\label{eq:F-N_plus-minus}
\end{align}
If a sequence $\bm{b}$ is included in $\mathfrak{F}_N^{+\alpha}$, \begin{align}
    F_N^{+\alpha}(\vartheta\bm{b}) \geq n\left[F_n(\vartheta\bm{b})-\alpha\right]
    \label{eq:inequality_F_+}
\end{align}
is satisfied for an integer $n$ in a range $0 \leq n \leq N-1$. 
In the same way, if $\bm{b}\in\mathfrak{F}_N^{-\beta}$ and $0 \leq n \leq N-1$ hold,
\begin{align}
    F_N^{-\beta}(\vartheta\bm{b}) \leq n\left[F_n(\vartheta\bm{b})-\beta\right]
    \label{eq:inequality_F_-}
\end{align}
is satisfied. 
Adding $f(\Psi_{1,\bm{b};\tilde{m}})-\alpha$ and $f(\Psi_{1,\bm{b};\tilde{m}})-\beta$ on both sides of Eqs.~\eqref{eq:inequality_F_+} and~\eqref{eq:inequality_F_-}, respectively, we can obtain
\begin{align}
    F_N^{+\alpha}(\vartheta\bm{b})+f(\Psi_{1,\bm{b};\tilde{m}})-\alpha &\geq (n+1)\left[F_{n+1}(\bm{b})-\alpha\right],\ \ 
    \bm{b}\in\mathfrak{F}_N^{+\alpha}\\ 
    F_N^{-\beta}(\vartheta\bm{b})+f(\Psi_{1,\bm{b};\tilde{m}})-\beta &\leq (n+1)\left[F_{n+1}(\bm{b})-\beta\right],\ \ 
    \bm{b}\in\mathfrak{F}_N^{-\beta}.
\end{align}
These lead to
\begin{align}
    f(\Psi_{1,\bm{b};\tilde{m}})-\alpha &\geq F_N^{+\alpha}(\bm{b})-F_N^{+\alpha}(\vartheta\bm{b}),\ \ \bm{b}\in\mathfrak{F}_N^{+\alpha},
    \label{eq:inequality_f-alpha_+}\\ 
    f(\Psi_{1,\bm{b};\tilde{m}})-\beta &\leq F_N^{-\beta}(\bm{b})-F_N^{-\beta}(\vartheta\bm{b}),\ \ \bm{b}\in\mathfrak{F}_N^{-\beta}.
    \label{eq:inequality_f-beta_-}
\end{align}
Here, we focus on sequences $\bm{b}$ included in an invariant set $\mathfrak{D}=\vartheta^{-1}\mathfrak{D}$. 
Then, taking the average of Eq.~\eqref{eq:inequality_f-alpha_+} over $\mathfrak{F}_N^{+\alpha} \cap \mathfrak{D}$, we can obtain 
\begin{align}
    \mathbb{E}_{\rho_\mathrm{ss}}^{\mathfrak{F}_N^{+\alpha} \cap \mathfrak{D}}\left[f(\Psi_{1,\bm{b};\tilde{m}})-\alpha\right]&\geq \mathbb{E}_{\rho_\mathrm{ss}}^{\mathfrak{F}_N^{+\alpha} \cap \mathfrak{D}}\left[F_N^{+\alpha}(\bm{b})-F_N^{+\alpha}(\vartheta\bm{b})\right]\nonumber\\&=\mathbb{E}_{\rho_\mathrm{ss}}^{\mathfrak{F}_N^{+\alpha} \cap \mathfrak{D}}\left[F_N^{+\alpha}(\bm{b})\right]-\mathbb{E}_{\rho_\mathrm{ss}}^{\vartheta^{-1}(\mathfrak{F}_N^{+\alpha} \cap \mathfrak{D})}\left[F_N^{+\alpha}(\vartheta\bm{b})\right]=0,
\end{align}
where $\vartheta^{-1}(\mathfrak{F}_N^{+\alpha} \cap \mathfrak{D})=\mathfrak{F}_N^{+\alpha} \cap \mathfrak{D}$ is used. 
In the same way, taking the average of Eq.~\eqref{eq:inequality_f-beta_-} over $\mathfrak{F}_N^{-\beta} \cap \mathfrak{D}$, we can obtain
\begin{align}
    \mathbb{E}_{\rho_\mathrm{ss}}^{\mathfrak{F}_N^{-\beta} \cap \mathfrak{D}}\left[f(\Psi_{1,\bm{b};\tilde{m}})-\beta\right]&\leq \mathbb{E}_{\rho_\mathrm{ss}}^{\mathfrak{F}_N^{-\beta} \cap \mathfrak{D}}\left[F_N^{-\beta}(\bm{b})-F_N^{-\beta}(\vartheta\bm{b})\right]\nonumber\\&=\mathbb{E}_{\rho_\mathrm{ss}}^{\mathfrak{F}_N^{-\beta} \cap \mathfrak{D}}\left[F_N^{-\beta}(\bm{b})\right]-\mathbb{E}_{\rho_\mathrm{ss}}^{\vartheta^{-1}(\mathfrak{F}_N^{-\beta} \cap \mathfrak{D})}\left[F_N^{-\beta}(\vartheta\bm{b})\right]=0.
\end{align}
These inequalities are equivalent to
\begin{align}
\mathbb{E}_{\rho_\mathrm{ss}}^{\mathfrak{F}_N^{+\alpha} \cap \mathfrak{D}}\left[f(\Psi_{1,\bm{b};m})\right]\geq\alpha P_{\rho_\mathrm{ss}}(\mathfrak{F}_N^{+\alpha} \cap \mathfrak{D}),\ \ 
\mathbb{E}_{\rho_\mathrm{ss}}^{\mathfrak{F}_N^{-\beta} \cap \mathfrak{D}}\left[f(\Psi_{1,\bm{b};m})\right]\leq \beta P_{\rho_\mathrm{ss}}(\mathfrak{F}_N^{-\beta} \cap \mathfrak{D}),
\label{eq:max-min_ergodic-theorem}
\end{align}
where $\tilde{m}$ is replaced by $m$. 
The replacement is valid because $\tilde{m}$ in Eq.~\eqref{eq:time-average_finite} can take arbitrary integers as long as $\varepsilon_{1,\bm{b};m\geq\tilde{m}}$ is nondegenerate almost surely. 
Here, we notice that $\mathfrak{C}_\gamma^\pm$ can be written as
\begin{align}
\mathfrak{C}_\gamma^\pm
=\lim_{N\rightarrow\infty}\mathfrak{F}_N^{\pm\gamma}.
\end{align}
Thus, taking the limit $N\rightarrow\infty$ in Eq.~\eqref{eq:max-min_ergodic-theorem}, we can obtain Eqs.~\eqref{eq:integral-inequality_plus} and~\eqref{eq:integral-inequality_minus}.

To show Eq.~\eqref{eq:theta-invariance_lim-sup_lim-inf}, we consider an identical equation
\begin{align}
\frac{n+1}{n}F_{n+1}(\bm{b})
=F_n(\vartheta\bm{b})+\frac{1}{n}f(\Psi_{1,\bm{b};\tilde{m}}).
\end{align}
Then, taking $\limsup_{n\rightarrow\infty}$ and $\liminf_{n\rightarrow\infty}$, we can obtain Eq.~\eqref{eq:theta-invariance_lim-sup_lim-inf} owing to $\limsup_{n\rightarrow\infty}f(\Psi_{1,\bm{b};\tilde{m}})/n=\liminf_{n\rightarrow\infty}f(\Psi_{1,\bm{b};\tilde{m}})/n=0$.

To show Eq.~\eqref{eq:lim-sup=lim-inf}, we consider a set
\begin{align}
\mathfrak{F}_{\alpha\beta}=\{\bm{b} : F^-(\bm{b})<\beta\ \mathrm{and}\ \alpha<F^+(\bm{b})\}.
\end{align}
From Eq.~\eqref{eq:theta-invariance_lim-sup_lim-inf}, we can understand that $\mathfrak{F}_{\alpha\beta}$ is the invariant set, $\mathfrak{F}_{\alpha\beta}=\vartheta^{-1}\mathfrak{F}_{\alpha\beta}$. 
Therefore, Eq.~\eqref{eq:integral-inequality_plus} leads to
\begin{align}
\mathbb{E}_{\rho_\mathrm{ss}}^{\mathfrak{F}_{\alpha\beta}}\left[f(\Psi_{1,\bm{b};m})\right]
=\mathbb{E}_{\rho_\mathrm{ss}}^{\mathfrak{F}_{\alpha\beta} \cap \mathfrak{C}_\alpha^+}\left[f(\Psi_{1,\bm{b};m})\right]
\geq\alpha P_\mathrm{\rho_\mathrm{ss}}(\mathfrak{F}_{\alpha\beta}).
\end{align}
Here, $\mathfrak{F}_{\alpha\beta} \in \mathfrak{C}_\alpha^+$ is used.
In the same way, Eq.~\eqref{eq:integral-inequality_minus} and $\mathfrak{F}_{\alpha\beta} \in \mathfrak{C}_\beta^-$ result in
\begin{align}
\mathbb{E}_{\rho_\mathrm{ss}}^{\mathfrak{F}_{\alpha\beta}}\left[f(\Psi_{1,\bm{b};m})\right]
=\mathbb{E}_{\rho_\mathrm{ss}}^{\mathfrak{F}_{\alpha\beta} \cap \mathfrak{C}_\beta^-}\left[f(\Psi_{1,\bm{b};m})\right]
\leq\beta P_\mathrm{\rho_\mathrm{ss}}(\mathfrak{F}_{\alpha\beta}).
\end{align}
Thus, $\alpha P_{\rho_\mathrm{ss}}(\mathfrak{F}_{\alpha\beta})\leq\beta P_{\rho_\mathrm{ss}}(\mathfrak{F}_{\alpha\beta})$ holds, which means
\begin{align}
P_\mathrm{\rho_\mathrm{ss}}(\mathfrak{F}_{\alpha\beta})=0
\label{eq:zero-measure}
\end{align}
if $\alpha>\beta$ is satisfied. 
Equation~\eqref{eq:zero-measure} indicates that Eq.~\eqref{eq:lim-sup=lim-inf} is satisfied almost surely.

To show Eq.~\eqref{eq:sample-average}, we consider invariant sets
\begin{align}
\mathfrak{D}_{nk}^+=\left\{\bm{b} : \frac{k}{n} < F^+(\bm{b}) \leq \frac{k+1}{n} \right\},\ \ \mathfrak{D}_{nk}^-=\left\{\bm{b} : \frac{k}{n} \leq F^-(\bm{b}) < \frac{k+1}{n} \right\}.
\end{align}
For arbitrary $\epsilon>0$, 
\begin{align}
\mathbb{E}_{\rho_\mathrm{ss}}^{\mathfrak{D}_{nk}^+}\left[f(\Psi_{1,\bm{b};m})\right]
=\mathbb{E}_{\rho_\mathrm{ss}}^{\mathfrak{D}_{nk}^+ \cap \mathfrak{C}_{\frac{k}{n}-\epsilon}^+}\left[f(\Psi_{1,\bm{b};m})\right]
\geq\left(\frac{k}{n}-\epsilon\right)P_{\rho_\mathrm{ss}}(\mathfrak{D}_{nk}^+)
\end{align}
is satisfied, which can be understood from Eq.~\eqref{eq:integral-inequality_plus}. 
This leads to
\begin{align}
\mathbb{E}_{\rho_\mathrm{ss}}^{\mathfrak{D}_{nk}^+}\left[F^+(\bm{b})\right]\leq\frac{k+1}{n}P_{\rho_\mathrm{ss}}(\mathfrak{D}_{nk}^+)
\leq\frac{1}{n}P_{\rho_\mathrm{ss}}(\mathfrak{D}_{nk}^+)+\mathbb{E}_{\rho_\mathrm{ss}}^{\mathfrak{D}_{nk}^+}\left[f(\Psi_{1,\bm{b};m})\right],
\label{eq:integral-inequality_plus_D}
\end{align}
where the limit $\epsilon\rightarrow0$ is taken. 
Taking the summation over all integers $k$ and taking the limit $n\rightarrow\infty$ in Eq.~\eqref{eq:integral-inequality_plus_D}, we can obtain
\begin{align}
\mathbb{E}_{\rho_\mathrm{ss}}\left[F^+(\bm{b})\right]\leq\mathbb{E}_{\rho_\mathrm{ss}}\left[f(\Psi_{1,\bm{b};m})\right].
\label{eq:average_lower-bound}
\end{align}
In the same way, Eq.~\eqref{eq:integral-inequality_minus} leads to
\begin{align}
\mathbb{E}_{\rho_\mathrm{ss}}^{\mathfrak{D}_{nk}^-}\left[f(\Psi_{1,\bm{b};m})\right]=\mathbb{E}_{\rho_\mathrm{ss}}^{\mathfrak{D}_{nk}^- \cap \mathfrak{C}_{\frac{k+1}{n}+\epsilon}^-}\left[f(\Psi_{1,\bm{b};m})\right]
\leq\left(\frac{k+1}{n}+\epsilon\right)P_{\rho_\mathrm{ss}}(\mathfrak{D}_{nk}^-),
\end{align}
from which we can obtain
\begin{align}
\mathbb{E}_{\rho_\mathrm{ss}}^{\mathfrak{D}_{nk}^-}\left[F^-(\bm{b})\right]\geq\frac{k}{n}P_{\rho_\mathrm{ss}}(\mathfrak{D}_{nk}^-)
\geq-\frac{1}{n}P_{\rho_\mathrm{ss}}(\mathfrak{D}_{nk}^-)+\mathbb{E}_{\rho_\mathrm{ss}}^{\mathfrak{D}_{nk}^-}\left[f(\Psi_{1,\bm{b};m})\right],
\label{eq:integral-inequality_minus_D}
\end{align}
with $\epsilon\rightarrow0$ and thus
\begin{align}
\mathbb{E}_{\rho_\mathrm{ss}}\left[f(\Psi_{1,\bm{b};m})\right]\leq\mathbb{E}_{\rho_\mathrm{ss}}\left[F^-(\bm{b})\right]
\label{eq:average_upper-bound}
\end{align}
is satisfied, where the summation over $k$ and the limit $n\rightarrow\infty$ are taken. 
From Eqs.~\eqref{eq:average_lower-bound}, \eqref{eq:average_upper-bound}, and~\eqref{eq:lim-sup=lim-inf}, we can obtain Eq.~\eqref{eq:sample-average}.

\section{Kingman's subadditive ergodic theorem used in Sec.~\ref{sec:typical-convergence_Lyapunov-spectrum}}
\label{app:proof_Kingman-theorem}
We show Eq.~\eqref{eq:Kingman-theorem} for a sequence of functions $\{f_n(\bm{b})\}_n$ that satisfy the subadditivity
\begin{align}
    f_{n+m}(\bm{b})\leq f_n(\bm{b})+f_m(\vartheta^n\bm{b}).
    \label{eq:subadditive-function_appendix}
\end{align}
The outline of the proof explained here is based on Ref.~\cite{steele1989kingman}. 
In the following, we assume that the initial state is a stationary state $\hat{\rho}_0=\hat{\rho}_\mathrm{ss}=\mathcal{E}(\hat{\rho}_\mathrm{ss})$, which leads to the invariant measure  $P_{\rho_0}(\mathfrak{B})=P_{\rho_0}(\vartheta^{-1}\mathfrak{B})$, where $\mathfrak{B}$ is an arbitrary set of $\bm{b}$. 
We also assume $f_n(\bm{b})\leq0$ for arbitrary $n$ and $\bm{b}$, which is satisfied when we choose the function as in Eq.~\eqref{eq:exterior-power_M}. 
We consider two limits
\begin{align}
    f^+(\bm{b})=\limsup_{n\rightarrow\infty}\frac{f_n(\bm{b})}{n},\ \ 
    f^-(\bm{b})=\liminf_{n\rightarrow\infty}\frac{f_n(\bm{b})}{n}.
\end{align}
If $f^+(\bm{b})=f^-(\bm{b})$ is satisfied, the limit of $f_n(\bm{b})/n$ exists. The limits are typically invariant under the shift,
\begin{align}
    f^\pm(\bm{b})=f^\pm(\vartheta\bm{b}).
    \label{eq:limit_sup-inf_invariant}
\end{align}
Equation~\eqref{eq:limit_sup-inf_invariant} originates from the subadditivity of $f_n(\bm{b})$, as explained below. 
Taking limits $\limsup_{n\rightarrow\infty}$ and $\liminf_{n\rightarrow\infty}$ of an inequality $\frac{n+1}{n}\frac{f_{n+1}(\bm{b})}{n+1} \leq \frac{f_1(\bm{b})+f_n(\vartheta\bm{b})}{n}$, which is obtained from Eq.~\eqref{eq:subadditive-function_appendix}, $f^\pm(\bm{b})\leq f^\pm(\vartheta\bm{b})$ is satisfied. 
This inequality results in $\{\bm{b}:f^\pm(\bm{b})>\alpha\}\in\{\bm{b}:f^\pm(\vartheta\bm{b})>\alpha\}$ where $\alpha$ is a real number. 
Since $\bm{b}\rightarrow\vartheta\bm{b}$ is a measure-preserving transformation, $\{\bm{b}:f^\pm(\bm{b})>\alpha\}$ and $\{\bm{b}:f^\pm(\vartheta\bm{b})>\alpha\}=\vartheta^{-1}\{\bm{b}:f^\pm(\bm{b})>\alpha\}$ are different at most by a measure-zero set, $P_{\rho_\mathrm{ss}}(\left\{\bm{b}:f^\pm(\bm{b})>\alpha\right\})=P_{\rho_\mathrm{ss}}(\left\{\bm{b}:f^\pm(\vartheta\bm{b})>\alpha\right\})$, for arbitrary $\alpha$. 
Therefore, Eq.~\eqref{eq:limit_sup-inf_invariant} is satisfied for typical sequences of $\{\bm{b}\}$. 

\begin{figure}[htbp]
\centering\includegraphics[width=\linewidth]{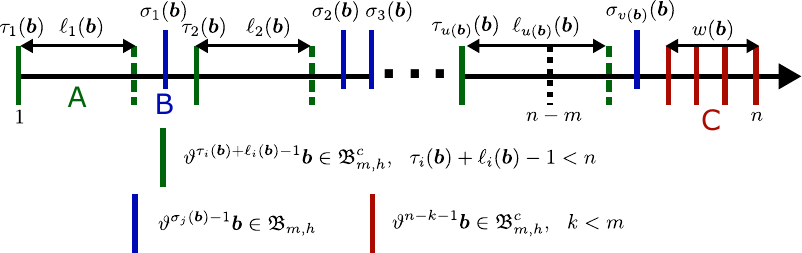}
\caption{How to separate the interval $[1,n]$ for the proof of Kingman's subadditive ergodic theorem.}
\label{fig:separation_kingman-theorem}
\end{figure}

To show the equivalence between $f^+(\bm{b})$ and $f^-(\bm{b})$, we consider a function
\begin{align}
    F_h(\bm{b})=\max\{h,f^-(\bm{b})\},
\end{align}
with $h<0$ being a real number. 
Owing to $f^-(\vartheta\bm{b})=f^-(\bm{b})$, the function is invariant by the shift of $\bm{b}$, $F_h(\vartheta\bm{b})=F_h(\bm{b})$. 
Defining a set of $\bm{b}$,
\begin{align}
    \mathfrak{B}_{m,h}=\left\{\bm{b}:f_\ell(\bm{b})/\ell
    >F_h(\bm{b})+\delta\ \mathrm{for\ all}\ 1 \leq \ell \leq m\right\},
\end{align}
we classify integers $i\in[1,n]$ into three classes A, B, and C, where the real number $\delta$ and the integer $m$ satisfy $\delta>0$ and $m<n$, respectively. 
For the classification, fixing a sequence of measurement outcomes $\bm{b}$, we sweep $i$ from $i=1$ and check whether $\vartheta^{i-1}\bm{b}$ is included in $\mathfrak{B}_{m,h}^c$. 
Here, $\mathfrak{B}_{m,h}^c$ is the complement of $\mathfrak{B}_{m,h}$. If $\vartheta^{i-1}\bm{b}\in \mathfrak{B}_{m,h}^c$ is satisfied, we label the integer $i$ as $\tau(\bm{b})$. 
In this case, there is a positive integer $\ell(\bm{b}) \le m$ which satisfies $f_{\ell(\bm{b})}(\vartheta^{\tau(\bm{b})-1}\bm{b})/\ell(\bm{b}) \leq F_h(\bm{b})+\delta$, where $F_h(\vartheta^{i-1}\bm{b})=F_h(\bm{b})$ is used. 
Then, if $\tau(\bm{b})+\ell(\bm{b})<n$ is satisfied, the interval $[\tau(\bm{b}),\tau(\bm{b})+\ell(\bm{b})-1]$ is classified into A and we continue sweeping from $\tau(\bm{b})+\ell(\bm{b})$. 
We write the number of intervals in the class A as $u(\bm{b})$; there are intervals $\{[\tau_j(\bm{b}),\tau_j(\bm{b})+\ell_j(\bm{b})-1]\}_j$ with $j=1,2,\ldots,u(\bm{b})$. 
If $\vartheta^{i-1}\bm{b}\in \mathfrak{B}_{m,h}$ is satisfied, such $i$ is classified into B and labeled as $\sigma(\bm{b})$. 
We write the number of integers in the class B as $v(\bm{b})$; there are integers $\{\sigma_j(\bm{b})\}_j$ with $j=1,2,\ldots,v(\bm{b})$. 
If $\vartheta^{i-1}\bm{b}\in \mathfrak{B}_{m,h}^c$ is satisfied but $i+\ell-1>n$ holds for the smallest $\ell$ that satisfies $f_\ell(\vartheta^{i-1}\bm{b})/\ell \leq F_h(\bm{b})+\delta$, such $i$ is classified into C. 
We write the number of integers in the class C as $w(\bm{b})$. 
The schematic picture for the classification is shown in Fig.~\ref{fig:separation_kingman-theorem}.
Then, through the subadditivity in Eq.~\eqref{eq:subadditive-function_appendix}, we can obtain an inequality
\begin{align}
    f_n(\bm{b})&\leq\sum_{i=1}^{u(\bm{b})}
    f_{\ell_i(\bm{b})}\left(\vartheta^{\tau_i(\bm{b})-1}\bm{b}\right)
    +\sum_{j=1}^{v(\bm{b})}f_1\left(\vartheta^{\sigma_j(\bm{b})-1}\bm{b}\right)
    +\sum_{k=1}^{w(\bm{b})}f_1(\vartheta^{n-k}\bm{b})\nonumber\\
    &\leq\left[F_h(\bm{b})+\delta\right]\sum_{i=1}^{u(\bm{b})}\ell_i(\bm{b})
    \label{eq:inequality_f}
\end{align}
where $f_1(\bm{b})\leq0$ is used. 
Here, by construction, an inequality 
\begin{align}
    n-m
    \leq\sum_{i=1}^{u(\bm{b})}\ell_i(\bm{b})
    +\sum_{j=1}^n\chi_{\mathfrak{B}_{m,h}}(\vartheta^{j-1}\bm{b})
    \label{eq:inequality_intervals}
\end{align}
is also satisfied, where $\chi_{\mathfrak{B}_{m,h}}(\bm{b})=1$ if $\bm{b} \in \mathfrak{B}_{m,h}$ and $\chi_{\mathfrak{B}_{m,h}}(\bm{b})=0$ if $\bm{b} \in \mathfrak{B}_{m,h}^c$. 
Dividing Eq.~\eqref{eq:inequality_intervals} by $n$ and taking the limit $n\rightarrow\infty$ lead to
\begin{align}
    1\leq
    \lim_{n\rightarrow\infty}\frac{1}{n}\sum_{i=1}^{u(\bm{b})}\ell_i(\bm{b})
    +\lim_{n\rightarrow\infty}\frac{1}{n}\sum_{j=1}^n\chi_{\mathfrak{B}_{m,h}}(\vartheta^{j-1}\bm{b}).
    \label{eq:inequality_ell}
\end{align}
Then, we also take the limit $m\rightarrow\infty$, which leads to $\chi_{\mathfrak{B}_{m,h}}(\bm{b})=0$. 
This means 
\begin{align}
    \lim_{m\rightarrow\infty}\lim_{n\rightarrow\infty}\frac{1}{n}\sum_{i=1}^{u(\bm{b})}\ell_i(\bm{b})=1,
    \label{eq:ell_sum}
\end{align}
since $\frac{1}{n}\sum_{i=1}^{u(\bm{b})}\ell_i(\bm{b})\leq1$ is always satisfied. 
On the basis of Eq.~\eqref{eq:ell_sum}, diving Eq.~\eqref{eq:inequality_f} by $n$ and taking the limit $\lim_{m\rightarrow\infty}\limsup_{n\rightarrow\infty}$, we can obtain 
\begin{align}
    f^+(\bm{b}) \leq F_h(\bm{b})+\delta.
    \label{eq:inequality_f+}
\end{align}
Since $h$ and $\delta$ can take arbitrary small values within the range $h<0$ and $\delta>0$, Eq.~\eqref{eq:inequality_f+} results in
\begin{align}
    f^+(\bm{b}) \leq f^-(\bm{b}),
    \label{eq:inequality_+-}
\end{align}
which means $f^+(\bm{b})=f^-(\bm{b})$. 
When the dynamics is ergodic with respect to the shift $\bm{b}\rightarrow\vartheta\bm{b}$, i.e., Eq.~\eqref{eq:ergodicity_theta} is satisfied, $f^\pm(\bm{b})=f^\pm(\vartheta\bm{b})$ almost surely becomes independent of $\bm{b}$. 
This is because $\{\vartheta^{-n}\mathfrak{B}\}_n$ cover almost all trajectories for an arbitrary set $\mathfrak{B}$ as long as $P_{\rho_\mathrm{ss}}(\mathfrak{B})\neq0$ is satisfied. 
Then, Eq.~\eqref{eq:inequality_+-} leads to the typical convergence of $f_n(\bm{b})/n$ to a $\bm{b}$-independent value $\gamma$,
\begin{align}
    \lim_{n\rightarrow\infty}\frac{f_n(\bm{b})}{n}=\gamma,
\end{align}
which corresponds to Eq.~\eqref{eq:Kingman-theorem} in the main text. 


\bibliographystyle{ptephy}

%

\vspace{0.2cm}
\noindent


\end{document}